\documentclass[
superscriptaddress,
 amsmath,amssymb,
 twocolumn,
prx,
]{revtex4-2}
\usepackage[T1]{fontenc}
\usepackage{graphicx}
\usepackage{dcolumn}
\usepackage{bm}
\usepackage{hyperref}
\usepackage{booktabs}
\usepackage{array}
\usepackage{tabularx}
\newcolumntype{C}{>{\centering\arraybackslash}X}
\usepackage{amsmath,amssymb,amsthm}
\theoremstyle{definition}
\newtheorem{definition}{Definition}[section]
\newtheorem{problem}{Problem}[section]
\newtheorem{theorem}{Theorem}[section]
\newtheorem{corollary}{Corollary}[theorem]

\newtheorem{lemma}[theorem]{Lemma}
\theoremstyle{remark}
\newtheorem*{remark}{Remark}

\newcommand{\argmax}{\mathop{\mathrm{argmax}}}
\newcommand{\CPTP}{\mathrm{CPTP}}
\usepackage{caption}
\usepackage{physics}
\usepackage{sublabel}
\usepackage{subcaption}
\usepackage{xcolor}
\usepackage[normalem]{ulem}
\usepackage[dvipsnames]{xcolor}

\hypersetup{colorlinks = true, linkcolor=blue, citecolor=blue, urlcolor=blue}

\newcommand{\red}[1]{{\color{RubineRed}{#1}}}
\newcommand{\orange}[1]{{\color{YellowOrange}{#1}}}
\usepackage{soul}

\soulregister\cite7
\soulregister\onlinecite7
\soulregister\ref7
\soulregister\eqref7
\soulregister\red7
\soulregister\orange7

\setstcolor{red}

\begin{document}
\let\addcontentslineorig\addcontentsline
\renewcommand{\addcontentsline}[3]{}

\preprint{APS/123-QED}

\title{Quantum channel learning with limited parallel access}

\author{Mahadevan Subramanian}
 \email{mahadevans@uchicago.edu}
 \affiliation{Pritzker School of Molecular Engineering, The University of Chicago, Chicago, Illinois 60637, USA 
}
\author{Hyukgun Kwon}%
 \email{kwon37hg@sejong.ac.kr}
 \affiliation{Department of Physics and Astronomy, Sejong University, 209 Neungdong-ro Gwangjin-gu, Seoul 05006, Republic of Korea}
\author{Liang Jiang}
 \email{liangjiang@uchicago.edu}
\affiliation{%
Pritzker School of Molecular Engineering, The University of Chicago, Chicago, Illinois 60637, USA 
}%

\date{\today}

\begin{abstract}
Quantum channels can be characterized by their action on an orthogonal operator basis, where these operators are related to observable properties of the quantum system. For qudit and multimode bosonic systems, this is encoded respectively in the Heisenberg--Weyl transfer matrix estimated from the Choi-state, and in the characteristic-function transfer map estimated from a two-mode squeezed vacuum based Choi-state. We derive sample-complexity bounds for estimating entries of the transfer matrix/map to additive accuracy $\epsilon$ with success probability $\ge1-\delta$, under different resources: access to the complex-conjugate channel $\mathcal{E}^*$ and/or parallel access to $c$ copies. In all settings, the learner uses parallel channel calls with adaptively chosen, ancilla-assisted input states and measurements. Absolute values of transfer-matrix entries can be learned efficiently with simultaneous access to $\mathcal{E}$ and $\mathcal{E}^*$, with tight scaling $\epsilon^{-4}$. Without conjugate access, any $c<d$ copies are insufficient for efficient learning, requiring sample complexity exponential in the number of ($d$-level) qudits $n$ (for prime $d$). Efficiency is recovered at $c=d$, with tight scaling $\epsilon^{-2d}$. For bosonic systems, exponential sample complexity holds in terms of an effective dimension induced by an energy constraint for all $c=O(1/\epsilon)$. Although the task is learning a particular state, these bounds carry stronger implications than standard state-learning bounds since the learner controls the inputs and has ancillary assistance. This establishes a hierarchy of channel-learning resources: self-complex-conjugate channels require two-copy ancilla-assisted access for efficient learning, while for every square-free $d$, some channels require $d$-copy access. As a corollary, we derive tighter lower bounds for state learning with limited multi-copy access.
\end{abstract}

\maketitle

\section{Introduction}
Characterizing quantum systems is essential for the development of quantum information science and technology. In particular, accurate knowledge of noise channels has a wide range of practical applications, including improving the performance of quantum error-correcting codes by tailoring them to the underlying noise model \cite{PhysRevA.56.2567,PhysRevA.75.012338,PhysRevLett.120.050505}, enhancing decoder performance through noise benchmarking \cite{Bausch2024,tsubouchi2026quantumadvantagessyndromeawarenoisy}, enabling process verification \cite{PRXQuantum.2.010102}, and supporting quantum error mitigation \cite{vandenBerg2023,PhysRevLett.131.210601,PRXQuantum.2.040330}. In this context, quantum learning theory provides a systematic framework for estimating parameters that characterize an unknown quantum state or channel. More specifically, it aims to estimate such parameters within a prescribed \emph{additive error} $\epsilon$ with \emph{success probability} at least $1-\delta$. Within this framework, the central objective is to design efficient estimation protocols that minimize the number of accesses to the quantum state or channel required to achieve this guarantee, a quantity referred to as the \textit{sample complexity}. 

A prominent example of such a sample-complexity advantage arises in quantum state learning. For the estimation of expectation values of Pauli operators, access to two copies of an unknown $n$-qubit quantum state instead of only single-copy access can yield an exponential-in-$n$ reduction of sample complexity \cite{huang2021information}. Generalizing to the learning of $n$-qudit ($d$-level systems), we consider the estimation of expectation values of the generators of the Heisenberg--Weyl group. For this task, access to $d$ copies of the unknown $n$-qudit state can yield an exponential advantage, whereas any protocol restricted to only $c<d$ copies cannot realize this advantage and continues to require exponential sample complexity \cite{ller2025infinitehierarchymulticopyquantum} for any $d$ that is square-free. In a similar vein, simultaneous access to the unknown state $\hat{\rho}$ and its complex conjugate $\hat{\rho}^*$ (in some fixed basis) yields an exponential-in-$n$ advantage over single-copy measurement protocols for the same learning task \cite{PRXQuantum.5.040301}. This principle of advantage with the complex conjugate state also generalizes to the learning of $n$-mode bosonic states \cite{coroi2025exponentialadvantagecontinuousvariablequantum}. These learning tasks are well motivated since in all these cases, the generators of the Heisenberg--Weyl group form a complete operator basis to represent quantum states through their expectation values.

\begin{figure}
    \centering
    \includegraphics[width=\linewidth]{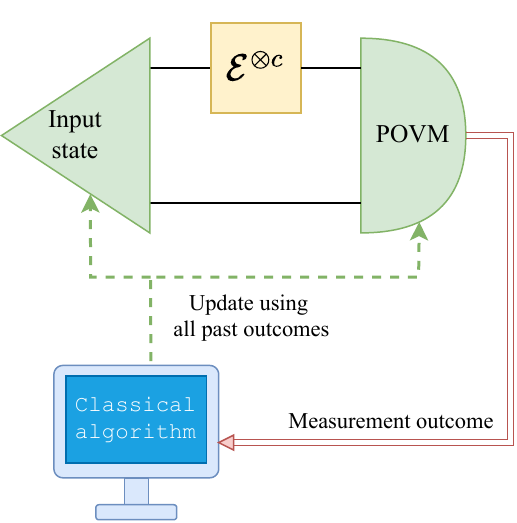}
    \caption{Representation of a learning protocol that uses $c$ copies of an unknown channel $\mathcal{E}$ which is probed by an input state and then measured with the assistance of ancillas. It is assumed that there are multiple steps of such a protocol where the measurement outcomes of all past steps are then used to update the input state and POVM for the next round. The processing is naturally done by a classical algorithm and since we are interested in the fundamental limits of sample complexity of channel usage, we assume unbounded classical memory and accurate control of the input states and POVMs. To study the case of complex-conjugate channels we can alternatively substitute $(\mathcal{E}\otimes\mathcal{E}^*)^{\otimes c}$ in place of the channel.}
    \label{fig:cMprotocolsimple}
\end{figure}

Another important direction in quantum learning theory concerns the characterization of unknown quantum channels. Prominent examples include multi-qubit Pauli channels~\cite{chen_tight_2024,seif2024entanglementenhancedlearningquantumprocesses,PhysRevA.105.032435} and continuous-variable random displacement channels~\cite{PhysRevLett.133.230604,doi:10.1126/science.adv2560}, which constitute practically relevant noise models for discrete and continuous-variable systems, respectively. In both cases, entanglement between the system and a noiseless ancilla can substantially reduce the sample complexity and is necessary for efficient estimation. A more general framework for 
quantum channel learning is provided by the Choi--Jamiołkowski isomorphism, which establishes a one-to-one correspondence between quantum channels and their associated Choi states. Under this correspondence, any linear functional of a channel can be represented as the expectation value of a suitable observable measured on the corresponding Choi state~\cite{PhysRevA.107.042403,wccm-zys6,PhysRevResearch.6.013029,he2026optimalclassicalshadowestimation}. Consequently, a broad class of channel-learning problems can be reformulated as the estimation of observable expectation values on Choi states. For qubit systems, a particularly useful representation in this context is the Pauli transfer matrix, which describes an arbitrary multi-qubit quantum channel in the Pauli operator basis. Each of its matrix elements can be expressed as the expectation value of an appropriate Pauli observable on the channel's Choi state~\cite{10.1145/3670418,chen_tight_2024}. Since multi-qubit Pauli operators form an orthogonal basis for the operator space, the Pauli transfer matrix uniquely specifies the action of the channel on any input operator and thus provides a complete characterization of the channel.

While general quantum-channel learning for qubit systems has been extensively studied in terms of the Pauli transfer matrix \cite{10.1145/3670418}, a comparably general framework for channel learning in qudit and bosonic systems has not yet been fully developed. Since the Heisenberg--Weyl operators form an operator basis for qudit and bosonic systems, analogous to the role of Pauli operators in qubit systems, this structural correspondence motivates us to extend the Pauli transfer matrix description of qubit channels to these more general settings.
Another important observation is that learning the transfer matrix of a quantum channel can be reformulated as learning the Choi state associated with the channel. Efficient general-state learning in qudit and bosonic systems can benefit substantially from access to complex-conjugated copies of the target state, as demonstrated in Refs.~\cite{coroi2025exponentialadvantagecontinuousvariablequantum,PRXQuantum.5.040301}. This naturally suggests considering the complex-conjugate channel as an additional resource for the learning task. For a given quantum channel $\mathcal{E}$, we define its complex-conjugate channel $\mathcal{E}^{*}$ satisfying $\left[\mathcal{E}^{*}\left(\hat{\rho}^{T}\right)\right]^{T} = \mathcal{E}(\hat{\rho})$ for every input state $\hat{\rho}$, where the transpose is taken with respect to a fixed reference basis.

\begin{table*}[ht]
\caption{Summary of lower bounds for the $c$-copy protocols (from Def.~\ref{def:cmprotinform}) that succeed in absolute value estimation for channel learning (Problem~\ref{prob:channel_learn}). In all cases it is assumed that $c=O(1/\epsilon)$. All stated complexities are for total channel uses ($c$ times the number of measurements). We observe an effective dimension of the bosonic mode in the case of the bosonic channel learning task which we label as $d_{\mathrm{in}}$ and $d_{\mathrm{out}}$ and assume that the two-mode squeezing parameter is sufficiently large to ensure $\cosh(2r)\geq 1.06\kappa m$}
\centering
\begin{tabular*}{\textwidth}{@{\extracolsep{\fill}} |l| l| l| l| }
\toprule
\parbox[t]{0.24\textwidth}{Type of Channel} & \parbox[t]{0.28\textwidth}{$c$-copy access to $\mathcal{E}$} & \parbox[t]{0.28\textwidth}{$1$-copy access to $\mathcal{E}$ satisfying                                                                                                                                 $\mathcal{E}\equiv\mathcal{E}^*$} & \parbox[t]{0.12\textwidth}{$c$-copy access to $\mathcal{E}\otimes\mathcal{E}^*$} \\
\midrule
\parbox[t]{0.24\textwidth}{Channels acting on $m$ qubits} & \parbox[t]{0.28\textwidth}{$\Omega(2^{2m}\epsilon^{-2})$ for $c=1$. $O(m\epsilon^{-4})$ for $c = 2$ to learn full transfer matrix \cite{10.1145/3670418}. $\Omega(c^{-3}\epsilon^{-4})$ for $c\geq 2$ and large enough $m$. (Thm.~\ref{thm:qudit2quditinformal})} & \parbox[t]{0.28\textwidth}{$\Omega(2^{2m}\epsilon^{-2})$ (Thm.~\ref{thm:qudit2quditinformalselfconj})} & \parbox[t]{0.12\textwidth}{$\Omega(c^{-3}\epsilon^{-4})$ for large enough $m$. (Thm.~\ref{thm:qudit2quditinformalwithconj})} \\
\midrule
\parbox[t]{0.24\textwidth}{Channels acting on $m$ qudits of local dimension $d\geq 3$ being prime} & \parbox[t]{0.28\textwidth}{$\Omega(d^{2m}c^{-1}\epsilon^{-2})$ for $c\leq d-1$. $\Omega(c(dc^{-1}\epsilon^{-1})^{2d})$ for $c\geq d$ and large enough $m$. (Thm.~\ref{thm:qudit2quditinformal})} & \parbox[t]{0.28\textwidth}{$\Omega(d^{2m}\epsilon^{-2})$ (Thm.~\ref{thm:qudit2quditinformalselfconj})} & \parbox[t]{0.12\textwidth}{$\Omega(c^{-3}\epsilon^{-4})$ for large enough $m$. (Thm.~\ref{thm:qudit2quditinformalwithconj})} \\
\midrule
\parbox[t]{0.24\textwidth}{Channels acting on $m$ bosonic modes, $m\geq 8$, and queries $\alpha,\beta\in \mathbb{C}^m$ ($|\beta|^2,|\alpha|^2\leq \kappa m$), input two-mode squeezed probes of parameter $r$} & \parbox[t]{0.28\textwidth}{$\Omega(d_{\mathrm{in}}^m d_{\mathrm{out}}^mc^{-1}\epsilon^{-2})$, $d_{\mathrm{in}} = \sqrt{1+(0.99\kappa\tanh^2(2r))^2}$, $d_{\mathrm{out}} = \sqrt{1+(0.99\kappa)^2}$ (Thm.~\ref{thm:boson2bosoninformal})} & \parbox[t]{0.28\textwidth}{$\Omega(d_{\mathrm{in}}^m d_{\mathrm{out}}^m\epsilon^{-2})$, $d_{\mathrm{in}} = \sqrt{1+(0.99\kappa\tanh^2(2r))^2}$, $d_{\mathrm{out}} = \sqrt{1+(0.99\kappa)^2}$ (Thm.~\ref{thm:boson2bosoninformalselfconj})} & \parbox[t]{0.12\textwidth}{$\Omega(c^{-3}\epsilon^{-4})$ for large enough $m$. (Thm.~\ref{thm:boson2bosoninformalwithconj})} \\
\bottomrule
\end{tabular*}\label{tab:results}
\end{table*}

This leads to our central question: can the elements of the generalized transfer matrix be estimated sample-efficiently while restricting the number of parallel uses of the unknown quantum channel per measurement, and how does additional access to the complex-conjugate channel affect this task? We consider learning protocols that, in each measurement round, are given parallel access either to $\mathcal{E}^{\otimes c}$ or to $(\mathcal{E}\otimes\mathcal{E}^{*})^{\otimes c}$. The protocols may employ arbitrary ancillary systems and may adapt both the input-state preparation and the measurement strategy across successive rounds; see Fig.~\ref{fig:cMprotocolsimple} for a schematic illustration. Following notation similar to that of Ref.~\cite{chen2024optimaltradeoffsestimatingpauli}, we refer to such schemes as $c$-copy protocols. Thus, depending on the resource model, $c$ denotes either the number of parallel uses of $\mathcal{E}$ or the number of parallel uses of the paired resource $\mathcal{E}\otimes\mathcal{E}^{*}$ available in each measurement round.

In this work, we derive sample-complexity bounds for estimating the absolute values of individual entries of these transfer matrices and functions to additive error $\epsilon$, with success probability at least $1-\delta$, under various operational resource assumptions. Specifically, we consider settings in which one has access to the complex-conjugate channel $\mathcal{E}^{*}$, simultaneous access to $c$ copies of the channel, or both. To show this, we first derive a general sample-complexity lower bound for all $c$-copy learning protocols that solve a broad class of channel-discrimination problems in Lemma~\ref{lem:master}. The bound applies to arbitrary input and output Hilbert spaces. Our result serves as a master lower bound that can be readily specialized to establish the hardness of channel-learning tasks. We then study three operational settings for learning the transfer-matrix representation of an unknown channel $\mathcal{E}$, which we show to be equivalent to learning its Choi state: (1) the learner has simultaneous access to $c<K/\epsilon$ (for an appropriately chosen constant $K$) copies of $\mathcal{E}$, while access to $\mathcal{E}^{*}$ is unavailable; (2) the channel satisfies $\mathcal{E}=\mathcal{E}^{*}$, while the learner is restricted to single-copy access; and (3) the learner has simultaneous $c$-copy access to $\mathcal{E}\otimes\mathcal{E}^{*}$. By applying the general lower bound established in Lemma~\ref{lem:master}, we prove that the first two settings require a number of samples that grows exponentially with the number of qudits (for $c<d$ for local qudit dimension $d$) or, in the bosonic setting, with the number of modes. In sharp contrast, in the third setting, efficient learning becomes possible: we construct an explicit learning protocol and derive a non-exponential upper bound on its sample complexity. We summarize our results in Table~\ref{tab:results}.

The transitions in lower bounds for the multi-qudit (of local dimension $d$) channel learning case once allowed $d$-copy access and similar transitions once access to the complex-conjugate channel is given (regardless of dimension) are both due to commuting observables being possible to efficiently measure, a feature captured in our universal lower bound. This observation is key to the sample complexity separations in the state learning tasks in \cite{coroi2025exponentialadvantagecontinuousvariablequantum,chen2024optimaltradeoffsestimatingpauli,ller2025infinitehierarchymulticopyquantum,PRXQuantum.5.040301}. Although the objective task is a state learning task (of the channel's Choi-state), we consider the learner to have access to the channel, which is much stronger than having access to the Choi-state. Hence the learner can probe the channel with any possible query, making our task inherently different and more encompassing than the state learning tasks studied in \cite{coroi2025exponentialadvantagecontinuousvariablequantum,PRXQuantum.5.040301,ller2025infinitehierarchymulticopyquantum} since state learning is the special case of channel learning when the channel is a complete replacement channel.

To contextualize our lower bounds with regard to the broad literature of channel learning, we note that the $c$-copy restriction is a clear distinction. Ref.~\cite{mele2026optimallearningquantumchannels} assumes a very large amount of parallel uses of the channel since it prepares multiple copies of the Choi-state and feeds those simultaneously into the random purification channel \cite{walter2025randompurificationchannelarbitrary,tang2026conjugatequerieshelp,girardi2026randompurificationchannelsimple}. Ref.~\cite{chen2026quantumchanneltomographyoptimal} similarly examines an unrestricted parallel-access setting for tomographic channel learning and observes a complexity phase transition dependent on the value of the Kraus rank multiplied by the ratio of the output dimension to the input dimension. From the work on bosonic state tomography \cite{mele_learning_2025}, one would expect that even the finite-energy Choi-state of a bosonic channel would be contained, to some degree of accuracy, in a finite-dimensional subspace. However, since the size of the effective finite-dimensional space from \cite{mele_learning_2025} itself depends on the target accuracy, this finite subspace effect does not cause a sharp complexity transition in the regime of $c= O(1/\epsilon)$. Hence our work offers useful bounds for the physically motivated restriction of limited quantum memory with multi-copy access to the channel. Additionally, we use our results to sharpen some of the state learning task lower bounds from \cite{coroi2025exponentialadvantagecontinuousvariablequantum,PRXQuantum.5.040301,ller2025infinitehierarchymulticopyquantum}.

This paper is organized as follows. In Sec.~\ref{sec:choistate_TM} we introduce the transfer matrix formalism we will be using for describing channels and how it relates to the Choi-state of the channel. In Sec.~\ref{sec:cmProt} we introduce our first main result (the master lemma) which offers a lower bound to all $c$-copy parallel access channel learning algorithms that are allowed to be ancilla-assisted and adaptive. In Sec.~\ref{sec:results} we present our results for sample complexity bounds for the channel learning task introduced in Sec.~\ref{sec:choistate_TM} derived using the master lemma introduced in Sec.~\ref{sec:cmProt} for various cases of channel access and complex-conjugate channel access. In Sec.~\ref{sec:learninghierarchy} we discuss the implications of these results in establishing a hierarchy of channel learning tasks based on parallel access to the channel and offer concluding remarks and discussion in Sec.~\ref{sec:conclusion}.

\section{Choi-state and transfer matrix for a quantum channel}\label{sec:choistate_TM}
In this section, we motivate our definition of the channel learning task, which is the learning of the Choi-state. Motivated by the Choi--Jamiołkowski isomorphism, we show how the channel can be fully characterized by learning expectation values of the Choi-state. A mathematically detailed introduction to relevant topics related to operator algebra can be found in Ref.~\cite{hall2013quantum} along with an introduction to the relevant concepts in Appendix~\ref{app:operatoralgebra}. To motivate the task of channel learning, we consider the task of observable estimation linked to a channel inspired by the works of Refs.~\cite{PhysRevA.107.042403,PRXQuantum.4.040337,he2026optimalclassicalshadowestimation}. 

We denote the set of all completely positive trace-preserving (CPTP) maps with an input Hilbert space $\mathcal{H}_{\mathrm{in}}$ and an output Hilbert space $\mathcal{H}_{\mathrm{out}}$ by $\CPTP(\mathcal{H}_{\mathrm{in}},\mathcal{H}_{\mathrm{out}})$ and now consider some $\mathcal{E}\in \CPTP(\mathcal{H}_{\mathrm{in}},\mathcal{H}_{\mathrm{out}})$. We can consider observable estimation of the output state $\mathcal{E}(\hat{\rho})$ for a chosen input state $\hat{\rho}\in D(\mathcal{H}_{\mathrm{in}})$ (here $D(\mathcal{H}_{\mathrm{in}})$ is the space of all valid density operators over Hilbert space $\mathcal{H}_{\mathrm{in}}$). Separately, we consider the larger set of trace-class operators on a Hilbert space $\mathcal{H}_{\mathrm{in}}$ (depicted by $L_1(\mathcal{H}_{\mathrm{in}})$) which is defined by all operators $\hat{\rho}$ that satisfy $\text{Tr}\left[\sqrt{\hat{\rho}^\dagger\hat{\rho}}\right] <\infty$ for the standard definition of the operator trace. Observable estimation now reduces to the estimation of the function 
\begin{equation}\label{eq:channel_function}
    f_{\mathcal{E}}(\hat{O},\hat{\rho}) = \text{Tr}\left[\hat{O}\mathcal{E}(\hat{\rho})\right],
\end{equation}
at various choices of $\hat{O},\hat{\rho}$. For this task to be valid, we require the above function to be bounded in value, or more generally must have a sense of continuity over the space of operators $\hat{O}$ and $\hat{\rho}$. This restricts the observables $\hat{O}$ to bounded operators on the Hilbert space $\mathcal{H}_{\mathrm{out}}$ (depicted by the set $B(\mathcal{H}_{\mathrm{out}})$). We define the operator norm for a bounded operator $\hat{O}\in B(\mathcal{H})$ as
\begin{equation}\label{eq:opnormdef}
    \|\hat{O}\|_{\mathrm{op}} = \max_{\ket{\psi}\in\mathcal{H},\langle\psi|\psi\rangle = 1}\sqrt{\langle\psi|\hat{O}^\dagger\hat{O}|\psi\rangle},
\end{equation}
which is always finite for a bounded operator. Unitary operators and the operators contained in the set defining a finite-outcome positive-operator-valued measure (POVM) will always be bounded operators. Hence we can consider learning the function $f_{\mathcal{E}}:B(\mathcal{H}_{\mathrm{out}})\times L_1(\mathcal{H}_{\mathrm{in}})\to \mathbb{C}$. The observables we will be largely concerned with are those that form the Heisenberg--Weyl group for their respective dimensions which we now proceed to define.

We consider $d$-level systems, referred to as qudits, that have a $d$-dimensional Hilbert space $\mathcal{H}_d$ with an orthonormal basis of $\{\ket
j\}$ for $j = 0,1,\dots ,d-1$. Assume that $d$ is a prime number so that we can define the Galois field $\mathbb{F}_d$ which we will interchangeably use for representing the basis states (this is equivalent to the integer ring $\mathbb{Z}_d = \mathbb{Z}\mod d$). We can then define the standard shift ($\hat{X}_d$) and phase ($\hat{Z}_d$) operations given by
\begin{equation}
\begin{aligned}
\hat{X}_d &= \sum_{j=0}^{d-1}|j+1\mod d\rangle\langle j|,\\ \hat{Z}_d &= \sum_{j=0}^{d-1}\exp(2\pi i \frac{j}{d})|j\rangle\langle j|.
\end{aligned}
\end{equation}
These operators generate the qudit Pauli group, as can be verified by the braiding relation $\hat{X}_d\hat{Z}_d = e^{-2\pi i/d}\hat{Z}_d\hat{X}_d$ \cite{Sarkar2024quditpauligroupnon}. We now consider displacement operations over the qudit phase space defined for $m$ qudit systems of local dimension $d$ giving a joint Hilbert space dimension $d^m$, where $d$ is a prime number and $m$ is a positive integer. Considering the Hilbert space of $\mathcal{H}_{d,m} = (\mathcal{H}_d)^{\otimes m}$, we can define the following displacement operations taking the previously defined $\hat{X}$ and $\hat{Z}$ and $\mathbf{q},\mathbf{p}\in \mathbb{F}^{m}_d$ to be
\begin{align}\label{eq:disp_qudit_defn}
    \hat{D}_{d,m}(\mathbf{q},\mathbf{p}) &= e^{i\pi\mathbf{q}\cdot\mathbf{p}/d}\bigotimes_{i=1}^m\hat{X}_d^{q_i}\hat{Z}_d^{p_i},\\ \hat{D}_{d,m}(\mathbf{q}',\mathbf{p}')\hat{D}_{d,m}(\mathbf{q},\mathbf{p})\nonumber\\ = e^{i\frac{2\pi}{d}(\mathbf{q}\cdot\mathbf{p}' - \mathbf{q}'\cdot\mathbf{p})}&\hat{D}_{d,m}(\mathbf{q},\mathbf{p})\hat{D}_{d,m}(\mathbf{q}',\mathbf{p}'),\\
    \hat{D}_{d,m}(-\mathbf{q},\mathbf{p}) &= \hat{D}_{d,m}^T(\mathbf{q},\mathbf{p}),
\end{align}
where transposition is defined in the standard diagonal basis of $\hat{Z}$. We also define the identity operation $\mathbb{I}_{d,m}$ which can be seen to be equal to $\hat{D}_{d,m}(0,0)$. These operators extend the displacement operations defined in \cite{PRXQuantum.5.040301,ller2025infinitehierarchymulticopyquantum}. An important aspect of these operators is their orthogonality with respect to the Hilbert--Schmidt inner product
\begin{equation}
    \text{Tr}\left[\left(\hat{D}_{d,m}(\mathbf{q}',\mathbf{p}')\right)^\dagger\hat{D}_{d,m}(\mathbf{q},\mathbf{p})\right] = d^m\delta_{\mathbf{q},\mathbf{q}'}\delta_{\mathbf{p},\mathbf{p}'},
\end{equation}
which in turn allows for any linear operator $\hat{\rho}$ acting on $\mathcal{H}$ to be represented as
\begin{equation}
    \hat{\rho} = \frac{1}{d^m}\sum_{\mathbf{q},\mathbf{p}\in\mathbb{F}^m_d}\text{Tr}\left[\hat{D}_{d,m}(\mathbf{q},\mathbf{p})\hat{\rho}\right]\hat{D}^\dagger_{d,m}(\mathbf{q},\mathbf{p}).\label{eq:mquditchar}
\end{equation}
This representation is unique due to the orthogonality of the displacement operations, making the learning of the expectation values $\text{Tr}\left[\hat{D}_{d,m}(\mathbf{q},\mathbf{p})\hat{\rho}\right]$ a meaningful state learning task for state $\hat{\rho}$ \cite{ller2025infinitehierarchymulticopyquantum,PRXQuantum.5.040301}.

\begin{table*}[ht]
\caption{Correspondence between qubit, qudit, and bosonic systems for the basic notions used in this work.}
\centering
\begin{tabular*}{\textwidth}{@{\extracolsep{\fill}} |l| l| l| l| }
\toprule
\parbox[t]{0.20\textwidth}{} & \parbox[t]{0.23\textwidth}{Qubit systems} & \parbox[t]{0.23\textwidth}{Qudit systems} & \parbox[t]{0.23\textwidth}{Bosonic systems} \\
\midrule
\parbox[t]{0.20\textwidth}{Heisenberg--Weyl group generators over $m$ systems} & \parbox[t]{0.23\textwidth}{$m$-qubit Pauli Group} & \parbox[t]{0.23\textwidth}{$\{\hat{D}_{d,m}(\mathbf{q},\mathbf{p})\}$ for $(\mathbf{q},\mathbf{p})\in \mathbb{F}_{d}^m\times\mathbb{F}_d^m$} & \parbox[t]{0.23\textwidth}{$\{\hat{D}(\alpha)\}$ for $\alpha\in \mathbb{C}^m$} \\
\midrule
\parbox[t]{0.20\textwidth}{Bell state} & \parbox[t]{0.23\textwidth}{$\ket{\Phi_2}$ (Eq.~\eqref{eq:bellstatequdit} at $d=2$)} & \parbox[t]{0.23\textwidth}{$\ket{\Phi_d}$ (Eq.~\eqref{eq:bellstatequdit})} & \parbox[t]{0.23\textwidth}{$\ket{\Phi_r^{\mathrm{TMSV}}}$ (Eq.~\eqref{eq:TMSVchardefn})} \\
\midrule
\parbox[t]{0.20\textwidth}{Bell measurement} & \parbox[t]{0.23\textwidth}{Measurement of $\hat{X}_2\otimes\hat{X}_2$ and $\hat{Z}_2\otimes\hat{Z}_2$.} & \parbox[t]{0.23\textwidth}{Measurement of $\hat{X}_d\otimes\hat{X}_d^\dagger$ and $\hat{Z}_d\otimes\hat{Z}_d^\dagger$.} & \parbox[t]{0.23\textwidth}{Measurement of $\hat{q}_1 - \hat{q}_2$ and $\hat{p}_1 + \hat{p}_2$ ($\hat{q}_i$ and $\hat{p}_i$ are the canonical position and momentum operators of the mode $i$).} \\
\midrule
\parbox[t]{0.20\textwidth}{Representing density matrices} & \parbox[t]{0.23\textwidth}{Pauli decomposition (Eq.~\eqref{eq:mquditchar} at $d=2$)} & \parbox[t]{0.23\textwidth}{Displacement decomposition Eq.~\eqref{eq:mquditchar}} & \parbox[t]{0.23\textwidth}{Characteristic function $\chi_{\hat{\rho}}(\alpha)$ (Eq.~\eqref{eq:bosonicchardensity})} \\
\midrule
\parbox[t]{0.20\textwidth}{Representing channels} & \parbox[t]{0.23\textwidth}{Pauli transfer matrix $C_{\mathcal{E}}$ (Eq.~\eqref{eq:cequdit} at $d=2$)} & \parbox[t]{0.23\textwidth}{Displacement transfer matrix $C_{\mathcal{E}^{\mathrm{qudit}}}$ (Eq.~\eqref{eq:cequdit})} & \parbox[t]{0.23\textwidth}{TMSV transfer function $C^{\mathrm{TMSV},r}_{\mathcal{E}^{\mathrm{boson}}}$ (Eq.~\eqref{eq:CETMSVr})} \\
\bottomrule
\end{tabular*}\label{tab:systems_comparison}
\end{table*}

We now consider phase-space displacement operations for $m$ bosonic modes (we represent the Hilbert space by $\mathcal{H}_{\infty,m}$) which are given by (for $\alpha\in \mathbb{C}^m$)
\begin{equation}\label{eq:bosonic_disp}
    \hat{D}(\alpha) = \bigotimes_{i=1}^m e^{\alpha_i\hat{a}_i^\dagger - \alpha_i^*\hat{a}_i}.
\end{equation}
Note that $\hat{D}^T(\alpha) = \hat{D}(-\alpha^*)$ with transposition in the Fock-basis for all modes. We also define the identity operator $\mathbb{I}_{\infty,m}$ which is equal to $\hat{D}(0)$. Defining the symplectic inner product $\Omega:\mathbb{C}^m\times \mathbb{C}^m\to \mathbb{R}$ by 
\begin{equation}
    \Omega(\alpha,\beta) = (\alpha^\dagger \beta - \beta^\dagger \alpha)/i,
\end{equation}
we have the braiding relation
\begin{equation}
    \hat{D}(\alpha')\hat{D}(\alpha) = e^{i\Omega(\alpha,\alpha')}\hat{D}(\alpha)\hat{D}(\alpha').
\end{equation}
There is a similar orthogonality relation that holds for displacements given by
\begin{equation}
    \text{Tr}\left[\hat{D}^\dagger(\beta)\hat{D}(\alpha)\right] = \pi^m\delta^{(2m)}(\alpha - \beta),
\end{equation}
where it is important to note that strictly speaking, the trace cannot be defined due to the operators not being trace class \cite{10415254}. However, we can still define it as a tempered distribution \footnote{A tempered distribution is a continuous linear functional on the space of Schwartz functions. An example is the Dirac delta function.}. Hence we have a unique representation for all trace-class operators $\hat{\rho}$ given by
\begin{equation}
\begin{aligned}\label{eq:bosonicchardensity}
    \hat{\rho} &= \frac{1}{\pi^m}\int d^{2m}\alpha \hat{D}^\dagger(\alpha)\chi_{\hat{\rho}}(\alpha),\\
    \chi_{\hat{\rho}}(\alpha)&=\text{Tr}\left[\hat{\rho}\hat{D}(\alpha)\right],
\end{aligned}
\end{equation}
which motivates the works of \cite{coroi2025exponentialadvantagecontinuousvariablequantum} to define bosonic state learning as the estimation of the function $\chi_{\hat{\rho}}(\alpha)$. We refer the reader to App.~\ref{app:qudit_disp} and App.~\ref{app:bosonic_disp} for further details on the properties of these operators and useful results related to them. To simplify notation, for any Hilbert space we consider a field $X_{\mathcal{H}}$. For the Hilbert space of $m$ qudits of local dimension $d$ we define $X_{\mathcal{H}_{d,m}} = \mathbb{F}_d^m\times \mathbb{F}_d^m$ and for the case of $m$ bosonic modes we define $X_{\mathcal{H}_{\infty,m}} = \mathbb{C}^m$. For such a field we can define an appropriate symplectic inner product $\Omega_{\mathcal{H}}:X_{\mathcal{H}}\times X_{\mathcal{H}}\to \mathbb{R}$ which then defines the Heisenberg--Weyl algebra \footnote{This is a $C^*$-algebra $W$ defined for symplectic space $(X,\sigma)$ such that the generators of this algebra satisfy, for all $f,g\in X$, $W(f)^\dagger = W(-f)$ and $W(f)W(g) = W(f+g)e^{i\sigma(g,f)/2}$. More information on this can be found in \cite{moretti2017spectral}.} with the generators matching the displacement operators defined in Eqs.~\eqref{eq:disp_qudit_defn} and~\eqref{eq:bosonic_disp}. We provide a comparison between the analogous notations as defined over qubit, qudit and bosonic systems in Table~\ref{tab:systems_comparison}, and collect the symbols used throughout in Table~\ref{app:notation}.

We now define our channel learning problem. Based on the estimation of the function $f_{\mathcal{E}}$ defined in Eq.~\eqref{eq:channel_function}, we can consider a few characteristic functions depending on the input and output Hilbert spaces of the channel. For a channel $\mathcal{E}^{\mathrm{qudit}}\in \CPTP(\mathcal{H}_{d,m},\mathcal{H}_{d',m'})$, we consider the function $C_{\mathcal{E}^\mathrm{qudit}}:\mathbb{F}_{d}^{m}\times \mathbb{F}_{d}^{m}\times \mathbb{F}_{d'}^{m'}\times \mathbb{F}_{d'}^{m'}\to \mathbb{C}$ defined by
\begin{equation}
\begin{aligned}
    &C_{\mathcal{E}^\mathrm{qudit}}((\mathbf{q}_i,\mathbf{p}_i),(\mathbf{q}_o,\mathbf{p}_o)) \\&= \text{Tr}\left[\hat{D}_{d',m'}(\mathbf{q}_o,\mathbf{p}_o)\mathcal{E}^\mathrm{qudit}\left(\hat{D}^\dagger_{d,m}(\mathbf{q}_i,\mathbf{p}_i)/d^m\right)\right],
\end{aligned}
\end{equation}
which represents a transfer matrix definition of the channel that generalizes the transfer matrix definition from qubits \cite{10.1145/3670418} to qudits. Taking the 2-qudit Bell pair $\ket{\Phi_d}_{ab} = \frac{1}{\sqrt{d}}\sum_{i\in\mathbb{F}_d}\ket{i}_a\ket{i}_b$, we can observe that the above function is equivalent to 
\begin{equation}
\begin{aligned}\label{eq:cequdit}
    &C_{\mathcal{E}^\mathrm{qudit}}((\mathbf{q}_i,\mathbf{p}_i),(\mathbf{q}_o,\mathbf{p}_o)) \\
    &= \text{Tr}\Bigg[\left(\hat{D}_{d',m'}(\mathbf{q}_o,\mathbf{p}_o)_A\otimes \hat{D}^*_{d,m}(\mathbf{q}_i,\mathbf{p}_i)_B\right)\\&\quad\quad\quad\times(\mathcal{E}^{\mathrm{qudit}}_{A}\otimes\mathbb{I}_B)(|\Phi_{d}\rangle\langle\Phi_{d}|_{AB}^{\otimes m})\Bigg],
\end{aligned}
\end{equation}
where the conjugation of $\hat{D}_{d,m}^*(\mathbf{q}_i,\mathbf{p}_i)_B = \hat{D}_{d,m}(\mathbf{q}_i,-\mathbf{p}_i)_B$ is defined in the diagonal basis of $\hat{Z}_d$ for each qudit. This shows that one succeeds in the channel learning task through learning the Choi-state of the channel that is defined by $(\mathcal{E}^{\mathrm{qudit}}_{A}\otimes\mathbb{I}_B)(|\Phi_{d}\rangle\langle\Phi_{d}|_{AB}^{\otimes m})$. Due to the nature of all linear operators having a unique representation in terms of the expectation values of the qudit displacement operations (see Eq.~\eqref{eq:mquditchar}), the function in Eq.~\eqref{eq:cequdit} is indeed a complete description of the quantum channel. To formally justify this claim, we can show that the value $f_{\mathcal{E}^{\mathrm{qudit}}}(\hat{O},\hat{\rho})$ can be estimated for any operators $\hat{O}\in B(\mathcal{H}_{d',m'})$ and $\hat{\rho}\in D(\mathcal{H}_{d,m})$ assuming each entry of $C_{\mathcal{E}^\mathrm{qudit}}((\mathbf{q}_i,\mathbf{p}_i),(\mathbf{q}_o,\mathbf{p}_o))$ is known to sufficient accuracy.

While one may be tempted to extend this principle for bosonic channel learning, one will be immediately met by the fact that a similar transfer function for $\mathcal{E}^{\mathrm{boson}}\in \CPTP(\mathcal{H}_{\infty,m},\mathcal{H}_{\infty,m'})$ given by $\Lambda_{\mathcal{E}^{\mathrm{boson}}}(\alpha,\beta) = \text{Tr}\left[\hat{D}(\beta)\mathcal{E}(\hat{D}^\dagger(\alpha))\right]$ is not a bounded function and has very tempered distribution like behaviors (see App.~\ref{app:bosonicnogo}). To avert this, we instead choose to consider the function $C^{\mathrm{TMSV},r}_{\mathcal{E}^{\mathrm{boson}}}:\mathbb{C}^m\times\mathbb{C}^{m'}\to\mathbb{C}$ defined using the Choi-state obtained through acting the channel on $\ket{\Phi^{\mathrm{TMSV}}_r}_{ab} = \mathrm{sech}(r)\sum_{i=0}^{\infty}(\tanh(r))^i\ket{i}_a\ket{i}_b$ given by
\begin{equation}
\begin{aligned}
    &C_{\mathcal{E}^{\mathrm{boson}}}^r(\alpha,\beta) \\
    &=\text{Tr}\Big[(\hat{D}(\beta)_A\otimes\hat{D}(-\alpha^*)_B)\\&\quad\quad\times(\mathcal{E}^{\mathrm{boson}}_A\otimes \mathbb{I}_{B})(|\Phi_r^{\mathrm{TMSV}}\rangle\langle\Phi_r^{\mathrm{TMSV}}|_{AB}^{\otimes m})\Big],
\end{aligned}
\end{equation}
which is equivalent to
\begin{equation}\label{eq:CETMSVr}
\begin{aligned}
    &C_{\mathcal{E}^{\mathrm{boson}}}^{\mathrm{TMSV},r}(\alpha,\beta) = e^{-\frac{|\alpha|^2}{2\cosh(2r)}}\text{Tr}\left[\hat{D}(\beta)\mathcal{E}\left(\hat{D}^\dagger(\alpha \tanh(2r))\hat{\rho}_{\mathrm{th},\alpha}\right)\right],\\
    &\hat{\rho}_{\mathrm{th},\alpha}= \frac{\hat{D} (\alpha\tanh(2r)/2)\tanh(r)^{2\sum\hat{n}_i}\hat{D}^\dagger(\alpha\tanh(2r)/2)}{\cosh^{2m}(r)},
\end{aligned}
\end{equation}
where $r>0$.

We require to assume that the squeezing $r$ is finite to ensure that the actual learning task is physical. However this restriction still allows the function $C^{\mathrm{TMSV},r}_{\mathcal{E}^{\mathrm{boson}}}$ to contain a complete description of the channel $\mathcal{E}^{\mathrm{boson}}$ albeit in a weaker sense. Take $\hat{O}\in B(\mathcal{H}_{\infty,m})$ and $\hat{\rho}\in L_1(\mathcal{H}_{\infty,m})$ and assume we have exact (or to sufficient accuracy) knowledge of the characteristic functions of $\hat{O}$ and $\hat{\rho}$. Using this, we can always find a finite $r_0$ such that for $r>r_0$, accurate knowledge of the function $C_{\mathcal{E}^{\mathrm{boson}}}^{r}$ can be used to estimate the value of $f_{\mathcal{E}^{\mathrm{boson}}}(\hat{O},\hat{\rho})$ where $r_0$ depends on how quickly the Wigner function of $\hat{\rho}$ decays at far away points in the bosonic phase space. This is merely a restatement of the Choi--Jamiołkowski isomorphism which also holds for infinite-dimensional systems \cite{serafini2023quantum}. The explicit mathematical relation to $f_{\mathcal{E}^{\mathrm{boson}}}(\hat{D}(\beta),\hat{\rho})$ is the characteristic function of $\mathcal{E}^{\mathrm{boson}}(\hat{\rho})$ which can be recovered from the relation we have in Eq.~\eqref{eq:char_by_CTMSV} which we discuss in Appendix~\ref{app:inadequacy}.

Using this, we finally define our channel learning task as follows
\begin{problem}[Absolute value estimation for channel learning]\label{prob:channel_learn}
Consider input and output Hilbert spaces of $\mathcal{H}_{\mathrm{in}}$ and $\mathcal{H}_{\mathrm{out}}$ with each Hilbert space having a representative set of displacement operators parameterized by a symplectic space $(X_{\mathcal{H}_{\mathrm{in}}},\Omega_{\mathcal{H}_{\mathrm{in}}})$ and $(X_{\mathcal{H}_{\mathrm{out}}},\Omega_{\mathcal{H}_{\mathrm{out}}})$ respectively. The channel learning task is the estimation of the absolute value of the function $C_{\mathcal{E}}:X_{\mathcal{H}_{\mathrm{in}}}\times X_{\mathcal{H}_{\mathrm{out}}}\to \mathbb{C}$. Hence with a success probability of $1-\delta$, the deviation of the estimate to $|C_{\mathcal{E}}(Q)|$ is no larger than $\epsilon$ for all queries $Q\in \mathcal{Q}$ which is a bounded sized subset $\mathcal{Q} \subseteq X_{\mathcal{H}_{\mathrm{in}}}\times X_{\mathcal{H}_{\mathrm{out}}}$. The functions are defined as follows
\begin{itemize}
    \item For $\mathcal{E}\in \CPTP(\mathcal{H}_{d,m},\mathcal{H}_{d',m'})$, 
    \begin{align*}
    &C_{\mathcal{E}}((\mathbf{q}_i,\mathbf{p}_i),(\mathbf{q}_o,\mathbf{p}_o))\\&= \text{Tr}\left[\hat{D}_{d',m'}(\mathbf{q}_o,\mathbf{p}_o)\mathcal{E}(\hat{D}^\dagger_{d,m}(\mathbf{q}_i,\mathbf{p}_i)/d^m)\right],
    \end{align*}
     for queries $(\mathbf{q}_i,\mathbf{p}_i)\in\mathbb{F}^{m}_{d}\times\mathbb{F}^{m}_{d}$ and $(\mathbf{q}_o,\mathbf{p}_o)\in \mathbb{F}^{m'}_{d'}\times\mathbb{F}^{m'}_{d'}$.
    \item For $\mathcal{E}\in \CPTP(\mathcal{H}_{d,m},\mathcal{H}_{\infty,m'})$, 
    \begin{align*}
        &C_{\mathcal{E}}((\mathbf{q},\mathbf{p}),\beta)\\ &= \text{Tr}\left[\hat{D}(\beta)\mathcal{E}(\hat{D}_{d,m}^\dagger(\mathbf{q},\mathbf{p})/d^m)\right],
    \end{align*}
    for queries $(\mathbf{q},\mathbf{p})\in \mathbb{F}^{m}_d\times\mathbb{F}^{m}_d$ and $\beta\in \mathbb{C}^{m'}$ with $|\beta|^2\leq \kappa m'$ ($\kappa>0$).
    \item For $\mathcal{E}\in \CPTP(\mathcal{H}_{\infty,m},\mathcal{H}_{\infty,m'})$ taking some choice $r>0$, 
    \begin{align*}
    &C_{\mathcal{E}}(\alpha,\beta) = C_{\mathcal{E}}^{\mathrm{TMSV},r}(\alpha,\beta) \\&= \text{Tr}\left[\hat{D}(\beta)\mathcal{E}(\hat{D}^\dagger(\alpha\tanh(2r))\hat{\rho}_{\mathrm{th},\alpha})\right],
    \end{align*}
    ($\hat{\rho}_{\mathrm{th},\alpha}$ defined in Eq.~\eqref{eq:CETMSVr}) for queries $\alpha\in\mathbb{C}^{m}$, $\beta\in \mathbb{C}^{m'}$ with $|\alpha|^2\leq \kappa m$ and $|\beta|^2\leq \kappa'm'$ ($\kappa,\kappa'>0$).
\end{itemize}
It is assumed that the queries are only revealed after the learner has used up all the copies of the unknown quantum channel made available to them.
\end{problem}

We only consider hardness for estimation of the absolute value for the problem at hand.
To the best of our knowledge, in the bosonic case, for the limited multi-copy-access setting, there is currently no way to assess the complete phase information of $C_{\mathcal{E}}$ with significant success probability without the learner having prior knowledge of the actual queries. The absolute value estimation for channel learning task (Problem~\ref{prob:channel_learn}) succeeds through the learning of the channel's Choi state by estimation of displacement observables on the state. This exact problem has been analyzed in \cite{PRXQuantum.5.040301,coroi2025exponentialadvantagecontinuousvariablequantum} where access to the conjugate of the state allows estimation up to a sign. Following this, Ref.~\cite{PRXQuantum.5.040301} prepares a hypothesis state using an efficient number of single-copy measurements on the state, relying on the fact that the set of all possible queries is bounded for finite-dimensional states. In the case that no conjugate state is available, using the methods proposed in \cite{ller2025infinitehierarchymulticopyquantum}, the estimation of displacement observables for an $m$-qudit state can be done up to a phase that is a power of $e^{2\pi i/d}$. In Sec. 5.2 of \cite{ller2025infinitehierarchymulticopyquantum}, a way to extend this to also learning the phase information is shown using methods inspired by \cite{doi:10.1137/1.9781611978322.27,11250187}. All these methods do not require knowledge of the queries since the set of possible queries itself is bounded due to the finite-dimensional nature.

This becomes an issue when working with multi-mode bosonic states. While Ref.~\cite{coroi2025exponentialadvantagecontinuousvariablequantum} shows that the characteristic function $\chi_{\hat{\rho}}$ for any state $\hat{\rho}\in D(\mathcal{H}_{\mathrm{in}})$ can be estimated up to a sign, sample-efficiently and with minimal quantum memory, given access to $\hat{\rho}^*$, gaining knowledge of the sign still requires knowledge of the queries using the method in \cite{PhysRevResearch.6.033280} and there is not at the moment a known way to construct a hypothesis state similar to \cite{PRXQuantum.5.040301}. We note that from the continuous-variable tomography results in \cite{mele_learning_2025}, there is a way to even deduce the sign albeit with a much larger amount of samples. We leave the exploration of the lower bound for sign deduction to future work.

\section{General lower bound for all $c$-copy protocols}\label{sec:cmProt}
In this section we will state the main result of our work, which establishes lower bounds for all learning protocols that succeed in a hypothesis testing task involving channel discrimination. We choose the discrimination task in a specific way to ensure that succeeding in the absolute value estimation for channel learning task in Problem~\ref{prob:channel_learn} is sufficient to succeed in the discrimination task. Hence any lower bound derived for protocols that succeed in this discrimination task is also a lower bound for protocols that succeed in Problem~\ref{prob:channel_learn}. 
We first formalize the learning protocols in the following definition.
\begin{figure*}
    \centering
    \includegraphics[width=\textwidth]{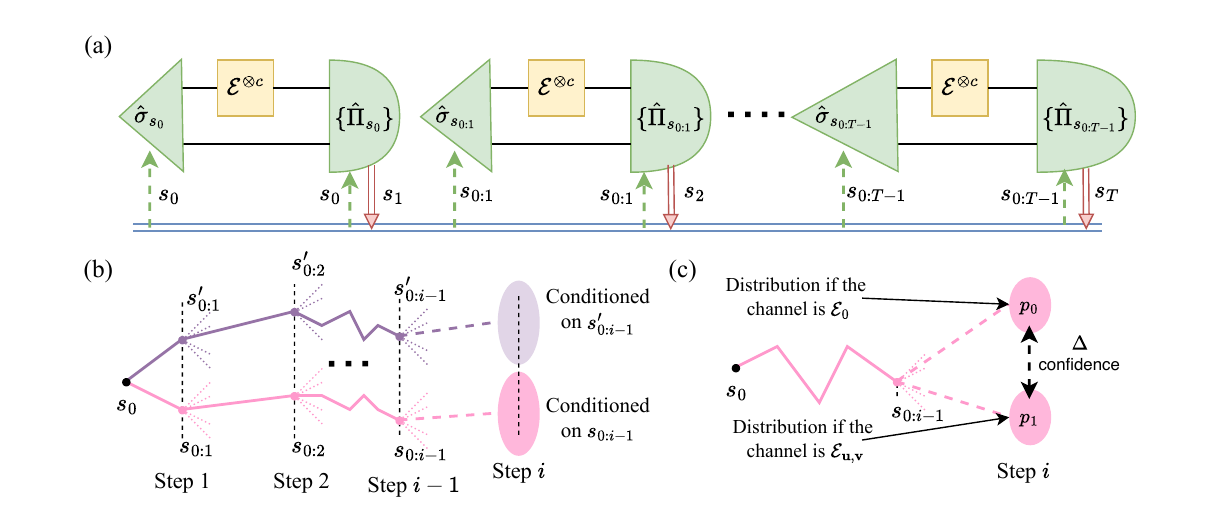}
    \caption{(a) Depiction of $c$-copy protocols as defined in Def.~\ref{def:cmprotinform}. Each learning step allows for ancilla assistance and adaptive state preparation and measurements. Each measurement $\{\hat{\Pi}_{s_{0:t}}\}$ is assumed to be an arbitrary POVM. The learning protocol is initialized with the classical memory with value $s_0$ which is appended with each measurement outcome to finally get the vector $s_{0:T}$ at the end of the learning protocol. This classical information is then processed to construct an appropriate description of the channel. (b) An example of two paths $s_{0:i}$ and $s_{0:i}'$ of the measurement outcomes where it can be noted that outcome $s_i$ ($s'_i$) is distributed as per a measure that is dependent on the previous outcomes $s_{0:i-1}$ ($s'_{0:i-1}$). This can be interpreted as a path being chosen at each step amongst the many possible paths (corresponding to getting a different measurement outcome at that point) from a tree-like structure. (c) When trying to discriminate the two channels $\mathcal{E}_0$ and $\mathcal{E}_{\mathbf{u},\mathbf{v}}$ (as described in discrimination task~\ref{prob:manyrevel}), at any given step one can quantify the confidence gained by quantifying the chi-squared divergence ($\Delta$)
    between the two possible probability distributions for the next outcome $s_i$ conditioned on previous outcomes $s_{0:i-1}$.}
    
    \label{fig:cMprotocol}
\end{figure*}
\begin{definition}[$c$-copy learning protocol]\label{def:cmprotinform}
    Consider an unknown quantum channel $\mathcal{E}\in \CPTP(\mathcal{H}_{\mathrm{in}},\mathcal{H}_{\mathrm{out}})$ and an ancillary space $\mathcal{H}_{\mathrm{anc}}$. A $c$-copy learning protocol consists of $T$ learning steps where at each step the learner prepares an input state to the channel $(\mathcal{E}^{\otimes c}\otimes \mathbb{I}_{\mathrm{anc}})$ and performs a POVM (allowed to be a general measurement) on the output state where the measurement outcomes are assumed to be contained in a measurable space $S$. At each step, the learner is allowed to use the outcomes of the previous steps to adaptively prepare the input state and choice of measurement. The learner finally uses all the classical information obtained from the measurements to perform the relevant channel learning task.
\end{definition}
A diagrammatic representation of such learning protocols is described in Fig.~\ref{fig:cMprotocol}(a). Consider that a certain number of steps of the learning protocol have been completed. At this stage, it can be seen that there are multiple possible measurement outcomes for the following step which are distributed according to a distribution that is conditioned on the previous outcomes. In this picture, the final state of the classical memory after $T$ steps can be understood as a path (or vector) of length $T$ with each entry of this path being the outcome of the measurement at that step. The path like interpretation follows from the fact that the $i$th step of the learning protocol is conditioned on the outcomes of the previous $i-1$ steps. We describe in full detail the probability distribution associated with the final state of the classical system after $T$ steps in Definition~\ref{def:path_rep_cM}.

We assume all the measurement outcomes are contained in the same space which we can do without loss of generality (see discussion in App.~\ref{app:measure_theory}). We make use of the following notation for a vector of length $T+1$ as $s_{0:T}$ where each $s_i\in S$ for $i = 0$ to $T$. We set $s_0$ to a trivial value to represent the state of the classical memory at the very start and $s_i$ for $i\geq 1$ is the outcome of the measurement at the $i$th step of the learning protocol. The notation $s_{i:j}$ is the sub-vector from the $i$th to $j$th entry including both $s_i$ and $s_j$. Using this we can see that at the $i$th learning step, the chosen input state and measurement are both functions of $s_{0:i-1}$. A diagrammatic representation of how two different paths evolve with each learning step is described in Fig.~\ref{fig:cMprotocol}(b). 

To establish lower bounds for the absolute value estimation for channel learning task of Problem~\ref{prob:channel_learn} we come up with a hypothesis testing task involving channel discrimination. The key motivation behind this definition is to describe a family of channels that have a sparse representation when expressed using the transfer matrix formalism of Problem~\ref{prob:channel_learn}. Hence learning the transfer matrix can help in discriminating such a channel from a complete replacement channel.
\begin{problem}[Many-one discrimination with revelation]\label{prob:manyrevel}
    Taking $x\in\{0,1,\dots l\}$ ($l\geq 1$), consider a set of operators $\{\hat{K}_{x,\mathbf{u}}\}$ and $\{\hat{W}_{x,\mathbf{v}}\}$ parameterized by $\mathbf{u}$ and $\mathbf{v}$ such that for all possible values of $x,\mathbf{u},\mathbf{v}$ the operators $\hat{K}_{x,\mathbf{u}}\in B(\mathcal{H}_{\mathrm{in}})$ and $\hat{W}_{x,\mathbf{v}}\in B(\mathcal{H}_{\mathrm{out}})$ and additionally the operators $\hat{K}_{0,\mathbf{u}} = \mathbb{I}_{\mathrm{in}}$ and $\hat{W}_{0,\mathbf{v}} = \mathbb{I}_{\mathrm{out}}$ which are the identity operators on the respective Hilbert spaces. Additionally, for $\hat{\rho}_0\in D(\mathcal{H}_{\mathrm{out}})$, the linear map \begin{equation}\label{eq:Epar1par2}\mathcal{E}_{\mathbf{u},\mathbf{v}}(\cdot) = \sum_{x\in\{0,\dots,l\}}\text{Tr}\left[(\cdot)\hat{K}_{x,\mathbf{u}}\right]\hat{\rho}_0^{1/2}\hat{W}_{x,\mathbf{v}}\hat{\rho}_0^{1/2},\end{equation} is always a valid quantum channel in $\CPTP(\mathcal{H}_{\mathrm{in}},\mathcal{H}_{\mathrm{out}})$. Consider the following problem setting where $A$ first picks random samples $\mathbf{u}$ and $\mathbf{v}$ according to distributions $p(\mathbf{u})$ and $p(\mathbf{v})$, respectively. With equal probability, $A$ now prepares as many copies as requested of one of the two channels $\mathcal{E}_0(\cdot) = \text{Tr}\left[(\cdot)\right]\hat{\rho}_0$ or $\mathcal{E}_{\mathbf{u},\mathbf{v}}$ and provides them to $B$ in batches of size $c$. Once $B$ has used up all the copies of the channel in their learning protocol, $A$ reveals the value of $\mathbf{u}$ and $\mathbf{v}$ following which $B$ has to guess whether $A$ had prepared copies of $\mathcal{E}_0$ or $\mathcal{E}_{\mathbf{u},\mathbf{v}}$.
\end{problem}
The form of the channel $\mathcal{E}_{\mathbf{u},\mathbf{v}}$ in Eq.~\eqref{eq:Epar1par2} can be used to describe any channel using a sufficiently large $l$ for finite-dimensional channels with this principle also being extendable to infinite-dimensional channels assuming an integral instead of a sum. We will only be exploring the regime of small $l$ in our analysis which means we have a low-rank representation of our channels in this form. Further, we find that the examples for the lower bounds turn out to be entanglement-breaking channels. The channel $\mathcal{E}_{\mathbf{u},\mathbf{v}}$ maps an operator with large overlap (in terms of the Hilbert--Schmidt inner product) with $\hat{K}_{x,\mathbf{u}}$ to a correspondingly large contribution of $\hat{\rho}_0^{1/2}\hat{W}_{x,\mathbf{v}}\hat{\rho}_0^{1/2}$ in the output. Consider that the operators $\hat{K}_{x,\mathbf{u}}$ and $\hat{W}_{x,\mathbf{v}}$ are each a sum of few displacements parameterized by $\mathbf{u}$ and $\mathbf{v}$. In this case, one can tell apart the hypotheses from queries dependent on $\mathbf{u},\mathbf{v}$ to transfer function $C_{\mathcal{E}}$ for $\mathcal{E}$ being one of $\mathcal{E}_0$ or $\mathcal{E}_{\mathbf{u},\mathbf{v}}$. For entanglement-breaking channels, the operators $\hat{K}_{x,\mathbf{u}}$ are related to the POVM applied on the input state and operators $\hat{W}_{x,\mathbf{v}}$ are related to the output state.

We now state our lower bound for the depth $T$ of all $c$-copy learning protocols that succeed with sufficiently high probability.
\begin{lemma}[Master lemma]\label{lem:master}
Any $c$-copy protocol with arbitrarily large ancillary assistance and adaptive state preparation and measurements will require a depth $T = \Omega(1/\Delta)$ to succeed with constant probability above $1/2$ in many-one channel discrimination with revelation (Problem~\ref{prob:manyrevel}) where
 \begin{equation}\label{eq:masterlemmadeltaup}
 \begin{aligned}
     \Delta = \Bigg(\sum_{\pmb{x}\in\{0,\dots,l\}^c\setminus\{\pmb{0}\}}\Big\{&\left\|\mathbb{E}_{\mathbf{v}}\left[\hat{W}_{\pmb{x},\mathbf{v}}^{\otimes 2}\right]\right\|_{\mathrm{op}}\\&\times\left\|\mathbb{E}_{\mathbf{u}}\left[\hat{K}_{\pmb{x},\mathbf{u}}^{\otimes 2}\right]\right\|_{\mathrm{op}}\Big\}^{1/2}\Bigg)^{2},
 \end{aligned}
    \end{equation}
    and we have defined
    \begin{equation}
        \hat{W}_{\pmb{x},\mathbf{v}} = \bigotimes_{i=1}^c\hat{W}_{x_i,\mathbf{v}},\quad\hat{K}_{\pmb{x},\mathbf{u}} = \bigotimes_{i=1}^c\hat{K}_{x_i,\mathbf{u}}.
    \end{equation}
\end{lemma}
A depth-$T$ $c$-copy protocol has a channel query complexity of $cT$ and requires the use of $T$ measurements. The proof for the above lemma can be found in Appendix~\ref{app:cm_protocols}. Our proof for the above lower bound borrows methodology from tree based learning protocols in \cite{chen2024optimaltradeoffsestimatingpauli}. The probability above random guess any learning protocol succeeds at many-one channel discrimination with revelation (Problem~\ref{prob:manyrevel}) is upper bounded by the total variation distance between the two hypotheses' probability distributions over the classical state $s_{0:T}$ of the learning protocol. We denote the likelihood ratio between probability distribution for hypothesis of channel $\mathcal{E}_{\mathbf{u},\mathbf{v}}$ and hypothesis of channel $\mathcal{E}_0$ as $\mathrm{LR}$. Using a one-sided bound (see Lemma~\ref{lem:1sidedLeCam}), the success probability above half is always upper bounded by the sum of $\beta$ and the probability that $\mathrm{LR}\leq 1-\beta$ for any $0<\beta<1$.

Suppose that the learner is at the $i$th learning step and has obtained outcomes $s_{0:i-1}$ so far. Conditioned on these outcomes, one obtains two possible probability distributions dictating the distribution of the next outcome $s_i$ depending on whether the channel was $\mathcal{E}_0$ or $\mathcal{E}_{\mathbf{u},\mathbf{v}}$ (see Fig.~\ref{fig:cMprotocol}(c)). The maximum possible variance of the likelihood ratio between these two conditional measures averaged over the distribution of the null hypothesis (paths arising when the channel is $\mathcal{E}_0$) and $\mathbf{u},\mathbf{v}$ gives the parameter $\Delta$. One can interpret this parameter as related to the possible confidence the learner gains in the $i$th step of the learning protocol. We now consider the likelihood ratio between the distributions related to the full classical memory of $\mathrm{LR}$. Using classical probability theory (inspired by the methods used in \cite{chen2024optimaltradeoffsestimatingpauli}), the probability of $\mathrm{LR}\leq 1-\beta$ is itself upper bounded by $\beta+z\Delta T$. Here $z$ is a constant dependent only on $\beta$. Hence for any scheme that succeeds with probability $1/2 + \Omega(1)$, one can find an appropriate value of $\beta$ which necessarily results in $\Delta\times T = \Omega(1)$ yielding the lower bound in question. The value $\Delta$ can be noted to be the maximum possible $\chi^2$-divergence (over $\mathbf{u}$ and $\mathbf{v}$ between the two marginal distributions at any learning step $i$ which naturally captures how distinguishable the two hypotheses are. For the sake of complete generality, the proof presented in App.~\ref{app:cm_protocols} makes use of probability measures and so we consider the Radon--Nikodym derivative between the two measures of the respective hypotheses which captures the same principle as the likelihood ratio.

Our work is an extension of past works that use similar principles of $c$-copy protocols, such as Ref.~\cite{chen2024optimaltradeoffsestimatingpauli} for the case of $c$-copy access to qubit states, which was later extended to multi-qudit states in Ref.~\cite{ller2025infinitehierarchymulticopyquantum}. We have extended these principles to apply to the channel learning scenario which also assumes arbitrary input states and ancilla assistance. Further, we also prove this for POVMs that have infinite accuracy, a non-trivial extension when considering bosonic systems since this allows for unbounded operators to be used to define the distribution corresponding to the probability measure for the outcomes (e.g., the single-mode homodyne measurement $|\hat{q}=q_0\rangle\langle\hat{q}=q_0|$ operator is not a bounded operator, although integrating it over an open interval for $q_0\in \mathbb{R}$ makes it bounded). We detail the subtlety involved in this in the proof in Appendix~\ref{app:cm_protocols}.

We make a few notable observations on the result in Lemma~\ref{lem:master}. The result assumes nothing about the input or output dimensions of the channel and further is always well defined since bounded operators have a well-defined operator norm even in the infinite-dimensional case. There is also no assumption on the size of the ancillary space since this result allows both the input state probes and measurements to be as entangled as desired. On an intuitive level, the use of the operator norm is a result of assuming the best possible input and output probes. If we assume the channel to be a complete replacement channel, we can make sharper restrictions (see Corollary~\ref{corr:statelearn}) allowing us to derive bounds related to state learning as well. The operator norm also informs us of instances where learning might become efficient since the terms that would dominate the value of $\Delta$ are those where the operator norm is large. The structure of the operators $\hat{W}_{\pmb{x},\mathbf{v}}$ and $\hat{K}_{\pmb{x},\mathbf{u}}$ being defined for $\pmb{x}\in \{0,\dots,l\}^{c}$ comes from the $c$-copy channel use, effectively giving these length $c$ strings that must be summed over. Over the averaging of the parameters of $\mathbf{u}$ and $\mathbf{v}$ we find that these correspond to the cases where the operators being averaged over all commute with each other, yielding sharp transitions in the complexity once $c$ reaches certain values or a resource such as the conjugate channel is made available.

\section{Results for channel learning}\label{sec:results}
In this section we will be examining lower bounds for $c$-copy learning protocols. Since the $c$-copy access might refer to either access to $\mathcal{E}^{\otimes c}$ or $(\mathcal{E}\otimes \mathcal{E}^*)^{\otimes c}$, the number of measurements (or depth of protocol $T$) means that the channel queries are $cT$ to the channel $\mathcal{E}$ or $\mathcal{E}\otimes\mathcal{E}^*$. Hence the channel query complexity is $\Theta(cT)$ where $T$ is the number of measurements.
\subsection{Sufficient conditions for efficient learning}
We first note a case where the learner can succeed in absolute value estimation for channel learning (Problem~\ref{prob:channel_learn}) using minimal quantum memory. This case assumes access to the complex-conjugate channel $\mathcal{E}^*$ along with access to the channel $\mathcal{E}$ and ancilla assistance, in effect using the same learning protocol from \cite{coroi2025exponentialadvantagecontinuousvariablequantum,PRXQuantum.5.040301} over the Choi state of the channel. This now provides us an upper bound to the sample complexity of Problem~\ref{prob:channel_learn} when given access to both $\mathcal{E}$ and $\mathcal{E}^*$ (further discussion can be found in App.~\ref{app:suff_cond_effic}). We assume that the complex-conjugation for defining channel $\mathcal{E}^*$ is in a fixed basis known to the learner.
\begin{theorem}\label{thm:efficientlearnconj}
    Consider the absolute value estimation for channel learning task described in Problem~\ref{prob:channel_learn} for channel $\mathcal{E}\in \CPTP(\mathcal{H}_{\mathrm{in}},\mathcal{H}_{\mathrm{out}})$ where the set of queries is of size $M$. There exists a non-adaptive ancilla assisted protocol which uses single copy access to the joint channel $(\mathcal{E}\otimes\mathcal{E}^*)$ that succeeds in Problem~\ref{prob:channel_learn} with accuracy $\epsilon$ and success probability of at least $1-\delta$ using $O(\log(M/\delta)\epsilon^{-4})$ measurements.
\end{theorem}
The proof for the above theorem follows from noting that the channel learning task in Problem~\ref{prob:channel_learn} is equivalent to the estimation of absolute values of displacement observables on some state $(\mathcal{E}\otimes\mathbb{I}_{\mathrm{in}})(\hat{\sigma})$ where $\hat{\sigma}\in D(\mathcal{H}_{\mathrm{in}}^{\otimes 2})$ is a highly entangled state that is either a collection of Bell pairs or two-mode squeezed vacuum probes (depending on input dimension). Assuming the complex-conjugation and state transposition is defined in the computational basis of the qudits (or the Fock basis for bosonic modes), the generalized Bell measurement scheme (described in detail in Appendix~\ref{app:qudit_disp} for the multi-qudit case and Appendix~\ref{app:bosonic_disp} for the bosonic mode case) can be applied on the state $(\mathcal{E}\otimes\mathbb{I}_{\mathrm{in}})(\hat{\sigma})\otimes(\mathcal{E}^*\otimes\mathbb{I}_{\mathrm{in}})(\hat{\sigma}^T)$ for efficient estimation of the square of the function $C_{\mathcal{E}}$. We can succeed by estimation of the square of the function to $\epsilon^2/3$ accuracy following the steps described in \cite{coroi2025exponentialadvantagecontinuousvariablequantum} and then guess the value to be zero if the estimate of the square is smaller than $2\epsilon^2/3$ which is sufficient to ensure $\epsilon$ accuracy for the absolute value.

The output of this learning protocol is the values of the function $C_{\mathcal{E}}$ accurate up to $\epsilon$ with a possible sign error. In the case that the channel $\mathcal{E}\in \CPTP(\mathcal{H}_{\mathrm{in}},\mathcal{H}_{\mathrm{out}})$ has both the Hilbert spaces being finite-dimensional, one can construct a hypothesis state $\hat{\sigma}\in D(\mathcal{H}_{\mathrm{out}}\otimes \mathcal{H}_{\mathrm{in}})$ which with joint measurements with the Choi-state can be used for also inferring the actual value of the function \cite{PRXQuantum.5.040301}. This is however intractable in the case of either Hilbert space being infinite-dimensional. Instead, using global measurements on up to $O(\log(M/\delta)/\epsilon^2)$ copies of the Choi-state in parallel, one can obtain the sign information as well. This is clearly prohibitive in the limited parallel channel use setting and hence any cases where we discuss bosonic channels we limit the channel learning task to only absolute value estimation. Regardless, this shows that the Choi-state learning (with phase information) will succeed given $O(\log(M/\delta)\epsilon^{-4})$ copies of $\mathcal{E}\otimes\mathcal{E}^*$, albeit requiring global measurements over $O(\log(M/\delta)/\epsilon^2)$ copies for sign information.

It is worth noting however that the complex-conjugate channel is a very non-trivial resource which might very rarely be available. An example where it is available is for channels with Kraus decompositions with each Kraus operator being real-valued in the basis of choice or equivalently the channel has a real-valued Stinespring dilation with an environment which begins in a self-conjugate state. There is a no-go theorem for the existence of a super-channel that is capable of transforming any number of copies of $\mathcal{E}$ into a single copy of $\mathcal{E}^*$ (see theorem 1 of \cite{zhu2026simulationadjointspetzrecovery}). We will now examine different lower bounds obtained based on the availability of this resource along with joint measurements. 

\subsection{Lower bounds in the absence of the complex-conjugate channel}
We consider the setting of $c$-copy access to the channel $\mathcal{E}$ with no access to the channel $\mathcal{E}^*$. We state our results on a case-wise basis for various input and output dimensions
\begin{theorem}\label{thm:qudit2quditinformal}
    Consider a $c$-copy learning protocol that succeeds in absolute value estimation for channel learning (Problem~\ref{prob:channel_learn}) for channels $\mathcal{E}\in \CPTP(\mathcal{H}_{d,m},\mathcal{H}_{d',m'})$ (with $d,d'$ being prime) for accuracy $\frac{\epsilon}{8e}$ (with $\epsilon<1$) with success probability of at least $2/3$ and $c\epsilon\leq1$ always.
    \begin{itemize}
        \item If $c\leq \min(d,d')-1$, the depth of the learning protocol must satisfy
        \begin{equation}
            T = \Omega(d^{m}d'^{m'}c^{-2}\epsilon^{-2}).
        \end{equation}
        \item Assuming $d = d'$ and $c \geq d$, the depth of the learning protocol must satisfy 
        \begin{equation}
            T = \Omega\left(\min\left(\frac{(d^m-1)(d^{m'}-1)}{(0.75c\epsilon)^2},\left(\frac{d}{c\epsilon}\right)^{2d}\right)\right).
        \end{equation}
    \end{itemize}
\end{theorem}
We defer the proof of this theorem along with a more complete statement for various possible cases of $c$ in relation to $d\neq d'$ to Thm.~\ref{thm:qudit2qudit} in Appendix~\ref{app:lowerbound_noconjugate} and provide an intuitive version of the proof here. The main idea, inspired from the lower bounds in \cite{chen2024optimaltradeoffsestimatingpauli,10.1145/3670418} is to construct a channel discrimination hypothesis testing task that one can succeed in by learning the channel according to Problem~\ref{prob:channel_learn}. Here the null hypothesis is some channel $\mathcal{E}_0$ and the hypothesis of channel $\mathcal{E}_{\mathbf{u},\mathbf{v}}$ where $\mathbf{u},\mathbf{v}$ are itself random variables hence following the framing of the task in Problem~\ref{prob:manyrevel}.

We define the channel $\mathcal{E}_{(\mathbf{q}_1,\mathbf{p}_1),(\mathbf{q}_2,\mathbf{p}_2)}\in \CPTP(\mathcal{H}_{d,m},\mathcal{H}_{d',m'})$ where $(\mathbf{q}_1,\mathbf{p}_1)\in \mathbb{F}_{d}^{m}\times \mathbb{F}_{d}^{m}\setminus\{(0,0)\}$ and $(\mathbf{q}_2,\mathbf{p}_2)\in \mathbb{F}_{d'}^{m'}\times \mathbb{F}_{d'}^{m'}\setminus\{(0,0)\}$. This channel follows the representation of Eq.~\eqref{eq:Epar1par2} for $l = 1$ where the operators 
\begin{align}
    \hat{K}_{1,(\mathbf{q}_1,\mathbf{p}_1)} &= \frac{e^{i\pi/4}\hat{D}_{d,m}(\mathbf{q}_1,\mathbf{p}_1)+e^{-i\pi/4}\hat{D}_{d,m}(\mathbf{q}_1,\mathbf{p}_1)^\dagger}{\sqrt{2}},\\
    \hat{W}_{1,(\mathbf{q}_2,\mathbf{p}_2)} &= \epsilon_0\frac{e^{i\pi/4}\hat{D}_{d',m'}(\mathbf{q}_2,\mathbf{p}_2)+e^{-i\pi/4}\hat{D}_{d',m'}(\mathbf{q}_2,\mathbf{p}_2)^\dagger}{2},
\end{align}
and $\hat{\rho}_0  =\mathbb{I}_{\mathrm{out}}/d'^{m'}$ which is the maximally mixed state. The parameters $(\mathbf{q}_1,\mathbf{p}_1),(\mathbf{q}_2,\mathbf{p}_2)$ are chosen uniformly randomly over the values they are allowed to take. Note that from the formulation of many-one channel discrimination task with revelation (Problem~\ref{prob:manyrevel}) we have $\hat{K}_{0,(\mathbf{q}_1,\mathbf{p}_1)} = \mathbb{I}_{d,m}$ and $\hat{W}_{0,(\mathbf{q}_2,\mathbf{p}_2)} = \mathbb{I}_{d',m'}$ always. Choosing an appropriate $\epsilon_0 = \Theta(\epsilon)$, a learner that succeeds in the absolute value estimation for channel learning task of Problem~\ref{prob:channel_learn} also succeeds in the discrimination task of Problem~\ref{prob:manyrevel} after the revelation of the randomized parameters $(\mathbf{q}_1,\mathbf{p}_1),(\mathbf{q}_2,\mathbf{p}_2)$ by learning this exact query to the function $C_{\mathcal{E}}$. The channel $\mathcal{E}_{(\mathbf{q}_1,\mathbf{p}_1),(\mathbf{q}_2,\mathbf{p}_2)}$ has a sparse representation when represented using the $C_{\mathcal{E}}$ function making the knowledge of the parameters useful in distinguishing this channel from $\mathcal{E}_0$. We now apply the master lemma (Lemma~\ref{lem:master}) to find a lower bound to the depth of any learning protocol that succeeds in Problem~\ref{prob:channel_learn}.
\begin{figure*}[ht]
    \centering
    \includegraphics[width=\textwidth]{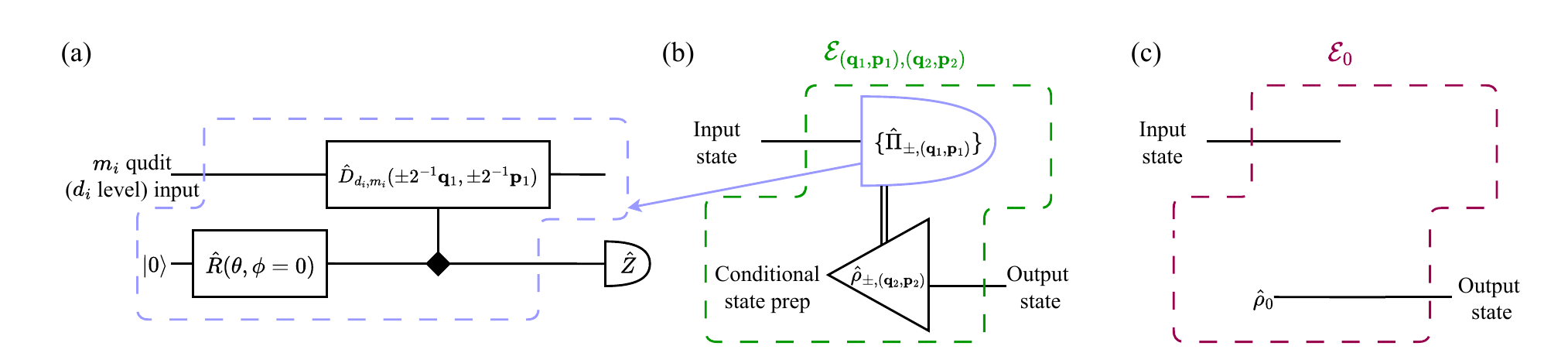}
    \caption{(a) Circuit to implement the two-outcome POVM $\{\hat{\Pi}_{\pm,(\mathbf{q}_1,\mathbf{p}_1)}\}$ in Eq.~\eqref{eq:Eq1p1q2p2} where $R(\theta,\phi)$ is a qubit rotation (we set $\theta = -\pi/4$ to obtain the exact POVM) and the controlled displacement operation performs $\hat{D}_{d,m}(\pm2^{-1}\mathbf{q},\pm2^{-1}\mathbf{p})$ conditioned on whether the qubit is in the $\ket{+}$ or $\ket{-}$ state (eigenstates of Pauli $\hat{X}$) and the two POVM outcomes correspond to the two outcomes of the $\hat{Z}$ basis measurement. (b) Construction for the parameterized channel $\mathcal{E}_{(\mathbf{q}_1,\mathbf{p}_1),(\mathbf{q}_2,\mathbf{p}_2)}$ (Eq.~\eqref{eq:Eq1p1q2p2}) for $(\mathbf{q}_1,\mathbf{p}_1)\in \mathbb{F}_{d}^{m}\times\mathbb{F}_{d}^{m}\setminus\{(0,0)\}$ and $(\mathbf{q}_2,\mathbf{p}_2)\in \mathbb{F}_{d'}^{m'}\times\mathbb{F}_{d'}^{m'}\setminus\{(0,0)\}$ realized using the described POVM  where the readout classically controls a state prepartion to prepare either $\hat{\rho}_{+,(\mathbf{q}_2,\mathbf{p}_2)}$ or $\hat{\rho}_{-,(\mathbf{q}_2,\mathbf{p}_2)}$. (c) The channel $\mathcal{E}_0$ is a complete replacement channel which only outputs the state $\hat{\rho}_0 = \frac{1}{2}(\hat{\rho}_{+,(\mathbf{q}_2,\mathbf{p}_2)}+\hat{\rho}_{-,(\mathbf{q}_2,\mathbf{p}_2)}) = \frac{\mathbb{I}}{d'^{m'}}$.}
    \label{fig:hardtolearn}
\end{figure*}

Following the result in Lemma~\ref{lem:master}, we find that now the complexity lower bound would be a function of the sum of certain operator norms. The operator norm for the operator
\begin{equation}
    \hat{O}_{d,m,k} = \sum_{\mathbf{q},\mathbf{p}\in\mathbb{F}^m_d}\hat{D}_{d,m}(\mathbf{q},\mathbf{p})^{\otimes 2k},
\end{equation}
has the property
\begin{equation}\label{eq:Odmknorm}
    \|\hat{O}_{d,m,k}\|_{\mathrm{op}} = \begin{cases}
        d^m, &k\neq 0\mod d\\
        d^{2m}, &k =0\mod d
    \end{cases},
\end{equation}
proven in Lemma~\ref{lem:discr_disp_tensor}. Applying this, we recover a sharp transition once $c$ becomes equal to $d$ in the case of $d = d'$. This is similar to the principle of the infinite learning hierarchy established in \cite{ller2025infinitehierarchymulticopyquantum} which is now generalized for the channel learning scheme. Further, we show that the scaling of $\epsilon^{-2d}$ shows up in the lower bound itself. We note trivially that through the use of the $d$-copy efficient learning scheme in \cite{ller2025infinitehierarchymulticopyquantum}, access to $d$ copies of the channel's Choi state in the case of $d = d'$ succeeds in both obtaining the absolute values for channel learning (Problem~\ref{prob:channel_learn}) as well as their phase information with a complexity of $O(d(m+m')\log(d/\delta)\epsilon^{-2d})$ (see Thm.~\ref{thm:quditeff}). This shows that the $\epsilon^{-2d}$ scaling is in fact tight within a certain regime for the Choi-state learning task. We discuss this further in Sec.~\ref{sec:learninghierarchy}.

Using a very similar construction of channels for the case of a bosonic multimode channel $\mathcal{E}_{\gamma_1,\gamma_2}\in \CPTP(\mathcal{H}_{\infty,m},\mathcal{H}_{\infty,m'})$, we are able to find the lower bound given as follows.
\begin{theorem}\label{thm:boson2bosoninformal}
    Consider a $c$-copy learning protocol that succeeds in absolute value estimation for channel learning (Problem~\ref{prob:channel_learn}) for channels $\mathcal{E}\in \CPTP(\mathcal{H}_{\infty,m},\mathcal{H}_{\infty,m'})$ with $m,m'\geq 8$ and $\min(\kappa^2 m,\kappa'^2m')\geq 2/0.99$, for accuracy $\epsilon \leq 0.05$, two-mode squeezing $r$ such that $\cosh(2r)\geq 1.06\kappa m$, with success probability of at least $2/3$ and $c = O(1/\epsilon)$. The depth of the learning protocol must satisfy
    \begin{equation}
        T = \Omega(d_{\mathrm{in}}^md_{\mathrm{out}}^{m'}c^{-2}\epsilon^{-2}),
    \end{equation}
    where we define effective mode dimensions $d_{\mathrm{in}} =\sqrt{1+(0.99\kappa\tanh^2(2r))^2}$ and $d_{\mathrm{out}} = \sqrt{1+(0.99\kappa')^2}$.
\end{theorem}
We defer the full proof of this theorem to Thm.~\ref{thm:boson2boson} in Appendix~\ref{app:lowerbound_noconjugate} and provide an intuitive version of the proof here. We define a channel $\mathcal{E}_{\gamma_1,\gamma_2}\in \CPTP(\mathcal{H}_{\infty,m},\mathcal{H}_{\infty,m'})$ for $\gamma_1\in \mathbb{C}^m,\gamma_2\in \mathbb{C}^{m'}$ that are both sampled from a complex-valued Gaussian distribution of spreads $2\sigma_1^2 = 0.99\kappa\tanh^2(2r)$ and $2\sigma_2^2 = 0.99\kappa'$ respectively. Similar to the qudit case, this channel follows the representation of Eq.~\eqref{eq:Epar1par2} for $l =1$ where the operators 
\begin{align}
    \hat{K}_{1,\gamma_1} = \frac{e^{i\pi/4}\hat{D}(\gamma_1) + e^{-i\pi/4}\hat{D}^\dagger(\gamma_1)}{\sqrt{2}},\\
    \hat{W}_{1,\gamma_2} = \epsilon_0\frac{e^{i\pi/2}\hat{D}(\gamma_2) + e^{-i\pi/2}\hat{D}^{\dagger}(\gamma_2)}{\sqrt{2}},
\end{align}
and $\hat{\rho}_0 = \nu^{2\hat{n}}(1-\nu^2)^{m'}$  is a thermal state (where $\hat{n} = \sum_{i=1}^{m'}\hat{n}_i$) for some $0<\nu<1$. By choosing an appropriate $\epsilon_0 = \Theta(\epsilon)$, a learner that succeeds in the absolute value estimation for channel learning task of Problem~\ref{prob:channel_learn} is also able to succeed in the discrimination task of Problem~\ref{prob:manyrevel} after the sampled values of $\gamma_1,\gamma_2$ are revealed since the learner now can query the function $C_{\mathcal{E}}^{\mathrm{TMSV},r}$ for argument $(\gamma_1/\tanh(2r),\gamma_2)$.
Similar to the multi-qudit case, we follow the result in Lemma~\ref{lem:master} and note that the operator norm in question is that of operators of form
\begin{equation}
    \hat{O}_{\infty,m,k} = \int d^{2m}\gamma \frac{e^{-\frac{|\gamma|^2}{2\sigma^2}}}{(2\pi\sigma^2)^{m}}\hat{D}^{\otimes k}(\gamma),
\end{equation}
which we show in Lemma~\ref{lem:dispopnorm} have their operator norm satisfy
\begin{equation}
    \|\hat{O}_{\infty,m,k}\|_{\mathrm{op}} \leq \frac{1}{(1+k^2\sigma^4)^{m/2}}.
\end{equation}
Note that the above bound has no sharp transition in its behavior with $k$. Further, the limit of $\sigma \to 0$ recovers the operator norm of $1$ which one would expect since the operator converges to identity. Plugging this norm into the main result of Lemma~\ref{lem:master} with further mathematical manipulation yields the lower bound in Thm.~\ref{thm:boson2boson}.

So far we have only highlighted the channels that are hard to learn; we will now show that these channels can be realized in practice by a measure-and-prepare channel. We highlight that the circuit in Fig.~\ref{fig:hardtolearn}(a) realizes a two-outcome POVM $\{\hat{\Pi}_{\pm,(\mathbf{q},\mathbf{p})}\}$ such that for any $(\mathbf{q}_1,\mathbf{p}_1)\in \mathbb{F}_{d}^m\times\mathbb{F}_d^m$, $\hat{\Pi}_{+,(\mathbf{q}_1,\mathbf{p}_1)} - \hat{\Pi}_{-,(\mathbf{q}_1,\mathbf{p}_1)} = \hat{K}_{1,(\mathbf{q}_1,\mathbf{p}_1)}$ when the input dimension $d>2$. For $d = 2$, this exact POVM can still be implemented since we note that $\hat{D}_{2,m}(\mathbf{q}_1,\mathbf{p}_1) = \hat{D}_{2,m}(\mathbf{q}_1,\mathbf{p}_1)^\dagger$ and so it can be interpreted as first measuring using the POVM $\frac{\mathbb{I}_{2,m}\pm \hat{D}_{2,m}(\mathbf{q}_1,\mathbf{p}_1)}{2}$ following which the classical information experiences a noise channel resulting in a fuzzy measurement of $\frac{\mathbb{I}_{2,m}\pm \hat{D}_{2,m}(\mathbf{q}_1,\mathbf{p}_1)/\sqrt{2}}{2}$ \cite{PhysRevD.33.2253}. Further we can prepare two states $\hat{\rho}_{\pm,(\mathbf{q}_2,\mathbf{p}_2)}$ such that for any $(\mathbf{q}_2,\mathbf{p}_2)\in \mathbb{F}_{d'}^{m'}\times\mathbb{F}_{d'}^{m'}$ we have $\hat{\rho}_{\pm,(\mathbf{q}_2,\mathbf{p}_2)} = \hat{\rho}_0 \pm \sqrt{\hat{\rho}_0}\hat{W}_{1,(\mathbf{q}_2,\mathbf{p}_2)}\sqrt{\hat{\rho}_0}$ which gives
\begin{equation}\label{eq:Eq1p1q2p2}
\begin{aligned}
    &\mathcal{E}_{(\mathbf{q}_1,\mathbf{p}_1),(\mathbf{q}_2,\mathbf{p}_2)}(\cdot) \\&=\text{Tr}\left[(\cdot)\hat{\Pi}_{+,(\mathbf{q}_1,\mathbf{p}_1)}\right]\hat{\rho}_{+,(\mathbf{q}_2,\mathbf{p}_2)}\\
    &\quad\quad + \text{Tr}\left[(\cdot)\hat{\Pi}_{-,(\mathbf{q}_1,\mathbf{p}_1)}\right]\hat{\rho}_{-,(\mathbf{q}_2,\mathbf{p}_2)}
    \\&= \sum_{x\in\{0,1\}}\text{Tr}\left[(\cdot)\hat{K}_{x,(\mathbf{q}_1,\mathbf{p}_1)}\right]\hat{W}_{x,(\mathbf{q}_2,\mathbf{p}_2)}/d'^{m'},
\end{aligned}
\end{equation}
and similarly for the bosonic case we have
\begin{equation}\label{eq:Egam1gam2}
\begin{aligned}
    &\mathcal{E}_{\gamma_1,\gamma_2}(\cdot) \\&=\text{Tr}\left[(\cdot)\hat{\Pi}_{+,\gamma_1}\right]\hat{\rho}_{+,\gamma_2}\\
    &\quad\quad + \text{Tr}\left[(\cdot)\hat{\Pi}_{-,\gamma_1}\right]\hat{\rho}_{-,\gamma_2}
    \\&= (1-\nu^2)^{m'}\sum_{x\in\{0,1\}}\text{Tr}\left[(\cdot)\hat{K}_{x,\gamma_1}\right]\nu^{\hat n}\hat{W}_{x,\gamma_2}\nu^{\hat n},
\end{aligned}
\end{equation}
where we define the POVM $\{\hat{\Pi}_{\pm,\gamma_1}\}$ by substituting $\theta=-\pi/4$ in Eq.~\eqref{eq:bosonic2outcome} in Appendix~\ref{app:bosonic_disp}. Fig.~\ref{fig:hardtolearn}(b) shows how the channel $\mathcal{E}_{(\mathbf{q}_1,\mathbf{p}_1),(\mathbf{q}_2,\mathbf{p}_2)}$ can be implemented using the POVM of $\{\hat{\Pi}_{\pm,(\mathbf{q}_1,\mathbf{p}_1)}\}$ which itself is implemented using an ancillary qubit for the measurement. 


A similar such construction can be used to realize the bosonic multimode channel described by $\mathcal{E}_{\gamma_1,\gamma_2}\in \CPTP(\mathcal{H}_{\infty,m},\mathcal{H}_{\infty,m'})$ which we detail the exact construction of in the proof of Thm.~\ref{thm:boson2boson} which uses three-peak states introduced in \cite{coroi2025exponentialadvantagecontinuousvariablequantum}.

We also explore another relevant case of channels that have input Hilbert spaces of $\mathcal{H}_{d,m}$ and output Hilbert spaces of $\mathcal{H}_{\infty,m'}$. A case where such channels are encountered is the case of noise acting on an encoded bosonic error correcting code that encodes qudits into multiple modes such as multi-mode Gottesman--Kitaev--Preskill codes \cite{conrad_gottesman-kitaev-preskill_2022,56vj-z7h1} or quantum spherical codes \cite{jain_quantum_2024}. These channels characterize the full leakage outside the codespace before any recovery can be applied on these channels. We present our result for such channels in Thm.~\ref{thm:qudit2boson} in Appendix~\ref{app:lowerbound_noconjugate}.

\subsection{Lower bounds for $1$-copy access to self-conjugate channel}
We now consider the case of learning channels that are actually self complex-conjugate, i.e., we have for some $\mathcal{E}\in \CPTP(\mathcal{H}_{\mathrm{in}},\mathcal{H}_{\mathrm{out}})$, $\mathcal{E} \equiv \mathcal{E}^*$ with conjugation defined in some appropriate basis. One might wonder if the power of the conjugate channel suffices even in the case of 1-copy access with ancilla assistance. We find that this is not the case, as stated in the following theorems and proved in App.~\ref{app:lowerbound_selfconjugate}.
\begin{theorem}\label{thm:qudit2quditinformalselfconj}
    Consider a $1$-copy learning protocol that succeeds in absolute value estimation for channel learning (Problem~\ref{prob:channel_learn}) for channels $\mathcal{E}\in \CPTP(\mathcal{H}_{d,m},\mathcal{H}_{d',m'})$ (with $d,d'$ being prime) that are promised to satisfy $\mathcal{E}\equiv \mathcal{E}^*$ and the estimation is to accuracy $\epsilon\leq 1/16$ with success probability of at least $2/3$. This learning protocol requires a depth of 
    \begin{equation}
        T = \Omega(d^{m}d'^{m'}\epsilon^{-2}).
    \end{equation}
\end{theorem}

Inspired by the channel constructions in Fig.~\ref{fig:hardtolearn}, we come up with a similarly minimal choice of channels with one added term to ensure the symmetry of $\mathcal{E} = \mathcal{E}^*$ where conjugation is defined in the standard computational basis of input and output dimensions (this is the diagonal basis for the phase operator $\hat{Z}$ for qudits). By choosing a 4-outcome POVM, we construct an instance of a parameterized channel in the form of Eq.~\eqref{eq:Epar1par2} which we refer to as $\mathcal{E}^{\text{self-conj}}_{(\mathbf{q}_1,\mathbf{p}_1),(\mathbf{q}_2,\mathbf{p}_2)}$ given by
\begin{equation}
\begin{aligned}
    \mathcal{E}^{\text{self-conj}}_{(\mathbf{q}_1,\mathbf{p}_1),(\mathbf{q}_2,\mathbf{p}_2)}(\cdot) &= \frac{1}{2}\mathcal{E}_{(\mathbf{q}_1,\mathbf{p}_1),(\mathbf{q}_2,\mathbf{p}_2)}(\cdot) \\&\quad\quad+\frac{1}{2}\mathcal{E}_{(-\mathbf{q}_1,\mathbf{p}_1),(-\mathbf{q}_2,\mathbf{p}_2)}(\cdot).
\end{aligned}
\end{equation}
We note that by choosing the diagonal basis for the $\hat{Z}$ operator for transposition, we have $\hat{D}_{d,m}(\mathbf{q},\mathbf{p})^T = \hat{D}_{d,m}(-\mathbf{q},\mathbf{p})$ which can be used to show that for all input states $\hat{\rho}\in D(\mathcal{H}_{d,m})$, we have $\mathcal{E}^{\text{self-conj}}_{(\mathbf{q}_1,\mathbf{p}_1),(\mathbf{q}_2,\mathbf{p}_2)}(\hat{\rho}^T)^T= \mathcal{E}^{\text{self-conj}}_{(\mathbf{q}_1,\mathbf{p}_1),(\mathbf{q}_2,\mathbf{p}_2)}(\hat{\rho})$ showing the channel to be equal to its own complex-conjugate. Revisiting the lower bound from Lemma~\ref{lem:master}, we note that $\mathcal{E}^{\text{self-conj}}_{(\mathbf{q}_1,\mathbf{p}_1),(\mathbf{q}_2,\mathbf{p}_2)}$ has the form of Eq.~\eqref{eq:Epar1par2} with $l=2$ where the $x=0,1$ operators are the same as those of $\mathcal{E}_{(\mathbf{q}_1,\mathbf{p}_1),(\mathbf{q}_2,\mathbf{p}_2)}$ with the difference that now there is a $\hat{K}_{2,(\mathbf{q}_1,\mathbf{p}_1)}=\hat{K}_{1,(\mathbf{q}_1,\mathbf{p}_1)}^T$ and $\hat{W}_{2,(\mathbf{q}_2,\mathbf{p}_2)}=\hat{W}_{1,(\mathbf{q}_2,\mathbf{p}_2)}^T$ and the parameter $\epsilon_0$ has a different constant of proportionality with $\epsilon$ to ensure validity of the state.

In the $c=1$ case, the lower bound in Lemma~\ref{lem:master} would have no cross terms between $\hat{K}_{x,\mathbf{u}}$ and $\hat{W}_{x,\mathbf{v}}$ operators of different values for $x$. As a result, of the operator norm bounds in Eq.~\eqref{eq:Odmknorm} only the $k=1$ case is relevant, and one can note that $\|\hat{O}^T\|_{\mathrm{op}} = \|\hat{O}\|_{\mathrm{op}}$ for any bounded operator $\hat{O}$. As a result, the $c = 1$ case for self-conjugate channels turns out to be no different from the $c=1$ case for channels that are not self-conjugate. The complete proof is detailed in Thm.~\ref{thm:qudit2quditselfconj}, with the qudit $\to$ bosonic case in Thm.~\ref{thm:qudit2bosonselfconj}. We note that a very similar behavior also holds true for the bosonic case as highlighted in the following theorem.

\begin{theorem}\label{thm:boson2bosoninformalselfconj}
    Consider a $1$-copy learning protocol that succeeds in absolute value estimation for channel learning (Problem~\ref{prob:channel_learn}) for channels $\mathcal{E}\in \CPTP(\mathcal{H}_{\infty,m},\mathcal{H}_{\infty,m'})$ (with $m,m'\geq 8$, $\min(\kappa^2 m,\kappa'^2 m')\geq 2/0.99$, and two-mode squeezing $r$ with $\cosh(2r)\geq 1.06\kappa m$) that are promised to satisfy $\mathcal{E}\equiv \mathcal{E}^*$ and the estimation is to accuracy $\epsilon <0.025$ with success probability of at least $2/3$. This learning protocol requires a depth of
    \begin{equation}
        T = \Omega(d_{\mathrm{in}}^md_{\mathrm{out}}^{m'}\epsilon^{-2}),
    \end{equation}
    where we define effective mode dimensions $d_{\mathrm{in}} =\sqrt{1+(0.99\kappa\tanh^2(2r))^2}$ and $d_{\mathrm{out}} = \sqrt{1+(0.99\kappa')^2}$.
\end{theorem}

Here we similarly prove hardness by the construction of a channel $\mathcal{E}_{\gamma_1,\gamma_2}^{\text{self-conj}}  = \frac{1}{2}(\mathcal{E}_{\gamma_1,\gamma_2}(\cdot) + \mathcal{E}_{-\gamma_1^*,-\gamma_2^*}(\cdot))$. Due to the exact same reasons we discussed in the qudit case, the $c=1$ case here turns out not to be any different from the $c=1$ case for channels that are not self-conjugate. The complete proof is detailed in Thm.~\ref{thm:boson2bosonselfconj}. Note that once $c=2$, efficient learning becomes possible using the scheme highlighted in Thm.~\ref{thm:efficientlearnconj}.  Notably, the scaling with respect to the estimation error parameter $\epsilon$ from the upper bound follows an $\epsilon^{-4}$ scaling. 

\subsection{Lower bounds for $c$-copy access to the channel along with the complex-conjugate channel}
Using Lemma~\ref{lem:master}, we can also verify that within the appropriate regime, the $\epsilon$ scaling in the upper bound from Thm.~\ref{thm:efficientlearnconj} is optimal for schemes that are given $c$-copy access to the channel $\mathcal{E}\otimes\mathcal{E}^*$. We state our result for the qudit channel case as follows.
\begin{theorem}\label{thm:qudit2quditinformalwithconj}
Consider a $c$-copy learning protocol where $c$-copy access is given to $\mathcal{E}\otimes\mathcal{E}^*$. Suppose that it succeeds in absolute value estimation for channel learning (Problem~\ref{prob:channel_learn}) for channel $\mathcal{E}\in \CPTP(\mathcal{H}_{d,m},\mathcal{H}_{d',m'})$ (where $d,d'$ are prime) to accuracy $\epsilon\leq \frac{1}{8}$ with success probability of at least $2/3$ and $c\epsilon\leq 1$. Further we assume that $(d^{m}-1)(d'^{m'}-1)c^2\epsilon^2\geq 64$ by ensuring large enough $m,m'$. The depth of the learning protocol must satisfy (for $c\leq \min(d',d)-1$)
\begin{equation}
    T = \Omega\left(\frac{1}{c^4\epsilon^4}\right).
\end{equation}
\end{theorem}

The proof for the above theorem is detailed in Thm.~\ref{thm:qudit2quditwithconj} which presents a more complete formulation of this result. We present an intuitive version of the proof in the following paragraphs. We only examine the case of $c\leq \min(d',d)$ since there are transitions that occur once $c$ crosses this amount for the case of access to only $\mathcal{E}$ (see Thm.~\ref{thm:qudit2qudit}). We revisit the parameterized channel $\mathcal{E}_{(\mathbf{q}_1,\mathbf{p}_1),(\mathbf{q}_2,\mathbf{p}_2)}\in \CPTP(\mathcal{H}_{d,m},\mathcal{H}_{d',m'})$ defined in Eq.~\eqref{eq:Eq1p1q2p2}. The conjugate of this channel can be obtained by simply replacing the operators $\hat{K}_{x,(\mathbf{q}_1,\mathbf{p}_1)}$ and $\hat{W}_{x,(\mathbf{q}_2,\mathbf{p}_2)}$ with $\hat{K}_{x,(\mathbf{q}_1,\mathbf{p}_1)}^T$ and $\hat{W}_{x,(\mathbf{q}_2,\mathbf{p}_2)}^T$ with transposition defined in the diagonal basis of the qudit phase operator. Note that $\mathcal{E}_{(\mathbf{q}_1,\mathbf{p}_1),(\mathbf{q}_2,\mathbf{p}_2)}\otimes \mathcal{E}_{(\mathbf{q}_1,\mathbf{p}_1),(\mathbf{q}_2,\mathbf{p}_2)}^*$ follows the exact parameterization as CPTP maps of the form of Eq.~\eqref{eq:Epar1par2} with $l = 3$ due to there being $4$ terms in this expansion and the parameters being $(\mathbf{q}_1,\mathbf{p}_1)\in \mathbb{F}_{d}^{m}\times\mathbb{F}_{d}^{m}\setminus\{(0,0)\}$ and $(\mathbf{q}_2,\mathbf{p}_2)\in \mathbb{F}_{d'}^{m'}\times\mathbb{F}_{d'}^{m'}\setminus\{(0,0)\}$. As a result we can apply Lemma~\ref{lem:master} in this case as well, with the only difference being that we have a different set of operator norms to deal with (see Corollary~\ref{corr:quditEtransopnorm}). Since we are now considering the channel $\mathcal{E}\otimes\mathcal{E}^*$, we now will be concerned with the operator norms of terms that look like $\hat{K}_{x_1,(\mathbf{q}_1,\mathbf{p}_1)}\otimes \hat{K}_{x_2,(\mathbf{q}_1,\mathbf{p}_1)}^T$ and $\hat{W}_{x_1,(\mathbf{q}_2,\mathbf{p}_2)}\otimes \hat{W}_{x_2,(\mathbf{q}_2,\mathbf{p}_2)}^T$. This can be reduced to studying the operator norms of the following operator
\begin{equation}
    \hat{O}_{d,m,k,k'} = \sum_{\mathbf{q},\mathbf{p}\in \mathbb{F}^m_d}\hat{D}_{d,m}(\mathbf{q},\mathbf{p})^{\otimes 2k}\otimes \hat{D}_{d,m}(-\mathbf{q},\mathbf{p})^{\otimes 2k'},
\end{equation}
as a result of $\hat{D}_{d,m}(-\mathbf{q},\mathbf{p}) = \hat{D}_{d,m}^T(\mathbf{q},\mathbf{p})$ for transposition in the standard computational basis. The operators have the property
\begin{equation}
    \|\hat{O}_{d,m,k,k'}\|_{\mathrm{op}} = \begin{cases}
        d^m, &k\neq k' \mod d\\
        d^{2m},& k = k'\mod d
    \end{cases},
\end{equation}
which we prove in Lemma~\ref{lem:discr_disp_transtensor}. The interesting property of this operator norm occurs for the case of $k=k'\mod d$ where the operator norm ends up saturating the maximum possible value it can take (by applying the triangle inequality). This can be seen to boil down to the fact that all the operators $\hat{D}_{d,m}(\mathbf{q},\mathbf{p})\otimes \hat{D}_{d,m}(-\mathbf{q},\mathbf{p})$ commute with each other for all $\mathbf{q},\mathbf{p}\in \mathbb{F}_{d}^m$ hence are all diagonalized in the same basis. We see that the intuition that expectation values of commuting operators can be sample-efficient to compute due to having a common basis where they diagonalize directly shows up in the operator norm. As a result, the $k= k'\mod d$ terms are the major contributors to the confidence parameter of $\Delta$ from Lemma~\ref{lem:master} when applied to $\mathcal{E}_{(\mathbf{q}_1,\mathbf{p}_1),(\mathbf{q}_2,\mathbf{p}_2)}\otimes \mathcal{E}_{(\mathbf{q}_1,\mathbf{p}_1),(\mathbf{q}_2,\mathbf{p}_2)}^*$ which result in a dominant scaling of $\epsilon^{-4}$. We find a similar result for the case of bosonic channels as well.

\begin{theorem}\label{thm:boson2bosoninformalwithconj}
    Consider a $c$-copy learning protocol where $c$-copy access is given to $\mathcal{E}\otimes\mathcal{E}^*$. Suppose that it succeeds in absolute value estimation for channel learning (Problem~\ref{prob:channel_learn}) for channel $\mathcal{E}\in \CPTP(\mathcal{H}_{\infty,m},\mathcal{H}_{\infty,m'})$ (where $m,m'\geq 8$, $\min(\kappa^2 m, \kappa'^2m')\geq 2/0.99$, and the two-mode squeezing $r$ satisfies $\cosh(2r) \geq 1.06\kappa m$) to accuracy $\epsilon\leq 0.05$ with success probability of at least $2/3$ and $c\epsilon=O(1)$ and $d_{\mathrm{in}}^{m}d_{\mathrm{out}}^{m'}c^2(\epsilon/0.11)^2\geq 1$  which can be ensured by choosing large enough $m,m'$ (here we define effective mode dimensions $d_{\mathrm{in}} =\sqrt{1+(0.99\kappa\tanh^2(2r))^2}$ and $d_{\mathrm{out}} = \sqrt{1+(0.99\kappa')^2}$). The depth of the learning protocol must satisfy 
\begin{equation}
    T = \Omega\left(\frac{1}{c^4\epsilon^4}\right).
\end{equation}
\end{theorem}
The proof for the above theorem is detailed in Thm.~\ref{thm:boson2bosonwithconj} which presents a more complete formulation of the result as well. We similarly revisit the channel $\mathcal{E}_{\gamma_1,\gamma_2}\in \CPTP(\mathcal{H}_{\infty,m},\mathcal{H}_{\infty,m'})$ defined in Eq.~\eqref{eq:Egam1gam2} and note that we can define $\mathcal{E}^*_{\gamma_1,\gamma_2}\equiv \mathcal{E}_{-\gamma_1^*,-\gamma_2^*}$ where we use the Fock basis to define the transpose operation. This shows that the channel $\mathcal{E}\otimes\mathcal{E}^*$ has the form of channels described by Eq.~\eqref{eq:Epar1par2} with $l = 3$ since there are 4 operators to consider when expanding this channel. Similar to the qudit case, we now must consider operators of form $\hat{K}_{x_1,\gamma_1}\otimes\hat{K}_{x_2,\gamma_1}^T$ and $\hat{W}_{x_1,\gamma_2}\otimes\hat{W}_{x_2,\gamma_2}^T$. As shown in the proof of Thm.~\ref{thm:boson2bosonwithconj} this boils down to the operator norm of the following operator
\begin{equation}
    \hat{O}_{\infty,m,k,k'} = \int d^{2m}\gamma\frac{e^{-\frac{|\gamma|^2}{2\sigma^2}}}{(2\pi\sigma^2)^{m}}\hat{D}(\gamma)^{\otimes k}\otimes \hat{D}(-\gamma^*)^{\otimes k'},
\end{equation}
which we show in Lemma~\ref{lem:disptransopnorm} (and, for the Heisenberg--Weyl observables, Corollary~\ref{corr:bosonicEnormconj}) have their operator norm satisfy
\begin{equation}
    \|\hat{O}_{\infty,m,k,k'}\|_{\mathrm{op}} \leq \frac{1}{(1+(k-k')^2\sigma^4)^{m/2}},
\end{equation}
which attain their highest values for the case of $k=k'$ as one would expect due to the fact that the operators $\hat{D}(\gamma)\otimes\hat D(-\gamma^*)$ commute with each other allowing for a common eigenstate which can maximize the operator norm. These terms end up dominating in the expression of $\Delta$ from Lemma~\ref{lem:master} when applied to the channel $\mathcal{E}_{\gamma_1,\gamma_2}\otimes\mathcal{E}^*_{\gamma_1,\gamma_2}$ which finally yields the $\epsilon^{-4}$ scaling in the appropriate conditions for Thm.~\ref{thm:boson2bosonwithconj}.

Similar results also hold for the case of channels in $\CPTP(\mathcal{H}_{d,m},\mathcal{H}_{\infty,m'})$ which we detail in Thm.~\ref{thm:qudit2bosonwithconj}. These results show that the $\epsilon^{-4}$ scaling is tight in an appropriate parameter regime since we know that $1$-copy access to $\mathcal{E}\otimes \mathcal{E}^*$ admits an efficient algorithm for learning with complexity scaling as $\epsilon^{-4}$ (see Thm.~\ref{thm:efficientlearnconj}). We now study the further implications of these results in establishing a hierarchy in learning based on the availability of multi-copy access as well as the conjugate channel.
\begin{figure}
    \centering
    \includegraphics[width=0.8\linewidth]{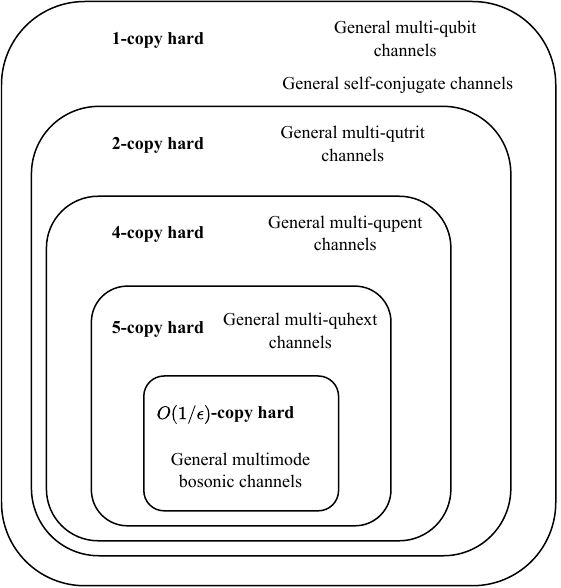}
    \caption{Visual representation of learning hierarchy over the set of all channels with the same input and output Hilbert spaces. By ``hard'' we refer to sample complexity that scales exponentially in the system size (i.e. number of qudits or modes). We note that multi-qubit channels \cite{10.1145/3670418} and self-conjugate channels (Thms.~\ref{thm:qudit2quditinformalselfconj},~\ref{thm:boson2bosoninformalselfconj}) are 1-copy hard, but are not 2-copy hard. Similarly we have results for multi-qudit channels being $(d-1)$-copy hard, but not $d$-copy hard (Thm.~\ref{thm:qudit2quditinformal}) for any prime $d$ which we depict for the cases of $d = 3$ and $5$ in the above figure (namely qutrits, qupents). Further, we discuss that channels that are $(d-1)$-copy hard but not $d$-copy hard can be realized for square-free $d$ so we depict this for the case of $d = 6$ as well (quhexts which are decomposed as a tensor product of a qubit and qutrit for each quhext). Finally, we always find multi-mode bosonic channels that are hard to learn as long as $c = O(1/\epsilon)$ (Thm.~\ref{thm:boson2bosoninformal}).}
    \label{fig:learninghierarchy}
\end{figure}
\section{Implications in learning hierarchy}\label{sec:learninghierarchy}
All lower bounds examined in this work focus on learning protocols with unbounded ancilla assistance. In this regard we note that the importance of the availability of ancilla has already been well understood in lower bounds for the Pauli channel learning problem where a quantum memory of $m$ qubits with entanglement is both necessary and sufficient for efficient learning of the Pauli eigenvalues of the Pauli channel \cite{chen_tight_2024}. The entanglement-enhanced learning advantage of using an ancilla holds true for bosonic displacement noise channels as well where there is an exponential separation between the sample complexity of ancilla-assisted learning and all learning schemes that make use of no ancillary space \cite{PhysRevLett.133.230604}. Because in both these cases the learning task only involves the estimation of commuting observables on the Choi state of the channel, efficient learning schemes are possible. In the case of Pauli channels this is the estimation of $\hat{D}_{2,m}(\mathbf{q},\mathbf{p})\otimes \hat{D}_{2,m}(-\mathbf{q},\mathbf{p})$ for $(\mathbf{q},\mathbf{p})\in \mathbb{F}^m_2\times\mathbb{F}^m_2$ for Pauli channel $\mathcal{E}^{\mathrm{pauli}}\in \CPTP(\mathcal{H}_{2,m},\mathcal{H}_{2,m})$ on the Choi state $(\mathcal{E}^{\mathrm{pauli}}_A\otimes\mathbb{I}_B)(|\Phi_2\rangle\langle\Phi_2|^{\otimes m}_{AB})$ and for displacement channels $\mathcal{E}^{\mathrm{disp}}\in \CPTP(\mathcal{H}_{\infty,m},\mathcal{H}_{\infty,m})$ this is the estimation of $e^{|\beta|^2e^{-2r}}\hat{D}(\beta)\otimes \hat{D}(\beta^*)$ on the state $(\mathcal{E}^{\mathrm{disp}}_A\otimes\mathbb{I}_B)(|\Phi_r^{\mathrm{TMSV}}\rangle\langle\Phi_r^{\mathrm{TMSV}}|^{\otimes m}_{AB})$ (note that the finite-energy requirement in the bosonic case necessitates the factor of $e^{|\beta|^2e^{-2r}}$). From this we can establish that there is always a channel which requires ancilla assistance with 1-copy access to be made efficient.

Restricting the exact channel learning task of Problem~\ref{prob:channel_learn} to CPTP maps acting on qubit systems, we find that this is the Pauli transfer matrix absolute value estimation problem studied in \cite{10.1145/3670418}. This work shows that the general channel $\mathcal{E}^{\mathrm{qubit}}\in \CPTP(\mathcal{H}_{2,m},\mathcal{H}_{2,m})$ requires $2$-copy access with ancillas for efficient learning, while even $1$-copy access with ancillas is insufficient for efficient learning showing it to be strictly harder than learning Pauli channels. To this end, Thms.~\ref{thm:qudit2quditinformalselfconj} and~\ref{thm:boson2bosoninformalselfconj} show that the same holds true for channels that are self-conjugate in the sense that $2$-copy access with ancillas is sufficient for efficient learning and $1$-copy access with ancillas does not admit efficient learning schemes. We find that this separation of $2$-copy access from $1$-copy access can be shown for channels of various possible input and output dimensions.

Finally we find that for all $d\geq 3$ being prime, there are learning tasks that do not have any efficient learning schemes if limited to $(d-1)$-copy access that do become efficient once given $d$-copy access which extends the known result in the state learning case to the channel learning scenario. We note that this extension is non-trivial since state learning can be shown to be a strictly easier task than the channel learning task we examine as we highlight by showing that state learning is the special case of channel learning where the channel is a complete replacement channel in Corollary~\ref{corr:statelearn}. 

Notably, the learning hierarchy established for qudit states in \cite{ller2025infinitehierarchymulticopyquantum} is also extended to the case of local dimension being a square-free integer. The key idea here is that using prime factorization $d = \prod_{i=1}^{f} d_i$ we can consider each $d$-level qudit to be composed of $f$ separate $d_i$-level qudits. Since the learning is defined over the composite system, one can only have the Heisenberg--Weyl group generators commute once given $d$ copies which is the lowest common multiple of the set $\{d_i\}$. For any $c<d$ copies, at least for one of the $i\in\{1,\dots,f\}$, $c$ would not be divisible by $d_i$ resulting in behavior similar to that in Eq.~\eqref{eq:Odmknorm}. This principle extends in a similar way to the channel learning scenario which can be seen as an extension of Thm.~\ref{thm:qudit2quditinformal} and is discussed in App.~\ref{app:squarefree}. As a result, we also have channel learning tasks that are $(d-1)$-copy hard but easy with $d$ copies for $d$ being square-free. For channel learning over $z$-level qudits where $z$ is no longer square-free, we show in App.~\ref{app:squarefree} that one requires $d$ copies for efficient learning where $d$ is the product of all the prime numbers that divide $z$ by the extension of this principle. Further, we also find a family of channels that are always hard to learn for any $c$-copy access as long as $c\epsilon = O(1)$ by examining the case of bosonic channel learning in Thm.~\ref{thm:boson2boson}. Fig.~\ref{fig:learninghierarchy} presents a visual representation of this learning hierarchy by showing the various sets of channels that are hard to learn for any $c$-copy learning protocol (as defined in Def~\ref{def:cmprotinform}).

The $\epsilon^{-4}$ scaling obtained in the upper bound of Thm.~\ref{thm:efficientlearnconj} for an unknown channel $\mathcal{E}$ with access to the channel $\mathcal{E}\otimes\mathcal{E}^*$ (using state learning ideas from \cite{coroi2025exponentialadvantagecontinuousvariablequantum,PRXQuantum.5.040301}) is tight (specifically for the absolute value estimation task), since we find a lower bound with the same scaling as the upper bound (see Thms.~\ref{thm:qudit2quditinformalwithconj} and~\ref{thm:boson2bosoninformalwithconj}). We also study the reduction of channel learning to state learning in Appendix~\ref{app:statelearn}. The $d$-copy learning scheme in \cite{ller2025infinitehierarchymulticopyquantum} has an $\epsilon^{-2d}$ scaling which we also find to be tight from the lower bound derived in Thm.~\ref{thm:qudit_statelearn} for the estimation along with phase information. We only claim tightness with respect to the $\epsilon$ scaling and leave extension to the scaling with respect to qudit dimension $d$ and number of qudits $m$ to future work. These lower bounds are derived using Corollary~\ref{corr:statelearn}, which follows from the master Lemma~\ref{lem:master}, further highlighting the strength and applicability of this lemma. We also find that the hardness for $c$-copy access to the state studied in \cite{coroi2025exponentialadvantagecontinuousvariablequantum} can have the lower bound improved by an exponential factor by applying Corollary~\ref{corr:statelearn}, which we note in Thm.~\ref{thm:boson_statelearn}.

\section{Conclusions and outlook}\label{sec:conclusion}
In this work we have attempted to exhaustively find lower bounds for the task of channel learning in the particular setting of limited parallel access to the unknown channel. Specifically, we considered learning protocols that are allowed $c$ copies of the unknown channel to use in parallel along with any amount of ancillary assistance and adaptive state preparation and measurement. We also consider a weaker task of Choi-state learning, which, while providing a complete description of the channel, has limitations that we discuss at length in Appendix~\ref{app:inadequacy}. Nevertheless, we find that even for this arguably easier learning task, this restriction of access to the channels can often lead to complexity lower bounds that scale with the dimension of the Hilbert space. We summarize some of our key findings in Table~\ref{tab:results}. Our main result of Lemma~\ref{lem:master} applies to channels regardless of input and output dimension, allowing us to find all our lower bounds for channels with bosonic output spaces as well and unifies existing results in state learning \cite{coroi2025exponentialadvantagecontinuousvariablequantum,ller2025infinitehierarchymulticopyquantum,PRXQuantum.5.040301}. We also generalize the principles used in \cite{chen2024optimaltradeoffsestimatingpauli} for their lower bounds for $c$-copy protocols by extending this to be independent of the dimensions of the Hilbert space. This allows for the measurements to also include POVMs that have infinite accuracy or uncountably many terms which requires subtlety to deal with since the operators describing these need not necessarily be bounded. We make a note of this carefully in our proof of Lemma~\ref{lem:master} in Appendix~\ref{app:cm_protocols}.

A more accurate way to characterize a quantum channel is tomographic learning, which produces a classical description of the channel with a guaranteed accuracy in diamond distance \cite{aharonov1998quantumcircuitsmixedstates}. This task has been explored in great detail in \cite{mele2026optimallearningquantumchannels,oufkir2026improvedlowerboundslearning} with tight lower and upper bounds on channel uses in the case of unitary operations \cite{10353133}. The upper bounds from \cite{mele2026optimallearningquantumchannels} for tomographic learning to diamond distance accuracy of $\epsilon$ for channels with a Kraus rank of $k$ have the number of channel uses scale as $O(d_{\mathrm{in}} d_{\mathrm{out}}k/\epsilon^2)$ for channels with finite-dimensional input and output Hilbert spaces of dimensions $d_{\mathrm{in}}$ and $d_{\mathrm{out}}$, respectively. Considering channel learning to be the estimation of a set of appropriate observables of some state, we find that there is hope in obtaining efficient complexity scalings in the setting of restricted parallel access. Due to the absolute value estimation for channel learning task in Problem~\ref{prob:channel_learn} being weaker than tomographic learning, we believe our methods can be readily extended to provide lower bounds for the tomographic learning problem in the restricted-parallel-access setting and leave this for future exploration.

The bosonic channel learning task we define in Problem~\ref{prob:channel_learn} requires some subtlety to show that it fully describes arbitrary channels (see Appendix~\ref{app:inadequacy}). Given that sets of bounded operators (such as the set of all phase space displacements) form a kind of orthogonal basis to represent quantum states, characterizing the way a channel transforms such sets of operators would yield an effective description. Unfortunately, we find that such an attempted transfer function turns out to be impossible to learn even when it is guaranteed to be bounded in value, which we detail in Thm.~\ref{thm:bosonnogo} of Appendix~\ref{app:bosonicnogo}. We note that the current bosonic channel learning tasks have tended to look at Gaussian channels \cite{fanizza2025efficientlearningbosonicgaussian} or displacement noise channels \cite{PhysRevLett.133.230604}. In the case of displacement noise channels, the learning task in \cite{PhysRevLett.133.230604} can be noted to be the same as the channel learning task in Problem~\ref{prob:channel_learn} along with phase information for only the queries of $\alpha = \beta$ and a multiplicative overhead of $e^{-e^{-2r}|\beta|^2}$. Notably phase information comes much easier here due to all displacements of form $\hat{D}(\beta)\otimes\hat{D}(-\beta^*)$ commuting with each other, making the task sample efficient with $1$-copy access. Ref.~\cite{fanizza2025efficientlearningbosonicgaussian} uses the fact that Gaussian unitaries have a finite description making it possible to learn by estimating $\text{Tr}[\hat{O}\mathcal{E}(\hat{\rho})]$ for some appropriate set of input states $\hat{\rho}$ and operators $\hat{O}$. The general bosonic channel however may not have a universal finite description making even defining the task seem somewhat challenging. We motivate our choice of definition for the absolute value estimation for channel learning problem in Problem~\ref{prob:channel_learn} by understanding the operator algebraic description for quantum channels in Appendix~\ref{app:operatoralgebra}. In addition, while results relating the trace distance between the Choi states of two channels and their diamond distances are known, we could not find an analogous result relating trace distance between TMSV-based Choi states and the energy-constrained diamond norm for bosonic channels. We answer this in Appendix~\ref{app:bosonicACID} by constructing
two entanglement-breaking channels that output only mixtures of Gaussian states, for which there is an exponential pre-factor between their TMSV based Choi-state trace distance and the energy-constrained diamond norm \cite{Shirokov2018,winter2017energyconstraineddiamondnormapplications}. We hope that this work can contribute to the discussion of finding a good metric for channel learning in the case of infinite-dimensional systems.

All of our lower bounds focus on the particular case of parallel use. As established in Ref.~\cite{PhysRevLett.127.200504}, sequential uses could greatly improve sample efficiency or ancilla requirements, as was found for Pauli channels in \cite{PRXQuantum.6.020323}. Notably, the case of sequential uses with unbounded ancilla is strictly stronger than parallel schemes with the same number of channel uses since one can always swap in fresh ancillas to process all the channel outputs in parallel. The case of interest would be that of sequential schemes with a bounded ancilla space. We leave the extension of our work to sequential protocols for future work.

\begin{acknowledgments}
We would like to thank Senrui Chen, Debayan Bandyopadhyay, and Mingxing Yao for useful discussions. We acknowledge support from the ARO (W911NF-23-1-0077), ARO MURI (W911NF-21-1-0325), AFOSR MURI (FA9550-21-1-0209, FA9550-23-1-0338), ONR MURI (N000142612102), DARPA (HR0011-24-9-0361), NSF (ERC-1941583, OMA-2137642, OSI-2326767, CCF-2312755, OSI-2426975). HK was supported by the IITP (RS-2025-25464252, RS-2024-00437191) and the NRF (RS-2025-25464492, RS-2024-00442710) funded by the Ministry of Science and ICT (MSIT), Korea.

This material is based upon work supported by the U.S. Department of Energy, Office of Science, National Quantum Information Science Research Centers and Advanced Scientific Computing Research (ASCR) program under contract number DE-AC02-06CH11357 as part of the InterQnet quantum networking project.

\textbf{A.I. use disclosure:} The LLM model Opus 5.0 developed by Anthropic was used for the proof-reading of this manuscript. The writing of this manuscript and the results of this work were derived solely by the authors without LLM assistance. The authors take full responsibility for the correctness of the presented work.
\end{acknowledgments}

\section{Data availability}
No data were created or produced in this work.
%

\onecolumngrid
\let\addcontentsline\addcontentslineorig

\appendix
\newpage
\tableofcontents
\begin{table}[ht]
\centering
\caption{Commonly used notation.}
\begin{tabular}{@{}|l|l|@{}}
\toprule
Notation & Explanation \\
\midrule
$\mathbb{F}_d^m$ & Vectors with $m$ entries with each entry being in Galois field $\mathbb{F}_d$ \\
$\mathcal{H}_{d,m}$ & Hilbert space of $m$ qudits with each having dimension $d$ \\
$\mathcal{H}_{\infty,m}$ & Hilbert space of $m$ bosonic modes.\\
$\mathbb{I}_{d,m}$ & Identity map $\mathcal{H}_{d,m}\to \mathcal{H}_{d,m}$.\\
$\mathbb{I}_{\infty,m}$ & Identity map $\mathcal{H}_{\infty,m}\to\mathcal{H}_{\infty,m}$.\\
$\|\cdot\|_{\mathrm{op}}$ & Operator norm for a linear operator $\mathcal{H}\to\mathcal{H}$ defined by the largest singular value.\\
$\|\cdot\|_1$ & $1$-norm which is the sum of all singular values of the linear operator.\\
$B(\mathcal{H})$ & The set of linear operators $\mathcal{H}\to\mathcal{H}$ that have bounded operator norm.\\
$L_1(\mathcal{H})$ & The set of linear operators $\mathcal{H}\to\mathcal{H}$ with bounded $1$-norm.\\
$D(\mathcal{H})$ & The set of linear operators $\mathcal{H}\to\mathcal{H}$ that represent a density operator. \\
$\|\cdot\|_{\diamond}$ & The diamond norm for superoperators \cite{aharonov1998quantumcircuitsmixedstates}. \\
$\|\cdot\|_{\diamond,E}$ & The energy-constrained diamond norm for infinite-dimensional superoperators \cite{winter2017energyconstraineddiamondnormapplications,Shirokov2018}.\\
$\CPTP(\mathcal{H}_{\mathrm{in}},\mathcal{H}_{\mathrm{out}})$ & The space of all valid completely positive trace-preserving linear maps $B(\mathcal{H}_{\mathrm{in}})\to B(\mathcal{H}_{\mathrm{out}})$.\\
$c$-copy & Learning protocol class with $c$-copy parallel access to the channel per step (Def.~\ref{def:cmprotinform}).\\
$\hat{X}_d,\,\hat{Z}_d$ & Qudit generalized shift and phase operators: $\hat{X}_d=\sum_{j}|j{+}1\rangle\langle j|$ and $\hat{Z}_d=\sum_j e^{2\pi i j/d}|j\rangle\langle j|$.\\
$\hat{D}_{d,m}(\mathbf{q},\mathbf{p})$, $(\mathbf{q},\mathbf{p})\in \mathbb{F}_d^m\times\mathbb{F}_d^m$ &  Phase space displacement operators for $m$ qudits of dimension $d$ (Eq.~\eqref{eq:disp_qudit_defn}).\\
$\hat{P}_d = \sum_{j\in\mathbb{F}_d}|-j\rangle\langle j|$ & Qudit parity flip operator satisfying $\hat{P}_d^{\otimes m}\hat{D}_{d,m}(\mathbf{q},\mathbf{p})\hat{P}_d^{\otimes m} = \hat{D}_{d,m}(\mathbf{q},\mathbf{p})^{\dagger}$ (Eq.~\eqref{eq:parityqudit}).\\
$\hat{D}(\alpha)$, $\alpha\in \mathbb{C}^m$ & Phase space displacement operators for $m$ bosonic modes (Eq.~\eqref{eq:bosonic_disp}).\\
$\hat{a}_i,\,\hat{a}_i^{\dagger}$ & Bosonic annihilation and creation operators for mode $i$.\\
$\hat{n}_i = \hat{a}_i^{\dagger}\hat{a}_i,\quad \hat{n}=\sum_{i=1}^{m}\hat{n}_i$ & Bosonic number operator for mode $i$ and total number operator.\\
$\Omega(\alpha,\beta) = (\alpha^{\dagger}\beta - \beta^{\dagger}\alpha)/i$ & Symplectic inner product on the bosonic phase space $\mathbb{C}^m$.\\
$\hat{E}_{d,m}(\mathbf{q},\mathbf{p})$, $(\mathbf{q},\mathbf{p})\in \mathbb{F}_d^m\times\mathbb{F}_d^m$ & The operator given by $\hat{E}_{d,m}(\mathbf{q},\mathbf{p}) = \frac{e^{i\pi/4}}{\sqrt{2}}\hat{D}_{d,m}(\mathbf{q},\mathbf{p}) + \frac{e^{-i\pi/4}}{\sqrt{2}}\hat{D}_{d,m}(-\mathbf{q},-\mathbf{p})$.\\
$\hat{E}(\gamma)$, $\gamma\in \mathbb{C}^m$ & The operator given by $\hat{E}(\gamma)=\frac{e^{i\pi/4}}{\sqrt{2}}\hat{D}(\gamma) + \frac{e^{-i\pi/4}}{\sqrt{2}}\hat{D}(-\gamma)$.\\
$\ket{\Phi_d}$ & The entangled 2-qudit ($d$ level) state $\ket{\Phi_d} = d^{-1/2}\sum_{i\in \mathbb{F}_d}\ket{i}\ket{i}$.\\
$\ket{\Phi^{\mathrm{TMSV}}_r}$ & Two-mode squeezed vacuum over 2 bosonic modes $\ket{\Phi_r^{\mathrm{TMSV}}} = \mathrm{sech}(r)\sum_{i=0}^{\infty}(\tanh(r))^i\ket{i}\ket{i}$.\\
$\mathcal{E}^*$ & Complex-conjugate channel with transpose defined in a fixed basis.\\
$C_{\mathcal{E}}$ & Transfer function description for channel $\mathcal{E}$ (Problem~\ref{prob:channel_learn}).\\
$C_{\mathcal{E}}^{\mathrm{TMSV},r}(\alpha,\beta)$ & TMSV-based channel learning function for bosonic channels (Eq.~\eqref{eq:CETMSVr}).\\
$\chi_{\hat{\rho}}(\alpha) = \mathrm{Tr}[\hat{\rho}\,\hat{D}(\alpha)]$ & Characteristic function of a quantum state $\hat{\rho}$.\\
$\mathcal{E}_{\mathbf{u},\mathbf{v}}$ & A parameterized CPTP map which takes form of Eq.~\eqref{eq:Epar1par2}.\\
$\Gamma(n,x)$ & The upper incomplete gamma function given by $\Gamma(n,x) = \int_{x}^{\infty} t^{n-1}e^{-t}dt$ and $\Gamma(n,0) = \Gamma(n)$.\\
$Q(n,x)$ & Regularized gamma function $Q(n,x) = \frac{\Gamma(n,x)}{\Gamma(n)}$.\\
\bottomrule
\end{tabular}
\label{app:notation}
\end{table}

\section{Mathematical preliminaries}

\subsection{Displacement operators for qudits}\label{app:qudit_disp}
We revisit the displacement operations $\hat{D}_{d,m}(\mathbf{q},\mathbf{p})$ for both $\mathbf{q},\mathbf{p}$ being in $\mathbb{F}^m_d$ defined in Eq.~\eqref{eq:disp_qudit_defn} that act on the $m$ qudit ($d$-level) Hilbert space $\mathcal{H}_{d,m}$. By defining operator conjugation in the basis $\{\ket{j}\}$ which diagonalizes $\hat{Z}$, we have
\begin{equation}
    \hat{D}_{d,m}(\mathbf{q},\mathbf{p})^*  = \hat{D}_{d,m}(\mathbf{q},-\mathbf{p}),\quad \hat{D}_{d,m}(\mathbf{q},\mathbf{p})^T = \hat{D}_{d,m}(-\mathbf{q},\mathbf{p}).
\end{equation}
Defining the parity flip operator $\hat{P}_d$, we obtain
\begin{equation}\label{eq:parityqudit}
    \hat{P}_d = \sum_{j\in\mathbb{F}_d}|-j\rangle\langle j|,\quad \hat{P}_d^{\otimes m}\hat{D}_{d,m}(\mathbf{q},\mathbf{p})\hat{P}_d^{\otimes m} = \hat{D}_{d,m}(\mathbf{q},\mathbf{p})^\dagger.
\end{equation}
Considering the maximally entangled state $\ket{\Phi_d} = \frac{1}{\sqrt{d}}\sum_{j\in\mathbb{F}_d}|j\rangle|j\rangle$, we have the following
\begin{equation}\label{eq:bellstatequdit}
    (|\Phi_d\rangle\langle\Phi_d|)^{\otimes m} = \frac{1}{d^{2m}}\sum_{\mathbf{q},\mathbf{p}\in\mathbb{F}^m_d}\hat{D}_{d,m}(\mathbf{q},\mathbf{p})\otimes \hat{D}_{d,m}(\mathbf{q},-\mathbf{p})
\end{equation}
which further gives us a complete POVM describing Bell-basis measurement given by
\begin{equation}
    \hat{\Pi}_{\mathbf{a},\mathbf{b}} = (\hat{D}_{d,m}(\mathbf{a},\mathbf{b})\otimes \mathbb{I})(|\Phi_d\rangle\langle\Phi_d|)^{\otimes m}(\hat{D}_{d,m}(\mathbf{a},\mathbf{b})\otimes \mathbb{I})^\dagger =  \frac{1}{d^{2m}}\sum_{\mathbf{q},\mathbf{p}\in\mathbb{F}^m_d}e^{\frac{2\pi i}{d}(\mathbf{q}\cdot\mathbf{b} - \mathbf{a}\cdot\mathbf{p})}\hat{D}_{d,m}(\mathbf{q},\mathbf{p})\otimes \hat{D}_{d,m}(\mathbf{q},-\mathbf{p})\label{eq:bellPOVMqudit}.
\end{equation}
Observe that for the state $\hat{\sigma}\in D(\mathcal{H}_{d,m}^{\otimes 2})$, we have
\begin{equation}
    \text{Tr}\left[(\hat{D}_{d,m}(\mathbf{q},\mathbf{p})\otimes \hat{D}_{d,m}(\mathbf{q},-\mathbf{p}))\hat{\sigma}\right] = \sum_{\mathbf{a},\mathbf{b}}e^{\frac{2\pi i}{d}(\mathbf{a}\cdot\mathbf{p}-\mathbf{q}\cdot\mathbf{b} )}p(\mathbf{a},\mathbf{b}),
\end{equation}
where $p(\mathbf{a},\mathbf{b}) = \text{Tr}\left[\hat{\Pi}_{\mathbf{a},\mathbf{b}}\hat{\sigma}\right]$ is the probability distribution obtained by measuring the state $\hat{\sigma}$ using the Bell-basis measurement. Treating $\mathbf{a},\mathbf{b}$ as random samples from the probability distribution $p(\mathbf{a},\mathbf{b})$, $e^{\frac{2\pi i}{d}(\mathbf{a}\cdot\mathbf{p}-\mathbf{q}\cdot\mathbf{b})}$ is an unbiased estimator for the value $\text{Tr}\left[(\hat{D}_{d,m}(\mathbf{q},\mathbf{p})\otimes \hat{D}_{d,m}(\mathbf{q},-\mathbf{p}))\hat{\sigma}\right]$. Using this fact (as observed also in \cite{PRXQuantum.5.040301}) we state the following.
\begin{lemma}\label{lem:effestquditchar}
    Using the qudit Bell measurement scheme (POVM defined by $\{\hat{\Pi}_{\mathbf{a},\mathbf{b}}\}$ in Eq.~\eqref{eq:bellPOVMqudit}), we can produce an estimate $\tilde{\chi}_{\hat{\sigma}}(\mathbf{q}_j,\mathbf{p}_j)$ for $M$ different $(\mathbf{q}_j,\mathbf{p}_j)\in\mathbb{F}_{d}^m\times \mathbb{F}_{d}^m$ (for $j=1,\dots,M$) to the function $\chi_{\hat{\sigma}}(\mathbf{q},\mathbf{p}) = \text{Tr}\left[(\hat{D}_{d,m}(\mathbf{q},\mathbf{p})\otimes \hat{D}_{d,m}(\mathbf{q},-\mathbf{p}))\hat{\sigma}\right]$ such that for all $j$, $|\tilde{\chi}_{\hat{\sigma}}(\mathbf{q}_j,\mathbf{p}_j) - \chi_{\hat{\sigma}}(\mathbf{q}_j,\mathbf{p}_j)| \leq \epsilon$ with success probability $1-\delta$ using $N$ copies of $\hat{\sigma}$ where $N = O(\epsilon^{-2}\log(M/\delta))$.
\end{lemma}
\begin{proof}
    This statement is proven in \cite{PRXQuantum.5.040301}, for which we restate a short proof for the sake of completeness. Note that if the algorithm can succeed in providing an $\epsilon$-accurate estimate for each query independently with probability at least $1-\frac{\delta}{M}$, by the union bound the success probability for succeeding for all queries together will be lower bounded by $1-\delta$.
    Suppose we obtain $N$ samples from the probability distribution $p(\mathbf{a},\mathbf{b})$ labeled as $(\mathbf{a}_k,\mathbf{b}_k)$ (for $k=1,\dots,N$) by separately measuring each copy of $\hat{\sigma}$. Define the estimators
    \begin{equation}
        \hat\chi_{\mathrm{re}}(\mathbf{q},\mathbf{p}) = \frac{1}{N}\sum_{k=1}^N\Re(e^{\frac{2\pi i}{d}(\mathbf{a}_k\cdot\mathbf{p}-\mathbf{q}\cdot\mathbf{b}_k)}),\quad \hat{\chi}_{\mathrm{im}}(\mathbf{q},\mathbf{p}) = \frac{1}{N}\sum_{k=1}^N\Im(e^{\frac{2\pi i}{d}(\mathbf{a}_k\cdot\mathbf{p}-\mathbf{q}\cdot\mathbf{b}_k)}).
    \end{equation}
    By the Hoeffding inequality, we have 
    \begin{equation}
        \Pr\left(|\hat\chi_{\mathrm{re}}(\mathbf{q},\mathbf{p}) - \Re(\chi(\mathbf{q},\mathbf{p}))|\geq\frac{\epsilon}{\sqrt{2}}\right) \leq 2\exp(-N\epsilon^2 /4),\quad \Pr\left(|\hat\chi_{\mathrm{im}}(\mathbf{q},\mathbf{p}) - \Im(\chi(\mathbf{q},\mathbf{p}))|\geq\frac{\epsilon}{\sqrt{2}}\right) \leq 2\exp(-N\epsilon^2 / 4),
    \end{equation}
    and further we have that
    \begin{equation}
    \begin{aligned}
        &\Pr\left(|\hat\chi_{\mathrm{re}}(\mathbf{q},\mathbf{p}) +i\hat\chi_{\mathrm{im}}(\mathbf{q},\mathbf{p}) - \chi(\mathbf{q},\mathbf{p})|\leq\epsilon\right)\\
        &\ge         \Pr\left(|\hat\chi_{\mathrm{re}}(\mathbf{q},\mathbf{p}) - \Re(\chi(\mathbf{q},\mathbf{p}))|\leq\frac{\epsilon}{\sqrt{2}},~~|\hat\chi_{\mathrm{im}}(\mathbf{q},\mathbf{p}) - \Im(\chi(\mathbf{q},\mathbf{p}))|\leq\frac{\epsilon}{\sqrt{2}}\right)\\
        & \ge 1- 4\exp(-N\epsilon^{2}/4)        .
    \end{aligned}\nonumber
    \end{equation}
    Hence to be able to succeed with probability $1-\delta/M$ for an arbitrary query we obtain
    \begin{equation}
        \exp(-N\epsilon^{2}/4) \leq \frac{\delta}{4M} \implies N\geq \frac{4}{\epsilon^2}\log(\frac{M}{\delta})+\frac{4\log(4)}{\epsilon^2} = O(\epsilon^{-2}\log(M/\delta)).
    \end{equation}
\end{proof}

A key feature allowing for this efficient estimation strategy lies in the fact that the set of operators $\{\hat{D}_{d,m}(\mathbf{q},\mathbf{p})\otimes \hat{D}_{d,m}(\mathbf{q},-\mathbf{p})\}$ all commute with each other. Note that the operators $\{\hat{D}_{d,m}(\mathbf{q},\mathbf{p})^{\otimes d}\}$ also commute with each other. We consider the following entangled states over $d$ qudits each of local dimension $d$ as defined in Sec. 4 of \cite{ller2025infinitehierarchymulticopyquantum}:
\begin{equation}
    \ket{\phi_{I,t}} = \frac{1}{\sqrt{d}}\sum_{k\in\mathbb{F}_d}e^{\frac{2\pi i}{d}kt}|(0,I)+k\mathbf{1}\rangle,\quad t\in\mathbb{F}_d,\ I\in\mathbb{F}_d^{d-1}.
\end{equation}
The above states form an orthonormal basis for $\mathcal{H}_{d,d}$ and are eigenstates of all $\hat{D}_{d,1}(q,p)^{\otimes d}$ for $q,p\in\mathbb{F}_d$ as we can observe that (using Eqs. (11) and (12) of \cite{ller2025infinitehierarchymulticopyquantum})
\begin{equation}
    \hat{D}_{d,1}(q,p)^{\otimes d}\ket{\phi_{I,t}} = e^{i\pi dqp/d}(\hat{X}^{q}\hat{Z}^{p})^{\otimes d}\ket{\phi_{I,t}} = (-1)^{qp}e^{\frac{2\pi i}{d}(p|I| - qt)}\ket{\phi_{I,t}},
\end{equation}
which allows us to define the complete POVM set for $\mathbf{s},\mathbf{t}\in\mathbb{F}^m_d$
\begin{equation}
    \hat{\Pi}_{\mathbf{s},\mathbf{t}} = \bigotimes_{j=1}^m\left(\sum_{I\in\mathbb{F}^{d-1}_d,|I| = s_j}|\phi_{I,t_j}\rangle\langle\phi_{I,t_j}|\right).
\end{equation}
By performing this measurement over the product state $\hat{\rho}^{\otimes d}$ for $\hat{\rho}$ being a state in $\mathcal{H}_{d,m}$, we have
\begin{equation}
    \text{Tr}\left[\hat{D}_{d,m}(\mathbf{q},\mathbf{p})^{\otimes d}\hat{\rho}^{\otimes d}\right] = \sum_{\mathbf{s},\mathbf{t}}(-1)^{\mathbf{q}\cdot\mathbf{p}}e^{\frac{2\pi i}{d}(\mathbf{p}\cdot\mathbf{s}-\mathbf{q}\cdot\mathbf{t})}\text{Tr}[\hat{\Pi}_{\mathbf{s},\mathbf{t}}\hat{\rho}^{\otimes d}].
\end{equation}
The method for implementing this POVM is described in Thm. 10 of \cite{ller2025infinitehierarchymulticopyquantum} where for $m$ qudits this can be implemented in constant depth using single- and two-qudit Clifford gates. Using this, we now restate the efficient estimation result from \cite{ller2025infinitehierarchymulticopyquantum}.
\begin{lemma}\label{lem:dcopylearning}
    Using the POVM set defined by $\{\hat{\Pi}_{\mathbf{s},\mathbf{t}}\}$ over the product state $\hat{\rho}^{\otimes d}$, we can provide an estimate $\hat{f}(\mathbf{q},\mathbf{p})$ of $\left(\text{Tr}\left[\hat{D}_{d,m}(\mathbf{q},\mathbf{p})\hat{\rho}\right]\right)^d$ for $M$ queries such that for each of the queries $(\mathbf{q}_j,\mathbf{p}_j)\in\mathbb{F}_d^m\times \mathbb{F}_d^m$, $\left|\hat{f}(\mathbf{q}_j,\mathbf{p}_j) - \left(\text{Tr}\left[\hat{D}_{d,m}(\mathbf{q}_j,\mathbf{p}_j)\hat{\rho}\right]\right)^d\right|\leq \epsilon$ with probability $1-\delta$ using a total of $O(\epsilon^{-2}\log(M/\delta))$ measurements over $\hat{\rho}^{\otimes d}$ and hence a sample complexity of $O(d\epsilon^{-2}\log(M/\delta))$ uses of state $\hat{\rho}$.
\end{lemma}
\begin{proof}
    We use the union bound along with the Hoeffding bound to obtain this result for the estimation of $\text{Tr}\left[\hat{D}_{d,m}(\mathbf{q},\mathbf{p})^{\otimes d}\hat{\rho}^{\otimes d}\right] = \left(\text{Tr}\left[\hat{D}_{d,m}(\mathbf{q},\mathbf{p})\hat{\rho}\right]\right)^d$ by sampling from the probability distribution $p(\mathbf{s},\mathbf{t}) = \text{Tr}\left[\hat{\Pi}_{\mathbf{s},\mathbf{t}}\hat{\rho}^{\otimes d}\right]$.
\end{proof}
Note that all the above protocols are non-adaptive and only use knowledge of the queries to process data obtained after measurements. Alternatively, one can estimate a single query using a specifically designed circuit similar to that of \cite{PhysRevLett.125.043602}. Consider the following two-outcome POVM for prime $d$:
\begin{equation}
    \hat{\Pi}_{\pm,\theta,(\mathbf{q},\mathbf{p})}^{(d,m)} = \frac{\mathbb{I} \pm \frac{e^{-i\theta}\hat{D}_{d,m}(\mathbf{q},\mathbf{p})+e^{i\theta}\hat{D}^\dagger_{d,m}(\mathbf{q},\mathbf{p})}{2}}{2}\label{eq:qudit2outcome}.
\end{equation}
By performing this two-outcome POVM on a state $\hat{\rho}$ in $D(\mathcal{H}_{d,m})$, we obtain
\begin{equation}
    p(+)-p(-)=\text{Tr}[\hat{\Pi}_{+,\theta,(\mathbf{q},\mathbf{p})}\hat{\rho}] - \text{Tr}[\hat{\Pi}_{-,\theta,(\mathbf{q},\mathbf{p})}\hat{\rho}] = \Re(e^{-i\theta}\text{Tr}[\hat{D}_{d,m}(\mathbf{q},\mathbf{p})\hat{\rho}]).
\end{equation}
which provides a direct estimator for the values of $\Re(\text{Tr}[\hat{D}_{d,m}(\mathbf{q},\mathbf{p})\hat{\rho}])$ and $\Im(\text{Tr}[\hat{D}_{d,m}(\mathbf{q},\mathbf{p})\hat{\rho}])$ by choosing $\theta = 0$ and $\theta = \pi/2$ respectively. 
It can be checked that the circuit depicted in Fig.~\ref{fig:hardtolearn}(a) implements this POVM with $\hat{\Pi}_{\pm,\theta,(\mathbf{q},\mathbf{p})}^{(d,m)}$ for $d > 2$ (allowing for $2^{-1}\mod d$ to be defined) corresponding to the ancillary qubit measurement of $\pm1$ eigenstate of $\hat{Z}$. Examining the case of $\theta = -\pi/4$, we obtain a useful set of Hermitian operators proportional to the difference of the two POVM operators given by
\begin{equation}
    \hat{E}_{d,m}(\mathbf{q},\mathbf{p}) = \frac{1+i}{2}\hat{D}_{d,m}(\mathbf{q},\mathbf{p}) + \frac{1-i}{2}\hat{D}_{d,m}(-\mathbf{q},-\mathbf{p})\label{eq:quditHWobs},
\end{equation}
which also happen to be orthogonal to each other with respect to the Hilbert--Schmidt inner product
\begin{equation}
    \text{Tr}\left[\left(\hat{E}_{d,m}(\mathbf{q}',\mathbf{p}')\right)^\dagger\hat{E}_{d,m}(\mathbf{q},\mathbf{p})\right]  = d^m\delta_{\mathbf{q},\mathbf{q}'}\delta_{\mathbf{p},\mathbf{p}'}.
\end{equation}
These operators were introduced in \cite{PhysRevA.94.010301} as a way to represent qudit states using a real-valued vector representation in terms of the above $d^{2m}$ operators similar to Eq.~\eqref{eq:mquditchar}. 

Notably, the above definition of Heisenberg--Weyl observables also holds valid for $d = 2$ where $\hat{E}_{2,m}(\mathbf{q},\mathbf{p}) = \hat{D}_{2,m}(\mathbf{q},\mathbf{p})$ due to the displacement operators already being Hermitian. In the case of $d = 2$, the set of operators $\hat{D}_{2,m}(\mathbf{q},\mathbf{p})$ corresponds to elements of the $m$-qubit Pauli group. As such, there is always a valid POVM given by $\hat{\Pi}_{\pm,(\mathbf{q},\mathbf{p})}^{(2,m)} = \frac{\mathbb{I}\pm \hat{D}_{2,m}(\mathbf{q},\mathbf{p})}{2}$. Since $\hat{D}_{2,m}(\mathbf{q},\mathbf{p})$ is a Pauli operator, there is a Clifford operation whose conjugation maps it to $\hat{Z}_2^{\otimes m}$, where $\hat{Z}_2$ is the standard single-qubit Pauli-$Z$ operator, hence showing a way to also implement this operation using a Clifford unitary followed by computational basis measurement. If we consider the fuzzy observable \cite{PhysRevD.33.2253} of $\cos(\theta)\hat{D}_{2,m}(\mathbf{q},\mathbf{p})$, this is measured by $\hat{\Pi}_{\pm,\theta,(\mathbf{q},\mathbf{p})}^{(2,m)}$. This is equivalent to considering a classical noise channel on the measurement outcome showing a way to perform this measurement for the case of $d =2$. For $d\geq3$ we are required to scale down the observable $\hat{E}_{d,m}(\mathbf{q},\mathbf{p})$ by $\sqrt{2}$ to ensure that the elements of the POVM $\hat{\Pi}_{\pm,-\pi/4,(\mathbf{q},\mathbf{p})}^{(d,m)}$ are positive operators. Hence we have shown that the POVM defined by $\hat{\Pi}_{\pm,\theta,(\mathbf{q},\mathbf{p})}^{(d,m)}$ can be validly implemented for all prime $d$. We make use of this in the hardness proofs for multi-qudit channels.

\subsection{Displacement operators for bosonic modes}\label{app:bosonic_disp}
The bosonic $m$-mode phase-space displacements $\hat{D}(\alpha)$ for $\alpha\in\mathbb{C}^m$ (defined in Eq.~\eqref{eq:bosonic_disp}) satisfy the following relations
\begin{equation}
    \hat{D}^\dagger({\alpha}) = \hat{D}(-{\alpha}),\quad \hat{D}^{*}({\alpha}) = \hat{D}({\alpha}^{*}),\quad \hat{D}({\alpha})\hat{D}({\beta}) = \hat{D}({\alpha}+{\beta})e^{\frac{1}{2}i\Omega({\beta},{\alpha})},\quad \hat{D}(\beta)\hat{D}(\alpha)\hat{D}^\dagger(\beta) = \hat{D}({\alpha})e^{i\Omega({\alpha},\beta)}.
\end{equation}
We will also make use of the resolution of the identity operation,
\begin{equation}
    \mathbb{I} = \frac{1}{\pi^m}\int d^{2m}{\alpha} |{\alpha}\rangle\langle{\alpha}|.
\end{equation}
We note the Fourier transform and inverse Fourier transform relation
\begin{equation}
    \tilde{f}(\tilde{{\alpha}}) = \int d^{2m}{\alpha}e^{i\Omega({\alpha},\tilde{{\alpha}})}f({\alpha}),\quad f({\alpha}) = \frac{1}{\pi^{2m}}\int d^{2m}\tilde{{\alpha}}e^{i\Omega(\tilde{{\alpha}},{\alpha})}\tilde{f}(\tilde{{\alpha}}).
\end{equation}
We also note the Fourier relation for Gaussian functions which we will be making extensive use of
\begin{equation}
    \int d^{2m}{\alpha}e^{i\Omega({\alpha},{\beta})}\exp(-\frac{|{\alpha}-{\mu}|^2}{\sigma^2}) = (\pi\sigma^2)^me^{-\sigma^2|{\beta}|^2}e^{i\Omega({\mu},{\beta})}.
\end{equation}
All CPTP channels mapping $m$-mode states to $m$-mode states can be written as $\mathcal{E}(\cdot) = \sum_{i}\hat{K}_i(\cdot)\hat{K}_i^\dagger$ as their Kraus decomposition. This can also be written as
\begin{equation}
    \mathcal{E}(\cdot) = \int d^{2m}{\alpha_1}d^{2m}{\alpha_2} K_{\mathcal{E}}({\alpha_1,{\alpha}_2}) \hat{D}({\alpha}_1)(\cdot) \hat{D}^\dagger({\alpha}_2),
\end{equation}
where the function $K_{\mathcal{E}}({\alpha}_1,{\alpha}_2)$ has tempered-distribution-like behavior. The Hermiticity of output states and the trace-preserving conditions result in 
\begin{equation}
    K_{\mathcal{E}}({\alpha}_1,{\alpha}_2) = (K_{\mathcal{E}}({\alpha_2},{\alpha}_1))^{*},\quad\int d^{2m}{\alpha}_2 K_{\mathcal{E}}({\alpha}_1+{\alpha}_2,{\alpha}_2)e^{\frac{1}{2}i\Omega({\alpha}_2,{\alpha}_1)} = \delta^{(2m)}({\alpha}_1).
\end{equation}
Note that the trace-preserving condition can be equivalently rewritten as
\begin{equation}
    \int d^{2m}{\xi} K_{\mathcal{E}}\left({\xi}+\frac{{\alpha}}{2},{\xi}-\frac{{\alpha}}{2}\right) e^{\frac{1}{2}i\Omega({\xi,{\alpha}})} = \delta^{(2m)}({\alpha}).
\end{equation}

An important state we will use is the $2m$-mode generalization of the two-mode squeezed vacuum which for squeezing of $r$ has the characteristic function
\begin{equation}\label{eq:TMSVchardefn}
    |\Psi_r\rangle\langle\Psi_r| = \frac{1}{\pi^{2m}}\int d^{2m}{\omega}_1 d^{2m}{\omega}_2 g_r({\omega}_1,{\omega}_2) \hat{D}^\dagger({\omega}_1)\otimes\hat{D}^\dagger({\omega}_2),
\end{equation}
\begin{equation}
    g_r({\omega}_1,{\omega}_2) = \exp(-\frac{1}{4}\left(e^{2r}|{\omega}_1-{\omega}_2^{*}|^2 + e^{-2r}|{\omega}_1+{\omega}_2^{*}|^2\right)).
\end{equation}
A useful observation about $g_r$ we make is as follows
\begin{equation}
    g_r({\omega}_1,{\omega}_2^{*}) = g_r({\omega}_2,{\omega}_1^{*}) = \exp(-\frac{|{\omega}_2|^2}{2\cosh(2r)})\exp(-\frac{\cosh(2r)}{2}|{\omega}_1 - {\omega}_2\tanh(2r)|^2).
\end{equation}
Taking the limit of $r\to\infty$, we can write down an (unnormalized) maximally entangled $2m$-mode state as
\begin{equation}
    |\Psi\rangle\langle\Psi| = \frac{1}{\pi^m}\int d^{2m}{\omega} \hat{D}({\omega})\otimes\hat{D}({\omega}^{*}),
\end{equation}
using which we can define the POVM set $\{\hat{\Pi}({\zeta})\}_{{\zeta}\in\mathbb{C}^m}$ given by
\begin{equation}\label{eq:bellPOVMboson}
    \hat{\Pi}({\zeta}) = \frac{1}{\pi^m}(\hat{D}({\zeta})\otimes\mathbb{I})|\Psi\rangle\langle\Psi|(\hat{D}^\dagger({\zeta})\otimes\mathbb{I}) = \frac{1}{\pi^{2m}}\int d^{2m}{\alpha}e^{i\Omega(\alpha,{\zeta})}\hat{D}({\alpha})\otimes\hat{D}({\alpha}^{*}).
\end{equation}
For a $2m$-mode state $\hat{\sigma}$, we can obtain specific points of its characteristic function by
\begin{equation}
    \chi_{\sigma}\left(\begin{pmatrix}
        {\alpha}\\
        {\alpha}^{*}
    \end{pmatrix}\right) = \int d^{2m}{\zeta} e^{i\Omega(\zeta,{\alpha})}p({\zeta}),
\end{equation}
where $p({\zeta})= \text{Tr}(\hat{\Pi}({\zeta})\hat{\sigma})$ is the probability associated with the POVM. Hence we can always define an unbiased estimator 
\begin{equation}
    \hat{\chi}_{\sigma}\left(\begin{pmatrix}
        {\alpha}\\{\alpha}^{*}
    \end{pmatrix}\right) = \frac{1}{N}\sum_{i=1}^N e^{i\Omega(\zeta_i,{\alpha})}\label{eq:bosoniccharestimator},
\end{equation}
by taking samples ${\zeta}_i$ from the distribution $p({\zeta})$. This is key to the polynomial sample complexity in \cite{coroi2025exponentialadvantagecontinuousvariablequantum,PhysRevLett.133.230604}, and is stated as Lemma~\ref{lem:effestbosonchar} below.
\begin{lemma}\label{lem:effestbosonchar}
    Using the Bell measurement scheme, we can estimate $M$ different points ${\alpha}_i$ for the function $\chi_{\hat{\sigma}}\left(\begin{pmatrix}
        {\alpha}\\{\alpha}^{*}
    \end{pmatrix}\right)$ up to $\epsilon$ error with success probability $1-\delta$ with $N$ copies of $\hat{\sigma}$ where $N = \mathcal{O}(\epsilon^{-2}\log(M/\delta))$.
\end{lemma}
\begin{proof}
    The proof of Lemma~\ref{lem:effestquditchar} directly generalizes here using the estimator defined in Eq.~\eqref{eq:bosoniccharestimator}.
\end{proof}
Similar to the POVM defined in Eq.~\eqref{eq:qudit2outcome}, we can define the following POVM which is exactly implemented by the circuit in Fig.~1 of \cite{PhysRevLett.125.043602} as
\begin{equation}\label{eq:bosonic2outcome}
    \hat{\Pi}_{\pm,\theta,\gamma} = \frac{\mathbb{I}\pm\frac{e^{-i\theta}\hat{D}(\gamma)+e^{i\theta}\hat{D}^\dagger(\gamma)}{2}}{2}.
\end{equation}
Through the use of this POVM, Ref.~\cite{PhysRevLett.125.043602} performs characteristic function tomography by noting
\begin{equation}
    p(+) - p(-) = \text{Tr}\left[\hat{\Pi}_{+,\theta,\gamma}\hat{\rho}\right] - \text{Tr}\left[\hat{\Pi}_{-,\theta,\gamma}\hat{\rho}\right] = \Re(\chi_{\hat{\rho}}(\gamma)e^{-i\theta}).
\end{equation}
Similar to Eq.~\eqref{eq:quditHWobs}, we can also define an orthogonal set of Hermitian operators as
\begin{equation}
    \hat{E}(\gamma) = \frac{1+i}{2}\hat{D}(\gamma) + \frac{1-i}{2}\hat{D}^\dagger(\gamma),\quad \text{Tr}\left[\hat{E}^\dagger(\gamma')\hat{E}(\gamma)\right] = \pi^m\delta^{(2m)}(\gamma-\gamma').
\end{equation}

\subsection{Operator algebra for describing quantum channels}\label{app:operatoralgebra}
In this subsection, we offer a brief introduction to operator algebra, heavily based on \cite{gupta2015functional}.
\begin{definition}[Linear operators over a Hilbert space]
    Consider a Hilbert space $\mathcal{H}$ equipped with inner product $\langle\cdot|\cdot\rangle$. We define the following sets of linear operators $\mathcal{H}\to\mathcal{H}$:
    \begin{itemize}
        \item $L(\mathcal{H})$ is the set of linear operators $f:\mathcal{H}\to\mathcal{H}$. This includes unbounded and discontinuous operators. 
        \item $B(\mathcal{H}) = \{\hat{O}\in L(\mathcal{H})|\text{ }\|\hat{O}\|_{\mathrm{op}} < \infty\}$ where $\|\hat{O}\|^2_{\mathrm{op}} = \sup_{\ket{\psi}\in\mathcal{H},\langle\psi|\psi\rangle=1}\langle\psi|\hat{O}^\dagger\hat{O}|\psi\rangle$. The space $(B(\mathcal{H}),\|\cdot\|_{\mathrm{op}})$ is a Banach space and so is complete under this norm. It is also a von Neumann algebra and a $C^*$-algebra. This admits an adjoint operation over itself where $\bra{\phi}\hat{O}^\dagger\ket{\psi} = (\bra{\psi}\hat{O}\ket{\phi})^*$ for all $\ket{\psi},\ket{\phi}\in\mathcal{H}$.
        \item We define the set $L_1(\mathcal{H}) = \{\hat{\rho}\in B(\mathcal{H})|\text{ }\|\hat{\rho}\|_1 < \infty\}$ where $\|\hat{\rho}\|_1 = \text{Tr}(\sqrt{\hat{\rho}^\dagger\hat{\rho}})$ which is the trace norm. This is also referred to as trace-class and corresponds to the set of nuclear operators acting on $\mathcal{H}$. The space $(L_1(\mathcal{H}),\|\cdot\|_1)$ is a Banach space but is not a $C^*$-algebra due to the properties of the trace norm.
        \item An operator $\hat{A}\in B(\mathcal{H})$ is positive if for all $\ket{\psi}\in\mathcal{H}$, $\langle\psi|\hat{A}|\psi\rangle \geq 0$. The set of positive operators in $B(\mathcal{H})$ will be represented by $B(\mathcal{H})_+$. We similarly define $L_1(\mathcal{H})_+$ and the operators in $L_1(\mathcal{H})_+$ with unit trace define the set $D(\mathcal{H})$ (density operators).
    \end{itemize}
\end{definition}

\begin{remark}
    For all operators $\hat{O}\in B(\mathcal{H})$ and $\hat{\rho}\in L_1(\mathcal{H})$, the operations $\hat{\rho}\hat{O}$ and $\hat{O}\hat{\rho}$ are contained in $L_1(\mathcal{H})$. The space $B(\mathcal{H})$ is isomorphic to the continuous dual space of $L_1(\mathcal{H})$ as can be seen by the continuity of the function $f_{\hat{O}}(\cdot) = \text{Tr}(\hat{O}(\cdot))$ over $L_1(\mathcal{H})$.
\end{remark}

Note that the Banach spaces $B(\mathcal{H})$ and $L_1(\mathcal{H})$ contain the set of unitary operations and the set of density matrices, respectively. We now consider a minimal description for a set of quantum channels.
\begin{definition}
    We define the following sets of linear maps over spaces and subspaces of $B(\mathcal{H})$:
    \begin{itemize}
        \item $L(B(\mathcal{H}))$ is the set of linear maps $f:B(\mathcal{H})\to B(\mathcal{H})$. While this is a vector space, it does not admit a canonical norm.
        \item $B(B(\mathcal{H})) = \{\Phi\in L(B(\mathcal{H}))|\text{ }\|\Phi\|_\mathrm{super} < \infty\}$ where $\|\Phi\|_{\mathrm{super}} = \sup_{\hat{O}\in B(\mathcal{H}),\|\hat{O}\|_{\mathrm{op}}=1}\|\Phi(\hat{O})\|_{\mathrm{op}}$. This similarly is a Banach space with this norm.
    \end{itemize} 
\end{definition}

\begin{definition}[Completely bounded and positive maps]
    We define the following norm for linear maps $\Phi\in B(B(\mathcal{H}))$ for complete boundedness:
    $$\|\Phi\|_{\mathrm{cb}} = \sup_{n\geq 1}\|\Phi\otimes\mathbb{I}_n\|_{\mathrm{super}}$$
    where the map $\Phi\otimes\mathbb{I}_n$ is defined as a map in $L(B(\mathcal{H}\otimes\mathbb{C}^{n}))$, with $\mathbb{C}^{n}$ an $n$-dimensional Hilbert space and $\mathbb{I}_n$ the identity operation over it.
    \\
    We define the set $CB(B(\mathcal{H})) = \{\Phi\in B(B(\mathcal{H}))|\text{ }\|\Phi\|_{\mathrm{cb}} < \infty\}$.
\\
    We similarly define the set $CP(B(\mathcal{H})) = \{\Phi\in CB(B(\mathcal{H}))|\text{ }\forall{\hat{A}\in B(\mathcal{H}\otimes \mathbb{C}^{n})_+},\text{ }(\Phi\otimes \mathbb{I}_n)(\hat{A})\in B(\mathcal{H}\otimes \mathbb{C}^{n})_+\}$.
\end{definition}

\begin{definition}[Completely positive trace-preserving map]
 We define the set of completely positive trace-preserving maps acting over a Hilbert space $\mathcal{H}$ as
 $$\CPTP(\mathcal{H}) = \{\mathcal{E}\in CP(B(\mathcal{H})) |\text{ } \forall\hat{\rho}\in L_1(\mathcal{H}), \text{ }\text{Tr}(\mathcal{E}(\hat{\rho})) = \text{Tr}(\hat{\rho})\}$$
 By this definition for all $\mathcal{E} \in \CPTP(\mathcal{H})$, $\mathcal{E}(\hat{\rho})\in L_1(\mathcal{H})$ for all $\hat{\rho}\in L_1(\mathcal{H})$. \\We define the unique adjoint of $\mathcal{E}\in \CPTP(\mathcal{H})$ as $\mathcal{E}^\dagger\in CP(B(\mathcal{H}))$ where $\text{Tr}(\hat{O}^\dagger\mathcal{E}(\hat{\rho})) = \text{Tr}((\mathcal{E}^\dagger(\hat{O}))^\dagger \hat{\rho})$ for all $\hat{\rho}\in L_1(\mathcal{H})$ and $\hat{O}\in B(\mathcal{H})$.
\end{definition}

\begin{lemma}[Fundamental factorization of CB maps]
    For all $\Phi \in CB(B(\mathcal{H}))$ the following two statements are equivalent:
    \begin{itemize}
        \item $\|\Phi\|_{\mathrm{cb}}\leq 1$
        \item There is a Hilbert space $\mathcal{K}$, a homomorphism $\varpi:B(\mathcal{H})\to B(\mathcal{K})$ and linear maps $V,W:\mathcal{K}\to\mathcal{H}$ with operator norms $\|V\|_\mathrm{op},\|W\|_{\mathrm{op}}\leq 1$ such that $\Phi(\cdot) = V^\dagger \varpi(\cdot)W$.
    \end{itemize}
    The map $\Phi$ is CP iff it admits a factorization with $V =W$.
\end{lemma}
\begin{proof}
    Refer to Sec. 1.1.3 of \cite{gupta2015functional}.
\end{proof}
If we restrict $\Phi$ to be trace-preserving as well, we can obtain $V$ and $W$ to be isometries by choosing $\varpi$ to be an embedding into a larger Hilbert space by tensor product. We now restate Corollary 1.1.19 of \cite{gupta2015functional}.
\begin{corollary}
    Any $\Phi\in CB(B(\mathcal{H}))$ can be represented as $\Phi = \phi_1-\phi_2 + i(\phi_3-\phi_4)$ where all $\phi_1,\phi_2,\phi_3,\phi_4\in CP(B(\mathcal{H}))$.
\end{corollary}
While the set of maps we are concerned with is $\CPTP(\mathcal{H})$, the span of this set will be contained in $CB(B(\mathcal{H}))$ making it a more relevant set to study. While complete boundedness is required as a physical constraint, we find that physically motivated learning tasks such as expectation values for output states from an unknown channel, lead us to having to study the even larger space of $B(B(\mathcal{H}))$.

\begin{definition}
    We define a learning task for an unknown CPTP map $\mathcal{E}$ acting on Hilbert space $\mathcal{H}$ as approximating the bilinear function $f_{\mathcal{E}}:B(\mathcal{H})\times L_1(\mathcal{H})\to\mathbb{C}$ over a finite subset of its domain, where $f_{\mathcal{E}}(\hat{O},\hat{\rho}) = \text{Tr}(\hat{O}\mathcal{E}(\hat{\rho})) = \text{Tr}(\hat{\rho}\mathcal{E}^\dagger(\hat{O}))$ for $\hat{\rho}\in L_1(\mathcal{H})$ and $\hat{O}\in B(\mathcal{H})$.
\end{definition}

This learning task is physically motivated since the function to be learned will always be continuous and bounded, because $\hat{O}$ is a bounded operator and $\hat{\rho}$ is trace class. While it may be tempting to choose the set of bounded linear maps that take $L_1(\mathcal{H})\to L_1(\mathcal{H})$ as the space containing the set of CPTP maps, it is important to note that by considering the space $CP(B(\mathcal{H}))$, we can have both $\mathcal{E}$ and $\mathcal{E}^\dagger$ contained in the same space with $\mathcal{E}^\dagger$ always being a unital map.

\begin{remark}
   For Banach spaces $E$ and $F$, we have the following observations:
   \begin{equation*}
       \begin{aligned}
           B(E,F^*) &= \{\text{all bounded linear maps from $E$ to $F^*$}\}\\
           &\simeq \mathrm{Bil}(E\times F) = \{\text{all bounded bilinear forms on $E\times F$}\}\\
           &\simeq (E\hat{\otimes}_\pi F)^*.
       \end{aligned}
   \end{equation*}
\end{remark}
\begin{corollary}
    The space of bilinear functions $B(\mathcal{H})\otimes L_1(\mathcal{H})\to \mathbb{C}$ is isomorphic to $B(B(\mathcal{H}))$ as well as $(B(\mathcal{H})\hat{\otimes}_\pi L_1(\mathcal{H}))^*$.
\end{corollary}

Note that while we have an operational understanding for $L_1(\mathcal{H})^*$ to be isomorphic to $B(\mathcal{H})$, the dual space $B(\mathcal{H})^*$ has no such direct interpretation since it is not isomorphic to $L_1(\mathcal{H})$.

\begin{lemma}
    Consider learning over a subset of $L_1(\mathcal{H})$ that is contained within the range of a finite-rank idempotent operator $\hat{P}\in L_1(\mathcal{H})$ (it is trace class since it is finite rank). There exists a trace-class operator $\hat{\sigma}_P\in L_1(\mathcal{H}^{\otimes 2})$ such that for all $\mathcal{E}\in \CPTP(\mathcal{H})$
    \begin{equation}
        f_{\mathcal{E}}(\hat{O},\hat{\rho}) = \text{Tr}((\hat{O}\otimes\hat
        \rho^T)(\mathcal{E}\otimes\mathbb{I})(\hat{\sigma}_P)),
    \end{equation}
    where the transposition of $\hat{\rho}$ is taken over a basis that diagonalizes $\hat{P}$. 
    \\
    If $\hat{P}=\mathbb{I}$, the operator $\hat{\sigma}_P$ is no longer trace class. In this situation the operator $(\mathcal{E}\otimes \mathbb{I})(\sigma_{\mathbb{I}})$ defines a bounded linear functional over the projective tensor product space $B(\mathcal{H})\otimes_\pi L_1(\mathcal{H})$ as can be seen by the function above which is guaranteed to be bounded when $\mathcal{E}(\hat{\rho})\in L_1(\mathcal{H})$ for all $\hat{\rho}\in L_1(\mathcal{H})$.
\end{lemma}
\begin{proof}
    Consider that the operator $\hat{P}$ has the following spectral decomposition
    \begin{equation}
        \hat{P} = \sum_{i=1}^{r}|e_i\rangle\langle e_i|.
    \end{equation}
    We can then define
    \begin{equation}
        \hat{\sigma}_P = \sum_{i,j=1}^{r}|e_i\rangle\langle e_j|\otimes|e_i\rangle\langle e_j|,
    \end{equation}
    and the rest of this follows under the assumption that $\mathcal{E}$ is a continuous map over $B(\mathcal{H})$. In the special case $\hat{P}=\mathbb{I}$, $\hat{\sigma}_P$ corresponds to an unnormalized maximally entangled state. For the case of an infinite-dimensional Hilbert space of $m$ bosonic modes, this operator takes the form
    \begin{equation}
        \hat{\sigma}_{\mathbb{I}} = \frac{1}{\pi^m}\int d^{2m}\beta\hat{D}(\beta)\otimes\hat{D}(\beta^*),
    \end{equation}
    which is a POVM element of a continuous POVM set (see Eq.~\eqref{eq:bellPOVMboson}) but is not actually a bounded operator and lies in the dual space $(B(\mathcal{H})\hat{\otimes}_\pi L_1(\mathcal{H}))^*$. In this situation, using the characteristic function of $\hat{\rho}$ and its uniqueness up to a measure-zero function, the equality holds trivially.
\end{proof}

To be able to learn this function over all possible input states $\hat{\rho}$, one method is to obtain a description of the operator $\int d^{2m}\beta\mathcal{E}(\hat{D}(\beta))\otimes\hat{D}(\beta^*)$ which is not a bounded operator. We now consider an altered learning route that does not attempt learning for arbitrary inputs.

\begin{remark}
    For $\mathcal{E}\in \CPTP(\mathcal{H})$, the linear functional defined by $g_{\mathcal{E},\hat{\sigma}}:B(\mathcal{H}^{\otimes 2})\to\mathbb{C}$ is continuous and bounded when defined as
    \begin{equation}
        g_{\mathcal{E},\hat{\sigma}}(\cdot) = \text{Tr}((\cdot) (\mathcal{E}\otimes\mathbb{I})(\hat{\sigma})),
    \end{equation}
    where $\hat{\sigma}\in L_1(\mathcal{H}^{\otimes 2})$. Since $B(\mathcal{H}^{\otimes 2})$ is the dual of $L_1(\mathcal{H}^{\otimes 2})$, the restriction of $\hat{\sigma}$ to the trace class is important for boundedness of the functional.
\end{remark}

\begin{definition}
    We define a reduced learning task for an unknown channel $\mathcal{E}\in CB(B(\mathcal{H}))$ as learning the linear function $g_{\mathcal{E},\hat{\sigma}}$ for a particular choice of $\hat{\sigma}$.
\end{definition}
We consider Choi-state learning since in a weak sense this might be the best we can do for learning an arbitrary channel. Channel learning protocols generally can be framed as finding some particular set of evaluations for $g_{\mathcal{E},\hat{\sigma}}$ since this captures the essence of finding expectation values of some operator over the output states from the channel for chosen input states.

\subsection{Measure theory for describing quantum measurements}\label{app:measure_theory}
In this subsection, we offer a concise introduction to a few concepts in measure theory from \cite{bogachev2007measure}. We also restate the concept of a POVM using this language adapted from \cite{hall2013quantum}.
\begin{definition}
   An algebra of sets $\Sigma_S$ is a class of subsets of a fixed set $S$ (referred to as a space) that satisfies:
   \begin{itemize}
       \item Both $S$ and the empty set $\varnothing$ belong to $\Sigma_S$.
       \item If $A,B\in\Sigma_S$, $A\cup B, A\cap B, A\setminus B$ are all also in $\Sigma_S$.
   \end{itemize}
   Further, if for every sequence of sets $A_n\in\Sigma_S$ we also have $\bigcup_{n=1}^{\infty}A_n\in\Sigma_S$, the algebra is referred to as a $\sigma$-algebra. A pair $(S,\Sigma_S)$, where $\Sigma_S$ is a $\sigma$-algebra of $S$, forms a measurable space.
\end{definition}
An easy example of a $\sigma$-algebra is the power set of the set $S$.

\begin{definition} A countably additive set function $\mu:\Sigma_S\to[0,1]$ for measurable space $(S,\Sigma_S)$ is referred to as a probability measure if it satisfies:
\begin{itemize}
    \item $\mu(S) = 1$ and $\mu(\varnothing) = 0$.
    \item For pairwise disjoint sets $A_n\in \Sigma_S$ with $\bigcup_{n=1}^\infty A_n \in \Sigma_S$, we have
    $$\mu\left(\bigcup_{n=1}^\infty A_n\right) = \sum_{n=1}^{\infty}\mu(A_n).$$
    \item $\mu$ is continuous at $\varnothing$ which is to say that for every sequence $A_n\in\Sigma_S$ with $A_{n+1}\subset A_n$ for all $n\in \mathbb{N}$ and $\bigcap_{n=1}^{\infty}A_n = \varnothing$, we have $\lim_{n\to\infty}\mu(A_n) = 0$.
    \item $\mu$ is continuous from below, which is to say that for every sequence $A_n\in \Sigma_S$ with $A_{n}\subset A_{n+1}$ for all $n\in \mathbb{N}$ with $\bigcup_{n=1}^{\infty}A_n \in \Sigma_S$, we have
    $$\mu\left(\bigcup_{n=1}^\infty A_n\right) = \lim _{n\to\infty}\mu(A_n).$$
\end{itemize}
If we further assume space $S$ to have an integrable volume element $ds$ centered around a particular point $s\in S$, we define the infinitesimal element $d\mu(s) = \mu(ds)$. 
\end{definition}

Measures that take on unbounded values can also be defined; however, we restrict our discussion to probability measures which by construction are finite measures. We will make use of the above definition of a probability measure to describe both discrete and continuous probability measures. Note that it is safe to assume only the space $S$ to be isomorphic to $\mathbb{R}^K$ for some $K\geq 1$ since through the use of Dirac delta measures we can capture discrete probability measures over the continuous space of $\mathbb{R}^K$. Since we will be specifically working in the space of $\mathbb{R}^K$, we now introduce the concept of a Borel $\sigma$-algebra.

\begin{definition}
    The Borel $\sigma$-algebra of $\mathbb{R}^K$, denoted by $\mathrm{Bor}(\mathbb{R}^K)$, is the $\sigma$-algebra generated by the open subsets of $\mathbb{R}^K$. The sets in $\mathrm{Bor}(\mathbb{R}^K)$ are referred to as Borel sets.
\end{definition}

\begin{definition}
    Let $\mu$ and $\nu$ be two finite positive-valued measures over measurable space $(S,\Sigma_S)$.
    \begin{itemize}
        \item The measure $\nu$ is called absolutely continuous with respect to $\mu$ if $\nu(A) = 0$ for all $A\in \Sigma_S$ with $\mu(A) = 0$. The notation used for this is $\nu\ll \mu$.
        \item The measure $\nu$ is called singular with respect to $\mu$ if there exists a set $\Omega\in \Sigma_S$ such that $\mu(\Omega) = 0$ and $\nu(S\setminus\Omega) = 0$. The notation used for this is $\nu\perp \mu$.
    \end{itemize}
\end{definition}

\begin{lemma}[Adapted from Theorems 3.2.3 and 5.8.8 of \cite{bogachev2007measure}]\label{thm:radonniko} Given two probability measures $\mu$ and $\nu$ defined on the Borel algebra of $\mathbb{R}^K$, $\nu$ can be written as the sum of two measures $\nu = \nu_{\mathrm{ac}} + \nu_{\mathrm{s}}$ where $\nu_{\mathrm{ac}}\ll \mu$ and $\nu_{\mathrm{s}}\perp \mu$. Defining the function $g:\mathbb{R}^K\to\mathbb{R}$
    \begin{equation}
        g(x) =\lim_{r\to0}\frac{\nu(B(x,r))}{\mu(B(x,r))} ,
    \end{equation}
    where $B(x,r)$ is the ball centered around $x\in \mathbb{R}^K$ with radius $r$. This function is integrable with respect to $d\mu(x)$ and is finite $\mu$-almost everywhere. As a shorthand we will refer to the function $g$ by $d\nu/d\mu$ (also known as the Radon--Nikodym derivative).
\end{lemma}
\begin{proof}
    This is in essence the Radon--Nikodym theorem. We refer the reader to the proof of Thm. 5.8.8 in \cite{bogachev2007measure} for a proof of the above lemma. The fact that $g(x)$ is finite $\mu$-almost everywhere is a consequence of the fact that the singular component $\nu_{\mathrm{s}}$ will be the only contribution to $g$ being ill-defined; however, this is exactly over the part of the space that is measure zero for $\mu$.
\end{proof}

\begin{remark}
    The concept of a likelihood ratio between two probability distributions is exactly the same as the Radon--Nikodym derivative as defined above. Assuming both distributions $\mu$ and $\nu$ are concentrated over a discrete set of points, the above definition reduces exactly to the more familiar likelihood ratio. If both $d\mu(x) = p_1(x)dx$ and $d\nu(x) = p_2(x)dx$, then the above definition reduces to $p_2(x)/p_1(x)$ as well.
\end{remark}

\begin{definition}
    Consider two probability measures $\mu$ and $\nu$ and a convex function $f:\mathbb{R}\to\mathbb{R}$ that satisfies $f(1) = 0$ and continuity of $f(0) = \lim_{t\to 0^+}f(t)$. The $f$-divergence between the two measures $\mu$ and $\nu$ is defined as follows:
    \begin{equation}
        D_f(\nu\|\mu) = \int d\mu f\left(\frac{d\nu}{d\mu}\right).
    \end{equation}
    As an example, for $f(x) = \frac{1}{2}|x-1|$, $D_f(\nu\|\mu)$ is the total variation distance between the two probability measures, and for $f(x)= (x-1)^2$, $D_f(\nu\|\mu)$ is the $\chi^2$-divergence.
\end{definition}

The fact that we can define the quantity $D_f(\nu\|\mu)$ in terms of an integral comes from being able to define $d\nu/d\mu$ as finite $\mu$-almost everywhere.

For the sake of completeness, we also define a positive-operator-valued measure (POVM), which can be used to describe a quantum measurement.
\begin{definition}
 For a measurable space $(S,\Sigma_S)$ with $\Sigma_S$ being a Borel algebra, we define a positive-operator-valued measure over Hilbert space $\mathcal{H}$ as a mapping $F:\Sigma_S\to B(\mathcal{H})$ such that:
 \begin{itemize}
 \item For all Borel sets $E\in \Sigma_S$, $F(E)$ is a positive bounded operator.
     \item $F(S) = \mathbb{I}$ and $F(\varnothing) = 0$.
     \item For a sequence of pairwise disjoint Borel sets $E_n\in \Sigma_S$, $F(\bigcup_{n=1}^\infty E_n) = \sum_{n=1}^{\infty} F(E_n)$.
 \end{itemize}
 For all states $\hat{\rho}\in D(\mathcal{H})$, we can define a measure $\mu_{\hat{\rho}}:\Sigma_S \to [0,1]$ such that for all Borel sets $E\in\Sigma_S$, $\mu_{\hat{\rho}}(E) = \text{Tr}\left[\hat{\rho}F(E)\right]$ and it can be seen that this measure is a valid probability measure.
\end{definition}

Note that here the operator $F(E)$ is always bounded, but can be obtained through the integration of unbounded operators over a volume element. As an example, we can examine the case of homodyne measurements for the quadrature $\hat{q}$. The operators take the form $\hat{\Pi}_{x} = |x\rangle\langle x|$, where $\ket{x}$ is technically contained in the rigged Hilbert space for the bosonic mode and as such the operator $\hat{\Pi}_{x}$ is unbounded. However the operator $\hat{F}(E) = \int_{x\in E} dx\hat{\Pi}_x$ for all Borel sets $E\subseteq \mathbb{R}$ will be a bounded operator. A similar property occurs for the POVM defined in Eq.~\eqref{eq:bellPOVMboson} and so, strictly speaking, we should not be allowed to define the probability using the trace definition, since we may no longer be dealing with trace-class operators. Within the settings we work in, this will not pose a problem; however, we will treat this case carefully when proving lower bounds for arbitrary learning schemes which are allowed to have access to these nonphysical yet mathematically well-defined POVMs.

\section{Sample complexity bounds for the Choi-state learning problem}
\subsection{$c$-copy learning protocols over channels}\label{app:cm_protocols}
We extend the tree based description for learning protocols that have $c$-copy access from \cite{chen2024optimaltradeoffsestimatingpauli} to be able to tell apart CPTP maps $\mathcal{E}_{\mathbf{u},\mathbf{v}}$ and $\mathcal{E}_0$ with each measurement step using $(\mathcal{E}_{\mathbf{u},\mathbf{v}}\otimes \mathbb{I})^{\otimes c}$ and $(\mathcal{E}_0\otimes\mathbb{I})^{\otimes c}$. We will then use the calculation of $\chi^2$-divergence from \cite{ller2025infinitehierarchymulticopyquantum} and extend it to our hypothesis test to establish channel discrimination lower bounds. To account for the fact that each step is allowed measurements with outcomes that are continuously valued, we define POVMs using the measure-theoretic description established in the previous subsection.

To start, we define the following channel discrimination problems.
\begin{problem}[Many-one discrimination]\label{prob:many}
    Consider the exact same setting as many-one channel discrimination with revelation (Problem~\ref{prob:manyrevel}) with the exception of there being no revelation of the parameters $\mathbf{u}$ and $\mathbf{v}$.
\end{problem}

The setting of the discrimination tasks in \cite{chen2024optimaltradeoffsestimatingpauli,ller2025infinitehierarchymulticopyquantum} follows that of Problem~\ref{prob:many} where there is no revelation. Our setting is that of Problem~\ref{prob:manyrevel} (similar to \cite{PhysRevLett.133.230604,coroi2025exponentialadvantagecontinuousvariablequantum}) to accommodate situations where the possible set of queries is unbounded (such as learning bosonic states). We will show that this difference will be inconsequential as far as the lower bounds we derive are concerned. We now establish a representation for learning protocols that use $c$-copy access to the channel by examining the probability distribution of the state of the classical memory at the final step of the learning protocol.

\begin{definition}[Path representation for $c$-copy learning protocols]\label{def:path_rep_cM}
    Consider a path described by $s_{0:T}$ which is a vector of length $T+1$ with each element $s_t$ contained in a measurable space $(S,\Sigma_S)$ with Borel $\sigma$-algebra $\Sigma_S$. (We assume without loss of generality that $S$ is isomorphic to some $\mathbb{R}^K$). Each path begins at $s_0$ which is a fixed value and the notation $s_{i:j}$ is the sequence from the $i$th term in the vector $s_{0:T}$ to the $j$th term including both ends.
    
    Given an unknown CPTP map $\mathcal{E}$ with input Hilbert space $\mathcal{H}_{\mathrm{in}}$ and output Hilbert space $\mathcal{H}_{\mathrm{out}}$ along with a fixed ancilla space $\mathcal{H}_{\mathrm{anc}}$, consider the learner to be allowed to use any POVM from the complete set of POVMs defined over $\sigma$-algebra $\Sigma_S$ for a measurable space $S$ which acts over Hilbert space $\mathcal{H}_{\mathrm{out}}^{\otimes c}\otimes \mathcal{H}_{\mathrm{anc}}$. A $c$-copy learning protocol will be defined as follows using a set of paths $\mathcal{T}$ each of length $T$.
    \begin{itemize}
        \item The value of $s_1$ is the measurement outcome of applying POVM $F_{s_0}:\Sigma_S\to B(\mathcal{H}_{\mathrm{out}}^{\otimes c}\otimes  \mathcal{H}_{\mathrm{anc}})$ on the state $(\mathcal{E}^{\otimes c}\otimes\mathbb{I})(\hat{\sigma}_{s_0})$ where $F_{s_0}$ and $\hat{\sigma}_{s_0}$ are the initial POVM and input state respectively. This defines a measure $\mu_{0,\mathcal{E}}:\Sigma_S\to [0,1]$ as the probability measure for the outcome $s_1$ based on this measurement. The infinitesimal element given by $d\mu_{0,{\mathcal{E}}}(s_1)$ for the space $S$ captures the probability distribution for $s_1$.
        \item At each later step $t\geq 2$, we adaptively prepare a state $\hat{\sigma}_{s_{0:t-1}}\in D(\mathcal{H}_{\mathrm{in}}^{\otimes c}\otimes\mathcal{H}_{\mathrm{anc}})$. Defining an adaptively chosen POVM $F_{s_{0:t-1}}:\Sigma_S\to B(\mathcal{H}^{\otimes c}_{\mathrm{out}}\otimes\mathcal{H}_{\mathrm{anc}})$, this performs measurements on $(\mathcal{E}^{\otimes c}\otimes\mathbb{I}_{\mathrm{anc}})(\hat{\sigma}_{s_{0:t-1}})$. We define a probability measure $\mu_{s_{0:t-1},\mathcal{E}}:\Sigma_S\to[0,1]$ corresponding to the outcome $s_t$.
        \item Along each path, at step $t$ the infinitesimal element $d\mu_{s_{0:t-1},\mathcal{E}}(s_t)$ represents the probability density corresponding to the path having $s_t$ at the $t$th step conditioned on having taken the path $s_{0:t-1}$ up to the $(t-1)$th step.
        \item Each completed path is of length $T$. Considering the space containing all paths $\mathcal{T}$, the volume element defined through the probability measure is given by $d\mu_{\mathcal{T},\mathcal{E}}(s_{1:T}) = \prod_{t=1}^Td\mu_{s_{0:t-1},\mathcal{E}}(s_t)$.
    \end{itemize}
\end{definition}

\begin{remark}
    The above path representation assumes that all the POVMs use the exact same $\sigma$-algebra $\Sigma_S$. At each step there is no loss of generality by assuming that the outcomes are contained in a space isomorphic to some $\mathbb{R}^{K_i}$. Hence without loss of generality, we can simply define the space $S = \mathbb{R}^{\max_{i}K_i}$ over all possible paths and appropriately redefine the measures to accommodate for this. Further, this allows us to simultaneously treat discrete-outcome and continuous-outcome POVMs with the same mathematical framework.
\end{remark}

\begin{lemma}\label{lem:finiteRNder}
    For the two channels defined in the many-one channel discrimination task with revelation (Problem~\ref{prob:manyrevel}), consider a measurable space $(S,\Sigma_S)$ with $\Sigma_S$ being a Borel $\sigma$-algebra and a POVM defined as $F:\Sigma_S\to B(\mathcal{H}_{\mathrm{out}}^{\otimes c}\otimes \mathcal{H}_{\mathrm{anc}})$ where $\mathcal{H}_{\mathrm{anc}}$ is an ancillary Hilbert space. Given an input state $\hat{\sigma}\in D(\mathcal{H}_{\mathrm{in}}^{\otimes c}\otimes\mathcal{H}_{\mathrm{anc}})$ consider the two probability measures obtained by acting with the POVM $F$ on the two states $(\mathcal{E}_{\mathbf{u},\mathbf{v}}^{\otimes c}\otimes\mathbb{I}_{\mathrm{anc}})(\hat{\sigma})$ and $(\mathcal{E}_{0}^{\otimes c}\otimes\mathbb{I}_{\mathrm{anc}})(\hat{\sigma})$ which we refer to as $\mu_{\mathcal{E}_{\mathbf{u},\mathbf{v}}}$ and $\mu_{\mathcal{E}_{0}}$. We have $\mu_{\mathcal{E}_{\mathbf{u},\mathbf{v}}}\ll \mu_{\mathcal{E}_0}$ and the derivative of $\mu_{\mathcal{E}_{\mathbf{u},\mathbf{v}}}$ with respect to $\mu_{\mathcal{E}_0}$ is necessarily finite-valued wherever it is defined.
\end{lemma}
\begin{proof}
We define the following bounded operators 
    \begin{equation}
        \hat{W}_{\pmb{x},\mathbf{v}} = \bigotimes_{i=1}^c\hat{W}_{x_i,\mathbf{v}},\quad\hat{K}_{\pmb{x},\mathbf{u}} = \bigotimes_{i=1}^c\hat{K}_{x_i,\mathbf{u}},
    \end{equation}
and further define the following operators in the ancillary space
    \begin{equation}
        \hat\lambda_{\pmb{x},\mathbf{u},\hat{\sigma}} = \text{Tr}_{\mathrm{in}^{\otimes c}}\left[(\hat{K}_{\pmb{x},\mathbf{u}}\otimes\mathbb{I}_{\mathrm{anc}})\hat{\sigma}\right],
    \end{equation}
    where we use $\text{Tr}_{\mathrm{in}^{\otimes c}}$ to represent tracing out $\mathcal{H}_{\mathrm{in}}^{\otimes c}$. Using these, we can write
    \begin{equation}
        (\mathcal{E}_{\mathbf{u},\mathbf{v}}^{\otimes c}\otimes\mathbb{I}_{\mathrm{anc}})(\hat{\sigma})= \sum_{\pmb{x}\in\{0,\dots,l\}^c}\left(\hat{\rho}_0^{1/2}\right)^{\otimes c}\hat{W}_{\pmb{x},\mathbf{v}}\left(\hat{\rho}_0^{1/2}\right)^{\otimes c}\otimes \hat{\lambda}_{\pmb{x},\mathbf{u},\hat{\sigma}},\quad (\mathcal{E}_0^{\otimes c}\otimes\mathbb{I}_{\mathrm{anc}})(\hat{\sigma}) = \hat{\rho}_0^{\otimes c}\otimes\hat{\sigma}_{\mathrm{anc}},
    \end{equation}
    where $\hat{\sigma}_{\mathrm{anc}}$ is the partial trace of $\hat{\sigma}$ in the ancillary space. Consider that the state $\hat{\sigma}$ has the following spectral decomposition
    \begin{equation}
        \hat{\sigma} = \sum_{k}w_k |\psi_k\rangle\langle\psi_k|,
    \end{equation}
    and each state $\ket{\psi_k} \in\mathcal{H}_{\mathrm{in}}^{\otimes c}\otimes\mathcal{H}_{\mathrm{anc}}$ has a Schmidt decomposition given by
    \begin{equation}
        \ket{\psi_k} = \sum_{i=1}^{r_k} \sqrt{q_{k}^{(i)}}\ket{e^{(i)}_k}_{\mathrm{in}^{\otimes c}}\ket{f^{(i)}_k}_{\mathrm{anc}}.
    \end{equation}
 Here, for each $k$, the sets $\cup_{i}\{\ket{e_k^{(i)}}\}$ and $\cup_i\{\ket{f_k^{(i)}}\}$ are orthonormal bases for a rank-$r_k$ subspace of $\mathcal{H}_{\mathrm{in}}^{\otimes c}$ and $\mathcal{H}_{\mathrm{anc}}$ respectively, and hence we can write an isometry $V_k:\mathcal{H}_{\mathrm{in}}^{\otimes c}\to \mathcal{H}_{\mathrm{anc}}$ such that $V_k\ket{e^{(i)}_k} = \ket{f_k^{(i)}}$. We note that from this we have
 \begin{equation}
     \hat{\sigma}_{\mathrm{anc}}^{(k)} = V_k\hat{\sigma}_{\mathrm{in}}^{(k)}V^\dagger_k, \quad \hat{\sigma}_{\mathrm{in}}^{(k)} = \sum_{i=1}^{r_k} q_k^{(i)}\ket{e^{(i)}_k}\bra{e^{(i)}_k},\quad \hat{\sigma}_{\mathrm{anc}} = \sum_{k}w_kV_k\hat{\sigma}_{\mathrm{in}}^{(k)}V^\dagger_k.
 \end{equation}
 In the orthonormal basis defined by $\ket{e^{(i)}_k}$, we define the transposition of the operator $\hat{K}_{\pmb{x},u}$ as  $\hat{K}_{\pmb{x},u}^{T^{(k)}}$. We can observe that the operators $\hat{\lambda}_{\pmb{x},\mathbf{u},\hat{\sigma}}$ simplify as
 \begin{equation}
     \hat{\lambda}_{\pmb{x},\mathbf{u},\hat{\sigma}} = \sum_{k}w_k\left(\sum_{i,j=1}^{r_k}\sqrt{q_k^{(i)}q_k^{(j)}}\bra{e^{(j)}_k}\hat{K}_{\pmb{x},\mathbf{u}}\ket{ e_k^{(i)}}\ket{f^{(i)}_k}\bra{f^{(j)}_k}_{\mathrm{anc}}\right) = \sum_{k}w_kV_k\left(\sqrt{\hat{\sigma}_{\mathrm{in}}^{(k)}}\hat{K}_{\pmb{x},\mathbf{u}}^{T^{(k)}}\sqrt{\hat{\sigma}_{\mathrm{in}}^{(k)}}\right)V^\dagger_k.
 \end{equation}
 Consider a POVM defined by $F:\Sigma_{S}\to B(\mathcal{H}_{\mathrm{out}}^{\otimes c}\otimes\mathcal{H}_{\mathrm{anc}})$ for a measurable space $(S,\Sigma_S)$ with Borel $\sigma$-algebra $\Sigma_{S}$. We define two measures $\mu_{\mathcal{E}_0}$ and $\mu_{\mathcal{E}_{\mathbf{u},\mathbf{v}}}$ obtained by acting this POVM on the states $(\mathcal{E}_0^{\otimes c}\otimes\mathbb{I}_{\mathrm{anc}})(\hat{\sigma})$ and $(\mathcal{E}_{\mathbf{u},\mathbf{v}}^{\otimes c}\otimes\mathbb{I}_{\mathrm{anc}})(\hat{\sigma})$, respectively. For all Borel sets $E\in \Sigma_S$, we have
    \begin{align}
        \mu_{\mathcal{E}_0}(E) &= \sum_{k}w_k\text{Tr}\left[F(E)(\hat{\rho}_0^{\otimes c}\otimes V_k\hat{\sigma}_{\mathrm{in}}^{(k)}V^\dagger_k)\right],\\
        \mu_{\mathcal{E}_{\mathbf{u},\mathbf{v}}}(E) &= \sum_{\pmb{x}\in\{0,\dots,l\}^c}\left(\sum_kw_k\text{Tr}\left[F(E)\left(\sqrt{\hat{\rho}_0^{\otimes c}}\hat{W}_{\pmb{x},\mathbf{v}}\sqrt{\hat{\rho}_0^{\otimes c}}\otimes V_k\sqrt{\hat{\sigma}_{\mathrm{in}}^{(k)}}\hat{K}_{\pmb{x},\mathbf{u}}^{T^{(k)}}\sqrt{\hat{\sigma}_{\mathrm{in}}^{(k)}}V^\dagger_k\right)\right]\right).
    \end{align}
    For each value of $k$, we now define a positive-operator-valued measure $F_{\hat{\sigma}^{(k)}}:\Sigma_S\to L_1(\mathcal{H}_{\mathrm{out}}^{\otimes c}\otimes\mathcal{H}_{\mathrm{anc}})$ that satisfies
    \begin{equation}\label{eq:POVMwithstate}
        F_{\hat{\sigma}^{(k)}}(E) = \left(\sqrt{\hat{\rho}_0^{\otimes c}}\otimes V_k\sqrt{\hat{\sigma}^{(k)}_{\mathrm{in}}}\right)^\dagger F(E)\left(\sqrt{\hat{\rho}_0^{\otimes c}}\otimes V_k\sqrt{\hat{\sigma}^{(k)}_{\mathrm{in}}}\right),
    \end{equation}
    for all Borel sets $E\in\Sigma_S$. Observe that the above operator is always trace class as a consequence of conjugation by the operators $\sqrt{\hat{\sigma}^{(k)}_{\mathrm{in}}}$ and $\sqrt{\hat{\rho}_0^{\otimes c}}$. From this, we obtain a much simpler expansion of the two relevant probability measures as
    \begin{align}
        \mu_{\mathcal{E}_0^{\otimes c}}(E) &= \sum_{k}w_k\text{Tr}\left[F_{\hat{\sigma}^{(k)}}(E)\right],\\
        \mu_{\mathcal{E}_{\mathbf{u},\mathbf{v}}^{\otimes c}}(E) &= \sum_{k}w_k\sum_{\pmb{x}\in\{0,\dots,l\}^c}\text{Tr}\left[F_{\hat{\sigma}^{(k)}}(E)\left(\hat{W}_{\pmb{x},\mathbf{v}}\otimes \hat{K}_{\pmb{x},\mathbf{u}}^{T^{(k)}}\right)\right].
    \end{align}
    Without loss of generality, we assume that the space $S$ is isomorphic to $\mathbb{R}^K$ for some $K$. Using Lemma~\ref{thm:radonniko} and assuming uniform boundedness of the ratio, we can define the Radon--Nikodym derivative of $\mu_{\mathcal{E}_{\mathbf{u},\mathbf{v}}}$ with respect to $\mu_{\mathcal{E}_{0}}$ at a particular point $s\in S$ as 
    \begin{equation}
        \frac{d\mu_{\mathcal{E}_{\mathbf{u},\mathbf{v}}}}{d\mu_{\mathcal{E}_{0}}}(s) = \lim_{r\to 0}\frac{\mu_{\mathcal{E}_{\mathbf{u},\mathbf{v}}}(B(s,r))}{\mu_{\mathcal{E}_0}(B(s,r))},
    \end{equation}
    where $B(s,r)$ is a ball of radius $r$ centered around $s$. We will now show that the ratio in the limit is uniformly bounded. For all $r>0$ the ball is a Borel set, hence the operator $F_{\hat{\sigma}}(B(s,r))$ will be trace class. Hence we can write 
    \begin{align}
        \frac{\mu_{\mathcal{E}_{\mathbf{u},\mathbf{v}}}(B(s,r))}{\mu_{\mathcal{E}_0}(B(s,r))}  
            & =\frac{\sum_{k}w_k\text{Tr}\left[F_{\hat{\sigma}^{(k)}}(B(s,r))\left(\sum_{\pmb{x}\in\{0,\dots,l\}^c}\hat{W}_{\pmb{x},\mathbf{v}}\otimes \hat{K}_{\pmb{x},\mathbf{u}}^{T^{(k)}}\right)\right]}{\sum_{k}\text{Tr}\left[F_{\hat{\sigma}^{(k)}}(B(s,r))\right]}\\
        &\leq \frac{\sum_{k}w_k\text{Tr}\left[F_{\hat{\sigma}^{(k)}}(B(s,r))\right]\left\|\sum_{\pmb{x}\in\{0,\dots,l\}^c}\hat{W}_{\pmb{x},\mathbf{v}}\otimes \hat{K}_{\pmb{x},\mathbf{u}}^{T^{(k)}}\right\|_{\mathrm{op}}}{\sum_k w_k\text{Tr}\left[F_{\hat{\sigma}^{(k)}}(B(s,r))\right]},\\
        &\leq \left\|\sum_{\pmb{x}\in\{0,\dots,l\}^c}\hat{W}_{\pmb{x},\mathbf{v}}\otimes \hat{K}_{\pmb{x},\mathbf{u}}^{T}\right\|_{\mathrm{op}}
    \end{align}
    where the transposition is defined in some fixed basis choice. Here we use the fact that for a positive trace class operator $A$ and bounded operator $B$, $\text{Tr}\left[AB\right] \leq \text{Tr}\left[A\right]\|B\|_{\mathrm{op}}$. Note that since there is always a unitary that can map the bases $\ket{e^{(i)}_{k}}$ to $\ket{e^{(i)}_{k}}$, this same unitary if conjugated around $\hat{K}_{\pmb{x},u}^{T^{(k)}}$ would map it to $\hat{K}_{\pmb{x},u}^{T^{(k')}}$. Hence all the operator norms in the weighed sum are equal giving the final reduction to the expression where the choice of basis for transposition doesn't matter. Due to the boundedness of the operators, the derivative will also be finite wherever the limit is defined as a consequence of every sequence along this limit having all of its terms bounded by the above operator norm. This upper bound further shows that $\mu_{\mathcal{E}_0}(E) = 0$ implies $\mu_{\mathcal{E}_{\mathbf{u},\mathbf{v}}}(E) = 0$ for any Borel set $E\in \Sigma_S$ which gives $\mu_{\mathcal{E}_{\mathbf{u},\mathbf{v}}}\ll \mu_{\mathcal{E}_0}$.
\end{proof}

\begin{lemma}[Le Cam's two-point method]\label{lem:lecam}
    The probability that the learning protocol represented by the tree $\mathcal{T}$ correctly solves the many-one discrimination task is upper bounded by 
    \begin{equation}
        p_{\mathrm{succ}} \leq \frac{1}{2}\left(1+\frac{1}{2}\int _{s_{0:T}\in \mathcal{T}} d\mu_{\mathcal{T},\mathcal{E}_0}(s_{0:T})\left|\mathbb{E}_{\mathbf{u},\mathbf{v}}\left[\frac{d\mu_{\mathcal{T},{\mathcal{E}_{\mathbf{u},\mathbf{v}}}}}{d\mu_{\mathcal{T},\mathcal{E}_0}}(s_{0:T})\right] - 1\right|\right),
    \end{equation}
    and the probability that it correctly solves the many-one discrimination with revelation task is upper bounded by
    \begin{equation}
        p_{\mathrm{succ},\mathrm{revel}} \leq\frac{1}{2}\left(1+\frac{1}{2}\mathbb{E}_{\mathbf{u},\mathbf{v}}\left[\int _{s_{0:T}\in \mathcal{T}} d\mu_{\mathcal{T},\mathcal{E}_0}(s_{0:T})\left|\frac{d\mu_{{\mathcal{T},\mathcal{E}_{\mathbf{u},\mathbf{v}}}}}{d\mu_{\mathcal{T},\mathcal{E}_0}}(s_{0:T}) - 1\right|\right]\right).
    \end{equation}
\end{lemma}
We point out that our success probability definition follows that of the probability of the correct guess given equal prior likelihood of either hypothesis which differs from the definition used in \cite{chen2024optimaltradeoffsestimatingpauli}. The definition used in \cite{chen2024optimaltradeoffsestimatingpauli} corresponds to the probability that the guess is correct with the promise that it is one of two hypotheses. This difference has no real effect on any results since the crucial aspect is only the upper bounding of the total variation distance. We now adapt another lemma from \cite{chen2024optimaltradeoffsestimatingpauli} to fit our definition of success probability. The upper bound on the success probability for Problem~\ref{prob:many} shown above is upper bounded by the corresponding bound for Problem~\ref{prob:manyrevel}, since revelation will always offer more information that can benefit the discrimination task.

This is also pointed out in Lemma 6 of \cite{chen2024optimaltradeoffsestimatingpauli}, which states two inequalities, the second of which is related to revelation. As it turns out, the upper bound introduced in Lemma 7 of \cite{chen2024optimaltradeoffsestimatingpauli} uses the second, weaker upper bound from Lemma 6 of the same work, which is valid for many-one channel discrimination with revelation (Problem~\ref{prob:manyrevel}). We state this result in detail for our chosen convention in the following lemma.

\begin{lemma}[One-sided bound suffices for Le Cam, adapted from \cite{chen2024optimaltradeoffsestimatingpauli}, Lemma 6]\label{lem:1sidedLeCam}
  The probability $p_{\mathrm{succ},\mathrm{revel}}$ that a $c$-copy learning protocol with set of paths $\mathcal{T}$ correctly solves the many-one discrimination-with-revelation task in Problem~\ref{prob:manyrevel} is upper bounded by (for all $\beta\in(0,1)$)
  \begin{equation}
      2(p_{\mathrm{succ},\mathrm{revel}}-1/2)\leq \Pr_{\mathbf{u},\mathbf{v}\sim p(\mathbf{u})p(\mathbf{v}), s_{0:T}\sim \mu_{\mathcal{T},\mathcal{E}_0}}\left[\frac{d\mu_{\mathcal{T},\mathcal{E}_{\mathbf{u},\mathbf{v}}}}{d\mu_{\mathcal{T},\mathcal{E}_0}}(s_{0:T}) \leq (1-\beta)\right] + \beta.
  \end{equation}
\end{lemma}
\begin{proof}
    By Lemma~\ref{lem:lecam}, and using the fact that $\forall{x}\in\mathbb{R}$, $\max(0,x) = (|x| + x)/2$ we get
    \begin{align}
        2(p_{\mathrm{succ},\mathrm{revel}}-1/2)&\leq \frac{1}{2}\int_{s_{0:T}\in\mathcal{T}}d\mu_{\mathcal{T},\mathcal{E}_0}(s_{0:T})\mathbb{E}_{\mathbf{u},\mathbf{v}}\left|\frac{d\mu_{\mathcal{T},\mathcal{E}_{\mathbf{u},\mathbf{v}}}}{d\mu_{\mathcal{T},\mathcal{E}_{0}}}(s_{0:T}) - 1\right|\\
        &=\mathbb{E}_{\mathbf{u},\mathbf{v}}\left[\int_{s_{0:T}\in\mathcal{T}}d\mu_{\mathcal{T},\mathcal{E}_0}(s_{0:T})\max\left(0,1  - \frac{d\mu_{\mathcal{T},\mathcal{E}_{\mathbf{u},\mathbf{v}}}}{d\mu_{\mathcal{T},\mathcal{E}_{0}}}(s_{0:T})\right)\right]\\
        &= \mathbb{E}_{\mathbf{u},\mathbf{v}}\mathbb{E}_{s_{0:T}\sim \mu_{\mathcal{T},\mathcal{E}_0}}\max\left(0,1  - \frac{d\mu_{\mathcal{T},\mathcal{E}_{\mathbf{u},\mathbf{v}}}}{d\mu_{\mathcal{T},\mathcal{E}_{0}}}(s_{0:T})\right)\\
        &\leq\Pr_{\mathbf{u},\mathbf{v}\sim p(\mathbf{u})p(\mathbf{v}), s_{0:T}\sim \mu_{\mathcal{T},\mathcal{E}_0}}\left[\frac{d\mu_{\mathcal{T},\mathcal{E}_{\mathbf{u},\mathbf{v}}}}{d\mu_{\mathcal{T},\mathcal{E}_{0}}}(s_{0:T}) \leq (1-\beta)\right] + \beta.
    \end{align}
\end{proof}

\begin{lemma}\label{lem:martingale}
There exists a constant $z$ for all $\beta\in(0,1/3)$ such that the following statement holds. Suppose $\mathcal{T}$ is the set of paths representation of a $c$-copy learning protocol for the many-one discrimination-with-revelation task defined in Problem~\ref{prob:manyrevel} such that for every partial path $s_{0:t-1}$ (with $t\geq 2$) we have
\begin{equation}
    \mathbb{E}_{\mathbf{u},\mathbf{v}\sim p(\mathbf{u})p(\mathbf{v})}\mathbb{E}_{s_t\sim \mu_{s_{0:t-1},\mathcal{E}_0}}\left[\left(\frac{d\mu_{s_{0:t-1},\mathcal{E}_{\mathbf{u},\mathbf{v}}}}{d\mu_{s_{0:t-1},\mathcal{E}_{0}}}(s_t) - 1\right)^2\right] \leq \Delta, 
\end{equation}
then we have
\begin{equation}
    \Pr_{\mathbf{u},\mathbf{v}\sim p(\mathbf{u})p(\mathbf{v}), s_{0:T}\sim \mu_{\mathcal{T},\mathcal{E}_0}}\left[\frac{d\mu_{\mathcal{T},\mathcal{E}_{\mathbf{u},\mathbf{v}}}}{d\mu_{\mathcal{T},\mathcal{E}_0}}(s_{0:T}) \leq (1-\beta)\right] \leq \beta + z\Delta T,
\end{equation}
which yields
\begin{equation}
    z\Delta T\geq 2(p_{\mathrm{succ,revel}}-1/2) - 2\beta.
\end{equation}
Also note that 
\begin{equation}
\begin{aligned}
    &\mathbb{E}_{\mathbf{u},\mathbf{v}\sim p(\mathbf{u})p(\mathbf{v})}\mathbb{E}_{s_t\sim \mu_{s_{0:t-1},\mathcal{E}_0}}\left[\left(\frac{d\mu_{s_{0:t-1},\mathcal{E}_{\mathbf{u},\mathbf{v}}}}{d\mu_{s_{0:t-1},\mathcal{E}_{0}}}(s_t) - 1\right)^2\right] \\&\leq \max_{F_{s_{0:t-1}},\hat{\sigma}_{s_{0:t-1}}}\mathbb{E}_{\mathbf{u},\mathbf{v}\sim p(\mathbf{u})p(\mathbf{v})}\left[\chi^2(\mu_{s_{0:t-1},\mathcal{E}_{\mathbf{u},\mathbf{v}}}\|\mu_{s_{0:t-1},\mathcal{E}_0})\right].
\end{aligned}
\end{equation}
\end{lemma}
\begin{proof}
Note that as per our definitions, we have
\begin{equation}
    \frac{d\mu_{\mathcal{T},\mathcal{E}_{\mathbf{u},\mathbf{v}}}}{d\mu_{\mathcal{T},\mathcal{E}_0}}(s_{1:T}) = \prod_{t=1}^T\frac{d\mu_{s_{0:t-1},\mathcal{E}_{\mathbf{u},\mathbf{v}}}}{d\mu_{s_{0:t-1},\mathcal{E}_{0}}}(s_t).
\end{equation}
Consider the set of paths $s_{0:T}$ that is distributed according to the probability measure $\mu_{\mathcal{T},\mathcal{E}_{0}}$. We introduce a new variable $r_0$ which encodes the values of $(\mathbf{u},\mathbf{v})$ and is itself randomly sampled following the distribution $p(\mathbf{u})p(\mathbf{v})$ and then label variables $r_t = s_t$ for $t = 1,\dots T$. Each path can then be seen as a sequence $(r_0,r_1,r_2,\dots)$ which is a stochastic process as a result of it being sampled according to a probability distribution. Notably, the learner only makes use of $s_{0:T}$ and doesn't have access to the sampled values of $\mathbf{u},\mathbf{v}$ since $s_0$ is a dummy variable and isn't related to $r_0$. We now define $Y_t$ as a function of $r_{0:t}$ given by
\begin{equation}
    Y_t(r_{0:t}) = \frac{d\mu_{s_{0:t-1},\mathcal{E}_{r_0}}}{d\mu_{s_{0:t-1},\mathcal{E}_{0}}}(r_t) - 1, \quad L_T = \prod_{t=1}^{T}(1+Y_t).
\end{equation}
Now referring to Lemma 21 from \cite{chen2024optimaltradeoffsestimatingpauli}, we note that for any stochastic process $(r_0,r_1,\dots)$ with random variables $Y_t=Y_t(r_{0:t})\in[-1,\infty)$ that satisfy $\mathbb{E}[Y_t|r_{0:t-1}]\geq -\varsigma$ and $\mathbb{E}[Y_t^2|r_{0:t-1}]\leq \Delta$ for $\varsigma\in[0,1]$ and $\Delta\geq 0$, there exist constants $c_3,c_4>0$ such that for all $T\in\mathbb{N}$, $c_1\in(0,1/3)$ and $c_2>0$ the random variable $L_T = \prod_{t=1}^{T}(1+Y_t)$ satisfies
\begin{equation}\label{eq:lemma21exact}
    \Pr\left[L_T \leq e^{-T\varsigma}(1-c_1)\right] \leq c_2e^{c_4T\varsigma^2} + c_3\Delta T.
\end{equation}
The expectation values here are taken over all possible stochastic sequences $r_{0:T}$. Note that we have this exact same setup, specifically with $\varsigma = 0$ and
\begin{equation}
    \mathbb{E}[Y_t^2|r_{0:t-1}] =  \mathbb{E}_{\mathbf{u},\mathbf{v}\sim p(\mathbf{u})p(\mathbf{v})}\mathbb{E}_{s_t\sim \mu_{s_{0:t-1},\mathcal{E}_0}}\left[\left(\frac{d\mu_{s_{0:t-1},\mathcal{E}_{\mathbf{u},\mathbf{v}}}}{d\mu_{s_{0:t-1},\mathcal{E}_{0}}}(s_t) - 1\right)^2\right] \leq \Delta. 
\end{equation}
The exact values of the constants will be
\begin{equation}
    c_3 = \eta_3^{-2} + \frac{12+3\eta_1^{-2}}{\eta_2} + \frac{1-\eta_3^{-2}}{\eta_1},\quad c_4 = \frac{6\eta_1^2}{(2\eta_2 + 4\eta_1\eta_3/3)^2} ,
\end{equation}
where we choose the constants $\eta_1,\eta_2$ and $\eta_3$ to satisfy
\begin{equation}
    e^{-2\eta_1} = 1-c_1, \quad -\frac{\eta_1^2}{2\eta_2 + 6T\varsigma^2 + 4\eta_1\eta_3/3} \leq  \ln(c_2) + c_4T\varsigma^2,
\end{equation}
to ensure the validity of Eq.~\eqref{eq:lemma21exact}. The reasons for these constants can be found in the proof of Lemma 21 in \cite{chen2024optimaltradeoffsestimatingpauli}. Now, on setting $\varsigma = 0$ and $c_1 = c_2 = \beta$, we obtain the value of $z$ in the stated lemma to equal $c_3$.
\end{proof}
\begin{lemma}[Extended version of Lemma~\ref{lem:master}]
    Any $c$-copy protocol with arbitrarily large ancillary assistance and adaptive state preparation and measurements will require a depth $T = \Omega(1/\Delta)$ to succeed with constant probability above $1/2$ in Problem~\ref{prob:manyrevel} where
 \begin{equation}
 \begin{aligned}
     \Delta &= \Bigg(\sum_{\pmb{x}\in\{0,\dots,l\}^c\setminus\{\pmb{0}\}}\max_{\hat{\sigma}\in D(\mathcal{H}_{\mathrm{in}}^{\otimes c}\otimes\mathcal{H}_{\mathrm{out}}^{\otimes c})}\left\{\mathbb{E}_{\mathbf{u},\mathbf{v}}\left[\text{Tr}\left[\hat{\sigma}\left(\hat{W}_{\pmb{x},\mathbf{v}}\otimes \hat{K}_{\pmb{x},\mathbf{u}}^T\right)\right]^2\right] \right\}^{1/2}\Bigg)^{2}\\
     &\leq \left(\sum_{\pmb{x}\in\{0,\dots,l\}^c\setminus\{\pmb{0}\}}\sqrt{\left\|\mathbb{E}_{\mathbf{v}}\left[\hat{W}_{\pmb{x},\mathbf{v}}^{\otimes 2}\right]\right\|_{\mathrm{op}}\left\|\mathbb{E}_{\mathbf{u}}\left[\hat{K}_{\pmb{x},\mathbf{u}}^{\otimes 2}\right]\right\|_{\mathrm{op}}}\right)^2,
 \end{aligned}
    \end{equation}
    where transposition is defined in some fixed basis choice, and we have defined
    \begin{equation}
        \hat{W}_{\pmb{x},\mathbf{v}} = \bigotimes_{i=1}^c\hat{W}_{x_i,\mathbf{v}},\quad\hat{K}_{\pmb{x},\mathbf{u}} = \bigotimes_{i=1}^c\hat{K}_{x_i,\mathbf{u}}.
    \end{equation}
\end{lemma}
\begin{proof}
We will first prove that for any $c$-copy learning protocol at partial path $s_{0:t-1}$, the averaged $\chi^2$-divergence over all $\mathbf{u},\mathbf{v}$ between the two measures $\mu_{s_{0:t-1},\mathcal{E}_0}$ and $\mu_{s_{0:t-1},\mathcal{E}_{\mathbf{u},\mathbf{v}}}$ is upper bounded by the value in Eq.~\eqref{eq:masterlemmadeltaup}.

We will drop the subscripts of $s_{0:t-1}$ assuming the dependence to be implicit. Consider that we have an input state $\hat{\sigma}\in D(\mathcal{H}_{\mathrm{in}}^{\otimes c}\otimes\mathcal{H}_{\mathrm{anc}})$ and a measurable space $(S,\Sigma_S)$ with POVM $F:\Sigma_S\to B(\mathcal{H}_{\mathrm{out}}^{\otimes c}\otimes\mathcal{H}_{\mathrm{in}}^{\otimes c})$. Let $\mu_{\mathcal{E}_0}$ and $\mu_{\mathcal{E}_{\mathbf{u},\mathbf{v}}}$ be the probability measures obtained by acting this POVM on the states $(\mathcal{E}_0^{\otimes c}\otimes\mathbb{I})(\hat{\sigma})$ and $(\mathcal{E}_{\mathbf{u},\mathbf{v}}^{\otimes c}\otimes\mathbb{I})(\hat{\sigma})$, respectively. One can verify that the relevant metric which we will be maximizing corresponding to $\chi^2(\mu_{\mathcal{E}_{\mathbf{u},\mathbf{v}}}\|\mu_{\mathcal{E}_0})$ is convex in terms of the input state $\hat{\sigma} = \sum_{k}w_k\ket{\psi_k}\bra{\psi_k}$. Hence for the rest of the discussion we will be treating the state $\hat{\sigma}$ to be pure and transposition of the operator $\hat{K}_{\pmb{x},\mathbf{u}}$ is performed in a basis that contains the state $\ket{\psi}$ for $\hat{\sigma} = \ket{\psi}\bra{\psi}$. By using the positive-operator-valued measure $F_{\hat{\sigma}}:\Sigma_S\to L_1(\mathcal{H}_{\mathrm{out}}^{\otimes c}\otimes\mathcal{H}_{\mathrm{anc}})$ defined in Eq.~\eqref{eq:POVMwithstate} we obtain the following simplification for $s\in S$ and $B(s,r)$ being a ball of radius $r$ centered around $s$
\begin{equation}
    \left(\frac{d\mu_{\mathcal{E}_{\mathbf{u},\mathbf{v}}}}{d\mu_{\mathcal{E}_{0}}}(s) - 1\right) = \sum_{\pmb{x}\in\{0,\dots,l\}^c\setminus\{\pmb{0}\}}\lim_{r\to0}\frac{\text{Tr}\left[F_{\hat{\sigma}}(B(s,r))(\hat{W}_{\pmb{x},\mathbf{v}}\otimes \hat{K}_{\pmb{x},\mathbf{u}}^T)\right]}{\text{Tr}\left[F_{\hat{\sigma}}(B(s,r))\right]}.
    \end{equation}
From Lemma~\ref{lem:finiteRNder} the Radon--Nikodym derivative is guaranteed to be finite wherever it is defined. We write down the $\chi^2$-divergence between the two measures
    \begin{equation}
        \chi^2(\mu_{\mathcal{E}_{\mathbf{u},\mathbf{v}}}\|\mu_{\mathcal{E}_0}) = \int_{s_t\in S}d\mu_{\mathcal{E}_{0}}(s_t) \left(\frac{d\mu_{\mathcal{E}_{\mathbf{u},\mathbf{v}}}}{d\mu_{\mathcal{E}_{0}}}(s_t) - 1\right)^2.
    \end{equation}
    We now try to maximize the averaged divergence to take into account the randomness involved in the choice of $\mathbf{u}$ and $\mathbf{v}$. To maximize distinguishability, we must maximize the quantity $\Delta$ over all possible POVMs $F$ and input states $\hat{\sigma}$,
    \begin{align}
    \Delta  &= \max_{F,{\hat{\sigma}}} \mathbb{E}_{\mathbf{u},\mathbf{v}}[\chi^2(\mu_{\mathcal{E}_{\mathbf{u},\mathbf{v}}}\|\mu_{\mathcal{E}_{0}})]\\
    &= \max_{F,{\hat{\sigma}}}\mathbb{E}_{\mathbf{u},\mathbf{v}}\left[\int_{s_t\in S} d\mu_{\mathcal{E}_0}\left(\sum_{\pmb{x}\in\{0,\dots,l\}^c\setminus\{\pmb{0}\}}\frac{\sqrt{v_{\pmb{x}}}}{\sqrt{v_{\pmb{x}}}}\lim_{r\to0}\frac{\text{Tr}\left[F_{\hat{\sigma}}(B(s_t,r))(\hat{W}_{\pmb{x},\mathbf{v}}\otimes \hat{K}_{\pmb{x},\mathbf{u}}^T)\right]}{\text{Tr}\left[F_{\hat{\sigma}}(B(s_t,r))\right]}\right)^2\right]\\
    &\leq \max_{F,{\hat{\sigma}}}\mathbb{E}_{\mathbf{u},\mathbf{v}}\left[\int_{s_t\in S} d\mu_{\mathcal{E}_0}(s_t)\left(\sum_{\pmb{x}\in\{0,\dots,l\}^c\setminus\{\pmb{0}\}}v_{\pmb{x}}\right)\left(\sum_{\pmb{x}\in\{0,\dots,l\}^c\setminus\{\pmb{0}\}}\frac{1}{v_{\pmb{x}}}\left(\lim_{r\to0}\frac{\text{Tr}\left[F_{\hat{\sigma}}(B(s_t,r))(\hat{W}_{\pmb{x},\mathbf{v}}\otimes \hat{K}_{\pmb{x},\mathbf{u}}^T)\right]}{\text{Tr}\left[F_{\hat{\sigma}}(B(s_t,r))\right]}\right)^2\right)\right]\\
    &\leq \left(\sum_{\pmb{x}\in\{0,\dots,l\}^c\setminus\{\pmb{0}\}}\frac{1}{v_{\pmb{x}}}\max_{F,\hat{\sigma}}\left(\int_{s_t\in S}d\mu_{\mathcal{E}_0}(s_t)\lim_{r\to0}\frac{\mathbb{E}_{\mathbf{u},\mathbf{v}}\text{Tr}\left[F_{\hat{\sigma}}(B(s_t,r))(\hat{W}_{\pmb{x},\mathbf{v}}\otimes \hat{K}_{\pmb{x},\mathbf{u}}^T)\right]^2}{\text{Tr}\left[F_{\hat{\sigma}}(B(s_t,r))\right]^2}\right)\right)\left(\sum_{\pmb{x}\in\{0,\dots,l\}^c\setminus\{\pmb{0}\}}v_{\pmb{x}}\right).
    \end{align}
    The above upper bound is valid for all choices of weights $v_{\pmb{x}}$. Hence to minimize the upper bound (following similar steps as \cite{ller2025infinitehierarchymulticopyquantum}), we define the weights $v_{\pmb{x}}$ as follows
    \begin{equation}
        v_{\pmb{x}} = \sqrt{\max_{F,\hat{\sigma}}\left(\int_{s_t\in S}d\mu_{\mathcal{E}_0}(s_t)\lim_{r\to0}\frac{\mathbb{E}_{\mathbf{u},\mathbf{v}}\text{Tr}\left[F_{\hat{\sigma}}(B(s_t,r))(\hat{W}_{\pmb{x},\mathbf{v}}\otimes \hat{K}_{\pmb{x},\mathbf{u}}^T)\right]^2}{\text{Tr}\left[F_{\hat{\sigma}}(B(s_t,r))\right]^2}\right)},
    \end{equation}
    which then gives
    \begin{equation}\label{eq:deltaupboundmaster}
        \Delta \leq  \left(\sum_{\pmb{x}\in\{0,\dots,l\}^c\setminus\{\pmb{0}\}}\sqrt{\max_{F,\hat{\sigma}}\left(\int_{s_t\in S}d\mu_{\mathcal{E}_0}(s_t)\lim_{r\to0}\frac{\mathbb{E}_{\mathbf{u},\mathbf{v}}\text{Tr}\left[F_{\hat{\sigma}}(B(s_t,r))(\hat{W}_{\pmb{x},\mathbf{v}}\otimes \hat{K}_{\pmb{x},\mathbf{u}}^T)\right]^2}{\text{Tr}\left[F_{\hat{\sigma}}(B(s_t,r))\right]^2}\right)}\right)^2.
    \end{equation}
    As we noted previously, the operator $F_{\hat{\sigma}}(B(s_t,r))$ stays trace class and positive. Hence we can define for $\pmb{x}\in\{0,\dots,l\}^c$ the state $\hat{\sigma}_{\pmb{x}}\in D(\mathcal{H}_{\mathrm{in}}^{\otimes c}\otimes \mathcal{H}_{\mathrm{out}}^{\otimes c})$ such that
    \begin{equation}
        \hat{\sigma}_{\pmb{x}} = \argmax_{\hat{\rho}\in D(\mathcal{H}_{\mathrm{in}}^{\otimes c}\otimes\mathcal{H}_{\mathrm{out}}^
        {\otimes c})}\mathbb{E}_{\mathbf{u},\mathbf{v}}\left[\text{Tr}\left[\hat{\rho}(\hat{W}_{\pmb{x},\mathbf{v}}\otimes \hat{K}_{\pmb{x},\mathbf{u}}^T)\right]^2\right]
    \end{equation}
    we have
    \begin{equation}    \lim_{r\to0}\frac{\mathbb{E}_{\mathbf{u},\mathbf{v}}\text{Tr}\left[F_{\hat{\sigma}}(B(s_t,r))(\hat{W}_{\pmb{x},\mathbf{v}}\otimes \hat{K}_{\pmb{x},\mathbf{u}}^T)\right]^2}{\text{Tr}\left[F_{\hat{\sigma}}(B(s_t,r))\right]^2} \leq \mathbb{E}_{\mathbf{u},\mathbf{v}}\left[\text{Tr}\left[\hat{\sigma}_{\pmb{x}}(\hat{W}_{\pmb{x},\mathbf{v}}\otimes \hat{K}_{\pmb{x},\mathbf{u}}^T)\right]^2\right]    
    \end{equation}
    For any probability measure $\mu$ and function $f:S\to\mathbb{R}$ with $f(s)\leq f_{\max}$, for all $s\in S$, we have $\int_{s\in S} d\mu(s) f(s)\leq f_{\max}$. Plugging this into the inequality for $\Delta$, we finally get
    \begin{equation}
    \begin{aligned}
        \Delta&\leq \Bigg(\sum_{\pmb{x}\in\{0,\dots,l\}^c\setminus\{\pmb{0}\}}\max_{\hat{\sigma}\in D(\mathcal{H}_{\mathrm{in}}^{\otimes c}\otimes\mathcal{H}_{\mathrm{out}}^{\otimes c})}\left\{\mathbb{E}_{\mathbf{u},\mathbf{v}}\left[\text{Tr}\left[\hat{\sigma}\left(\hat{W}_{\pmb{x},\mathbf{v}}\otimes \hat{K}_{\pmb{x},\mathbf{u}}^T\right)\right]^2\right] \right\}^{1/2}\Bigg)^{2}\\
        &\leq \left(\sum_{\pmb{x}\in\{0,\dots,l\}^c\setminus\{\pmb{0}\}}\sqrt{\left\|\mathbb{E}_{\mathbf{v}}\left[\hat{W}_{\pmb{x},\mathbf{v}}^{\otimes 2}\right]\right\|_{\mathrm{op}}\left\|\mathbb{E}_{\mathbf{u}}\left[\hat{K}_{\pmb{x},\mathbf{u}}^{\otimes 2}\right]\right\|_{\mathrm{op}}}\right)^2.
    \end{aligned}
    \end{equation}
For the last inequality we use the fact that 
\begin{equation}
    \begin{aligned}
    &\lim_{r\to0}\frac{\mathbb{E}_{\mathbf{u},\mathbf{v}}\text{Tr}\left[F_{\hat{\sigma}}(B(s_t,r))(\hat{W}_{\pmb{x},\mathbf{v}}\otimes \hat{K}_{\pmb{x},\mathbf{u}}^T)\right]^2}{\text{Tr}\left[F_{\hat{\sigma}}(B(s_t,r))\right]^2} \\&= \lim_{r\to0}\frac{\text{Tr}\left[(F_{\hat{\sigma}}(B(s_t,r))\otimes F_{\hat{\sigma}}(B(s_t,r)))\mathbb{E}_{\mathbf{u},\mathbf{v}}\left[\hat{W}_{\pmb{x},\mathbf{v}}\otimes \hat{K}_{\pmb{x},\mathbf{u}}^T\otimes \hat{W}_{\pmb{x},\mathbf{v}}\otimes \hat{K}_{\pmb{x},\mathbf{u}}^T\right]\right]}{\text{Tr}\left[(F_{\hat{\sigma}}(B(s_t,r))\otimes F_{\hat{\sigma}}(B(s_t,r)))\right]}
        \\&\leq \left\|\mathbb{E}_{\mathbf{v}}[\hat{W}_{\pmb{x},\mathbf{v}}^{\otimes 2}]\otimes \mathbb{E}_{\mathbf{u}}[\hat{K}_{\pmb{x},\mathbf{u}}^{\otimes 2}]\right\|_{\mathrm{op}}\\
            &\leq \left\|\mathbb{E}_{\mathbf{v}}[\hat{W}_{\pmb{x},\mathbf{v}}^{\otimes 2}]\right\|_{\mathrm{op}}\left\|\mathbb{E}_{\mathbf{u}}[\hat{K}_{\pmb{x},\mathbf{u}}^{\otimes 2}]\right\|_{\mathrm{op}},
    \end{aligned}
\end{equation}
where we have used the fact that for any bounded operator $\hat
    O\in B(\mathcal{H})$ and trace-class positive $\hat{\rho}\in L_1(\mathcal{H})_+$, we have $\frac{\text{Tr}(\hat{O}\hat{\rho})}{\text{Tr}(\hat
    \rho)}\leq \|\hat{O}\|_{\mathrm{op}}$ and for two bounded operators $\hat{O}_1,\hat{O}_2$, we have $\|\hat{O}_1\otimes\hat{O}_2^T\|_{\mathrm{op}}\leq \|\hat{O}_1\|_{\mathrm{op}}\times\|\hat{O}_2\|_{\mathrm{op}}$. From this we have effectively proven that for all $c$-copy learning protocols, at all partial paths $s_{0:t-1}$ the following holds
    \begin{equation}
    \begin{aligned}
        &\mathbb{E}_{\mathbf{u},\mathbf{v}\sim p(\mathbf{u})p(\mathbf{v})}\mathbb{E}_{s_t\sim \mu_{s_{0:t-1},\mathcal{E}_0}}\left[\left(\frac{d\mu_{s_{0:t-1},\mathcal{E}_{\mathbf{u},\mathbf{v}}}}{d\mu_{s_{0:t-1},\mathcal{E}_{0}}}(s_t) - 1\right)^2\right]\\
        &\leq \Bigg(\sum_{\pmb{x}\in\{0,\dots,l\}^c\setminus\{\pmb{0}\}}\max_{\hat{\sigma}\in D(\mathcal{H}_{\mathrm{in}}^{\otimes c}\otimes\mathcal{H}_{\mathrm{out}}^{\otimes c})}\left\{\mathbb{E}_{\mathbf{u},\mathbf{v}}\left[\text{Tr}\left[\hat{\sigma}\left(\hat{W}_{\pmb{x},\mathbf{v}}\otimes \hat{K}_{\pmb{x},\mathbf{u}}^T\right)\right]^2\right] \right\}^{1/2}\Bigg)^{2}\\
        &\leq \left(\sum_{\pmb{x}\in\{0,\dots,l\}^c\setminus\{\pmb{0}\}}\sqrt{\left\|\mathbb{E}_{\mathbf{v}}\left[\hat{W}_{\pmb{x},\mathbf{v}}^{\otimes 2}\right]\right\|_{\mathrm{op}}\left\|\mathbb{E}_{\mathbf{u}}\left[\hat{K}_{\pmb{x},\mathbf{u}}^{\otimes 2}\right]\right\|_{\mathrm{op}}}\right)^2.
    \end{aligned}
\end{equation}
    
    Hence if any learning protocol succeeds with probability $1/2+\eta$, we can choose $\beta = 2\eta/3$ (this will be less than 1/3) and then apply Lemma~\ref{lem:martingale} to obtain $z\Delta T\geq 2\eta/3$ for some constant $z>0$. This will finally give the lower bound on $T$ to be $\Omega(1/\Delta)$, which gives us the claimed lower bound.
\end{proof}

\begin{corollary}\label{corr:master}
    Assume that for the $c$-copy learning protocol that succeeds in many-one channel discrimination with revelation (Problem~\ref{prob:manyrevel}) we have a specification that $\mathcal{E}_{\mathbf{u},\mathbf{v}}(\cdot) = \text{Tr}(\cdot)\hat{\rho}_0 + \epsilon_0\text{Tr}\left[(\hat{\Pi}_{+,\mathbf{u}}-\hat{\Pi}_{-,\mathbf{u}})(\cdot)\right]\hat{\rho}_0^{1/2}(\hat{W}_{\mathbf{v}})\hat{\rho}_0^{1/2}$ where $\{\hat{\Pi}_{\pm,\mathbf{u}}\}$ is a 2-outcome POVM parameterized by $\mathbf{u}$ and $\hat{W}_{\mathbf{v}}$ is a bounded operator parameterized by $\mathbf{v}$. If we have
    \begin{equation}
        \left\|\mathbb{E}_{\mathbf{u}}\left[(\hat{\Pi}_{+,\mathbf{u}} - \hat{\Pi}_{-,\mathbf{u}})^{\otimes 2|\pmb{x}|}\right]\right\|_{\mathrm{op}} \leq \frac{c_1}{d_{\mathrm{in}}},\quad\left\|\mathbb{E}_{\mathbf{v}}\left[\hat{W}_{\mathbf{v}}^{\otimes 2|\pmb{x}|}\right]\right\|_{\mathrm{op}} \leq \frac{c_2(c_3)^{2|\pmb{x}|}}{d_{\mathrm{out}}},
    \end{equation}
    we then have
    \begin{equation}
        \Delta \leq \frac{c_1c_2}{d_{\mathrm{in}}d_{\mathrm{out}}}\left(\sum_{\pmb{x}\in\{0,1\}^c\setminus\{\pmb{0}\}}c_3^{|\pmb{x}|}\epsilon_0^{|\pmb{x}|}\right)^2= \frac{c_1c_2}{d_{\mathrm{in}}d_{\mathrm{out}}}((1+c_3\epsilon_0)^c-1)^2.
    \end{equation}
    Hence if we assume that $c \leq c_4/\epsilon_0$, we have $((1+c_3\epsilon_0)^c-1)^2 \leq c^2c_3^2\epsilon_0^2e^{2c_3c_4}$ (using Lemma~\ref{lem:epssum}). Assuming all $c_1,c_2,c_3,c_4\in O(1)$, this gives $T = \Omega(d_{\mathrm{in}}d_{\mathrm{out}}c^{-2}\epsilon_0^{-2})$ for any learning protocol that succeeds in Problem~\ref{prob:manyrevel} with success probability strictly above 1/2.
\end{corollary}

\begin{corollary}[Reduction for state learning]\label{corr:statelearn}
    Consider a specific case of many-one channel discrimination with revelation (Problem~\ref{prob:manyrevel}) where for all $x\in\{0,\dots,l\}$ the operators $\hat{K}_{x,\mathbf{u}} = \mathbb{I}$. This now reduces to a state discrimination problem. Any learning protocol capable of discriminating the state $\hat{\rho}_0$ and $\hat{\rho}_{\mathbf{v}} = \sum_{x\in\{0,\dots,l\}}\sqrt{\hat{\rho}_0}\hat{W}_{x,\mathbf{v}}\sqrt{\hat{\rho}_0}$ with success probability above $1/2$ requires $T = \Omega(1/\Delta)$ measurements where
    \begin{equation}
        T = \Omega\left[\left(\sum_{\pmb{x}\in\{0,\dots,l\}^c\setminus\{\pmb{0}\}}\sqrt{\max_{\hat{\Pi}\in B(\mathcal{H}_{\mathrm{out}}^{\otimes c})_+}\frac{\text{Tr}\left[(\hat{\Pi}\otimes\hat
        \Pi)\mathbb{E}_{\mathbf{v}}\left[\left(\sqrt{\hat{\rho}_0^{\otimes c}}\hat{W}_{\pmb{x},\mathbf{v}}\sqrt{\hat{\rho}_0^{\otimes c}}\right)^{\otimes 2}\right]\right]}{\text{Tr}\left[(\hat{\Pi}\otimes\hat
        \Pi)(\hat{\rho}_0^{\otimes c}\otimes\hat{\rho}_0^{\otimes c})\right]}}\right)^{-2}\right],
    \end{equation}
    which also gives a valid possibly weaker lower bound of
    \begin{equation}
        T = \Omega\left[\left(\sum_{\pmb{x}\in\{0,\dots,l\}^c\setminus\{\pmb{0}\}}\sqrt{\left\|\mathbb{E}_{\mathbf{v}}\left[\hat{W}_{\pmb{x},\mathbf{v}}^{\otimes 2}\right]\right\|_{\mathrm{op}}}\right)^{-2}\right].
    \end{equation}
\end{corollary}
\begin{proof}
    We revisit Eq.~\eqref{eq:deltaupboundmaster} in the proof of Lemma~\ref{lem:master}, substituting $\hat{K}_{\pmb{x},\mathbf{u}} = \mathbb{I}$, which gives
    \begin{equation}
        \Delta \leq  \left(\sum_{\pmb{x}\in\{0,\dots,l\}^c\setminus\{\pmb{0}\}}\sqrt{\max_{F,\hat{\sigma}}\left(\int_{s_t\in S}d\mu_{\mathcal{E}_0}(s_t)\lim_{r\to0}\frac{\mathbb{E}_{\mathbf{u},\mathbf{v}}\text{Tr}\left[F_{\hat{\sigma}}(B(s_t,r))(\hat{W}_{\pmb{x},\mathbf{v}}\otimes \mathbb{I})\right]^2}{\text{Tr}\left[F_{\hat{\sigma}}(B(s_t,r))\right]^2}\right)}\right)^2.
    \end{equation}
    Observe that
    \begin{equation}
    \begin{aligned}\label{eq:replacementchannel}
        \frac{\mathbb{E}_{\mathbf{u},\mathbf{v}}\text{Tr}\left[F_{\hat{\sigma}}(B(s_t,r))(\hat{W}_{\pmb{x},\mathbf{v}}\otimes \mathbb{I})\right]^2}{\text{Tr}\left[F_{\hat{\sigma}}(B(s_t,r))\right]^2} &= \frac{\mathbb{E}_{\mathbf{u},\mathbf{v}}\text{Tr}\left[F(B(s_t,r))\left(\sqrt{\hat{\rho}_0^{\otimes c}}\hat{W}_{\pmb{x},\mathbf{v}}\sqrt{\hat{\rho}_0^{\otimes c}}\otimes \hat{\sigma}_{\mathrm{anc}}\right)\right]^2}{\text{Tr}\left[F(B(s_t,r))\left(\hat{\rho}_0^{\otimes c}\otimes \hat{\sigma}_{\mathrm{anc}}\right)\right]^2}\\
        &= \frac{\mathbb{E}_{\mathbf{u},\mathbf{v}}\text{Tr}\left[\hat{\Pi}_{F,\hat{\sigma}}(B(s_t,r))\sqrt{\hat{\rho}_0^{\otimes c}}\hat{W}_{\pmb{x},\mathbf{v}}\sqrt{\hat{\rho}_0^{\otimes c}}\right]^2}{\text{Tr}\left[\hat{\Pi}_{F,\hat{\sigma}}(B(s_t,r))\hat{\rho}_0^{\otimes c}\right]^2},
    \end{aligned}
    \end{equation}
    where we define the positive-operator-valued measure $\hat{\Pi}_{F,\hat{\sigma}}:\Sigma_S\to B(\mathcal{H}_{\mathrm{out}}^{\otimes c})$ as
    \begin{equation}
        \hat{\Pi}_{F,\hat{\sigma}}(E) = \text{Tr}_{\mathrm{anc}}\left[F(E)\left(\mathbb{I}\otimes \hat{\sigma}_{\mathrm{anc}}\right)\right].
    \end{equation}
    We are able to define the partial trace since the operator $F(E)\left(\mathbb{I}\otimes \hat{\sigma}_{\mathrm{anc}}\right)$ is actually in the projective tensor product space $B(\mathcal{H}_{\mathrm{out}}^{\otimes c})\hat{\otimes}L_1(\mathcal{H}_{\mathrm{anc}})$. The above simplification is the direct consequence of the fact that both the channels are complete replacement channels. Hence no ancilla assistance or input state modification has any effect on the learning protocol, making it depend only on the POVM applied on the space $\mathcal{H}_{\mathrm{out}}^{\otimes c}$. Since $\Delta$ is upper bounded by a value which is obtained by a limit, we can note that finding an upper bound for every possible ball of radius $r>0$ would suffice to work as an upper bound of the limit. Using Eq.~\eqref{eq:replacementchannel}, we can directly note that this gives
    \begin{equation}
        \Delta \leq  \left(\sum_{\pmb{x}\in\{0,\dots,l\}^c\setminus\{\pmb{0}\}}\sqrt{\max_{\hat{\Pi}\in B(\mathcal{H}_{\mathrm{out}}^{\otimes c})_+}\frac{\mathbb{E}_{\mathbf{u},\mathbf{v}}\text{Tr}\left[\hat{\Pi}\sqrt{\hat{\rho}_0^{\otimes c}}\hat{W}_{\pmb{x},\mathbf{v}}\sqrt{\hat{\rho}_0^{\otimes c}}\right]^2}{\text{Tr}\left[\hat{\Pi}\hat{\rho}_0^{\otimes c}\right]^2}}\right)^2,
    \end{equation}
    following which the claims in the corollary now directly follow.
\end{proof}

As a guide to reading the proofs of the theorems that can be found in the subsections that follow, we note that most of the heavy lifting is done by the master lemma (Lemma~\ref{lem:master}). Specifically, all the proofs are structured as follows
\begin{itemize}
    \item We first define two channels: one is a complete replacement channel and the other one differs from the complete replacement channel only if the input state has overlap with a particular 2-outcome POVM. Both the POVM and the output state differing from the replacement channel are parameterized by random variables.
    \item We show that a learner should be able to tell apart these two cases with probability strictly above $1/2$ if framed as an instance of many-one channel discrimination with revelation (Problem~\ref{prob:manyrevel}) by learning a specific query to the channel learning task.
    \item By finding appropriate operator norms and other upper bounds, we find a way to upper bound the value of $\Delta$ from the corollaries of Lemma~\ref{lem:master}. 
    \item This gives us a lower bound on the total number of channel queries $N = cT$, which can depend on the range of $c$ considered which we state in each respective theorem.
\end{itemize}

For assistance in these proofs, we also introduce a few additional lemmata.

\begin{lemma}\label{lem:epssum}
    For $x\in [0,1)$ the following hold for $c\in \mathbb{N}$ ($L_c(x)$ is the standard $c$th Legendre polynomial)
    \begin{itemize}
        \item $cxe^{cx}\geq (1+x)^c-1$
        \item $c^2x^2e^{2cx} \geq (1-x^2)^cL_c\left(\frac{1+x^2}{1-x^2}\right) -1$
    \end{itemize}
\end{lemma}
\begin{proof}
    Note that the binomial coefficient always is upper bounded as 
    \begin{equation}
        \binom{c}{k}\leq \frac{c^k}{k!},
    \end{equation}
    using which we have
    \begin{equation}
        \begin{aligned}
            (1+x)^c - 1 &= \sum_{k=1}^{\infty}\binom{c}{k}x^k\leq \sum_{k=1}^{\infty}\frac{c^k}{k!}x^k \leq cx \sum_{k=0}^{\infty}\frac{(cx)^k}{k!}
            = cxe^{cx},
        \end{aligned}
    \end{equation}
    and similarly
    \begin{equation}
    \begin{aligned}
        (1-x^2)^cL_c\left(\frac{1+x^2}{1-x^2}\right) -1 &= \sum_{k=1}^{c}\binom{c}{k}^2 x^{2k}\leq \left(\sum_{k=1}^c \binom{c}{k}x^k\right)^2
        \leq (cx)^2 e^{2cx}.
    \end{aligned}
    \end{equation}
\end{proof}

\begin{lemma}\label{lem:gaussprob}
    For $\gamma\in\mathbb{R}^n$ such that each $\gamma_i$ is sampled from a Gaussian distribution of spread $\sigma$ centered at $0$, we have (for all $n\geq 8$)
    \begin{equation}
        \Pr_{\gamma}\left[2\sigma^2\leq |\gamma|^2\leq \frac{n\sigma^2}{0.99}\right] \geq 0.546.
    \end{equation}
\end{lemma}
\begin{proof}
\begin{figure}
    \centering
    \includegraphics[width=0.4\linewidth]{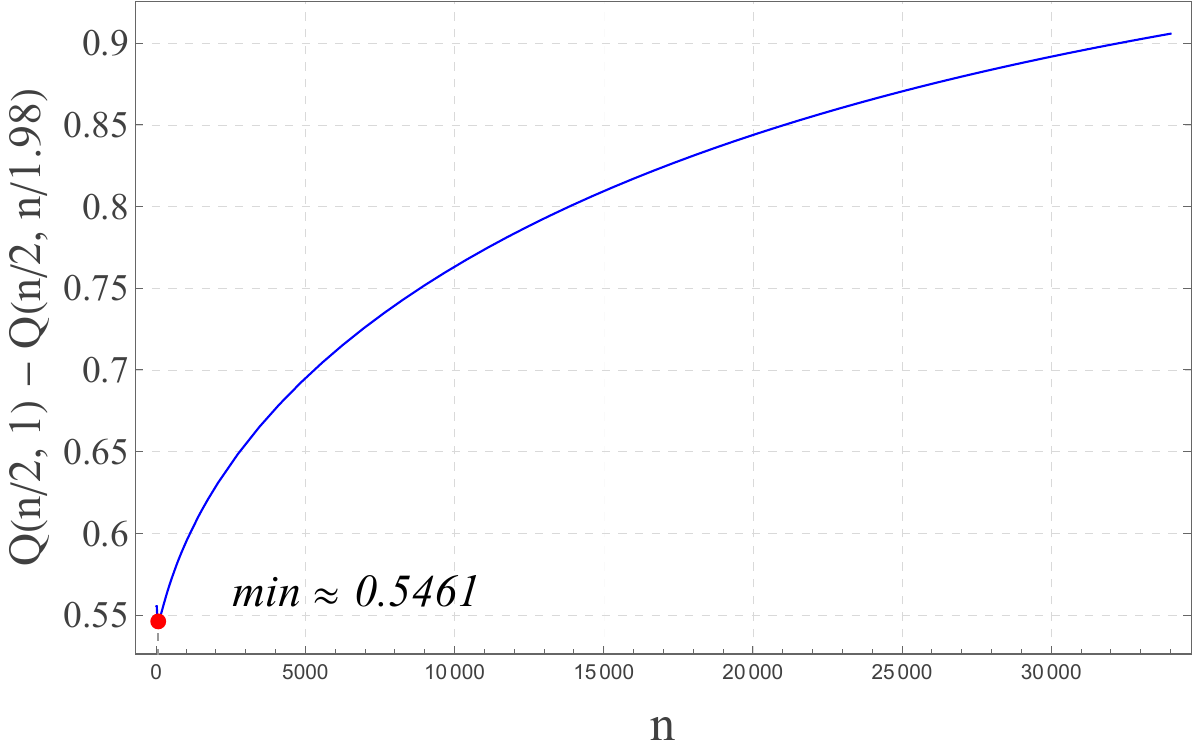}
    \caption{Manual verification for Lemma~\ref{lem:gaussprob} for the range $8\leq n\leq 34000$.}
    \label{fig:gaussverify}
\end{figure}
    Since $\frac{|\gamma|^2}{\sigma^2}$ follows a $\chi^2$ distribution with $n$ degrees of freedom, we have 
    \begin{equation}
        \Pr_{\gamma}\left[2\sigma^2\leq |\gamma|^2\leq \frac{n\sigma^2}{0.99}\right] = Q(n/2,1) - Q(n/2,n/1.98),
    \end{equation}
    where $Q(n,x) = \frac{\Gamma(n,x)}{\Gamma(n)}$ is the regularized gamma function ($\Gamma(n,x) = \int_{x}^\infty t^{n-1}e^{-t}dt$ and $\Gamma(n) = \Gamma(n,0)$). Since the above is only dependent on $n$, we manually verify the above inequality in the range $8\leq n\leq 34000$ using Fig.~\ref{fig:gaussverify}.

    For $n\geq 34000$, we make use of the tail bound (taken from \cite{ghosh2021exponential}) and a trivial lower bound for $Q(n/2,1)$ as 
    \begin{equation}
    \begin{aligned}
        Q(n/2,kn/2) &\leq (ke^{1-k})^{n/2},\forall k>1,
        \\Q(n/2,1) &= 1-\frac{\int_{0}^1 t^{-1+n/2}e^{-t}dt}{\Gamma(n/2)} \geq 1 - \frac{2}{n\Gamma(n/2)},
    \end{aligned}
    \end{equation}
    using which we get for $n\geq 34000$ (by numerical verification)
    \begin{equation}
        \Pr_{\gamma}\left[2\sigma^2\leq |\gamma|^2\leq \frac{n\sigma^2}{0.99}\right] \geq 1 - \frac{8500}{\Gamma(17000)} - (e^{-1/99}/0.99)^{17000} > 0.577.
    \end{equation}
    
\end{proof}

\subsection{Lower bounds in the absence of the complex-conjugate channel}\label{app:lowerbound_noconjugate}
To establish the lower bounds for the following tasks, we will first find bounds for operator norms which we can directly then plug into Corollary~\ref{corr:master} to obtain our results. Specifically we show hardness of the channel learning task by finding the hardness of a strictly easier task which is to succeed in particular examples of many-one channel discrimination with revelation (Problem~\ref{prob:manyrevel}). 

\begin{lemma}\label{lem:discr_disp_tensor}
    Assuming prime $d$ and $k$ being a positive integer,
    \begin{equation}
        \left\|\sum_{\mathbf{q},\mathbf{p}\in\mathbb{F}^{m}_d}\hat{D}_{d,m}(\mathbf{q},\mathbf{p})^{\otimes 2k}\right\|_\mathrm{op} = \begin{cases}
            d^m,& k\neq 0 \mod d\\
            d^{2m}, & k =0\mod d
        \end{cases}.
    \end{equation}
\end{lemma}
\begin{proof}
We first examine how this operator acts on a basis state in $\mathcal{H}_{d,m}$ defined as follows where $\mathbf{a}^{(i)}\in\mathbb{F}^m_d$
    \begin{equation}
        \begin{aligned}
            \sum_{\mathbf{q},\mathbf{p}\in\mathbb{F}^{m}_d}\hat{D}_{d,m}(\mathbf{q},\mathbf{p})^{\otimes 2k}\bigotimes_{i=1}^{2k}\ket{\mathbf{a}^{(i)}} =\sum_{\mathbf{q},\mathbf{p}\in\mathbb{F}^{m}_d}e^{i\frac{2\pi}{d}\mathbf{p}\cdot(k\mathbf{q} + \sum_{i'=1}^{2k}\mathbf{a}^{(i')})}\bigotimes_{i=1}^{2k}\ket{\mathbf{a}^{(i)} + \mathbf{q}}.
        \end{aligned}
    \end{equation}
    Observe that
    \begin{equation}
        \label{eq:disp_tensor_phase_sum}\sum_{\mathbf{p}\in\mathbb{F}^{m}_d}e^{i\frac{2\pi}{d}\mathbf{p}\cdot(k\mathbf{q} + \sum_{i=1}^{2k}\mathbf{a}^{(i)})} = \prod_{j=1}^m\left(\sum_{p\in\mathbb{F}_d}e^{i\frac{2\pi}{d}p(kq_j + \sum_{i=1}^{2k}a^{(i)}_j)}\right)= d^{m}\delta_{k\mathbf{q} + \sum_{i=1}^{2k}\mathbf{a}^{(i)}\mod d,0}.
    \end{equation}
    Suppose $k \neq 0\mod d$. We have a $g\in\mathbb{F}_{d}$ such that $gk = 1\mod d$ which then means that there is a unique $q_j = -g\sum_{i=1}^{2k}a^{(i)}_j$ for all $j$ hence giving
    \begin{equation}
        \sum_{\mathbf{q},\mathbf{p}\in\mathbb{F}^{m}_d}\hat{D}_{d,m}(\mathbf{q},\mathbf{p})^{\otimes 2k}\bigotimes_{i=1}^{2k}\ket{\mathbf{a}^{(i)}} = d^m\bigotimes_{i=1}^{2k}\left|\mathbf{a}^{(i)} -g\sum_{i'=1}^{2k}\mathbf{a}^{(i')}\right\rangle.
    \end{equation}
which means that this operation maps one basis state into another basis state. This can be expressed as a matrix $P$ over $\mathbb{F}^{2km}_{d}$ given by $P = \mathbb{I}_{2km} - g\mathbf{1}_{2k}\otimes\mathbb{I}_m$ (where $\mathbf{1}_{2k}$ is a $2k\times 2k$ matrix with all entries equal to 1) such that
\begin{equation}
    \sum_{\mathbf{q},\mathbf{p}\in\mathbb{F}^{m}_d}\hat{D}_{d,m}(\mathbf{q},\mathbf{p})^{\otimes 2k}\ket{\mathbf{a}} = d^m\ket{P\mathbf{a}}.
\end{equation}
Note that matrix $P$ is self-inverse over $\mathbb{F}^{2km}_d$ showing that the operator in question is a permutation matrix in the standard basis and hence leading to operator norm of $d^m$.

Suppose $k = 0\mod d$. In this case we have
\begin{equation}
    \sum_{\mathbf{q},\mathbf{p}\in\mathbb{F}^{m}_d}\hat{D}_{d,m}(\mathbf{q},\mathbf{p})^{\otimes 2k}\bigotimes_{i=1}^{2k}\ket{\mathbf{a}^{(i)}} = \begin{cases}
        d^m\sum_{\mathbf{q}\in\mathbb{F}^m_d}\bigotimes_{i=1}^{2k}\left|\mathbf{a}^{(i)}+\mathbf{q}\right\rangle & \sum_{i=1}^{2k}\mathbf{a}^{(i)} = 0\mod d\\
        0 & \mathrm{otherwise}
    \end{cases}.
\end{equation}
We can check that a trivial eigenstate of this operator is the fully entangled input state  given by
\begin{equation}
    \ket{\Psi} =\frac{1}{d^{m/2}} \sum_{\mathbf{z}\in\mathbb{F}^m_d}\ket{\mathbf{z}}^{\otimes 2k},\quad \sum_{\mathbf{q},\mathbf{p}\in\mathbb{F}^{m}_d}\hat{D}_{d,m}(\mathbf{q},\mathbf{p})^{\otimes 2k}\ket{\Psi}  = d^{2m}\ket{\Psi},
\end{equation}
which can be checked to yield an operator norm of $d^{2m}$. Note that since the operator in question is summing $d^{2m}$ unitary operators, the triangle inequality yields an upper bound of $d^{2m}$ which can be shown to be saturated by the choice of this input state.
\end{proof}

\begin{corollary}\label{corr:dispparityflip}
     Assuming $d$ prime and $k,r$ positive integers with $1\leq r\leq 2k$,
    \begin{equation}
        \left\|\sum_{\mathbf{q},\mathbf{p}\in\mathbb{F}^{m}_d}\hat{D}_{d,m}(\mathbf{q},\mathbf{p})^{\otimes r}\otimes \hat{D}_{d,m}(-\mathbf{q},-\mathbf{p})^{\otimes 2k-r}\right\|_\mathrm{op} = \begin{cases}
            d^m,& k\neq 0 \mod d\\
            d^{2m}, & k =0\mod d
        \end{cases}.
    \end{equation}
\end{corollary}
\begin{proof}
    There is always a parity flip operator $\hat{P}_d = \sum_{j\in\mathbb{F}_d} \ket{-j}\bra{j}$ which will yield $\hat{P}_d^{\otimes m}\hat{D}_{d,m}(\mathbf{q},\mathbf{p})\hat{P}_d^{\otimes m} = \hat{D}_{d,m}(-\mathbf{q},-\mathbf{p})$. Hence the above operator is related by a unitary to the operator in Lemma~\ref{lem:discr_disp_tensor} and so has the exact same operator norm.
\end{proof}

\begin{corollary}\label{corr:quditEopnorm}
    Defining, for $d$ prime and any $\theta\in\mathbb{R}$, the operator \begin{equation}
        \hat{O} = \frac{1}{d^{2m} - 1}\sum_{\mathbf{q},\mathbf{p} \in\mathbb{F}^m_{d},(\mathbf{q},\mathbf{p})\neq (\mathbf{0},\mathbf{0})} \left(\frac{e^{i\theta}\hat{D}_{d,m}(\mathbf{q},\mathbf{p}) + e^{-i\theta}\hat{D}_{d,m}(-\mathbf{q},-\mathbf{p})}{\sqrt{2}}\right)^{\otimes 2k},
    \end{equation}
    we have
    \begin{equation}
        \|\hat{O}\|_{\mathrm{op}} \leq \begin{cases}\frac{2^{k}}{d^m-1}& k\neq 0 \mod d\\
        2^k & k = 0\mod d
        \end{cases}.
    \end{equation}
\end{corollary}
\begin{proof}
    Note that
    \begin{equation}
        \hat{O} = \frac{1}{d^{2m} - 1}\sum_{\mathbf{q},\mathbf{p} \in\mathbb{F}^m_{d}} \left(\frac{e^{i\theta}\hat{D}_{d,m}(\mathbf{q},\mathbf{p}) + e^{-i\theta}\hat{D}_{d,m}(-\mathbf{q},-\mathbf{p})}{\sqrt{2}}\right)^{\otimes 2k} - \frac{(\sqrt{2}\cos(\theta))^{2k}}{d^{2m}-1}\mathbb{I}.
    \end{equation}
    By expanding the above sum in terms of displacements and then applying the triangle inequality we have the following for $k\neq 0\mod d$
    \begin{equation}
    \begin{aligned}
        \|\hat{O}\|_{\mathrm{op}} &\leq \frac{1}{2^{k}(d^{2m}-1)}\sum_{r=0}^{2k} \binom{2k}{r}\left\|\sum_{\mathbf{q},\mathbf{p}\in\mathbb{F}^{m}_d}\hat{D}_{d,m}(\mathbf{q},\mathbf{p})^{\otimes r}\otimes \hat{D}_{d,m}(-\mathbf{q},-\mathbf{p})^{\otimes 2k-r}\right\|_\mathrm{op} + \frac{2^k\cos^{2k}(\theta)}{d^{2m}-1} \\&\leq \frac{2^{k}d^m}{d^{2m}-1} + \frac{2^k}{d^{2m}-1} =\frac{2^k}{d^{m}-1}.
    \end{aligned}
    \end{equation}
    For the case of $k = 0\mod d$, we can note that the operator $\hat{O}$ can be expressed as the sum of $2^{2k}(d^{2m}-1)$ operators, each of the form $\frac{\hat{D}}{2^k(d^{2m}-1)}$ where $\hat{D}$ is a displacement. Hence by the triangle inequality it trivially follows that $\|\hat{O}\|_{\mathrm{op}}\leq 2^k$ which holds for any value of $k$.
\end{proof}

\begin{lemma}\label{lem:dispopnorm}
Taking $P(\gamma)$ as a Gaussian distribution for $\gamma\in\mathbb{C}^m$ with spread $\sigma$, the following holds for $k\geq 1$
\begin{equation}
    \left\|\int d^{2m}\gamma P(\gamma)\hat{D}^{\otimes k}(\gamma)\right\|_{\mathrm{op}} \leq \frac{1}{(1+k^2\sigma^4)^{m/2}}.
\end{equation}
\end{lemma}
\begin{proof}
    We define the operator $\hat{O}$ to be given by
    \begin{equation}
        \hat{O} = \int d^{2m}\gamma P(\gamma)\hat{D}^{\otimes k}(\gamma).
    \end{equation}
    By the definition of the operator norm, we have
    \begin{align}
\|\hat{O}\|^2_{\mathrm{op}} &= \sup_{\hat{\rho}\in D(\mathcal
        H^{\otimes k})}\text{Tr}\left[\hat{O}\hat{\rho}\hat{O}^\dagger\right]\\
        &= \sup_{\hat{\rho}\in D(\mathcal
        H^{\otimes k})}\int d^{2m}\gamma_1 d^{2m}\gamma_2 P(\gamma_1)P(\gamma_2)e^{\frac{k}{2}i\Omega(\gamma_2,\gamma_1)}\text{Tr}\left[\hat{D}^{\otimes k}(\gamma_1-\gamma_2)\hat{\rho}\right].
    \end{align}
    We define $\gamma_+ = \frac{\gamma_1+\gamma_2}{2}$ and $\gamma_- = \gamma_1 - \gamma_2$. Note that
    \begin{align}
        P(\gamma_1)P(\gamma_2) &= \frac{1}{(2\pi\sigma^2)^{2m}}\exp(-\frac{|\gamma_1|^2 + |\gamma_2|^2}{2\sigma^2})\\
        &= \left(\frac{1}{(\pi\sigma^2)^m}\exp(-\frac{|\gamma_+|^2}{\sigma^2})\right)\times\left(\frac{1}{(4\pi\sigma^2)^m}\exp(-\frac{|\gamma_-|^2}{4\sigma^2})\right).
    \end{align}
    and note that $i\Omega(\gamma_2,\gamma_1) = i\Omega(\gamma_+,\gamma_-)$, from which we use the Gaussian Fourier transform formula to note that
    \begin{align}
        \|\hat{O}\|_{\mathrm{op}}^2 &= \sup_{\hat{\rho}\in D(\mathcal
        H^{\otimes k})}\int d^{2m}\gamma_+ d^{2m}\gamma_- \left(\frac{1}{(\pi\sigma^2)^m}\exp(-\frac{|\gamma_+|^2}{\sigma^2})\right)\times\left(\frac{1}{(4\pi\sigma^2)^m}\exp(-\frac{|\gamma_-|^2}{4\sigma^2})\right) e^{\frac{k}{2}i\Omega(\gamma_+,\gamma_-)}\text{Tr}\left[\hat{D}^{\otimes k}(\gamma_-)\hat{\rho}\right]\\
        &= \sup_{\hat{\rho}\in D(\mathcal
        H^{\otimes k})}\int d^{2m}\gamma_- \frac{1}{(4\pi\sigma^2)^m}\exp(-\frac{|\gamma_-|^2}{4}\left(\frac{1}{\sigma^2}+k^2\sigma^2\right)) \text{Tr}\left[\hat{D}^{\otimes k}(\gamma_-)\hat{\rho}\right]\\
        &\leq \int d^{2m}\gamma_- \frac{1}{(4\pi\sigma^2)^m}\exp(-\frac{|\gamma_-|^2}{4}\left(\frac{1}{\sigma^2}+k^2\sigma^2\right))
        \\&= \frac{1}{(1+k^2\sigma^4)^{m}}.
    \end{align}
    For the final inequality we use the fact that the characteristic function is a complex number with norm bounded by 1. Note that this holds true for arbitrary states $\hat{\rho}$ without any assumptions on the energy of the state.
\end{proof}
\begin{corollary}\label{corr:bosonicEnorm}
Taking $P(\gamma)$ as a Gaussian distribution for $\gamma\in\mathbb{C}^m$ with spread $\sigma$, the following holds for $k\geq 1$ and any $\theta\in\mathbb{R}$

\begin{equation}
    \left\|\int d^{2m}\gamma P(\gamma)\hat{H}^{\otimes k}(\gamma)\right\|_{\mathrm{op}} \leq \frac{2^{k/2}}{(1+k^2\sigma^4)^{m/2}},\quad \hat{H}(\gamma) = \frac{e^{i\theta}\hat{D}(\gamma)  + e^{-i\theta}\hat{D}^\dagger(\gamma)}{\sqrt{2}}. 
\end{equation}
\end{corollary}
\begin{proof}
    Note that 
    \begin{align}
        \hat{H}^{\otimes k}(\gamma) &= \left(\frac{e^{i\theta}\hat{D}(\gamma) + e^{-i\theta}\hat{D}^\dagger(\gamma)}{\sqrt{2}}\right)^{\otimes k}\\
        &= \frac{1}{2^{k/2}}\sum_{s\in\{-1,1\}^{k}}e^{i\sum_{i=1}^ks_i\theta}\bigotimes_{i = 1}^{k}\hat{D}(s_i\gamma).
    \end{align}
    Note that $\|\hat{O}\|_{\mathrm{op}} = \|\hat{U}^\dagger\hat{O}\hat{U}\|_{\mathrm{op}}$ where $\hat{U}$ is a unitary and $\hat{D}(\gamma) = e^{i\pi\hat{n}}\hat{D}^\dagger(\gamma)e^{-i\pi\hat{n}}$. As a result, for all $s\in\{-1,1\}^k$ we have $\|\int d^{2m}\gamma P(\gamma)\hat{D}^{\otimes k}(\gamma)\|_{\mathrm{op}} = \|\int d^{2m}\gamma P(\gamma)\bigotimes_{i = 1}^{k}\hat{D}(s_i\gamma)\|_{\mathrm{op}}$. By triangle inequality, we have
    \begin{align}
       \left\|\int d^{2m}\gamma P(\gamma)\hat{H}^{\otimes k}(\gamma)\right\|_{\mathrm{op}} &\leq \frac{1}{2^{k/2}}\sum_{s\in\{-1,1\}^k}\left\|\int d^{2m}\gamma P(\gamma)e^{i\sum_{i=1}^ks_i\theta}\bigotimes_{i = 1}^{k}\hat{D}(s_i\gamma)\right\|_{\mathrm{op}}\\
       &= 2^{k/2}\left\|\int d^{2m}\gamma P(\gamma)\hat{D}^{\otimes k}(\gamma)\right\|_{\mathrm{op}} \leq \frac{2^{k/2}}{(1+k^2\sigma^4)^{m/2}}.
    \end{align}
\end{proof}
Now through the use of the proven bounds for various operator norms, we will establish our results for the hardness of channel learning under the setting of $c$-copy access to the channel for adaptive ancilla-assisted protocols.

\begin{theorem}[No access to the complex-conjugate channel: qudit $\to$ qudit]\label{thm:qudit2qudit}
Let $\mathcal{E}\in \CPTP(\mathcal{H}_{d,m},\mathcal{H}_{d',m'})$ where $d',d$ are prime numbers. Consider any adaptive ancilla-assisted learning scheme that is allowed to perform arbitrary measurements along with arbitrary inputs to parallel uses of the channel $\mathcal{E}^{\otimes c}$ where $c\leq \frac{1}{\epsilon} $ (here $\epsilon <1$). If the learner can use this scheme to produce an estimate of $|\hat{C}_{\mathcal{E}}(\mathbf{q},\mathbf{p},\mathbf{q}',\mathbf{p}')|$ such that for a query of $\mathbf{q},\mathbf{p}\in\mathbb{F}_{d}^{m}$, $\mathbf{q}',\mathbf{p}'\in\mathbb{F}_{d'}^{m'}$ it satisfies $\left||\hat{C}_{\mathcal{E}}(\mathbf{q},\mathbf{p},\mathbf{q}',\mathbf{p}')| - |C_{\mathcal{E}}(\mathbf{q},\mathbf{p},\mathbf{q}',\mathbf{p}')|\right|\leq \frac{\epsilon}{8e}$ with a 2/3 success probability, then this learning scheme requires at least $N$ uses of $\mathcal{E}$ where
\begin{itemize}
    \item Assuming $d'\neq d$, if $c\leq\min(d',d) -1$, $$N =\Omega(d^{m}d'^{m'}c^{-1}\epsilon^{-2}).$$
    \item Assuming $d'\neq d$, if $\min(d',d)\leq c\leq\max(d',d) -1$, taking $d_< = \min(d,d')$ and $d_> = \max(d,d')$, with $m_<$ and $m_>$ the corresponding numbers of qudits ( $(d_<,m_<) = (d,m)$ and $(d_>,m_>) = (d',m')$ if $d<d'$, and the other way round if $d>d'$), $$N  = \Omega\left(c\min\left[\frac{(d^{m}-1)(d'^{m'}-1)}{c^2\epsilon^2},(d_>^{m_>}-1)\left(\frac{d_<}{c\epsilon}\right)^{2d_<}\left(\frac{3}{4}\right)^2\right]\right).$$
    \item Assuming $d'\neq d$, if $\max(d',d)\leq c\leq d'd - 1$, $$N  = \Omega\left(c\min\left[\frac{(d^{m}-1)(d'^{m'}-1)}{c^2\epsilon^2},(d'^{m'}-1)\left(\frac{d}{c\epsilon}\right)^{2d}\left(\frac{3}{4}\right)^2,(d^{m}-1)\left(\frac{d'}{c\epsilon}\right)^{2d'}\left(\frac{3}{4}\right)^2\right]\right).$$
    \item Assuming $d'\neq d$, if $c\geq dd'$,
     $$N  = \Omega\left(c\min\left[\frac{(d^{m}-1)(d'^{m'}-1)}{c^2\epsilon^2},(d'^{m'}-1)\left(\frac{d}{c\epsilon}\right)^{2d}\left(\frac{3}{4}\right)^2,(d^{m}-1)\left(\frac{d'}{c\epsilon}\right)^{2d'}\left(\frac{3}{4}\right)^2,\left(\frac{dd'}{c\epsilon}\right)^{2dd'}\left(\frac{255}{256}\right)^2\right]\right).$$
    \item Assuming $d = d'$, if $c\geq d$, $$N = \Omega\left(c\min\left[\frac{(d^{m}-1)(d'^{m'}-1)}{c^2\epsilon^2},\left(\frac{d}{c\epsilon}\right)^{2d}\left(\frac{3}{4}\right)^2\right]\right).$$
\end{itemize}
\end{theorem}
\begin{proof}
    Consider the following two channels
    \begin{gather}
        \mathcal{E}_0(\cdot) = \text{Tr}(\cdot)\frac{\mathbb{I}}{d'^{m'}},\quad\mathcal{E}_{(\mathbf{q}_1,\mathbf{p}_1),(\mathbf{q}_2,\mathbf{p}_2)}(\cdot) = \text{Tr}\left[\hat{\Pi}_{+,(\mathbf{q}_1,\mathbf{p}_1)}(\cdot)\right]\hat{\rho}_{+,(\mathbf{q}_2,\mathbf{p}_2)} + \text{Tr}\left[\hat{\Pi}_{-,(\mathbf{q}_1,\mathbf{p}_1)}(\cdot)\right]\hat{\rho}_{-,(\mathbf{q}_2,\mathbf{p}_2)},\\
        \hat{\Pi}_{\pm,(\mathbf{q}_1,\mathbf{p}_1)}  = \frac{\mathbb{I} \pm \hat{E}_{d,m}(\mathbf{q}_1,\mathbf{p}_1)/\sqrt{2}}{2}\label{eq:quditpovm},
        \end{gather}
        \begin{equation}\label{eq:quditstates}
            \hat{\rho}_{\pm,(\mathbf{q}_2,\mathbf{p}_2)} = \frac{\mathbb{I} \pm \epsilon_0\hat{E}_{d',m'}(\mathbf{q}_2,\mathbf{p}_2)}{d'^{m'}}.
        \end{equation}
    Note that the above channel can always be implemented for qudits of prime dimensions $d$ as discussed in App.~\ref{app:qudit_disp}; both the POVM and the output states are valid for $\epsilon_0\leq 1/\sqrt{2}$. This simplifies the channel $\mathcal{E}_{(\mathbf{q}_1,\mathbf{p}_1),(\mathbf{q}_2,\mathbf{p}_2)}(\cdot)$ to the form of Eq.~\eqref{eq:Epar1par2} with $l =1$ given by
    \begin{equation}
    \begin{aligned}
        \mathcal{E}_{(\mathbf{q}_1,\mathbf{p}_1),(\mathbf{q}_2,\mathbf{p}_2)}(\cdot)&= \text{Tr}(\cdot)\frac{\mathbb{I}}{d'^{m'}} + \frac{\epsilon_0}{d'^{m'}\sqrt{2}}\text{Tr}(\hat{E}_{d,m}(\mathbf{q}_1,\mathbf{p}_1)(\cdot))\hat{E}_{d',m'}(\mathbf{q}_2,\mathbf{p}_2) \\&= \sum_{x\in\{0,1\}}\text{Tr}\left[\hat{K}_{x,(\mathbf{q}_1,\mathbf{p}_1)}(\cdot)\right]\sqrt{\hat{\rho}_0}\hat{W}_{x,(\mathbf{q}_2,\mathbf{p}_2)}\sqrt{\hat{\rho}_0},
    \end{aligned}
    \end{equation}
    for $\hat{\rho}_0 = \mathbb{I}_{d',m'}/d'^{m'}$ and
    \begin{equation}
        \hat{K}_{0,(\mathbf{q}_1,\mathbf{p}_1)} = \mathbb{I}_{d,m},\quad \hat{K}_{1,(\mathbf{q}_1,\mathbf{p}_1)} = \hat{E}_{d,m}(\mathbf{q}_1,\mathbf{p}_1)/\sqrt{2},\quad \hat{W}_{0,(\mathbf{q}_2,\mathbf{p}_2)} =\mathbb{I}_{d',m'},\quad \hat{W}_{1,(\mathbf{q}_2,\mathbf{p}_2)} = \epsilon_0\hat{E}_{d',m'}(\mathbf{q}_2,\mathbf{p}_2).
    \end{equation}
\textbf{Success in Problem~\ref{prob:manyrevel} by learning $|\hat{C}_{\mathcal{E}}(\mathbf{q},\mathbf{p},\mathbf{q}',\mathbf{p}')|$}: 
Consider we have the setting of Problem~\ref{prob:manyrevel} where we relabel $\mathbf{u}$ by $(\mathbf{q}_1,\mathbf{p}
_1) \in\mathbb{F}_{d}^{m}\times\mathbb{F}_{d}^{m}$ and relabel $\mathbf{v}$ by $(\mathbf{q}_2,\mathbf{p}
_2) \in\mathbb{F}_{d'}^{m'}\times\mathbb{F}_{d'}^{m'}$. So the setting is now that A prepares an unknown channel parameterized by $(\mathbf{q}_1,\mathbf{p}_1),(\mathbf{q}_2,\mathbf{p}_2)$, which B has to guess by performing a hypothesis test. Note that this relabeling is only for our convenience of applying Lemma~\ref{lem:master}. The variables $(\mathbf{q}_1,\mathbf{p}_1)$ follow the probability distribution
\begin{equation}
    \Pr((\mathbf{q}_1,\mathbf{p}_1) ) = \begin{cases}0& (\mathbf{q}_1,\mathbf{p}_1) = (0,0)\\
        \frac{1}{{d}^{2m} - 1}& \text{otherwise}
    \end{cases},
\end{equation}
with $\hat{\Pi}_{+,(\mathbf{q}_1,\mathbf{p}_1)} - \hat{\Pi}_{-,(\mathbf{q}_1,\mathbf{p}_1)} = \frac{\hat{E}_{d,m}(\mathbf{q}_1,\mathbf{p}_1)}{\sqrt{2}}$ from Eq.~\eqref{eq:quditpovm}. The variables $(\mathbf{q}_2,\mathbf{p}_2)$ follow the probability distribution 
\begin{equation}
    \Pr((\mathbf{q}_2,\mathbf{p}_2)) = \begin{cases}0& (\mathbf{q}_2,\mathbf{p}_2) = (0,0)\\
        \frac{1}{{d'}^{2m'} - 1}& \text{otherwise}
    \end{cases},
\end{equation}
and states parameterized by $\hat{\rho}_{(\mathbf{q}_2,\mathbf{p}_2)} = \hat{\rho}_{+,(\mathbf{q}_2,\mathbf{p}_2)}$ from Eq.~\eqref{eq:quditstates}. 
Note that for the two channels, we have
\begin{equation}
    |\hat{C}_{\mathcal{E}_{(\mathbf{q}_1,\mathbf{p}_1),(\mathbf{q}_2,\mathbf{p}_2)}}(\mathbf{q}_1,\mathbf{p}_1,\mathbf{q}_2,\mathbf{p}_2)| = \frac{\epsilon_0}{2\sqrt{2}},\quad |\hat{C}_{\mathcal{E}_0}(\mathbf{q}_1,\mathbf{p}_1,\mathbf{q}_2,\mathbf{p}_2)| = 0,
\end{equation}
and so if we have $\frac{\epsilon}{8e}= \frac{\epsilon_0}{4\sqrt{2}}$ ($\epsilon = e\epsilon_0\sqrt{2}$), we can ensure that estimating $|\hat{C}_{\mathcal{E}}(\mathbf{q},\mathbf{p},\mathbf{q}',\mathbf{p}')|$ to $\frac{\epsilon}{8e}$ accuracy with success probability $2/3$ means that the hypothesis test succeeds with probability $2/3$. Note that for $\epsilon_0 \leq 1/\sqrt{2}$, the state $\hat{\rho}_{(\mathbf{q}_2,\mathbf{p}_2)}$ will always be a valid state. Hence we can construct such a problem for all $\epsilon\leq 1 < e$.

\textbf{Application of Corollary~\ref{corr:master} from Lemma~\ref{lem:master}}:
We define the following operators
\begin{equation}
        \hat{\omega}_{\pmb{x},(\mathbf{q}_2,\mathbf{p}_2)} = \epsilon_0^{|\pmb{x}|}\hat{\rho}_0^{1/2}\hat{W}_{\pmb{x},(\mathbf{q}_2,\mathbf{p}_2)}\hat{\rho}_0^{1/2},\quad \hat{W}_{\pmb{x},(\mathbf{q}_2,\mathbf{p}_2)} = \bigotimes_{i=1}^{c}\hat{E}_{d',m'}(x_i\mathbf{q}_2,x_i\mathbf{p}_2).
    \end{equation}
Using Corollary~\ref{corr:quditEopnorm} we have
\begin{equation}\label{eq:quditoutputopnorm}
\begin{aligned}
     \frac{1}{d'^{2m'}-1}\left\|\epsilon_0^{2|\pmb{x}|}\sum_{\mathbf{q}_2,\mathbf{p}_2\in\mathbb{F}^{m'}_{d'},(\mathbf{q}_2,\mathbf{p}_2)\neq (\mathbf{0},\mathbf{0})}\hat{E}_{d',m'}(\mathbf{q}_2,\mathbf{p}_2)^{\otimes 2|\pmb{x}|}\right\|_{\mathrm{op}}\leq \begin{cases}\frac{(\epsilon_0\sqrt{2})^{2|\pmb{x}|}}{d'^{m'}-1}&|\pmb{x}| \neq 0\mod d'\\
      (\epsilon_0\sqrt{2})^{2|\pmb{x}|}&|\pmb{x}| = 0\mod d'\end{cases}.
\end{aligned}
\end{equation}
Since we have the setting of a two-outcome POVM with state preparation conditioned on the outcome, we also note that
\begin{equation}\label{eq:quditinputopnorm}
    \|\mathbb{E}_{(\mathbf{q}_1,\mathbf{p}_1)}[(\hat{\Pi}_{+,(\mathbf{q}_1,\mathbf{p}_1)}-\hat{\Pi}_{-,(\mathbf{q}_1,\mathbf{p}_1)})^{\otimes 2|\pmb{x}|}]\|_{\mathrm{op}} = \left\|\frac{1}{2^{|\pmb{x}|}}\mathbb{E}_{(\mathbf{q}_1,\mathbf{p}_1)}\left[\hat{E}_{d,m}(\mathbf{q}_1,\mathbf{p}_1)^{\otimes 2|\pmb{x}|}\right]\right\|_{\mathrm{op}}\leq \begin{cases}\frac{1}{d^{m}-1} & |\pmb{x}|\neq 0\mod d\\
    1&|\pmb{x}| = 0\mod d\end{cases}.
\end{equation}
For all the final simplifications for the following lower bounds, we use the fact that since $d,d'\geq2$, we always have $\frac{1}{d^{m} - 1}\leq \frac{3}{d^{m}}$, $\frac{1}{d'^{m'} - 1}\leq \frac{3}{d'^{m'}}$ and $\frac{1}{(d^{m}-1)(d'^{m'}-1)} \leq \frac{9}{d^{m}d'^{m'}}$. This allows us to write the bounds in a cleaner form.

\textbf{Sample complexity lower bound}:
We will make use of the quantity $\Delta$ defined in Lemma~\ref{lem:master}. Due to the form of our channel, we can directly apply the weaker bound from Corollary~\ref{corr:master} and the inequalities~\eqref{eq:quditoutputopnorm} and~\eqref{eq:quditinputopnorm}. Simplifying the upper bound on $\Delta$, and noting that each term only depends on the value of $|\pmb{x}|$, we get
\begin{equation}
    \Delta\leq \left(\sum_{k = 1}^c\binom{c}{k}\epsilon_0^k \sqrt{\left\|\frac{1}{2^{k}}\mathbb{E}_{(\mathbf{q}_1,\mathbf{p}_1)}\left[\hat{E}_{d,m}(\mathbf{q}_1,\mathbf{p}_1)^{\otimes 2k}\right]\right\|_{\mathrm{op}}\left\|\mathbb{E}_{(\mathbf{q}_2,\mathbf{p}_2)}\left[\hat{E}_{d',m'}(\mathbf{q}_2,\mathbf{p}_2)^{\otimes 2k}\right]\right\|_{\mathrm{op}}}\right)^2.
\end{equation}
We now assume that $d \neq d'$ which reduces the above equation into 4 separate cases.
\begin{itemize}
    \item Case 1: $c\leq \min(d',d) - 1$. In this case both operator norms will only satisfy the case of $k \neq 0\mod d$ and $k\neq 0\mod d'$. This simplifies to give
    \begin{equation}
        \Delta\leq \frac{1}{(d^{m}-1)(d'^{m'}-1)}((1+\sqrt{2}\epsilon_0)^c-1)^2\leq \frac{2c^2\epsilon_0^2e^{2\sqrt{2}c\epsilon_0}}{(d^{m}-1)(d'^{m'}-1)}\implies N = \Omega\left(\frac{d^{m}d'^{m'}}{c\epsilon^2}\right).
    \end{equation}
    \item Case 2: $\min(d',d)\leq c\leq \max(d',d) - 1$. In this case one of the operator norms would have the case of $k = 0\mod \min(d',d)$. We relabel $d_< = \min (d',d)$ and $d_>= \max(d',d)$ and $m_<$ and $m_>$ as the corresponding numbers of qudits. We get the following expression for $\Delta$:
    \begin{equation}
    \begin{aligned}
        \Delta &\leq \left(\sum_{k=1,k \neq 0\mod d_<}^c\binom{c}{k}\frac{(\epsilon_0\sqrt{2})^k}{(d^{m}-1)^{1/2}(d'^{m'}-1)^{1/2}} + \sum_{k=1}^{\lfloor c/d_<\rfloor}\binom{c}{kd_<}\frac{(\epsilon_0\sqrt{2})^{kd_<}}{(d_>^{m_>}-1)^{1/2}}\right)^2\\
        &\leq \left(\sum_{k=1}^c\binom{c}{k}\frac{(\epsilon_0\sqrt{2})^k}{(d^{m}-1)^{1/2}(d'^{m'}-1)^{1/2}} + \sum_{k=1}^{\lfloor c/d_<\rfloor}\left(\frac{ec}{kd_<}\right)^{kd_<}\frac{(\epsilon_0\sqrt{2})^{kd_<}}{(d_>^{m_>}-1)^{1/2}}\right)^2\\
        &\leq\left(\frac{c\epsilon_0 e^{c\epsilon_0\sqrt{2}}\sqrt{2}}{(d^{m}-1)^{1/2}(d'^{m'}-1)^{1/2}} + \sum_{k=1}^{\infty}\left(\frac{ec\epsilon_0\sqrt{2}}{d_<}\right)^{kd_<}\frac{1}{(d_>^{m_>}-1)^{1/2}}\right)^2 \\&\leq \left(\frac{c\epsilon}{(d^{m}-1)^{1/2}(d'^{m'}-1)^{1/2}} + \frac{4}{3}\left(\frac{c\epsilon}{d_<}\right)^{d_<}\frac{1}{(d_>^{m_>}-1)^{1/2}}\right)^2.
    \end{aligned}
    \end{equation}
    Here we make use of the fact that $\binom{n}{k}\leq (\frac{en}{k})^k$ and $ec\epsilon_0\sqrt{2}/d_< \leq 1/2$ and $d_<\geq 2$, since we assume $ec\epsilon_0\sqrt{2}\leq 1$. From this we get a lower bound for $N$,
    \begin{equation}
        N  = \Omega\left(c\min\left[\frac{(d^{m}-1)(d'^{m'}-1)}{c^2\epsilon^2},(d_>^{m_>}-1)\left(\frac{d_<}{c\epsilon}\right)^{2d_<}\left(\frac{3}{4}\right)^2\right]\right).
    \end{equation}
    \item Case 3: $\max(d',d) \leq c\leq d'd - 1$. In this case, there will be possibility that both $k = 0\mod d'$ and $k = 0\mod d$ for different values of $k$. We use a similar method for upper bounding $\Delta$ as the previous case to obtain
    \begin{equation}
    \begin{aligned}
        \Delta &\leq \left(\sum_{k=1,k \neq 0\mod d,d'}^c\binom{c}{k}\frac{(\epsilon_0\sqrt{2})^k}{(d^{m}-1)^{1/2}(d'^{m'}-1)^{1/2}} + \sum_{k=1}^{\lfloor c/d\rfloor}\binom{c}{kd}\frac{(\epsilon_0\sqrt{2})^{kd}}{(d'^{m'}-1)^{1/2}} + \sum_{k=1}^{\lfloor c/d'\rfloor}\binom{c}{kd'}\frac{(\epsilon_0\sqrt{2})^{kd'}}{(d^{m}-1)^{1/2}}\right)^2\\
        & \leq \left(\frac{c\epsilon }{(d^{m}-1)^{1/2}(d'^{m'}-1)^{1/2}} + \frac{4}{3}\left(\frac{c\epsilon}{d}\right)^{d}\frac{1}{(d'^{m'}-1)^{1/2}}+\frac{4}{3}\left(\frac{c\epsilon}{d'}\right)^{d'}\frac{1}{(d^{m}-1)^{1/2}}\right)^2,
    \end{aligned}
    \end{equation}
    which gives the lower bound on $N$ as 
    \begin{equation}
        N  = \Omega\left(c\min\left[\frac{(d^{m}-1)(d'^{m'}-1)}{c^2\epsilon^2},(d'^{m'}-1)\left(\frac{d}{c\epsilon}\right)^{2d}\left(\frac{3}{4}\right)^2,(d^{m}-1)\left(\frac{d'}{c\epsilon}\right)^{2d'}\left(\frac{3}{4}\right)^2\right]\right).
    \end{equation}
    \item Case 4: $c\geq d'd$. In this case, the value of $k$ can simultaneously satisfy $k = 0\mod d'$ and $k  = 0\mod d$. Hence we obtain the following bound for $\Delta$
    \begin{equation}
        \begin{aligned}
            \Delta &\leq \Bigg(\sum_{k=1,k \neq 0\mod d,d'}^c\binom{c}{k}\frac{(\epsilon_0\sqrt{2})^k}{(d^{m}-1)^{1/2}(d'^{m'}-1)^{1/2}} + \sum_{k=1,k\neq 0\mod d'}^{\lfloor c/d\rfloor}\binom{c}{kd}\frac{(\epsilon_0\sqrt{2})^{kd}}{(d'^{m'}-1)^{1/2}} \\&\quad\quad\quad+ \sum_{k=1,k\neq 0\mod d}^{\lfloor c/d'\rfloor}\binom{c}{kd'}\frac{(\epsilon_0\sqrt{2})^{kd'}}{(d^{m}-1)^{1/2}} + \sum_{k=1}^{\lfloor c/(dd')\rfloor}\binom{c}{kd'd}(\epsilon_0\sqrt{2})^{kd'd}\Bigg)^2\\
        & \leq \left(\frac{c\epsilon }{(d^{m}-1)^{1/2}(d'^{m'}-1)^{1/2}} + \frac{4}{3}\left(\frac{c\epsilon}{d}\right)^{d}\frac{1}{(d'^{m'}-1)^{1/2}}+\frac{4}{3}\left(\frac{c\epsilon}{d'}\right)^{d'}\frac{1}{(d^{m}-1)^{1/2}} + \frac{256}{255}\left(\frac{c\epsilon}{d'd}\right)^{d'd}\right)^2,
        \end{aligned}
    \end{equation}
    which gives the lower bound on $N$ as 
    \begin{equation}
        N  = \Omega\left(c\min\left[\frac{(d^{m}-1)(d'^{m'}-1)}{c^2\epsilon^2},(d'^{m'}-1)\left(\frac{d}{c\epsilon}\right)^{2d}\left(\frac{3}{4}\right)^2,(d^{m}-1)\left(\frac{d'}{c\epsilon}\right)^{2d'}\left(\frac{3}{4}\right)^2,\left(\frac{dd'}{c\epsilon}\right)^{2dd'}\left(\frac{255}{256}\right)^2\right]\right).
    \end{equation}
\end{itemize}
Now we consider the setting of $d = d'$ and $c\geq d$. In this scenario, we get
\begin{equation}
        \begin{aligned}
            \Delta &\leq \Bigg(\sum_{k=1,k \neq 0\mod d}^c\binom{c}{k}\frac{(\epsilon_0\sqrt{2})^k}{(d^{m}-1)^{1/2}(d^{m'}-1)^{1/2}} + \sum_{k=1}^{\lfloor c/d\rfloor}\binom{c}{kd}(\epsilon_0\sqrt{2})^{kd}\Bigg)^2\\
        & \leq \left(\frac{c\epsilon}{(d^{m}-1)^{1/2}(d^{m'}-1)^{1/2}} + \frac{4}{3}\left(\frac{c\epsilon}{d}\right)^{d}\right)^2,
        \end{aligned}
    \end{equation}
    which gives the lower bound on $N$ as
    \begin{equation}
        N = \Omega\left(c\min\left[\frac{(d^{m}-1)(d^{m'}-1)}{c^2\epsilon^2},\left(\frac{d}{c\epsilon}\right)^{2d}\left(\frac{3}{4}\right)^2\right]\right).
    \end{equation}
\end{proof}

\begin{theorem}[No access to the complex-conjugate channel: qudit $\to$ bosonic]\label{thm:qudit2boson}
Let $\mathcal{E}\in \CPTP(\mathcal{H}_{d,m},\mathcal{H}_{\infty,m'})$ where $d$ is prime and $m'\geq 8$. Taking $ec\epsilon < 0.1225$ and $\kappa\geq 2/m'$, consider any adaptive ancilla-assisted learning scheme that is allowed uses of the channel $\mathcal{E}^{\otimes c}$ for each measurement. If the learner can use this scheme to produce an estimate $|\hat{C}_{\mathcal{E}}((\mathbf{q},\mathbf{p}),\beta)|$ such that for a query of $\mathbf{q},\mathbf{p}\in\mathbb{F}_d^m$ and $|\beta|^2\leq \kappa m'$ it satisfies $\left||\hat{C}_{\mathcal{E}}((\mathbf{q},\mathbf{p}),\beta)| - |C_{\mathcal{E}}((\mathbf{q},\mathbf{p}),\beta)|\right|\leq \epsilon$ with a 2/3 success probability, then this learning scheme requires at least $N$ uses of $\mathcal{E}$ where 
\begin{itemize}
    \item If $c\leq d-1$, we have
    \begin{equation}
        N = \Omega\left(\frac{d^m (1+(0.99\kappa)^2)^{m'/2}}{c\epsilon^2}\right).
    \end{equation}
    \item If $c\geq d$, we have
    \begin{equation}
        N = \Omega\left(c(1+(0.99\kappa)^2)^{m'/2}\min\left(\frac{0.1225^2(d^m-1)}{c^2\epsilon^2e^{2/e}},\left(\frac{3}{4}\right)^2\left(\frac{0.1225d}{ec\epsilon}\right)^{2d}\right)\right).
    \end{equation}
\end{itemize}
\end{theorem}
\begin{proof}
    Consider the following two channels 
    \begin{gather}
        \mathcal{E}_0(\cdot) = \text{Tr}(\cdot)\hat{\rho}_0,\quad\mathcal{E}_{(\mathbf{q},\mathbf{p}),\gamma}(\cdot) = \text{Tr}\left[\hat{\Pi}_{+,(\mathbf{q},\mathbf{p})}(\cdot)\right]\hat{\rho}_\gamma + \text{Tr}\left[\hat{\Pi}_{-,(\mathbf{q},\mathbf{p})}(\cdot)\right]\hat{\rho}_{-\gamma},\\
        \hat{\Pi}_{\pm,(\mathbf{q},\mathbf{p})}  = \frac{\mathbb{I} \pm \hat{E}_{d,m}(\mathbf{q},\mathbf{p})/\sqrt{2}}{2},\\
        \hat{\rho}_{\gamma} = (1-\nu^2)^{m'}\nu^{\hat{n}}(\hat{D}(0) + 2i\epsilon_0(\hat{D}(\gamma) - \hat{D}(-\gamma)))\nu^{\hat{n}},
    \end{gather}
    where the operator $\hat{n}$ is the total photon number over $m'$ bosonic modes. We can rewrite $\mathcal{E}_{(\mathbf{q},\mathbf{p}),\gamma}$ in the form of Eq.~\eqref{eq:Epar1par2} with $l = 1$ as follows
    \begin{equation}
        \mathcal{E}_{(\mathbf{q},\mathbf{p}),\gamma}(\cdot) = \text{Tr}(\cdot)\hat{\rho}_0 + \text{Tr}(\hat{E}_{d,m}(\mathbf{q},\mathbf{p})(\cdot))( \hat{\rho}_{\gamma} - \hat{\rho}_0)/\sqrt{2} = \sum_{x\in\{0,1\}}\text{Tr}\left[(\cdot)\hat{K}_{x,(\mathbf{q},\mathbf{p})}\right]\sqrt{\hat{\rho}_0}\hat{W}_{x,\gamma}\sqrt{\hat{\rho}_0}.
    \end{equation}
    for $\hat{\rho}_0 = (1-\nu^2)^{m'}\nu^{2\hat{n}}$ and
    \begin{equation}
        \hat{K}_{0,(\mathbf{q},\mathbf{p})} = \mathbb{I}_{d,m},\quad \hat{K}_{1,(\mathbf{q},\mathbf{p})} = \hat{E}_{d,m}(\mathbf{q},\mathbf{p})/\sqrt{2},\quad \hat{W}_{0,\gamma} =\mathbb{I}_{\infty,m'},\quad \hat{W}_{1,\gamma} = 2\sqrt{2}\epsilon_0\frac{i(\hat{D}(\gamma) - \hat{D}^\dagger(\gamma))}{\sqrt{2}}.
    \end{equation}
    Suppose that $A$ randomly samples $\gamma$ and $(\mathbf{q},\mathbf{p})$ from the distribution
    \begin{equation}
        P(\gamma) = \frac{1}{(2\pi\sigma_{\gamma}^2)^{m'}}e^{-|\gamma|^2/2\sigma^2_{\gamma}}, \quad P(\mathbf{q},\mathbf{p}) = \begin{cases}
            0,& \text{if } (\mathbf{q},\mathbf{p})=(0,0)\\
            \frac{1}{d^{2m} - 1},& \text{otherwise}
        \end{cases},
    \end{equation}
    and sends this channel to $B$ as highlighted in Problem~\ref{prob:manyrevel}.
    
    \textbf{Success in Problem~\ref{prob:manyrevel} through learning $|C_{\mathcal{E}_{(\mathbf{q},\mathbf{p}),\gamma}}((\mathbf{q},\mathbf{p}),\beta)|$}: Note that we have
    \begin{equation}
        |C_{\mathcal{E}_{(\mathbf{q},\mathbf{p}),\gamma}}((\mathbf{q},\mathbf{p}),\gamma)| = \frac{1}{2}|\chi_{\hat{\rho}_{\gamma}- \hat{\rho}_0}(\gamma)|, \quad |C_{\mathcal{E}_{0}}((\mathbf{q},\mathbf{p}),\gamma)| = 0.
    \end{equation}
    We use the discrimination procedure for the revealed hypothesis test that tells apart $\chi_{\hat{\rho}_{\gamma}}$ from $\chi_{\hat{\rho}_{0}}$, given in Section S6.A of \cite{coroi2025exponentialadvantagecontinuousvariablequantum}. The required range of $\gamma$ is $2\sigma_\gamma^2\leq |\gamma|^2 \leq \kappa m'$, while $\epsilon
    _0  =2\epsilon e^{\kappa m'/\Sigma^2}/0.98$, where $\Sigma^2 =\frac{1+\nu}{1-\nu}$ is the variance parameter of the thermal state $\hat{\rho}_0$ and $L = \kappa m'/\Sigma^2$, giving $\epsilon_0 =2\epsilon e^L/0.98$. The hardest setting for the task is to take $L\to 0$ and since we require $\epsilon_0\leq 1/4$ this gives the upper bound on $\epsilon$ to be $0.98/8 = 0.1225$. 
    
    \textbf{Application of Corollary~\ref{corr:master} of Lemma~\ref{lem:master}}: Note that the channel $\mathcal{E}_{(\mathbf{q},\mathbf{p}),\gamma}$ takes on the exact form as the case examined in Corollary~\ref{corr:master}. We now evaluate the operator norms that are relevant to this setting. We define the following operators for $\pmb{x}\in\{0,1\}^c$
    \begin{equation}
        \hat{\omega}_{\pmb{x},\gamma} = (2\sqrt{2}\epsilon_0)^{|\pmb{x}|}\hat{\rho}_0^{1/2}\hat{W}_{\pmb{x},\gamma}\hat{\rho}_0^{1/2},\quad\hat{W}_{\pmb{x},\gamma} = \bigotimes_{i=1}^{c}\hat{o}_{x_i},\quad \hat{o}_0 =\mathbb{I},\quad\hat{o}_1 = \frac{i(\hat{D}(\gamma) - \hat{D}^\dagger(\gamma))}{\sqrt{2}}.
    \end{equation}
     From this we note that by using Corollary~\ref{corr:bosonicEnorm}, we have
     \begin{equation}
         \left\|\mathbb{E}_{\gamma}[\hat{W}_{\pmb{x},\gamma}^{\otimes 2}]\right\|_{\mathrm{op}} \leq \frac{(2\sqrt{2}\epsilon_0)^{2|\pmb{x}|}(\sqrt{2})^{2|\pmb{x}|}}{(1+4|\pmb{x}|^2\sigma_{\gamma}^4)^{m'/2}}\leq \frac{(\epsilon/0.1225)^{2|\pmb{x}|}}{(1+(2\sigma_{\gamma}^2)^2)^{m'/2}},
     \end{equation}
   and for the POVMs, we have
\begin{equation}
    \|\mathbb{E}_{(\mathbf{q},\mathbf{p})}[(\hat{\Pi}_{+,(\mathbf{q},\mathbf{p})}-\hat{\Pi}_{-,(\mathbf{q},\mathbf{p})})^{\otimes 2|\pmb{x}|}]\|_{\mathrm{op}} = \left\|\frac{1}{2^{|\pmb{x}|}}\mathbb{E}_{(\mathbf{q},\mathbf{p})}\left[\hat{E}_{d,m}(\mathbf{q},\mathbf{p})^{\otimes 2|\pmb{x}|}\right]\right\|_{\mathrm{op}}\leq \begin{cases}\frac{1}{d^{m}-1} & |\pmb{x}|\neq 0\mod d\\
    1&|\pmb{x}| = 0\mod d\end{cases}.
\end{equation}
\textbf{Sample complexity lower bound}:
\begin{itemize}
    \item Case 1: $c\leq d-1$. Using the previously derived Lemma~\ref{lem:master} and Corollary~\ref{corr:master}, we note that we have the inequalities satisfied for $c_1=O(1)$, $c_2 = 1$ and $d_{\mathrm{out}} = (1+(0.99\kappa)^2)^{m'/2}$, $d_{\mathrm{in}} = d^{m}-1$. This finally yields a learning tree of depth $T$ and a number of channel uses $N$ satisfying
\begin{equation}
    T= \Omega(d^{m}(1+(0.99\kappa)^2)^{m'/2}c^{-2}\epsilon^{-2}) \implies N= \Omega(d^{m}(1+(0.99\kappa)^2)^{m'/2}c^{-1}\epsilon^{-2}).
\end{equation}
\item Case 2: $c\geq d$. Using Lemma~\ref{lem:master}, we get that
\begin{equation}
\begin{aligned}
    \Delta &\leq \left(\sum_{k=1}^c\binom{c}{k}\sqrt{\frac{ (\epsilon/0.1225)^{2k}}{(1+(0.99\kappa)^2)^{m'/2}}\left\|\frac{1}{2^{k}}\mathbb{E}_{(\mathbf{q},\mathbf{p})}\left[\hat{E}_{d,m}(\mathbf{q},\mathbf{p})^{\otimes 2k}\right]\right\|_{\mathrm{op}}}\right)^2\\
    &\leq\frac{1}{(1+(0.99\kappa)^2)^{m'/2}}\left(\sum_{k=1,k\neq 0\mod d}^c\binom{c}{k}\frac{(\epsilon/0.1225)^{k}}{(d^{m}-1)^{1/2}} + \sum_{k=1}^{\lfloor c/d\rfloor}\binom{c}{kd}(\epsilon/0.1225)^{kd}\right)^2\\
    &\leq \frac{1}{(1+(0.99\kappa)^2)^{m'/2}}\left(\frac{c\epsilon e^{1/e}}{0.1225(d^m-1)^{1/2}} + \frac{4}{3}\left(\frac{ec\epsilon}{0.1225d}\right)^d\right)^2,
\end{aligned}
\end{equation}
which gives the lower bound on $N$ as 
\begin{equation}
    N = \Omega\left(c(1+(0.99\kappa)^2)^{m'/2}\min\left(\frac{0.1225^2(d^m-1)}{c^2\epsilon^2e^{2/e}},\left(\frac{3}{4}\right)^2\left(\frac{0.1225d}{ec\epsilon}\right)^{2d}\right)\right).
\end{equation}
\end{itemize}
\end{proof}

\begin{theorem}[No access to the complex-conjugate channel: bosonic $\to$ bosonic]\label{thm:boson2boson}
Let $\mathcal{E}\in \CPTP(\mathcal{H}_{\infty,m},\mathcal{H}_{\infty,m'})$ for $m,m'\geq 8$. Taking $\epsilon \leq 0.05$, $c =O(1/\epsilon)$, $\min(\kappa^2 m,\kappa'^2m') > 2/0.99$, consider any adaptive ancilla-assisted learning scheme that is allowed uses of the channel $\mathcal{E}^{\otimes c}$. If the learner can use this scheme to produce an estimate $|\hat{C}^{\mathrm{TMSV},r}_{\mathcal{E}}(\alpha,\beta)|$ (where $\cosh(2r)\geq 1.06\kappa m$) such that for a query of $\alpha\in\mathbb{C}^m,\beta\in \mathbb{C}^{m'}$ with $|\alpha|^2\leq \kappa m,|\beta|^2\leq \kappa 'm'$ it satisfies $$\left||\hat{C}^{\mathrm{TMSV},r}_{\mathcal{E}}(\alpha,\beta)| - |C^{\mathrm{TMSV},r}_{\mathcal{E}}(\alpha,\beta)|\right|\leq \epsilon,$$ with a 2/3 success probability, then this learning scheme requires at least $N$ uses of $\mathcal{E}$ where $$N = \Omega(c^{-1}\epsilon^{-2}(1 + (0.99\kappa')^2)^{m'/2}(1+(0.99\tanh^2(2r)\kappa)^2)^{m/2}).$$ 
\end{theorem}
\begin{proof}
    We consider two channels defined as follows
    \begin{gather}
        \mathcal{E}_0(\cdot) = \text{Tr}(\cdot)\hat{\rho}_0,\quad\mathcal{E}_{\gamma_1,\gamma_2}(\cdot) = \text{Tr}\left[\hat{\Pi}_{+,\gamma_1}(\cdot)\right]\hat{\rho}_{\gamma_2} + \text{Tr}\left[\hat{\Pi}_{-,\gamma_1}(\cdot)\right]\hat{\rho}_{-\gamma_2},\\
        \hat{\Pi}_{\pm,\gamma_1} = \frac{\mathbb{I} \pm \hat{E}(\gamma_1)/\sqrt{2}}{2},\quad\hat{\rho}_{\gamma_2} = (1-\nu^2)^{m'}\nu^{\hat n}(\hat{D}(0) + 2i\epsilon_0(\hat{D}(\gamma_2) - \hat{D}^\dagger(\gamma_2)))\nu^{\hat n},
    \end{gather}
    where the operator $\hat{n}$ is the total photon number for $m'$ bosonic modes. We can rewrite $\mathcal{E}_{\gamma_1,\gamma_2}$ taking the form of Eq.~\eqref{eq:Epar1par2} with $l = 1$
    \begin{equation}
        \mathcal{E}_{\gamma_1,\gamma_2}(\cdot) = \text{Tr}(\cdot)\hat{\rho}_{0} + \text{Tr}(\hat{E}(\gamma_1)(\cdot))(\hat{\rho}_{\gamma_2} - \hat{\rho}_0)/\sqrt{2} = \sum_{x\in\{0,1\}}\text{Tr}\left[(\cdot)\hat{K}_{x,\gamma_1}\right]\sqrt{\hat{\rho}_0}\hat{W}_{x,\gamma_2}\sqrt{\hat{\rho}_0},
    \end{equation}
   for $\hat{\rho}_0 = (1-\nu^2)^{m'}\nu^{2\hat{n}}$ and
    \begin{equation}
        \hat{K}_{0,\gamma_1} = \mathbb{I}_{\infty,m},\quad \hat{K}_{1,\gamma_1} = \hat{E}(\gamma_1)/\sqrt{2},\quad \hat{W}_{0,\gamma_2}=\mathbb{I}_{\infty,m'},\quad\hat{W}_{1,\gamma_2} = 2\sqrt{2}\epsilon_0\frac{i(\hat{D}(\gamma_2) - \hat{D}^\dagger(\gamma_2))}{\sqrt{2}}.
    \end{equation}

    \textbf{Success in Problem~\ref{prob:manyrevel} through learning $C^{\mathrm{TMSV},r}_{\mathcal{E}}$}: We consider the setting of Problem~\ref{prob:manyrevel} where $\gamma_1\in \mathbb{C}^m$ and $\gamma_2\in \mathbb{C}^{m'}$ are sampled from the following two distributions
    \begin{equation}
        P(\gamma_1) = \frac{1}{(2\pi\sigma_1^2)^{m}}e^{-\frac{|\gamma_1|^2}{2\sigma_1^2}},\quad P(\gamma_2) = \frac{1}{(2\pi\sigma_2^2)^{m'}}e^{-\frac{|\gamma_2|^2}{2\sigma_2^2}}.
    \end{equation}
    We now define the following parameters
    \begin{equation}
        \Sigma^2 = \frac{1+\nu}{1-\nu} = \frac{\kappa' m'}{L},\quad 2\sigma^2 = \frac{1}{\nu} - \nu = \frac{\frac{L}{\kappa'm'}}{\sqrt{1 - \left(\frac{L}{\kappa' m'}\right)^2}},
    \end{equation}
    where we choose constant $L>0$. We now have
    \begin{align}
    C^{\mathrm{TMSV},r}_{\mathcal{E}_{\gamma_1,\gamma_2}}(\gamma_1/\tanh(2r),\gamma_2) &= e^{-\frac{|\gamma_1|^2\cosh(2r)}{2\sinh^2(2r)}}\left(\frac{1+i}{2\sqrt{2}} + \frac{1-i}{2\sqrt{2}}e^{-2|\gamma_1|^2\cosh(2r)}\right)(\chi_{\hat{\rho}_{\gamma_2}}(\gamma_2) - \chi_{\hat{\rho}_0}(\gamma_2)) + e^{-\frac{|\gamma_1|^2\cosh(2r)}{2\tanh^2(2r)}}\chi_{\hat{\rho}_0}(\gamma_2)\nonumber\\
    &= -\frac{i\epsilon_0}{\sqrt{2}}e^{-\frac{|\gamma_1|^2\cosh(2r)}{2\sinh^2(2r)}}e^{-\frac{|\gamma_2|^2}{\Sigma^2}}(1+e^{-2|\gamma_1|^2\cosh(2r)})(1-e^{-2|\gamma_2|^2/\sigma^2})  + e^{-\frac{|\gamma_1|^2\cosh(2r)}{2\tanh^2(2r)}}e^{-\frac{|\gamma_2|^2}{2\Sigma^2}}e^{-\frac{|\gamma_2|^2}{2\sigma^2}}\nonumber\\
    &\quad\quad+\frac{\epsilon_0}{\sqrt{2}}e^{-\frac{|\gamma_1|^2\cosh(2r)}{2\sinh^2(2r)}}e^{-\frac{|\gamma_2|^2}{\Sigma^2}}(1-e^{-2|\gamma_1|^2\cosh(2r)})(1-e^{-2|\gamma_2|^2/\sigma^2}),
    \\C^{\mathrm{TMSV},r}_{\mathcal{E}_{0}}(\gamma_1/\tanh(2r),\gamma_2) &=e^{-\frac{|\gamma_1|^2\cosh(2r)}{2\tanh^2(2r)}}\chi_{\hat{\rho}_0}(\gamma_2) = e^{-\frac{|\gamma_1|^2\cosh(2r)}{2\tanh^2(2r)}}e^{-\frac{|\gamma_2|^2}{2\Sigma^2}}e^{-\frac{|\gamma_2|^2}{2\sigma^2}}.
    \end{align}
    We can simplify the above expressions as follows
    \begin{equation}
        C^{\mathrm{TMSV},r}_{\mathcal{E}_{\gamma_1,\gamma_2}}(\gamma_1/\tanh(2r),\gamma_2) = -i\epsilon_0t_+ + t_0 + \epsilon_0t_-,\quad C^{\mathrm{TMSV},r}_{\mathcal{E}_{0}}(\gamma_1/\tanh(2r),\gamma_2) = t_0,
    \end{equation}
    where we define
    \begin{equation}
            t_\pm  = \frac{1}{\sqrt{2}}e^{-\frac{|\gamma_1|^2\cosh(2r)}{2\sinh^2(2r)}}e^{-\frac{|\gamma_2|^2}{\Sigma^2}}(1\pm e^{-2|\gamma_1|^2\cosh(2r)})(1-e^{-2|\gamma_2|^2/\sigma^2}),\quad t_0 = e^{-\frac{|\gamma_1|^2\cosh(2r)}{2\tanh^2(2r)}}e^{-\frac{|\gamma_2|^2}{2\Sigma^2}}e^{-\frac{|\gamma_2|^2}{2\sigma^2}}.
    \end{equation}
    For sufficient distinguishability, we choose that the parameters $\gamma_1$ and $\gamma_2$ satisfy
    \begin{equation}\label{eq:b10succcondition}
    \frac{2}{\cosh(2r)}\leq |\gamma_1|^2\leq\kappa m\tanh^2(2r),\quad 2\sigma^2\leq|\gamma_2|^2\leq \kappa' m',
     \end{equation}
     which along with ensuring $\tanh(2r)\sinh(2r)\geq \kappa m$ (giving $\cosh(2r)\geq \kappa m$) yields
     \begin{equation}
       t_-\geq \frac{1}{\sqrt{2e}}(1-e^{-4})^2e^{-L}.
     \end{equation}
     Note that the values $t_0,t_\pm$ are all positive and so we clearly have 
     \begin{equation}
         \left||C^{\mathrm{TMSV},r}_{\mathcal{E}_{\gamma_1,\gamma_2}}(\gamma_1/\tanh(2r),\gamma_2)| - |C^{\mathrm{TMSV},r}_{\mathcal{E}_{0}}(\gamma_1/\tanh(2r),\gamma_2)|\right| \geq \epsilon_0t_-.
     \end{equation}
     Hence for success in the learning scheme, we have to ensure that $2\epsilon < \epsilon_0t_-$, which we can ensure by taking
     \begin{equation}
         \epsilon = \epsilon_0\frac{1}{2\sqrt{2e}}(1-e^{-4})^2e^{-L}.
     \end{equation}
    Note that $\epsilon_0\leq 1/4$ and the hardest setting of the problem is to have $L\to 0$. Hence, we can always construct a hard learning task as long as $\epsilon\leq 0.05$. If $B$ uses this learning scheme when the norm bounds on $\gamma_1$ and $\gamma_2$ are satisfied, the probability of success for the discrimination is given by
    \begin{equation} 
        \begin{aligned}
            \Pr(\mathrm{success}) &\geq 1/2(1 - \Pr_{\gamma_1,\gamma_2}(2\sigma^2 < |\gamma_2|^2\leq \kappa' m', 2/\cosh(2r)<|\gamma_1|^2\leq \kappa m\tanh^2(2r))) \\&\quad\quad+ (2/3)\Pr_{\gamma_1,\gamma_2}(2\sigma^2 < |\gamma_2|^2\leq \kappa' m', 2/\cosh(2r)<|\gamma_1|^2\leq \kappa m\tanh^2(2r))\\
            &= \frac{1}{2} + \frac{1}{6}\Pr_{\gamma_1,\gamma_2}(2\sigma^2 < |\gamma_2|^2\leq \kappa' m', 2/\cosh(2r)<|\gamma_1|^2\leq \kappa m\tanh^2(2r)).
        \end{aligned}
    \end{equation}
   Here we set $2\sigma_1^2 = 0.99\kappa \tanh^2(2r)$ and $2\sigma_2^2 = 0.99\kappa'$. In the range of $L < \kappa' m'/1.16$, we will always have $2\sigma^2  < \frac{2L}{\kappa'm'}$. Hence on ensuring both $0.99\kappa'^2m'\geq 2$ and $0.99\kappa^2m\geq 2$, we have $0.99\kappa'\geq 2\sigma^2$ and $0.99\kappa\tanh^2(2r)\geq \frac{2}{\cosh(2r)}$ (using $\sinh(2r)\tanh(2r)\geq \kappa m$). This gives $\Pr_{\gamma_2}(2\sigma^2<|\gamma_2|^2\leq \kappa' m') \geq \Pr_{\gamma_2}(2\sigma_2^2<|\gamma_2|^2\leq 2m'\sigma_2^2/0.99)$ and $\Pr_{\gamma_1}(\tanh^2(2r)<|\gamma_2|^2\leq \kappa m\tanh^2(2r)) \geq \Pr_{\gamma_2}(2\sigma_1^2<|\gamma_2|^2\leq 2m\sigma_1^2/0.99)$. From this we can verify using Lemma~\ref{lem:gaussprob} (since both $m,m'\geq8$) that 
    \begin{equation}
        \Pr(\mathrm{success})\geq \frac{1}{2}\left(1 + \frac{0.298}{3}\right),
    \end{equation}
    which gives a lower bound on the total variational distance if the learning scheme succeeds with probability $2/3$. Since we also have $m>8$ and $\kappa^2 m>2$, by having $\cosh(2r)>1.06\kappa m$, we can ensure $\sinh(2r)\tanh(2r)\geq \kappa m$ in the relevant parameter range.
    
    \textbf{Application of Lemma~\ref{lem:master}}: Note that the channel $\mathcal{E}_{\gamma_1,\gamma_2}$ takes on the exact form as the case examined in Corollary~\ref{corr:master}. We now evaluate the operator norms that are relevant to this setting. We define the following operators for $\pmb{x}\in\{0,1\}^c$
    \begin{equation}
        \hat{K}_{\pmb{x},\gamma_1} = \bigotimes_{i=1}^{c}\hat{K}_{x_i,\gamma_1},\quad \hat{W}_{\pmb{x},\gamma_2} = \bigotimes_{i=1}^c \hat{W}_{x_i,\gamma_2}.
    \end{equation}
     From this we note that by using Corollary~\ref{corr:bosonicEnorm}, we have
     \begin{equation}
         \left\|\mathbb{E}_{\gamma_2}[\hat{W}_{\pmb{x},\gamma_2}^{\otimes 2}]\right\|_{\mathrm{op}} \leq \frac{(2\sqrt{2}\epsilon_0)^{2|\pmb{x}|}2^{|\pmb{x}|}}{(1+4|\pmb{x}|^2\sigma_2^4)^{m'/2}}\leq \frac{(4\epsilon_0)^{2|\pmb{x}|}}{(1+4\sigma_2^4)^{m'/2}},
     \end{equation}
   and we further have
\begin{equation}
    \|\mathbb{E}_{\gamma_1}\hat{K}_{\pmb{x},\gamma_1}^{\otimes 2}\|_{\mathrm{op}} = \frac{1}{2^{|\pmb{x}|}}\left\|\int d^{2m}\gamma_1 P(\gamma_1)\hat{E}^{\otimes 2|\pmb{x}|}(\gamma_1)\right\|_{\mathrm{op}} \leq \frac{1}{(1+4|\pmb{x}|^2\sigma
    _1^4)^{m/2}}.
\end{equation}
\textbf{Sample complexity lower bound}:
Using the previously derived Lemma~\ref{lem:master} and Corollary~\ref{corr:master}, we note that we have the inequalities satisfied for $c_1,c_2,c_3,c_4 = O(1)$ and $d_{\mathrm{out}} = (1+0.99^2\kappa'^2)^{m'/2}$, $d_{\mathrm{in}} = (1+0.99^2\kappa^2\tanh^4(2r))^{m/2}$. This finally yields a learning tree of depth $T$ and a number of channel uses $N$ satisfying
\begin{equation}
\begin{aligned}
    &T= \Omega((1+0.99^2\tanh^4(2r)\kappa^2)^{m/2}(1+0.99^2\kappa'^2)^{m'/2}c^{-2}\epsilon^{-2}) \\&\implies N= \Omega((1+0.99^2\tanh^4(2r)\kappa^2)^{m/2}(1+0.99^2\kappa'^2)^{m'/2}c^{-1}\epsilon^{-2}).
\end{aligned}
\end{equation}
\end{proof}

\subsection{Lower bounds for single-copy access to self-complex-conjugate channels}\label{app:lowerbound_selfconjugate}

\begin{lemma}
Defining, for $d$ prime and any $\theta\in\mathbb{R}$, the operator \begin{equation}
        \hat{O} = \frac{1}{(d^{m} - 1)^2}\sum_{\mathbf{q}\in \mathbb{F}_d^m\setminus\{\mathbf{0}\},\mathbf{q}\in \mathbb{F}_d^m\setminus\{\mathbf{0}\}} \left(\frac{e^{i\theta}\hat{D}_{d,m}(\mathbf{q},\mathbf{p}) + e^{-i\theta}\hat{D}_{d,m}(-\mathbf{q},-\mathbf{p})}{\sqrt{2}}\right)^{\otimes 2k},
    \end{equation}
    we have
    \begin{equation}
        \|\hat{O}\|_{\mathrm{op}} \leq \begin{cases}\frac{2^{k}(d^m+1)}{(d^m-1)^2}& k\neq 0 \mod d\\
        2^k & k = 0\mod d
        \end{cases}.
    \end{equation}
\end{lemma}
\begin{proof}
    We first consider the action of the operator $\sum_{\mathbf{q},\mathbf{p}\in \mathbb{F}_d^m\setminus\{\mathbf{0}\}}\hat{D}_{d,m}(\mathbf{q},\mathbf{p})^{\otimes 2k}$ on the basis states of $\mathcal{H}_{d,m}$ where $\mathbf{a}^{(i)}\in\mathbb{F}^m_d$ giving,
    \begin{equation}
        \begin{aligned}
            \sum_{\mathbf{q},\mathbf{p}\in\mathbb{F}^{m}_d\setminus\{\mathbf{0}\}}\hat{D}_{d,m}(\mathbf{q},\mathbf{p})^{\otimes 2k}\bigotimes_{i=1}^{2k}\ket{\mathbf{a}^{(i)}} =\sum_{\mathbf{q},\mathbf{p}\in\mathbb{F}^{m}_d\setminus\{\mathbf{0}\}}e^{i\frac{2\pi}{d}\mathbf{p}\cdot(k\mathbf{q} + \sum_{i'=1}^{2k}\mathbf{a}^{(i')})}\bigotimes_{i=1}^{2k}\ket{\mathbf{a}^{(i)} + \mathbf{q}}.
        \end{aligned}
    \end{equation}
    Similar to Eq.~\eqref{eq:disp_tensor_phase_sum} we have
    \begin{equation}
        \sum_{\mathbf{p}\in\mathbb{F}^{m}_d\setminus\{\mathbf{0}\}}e^{i\frac{2\pi}{d}\mathbf{p}\cdot(k\mathbf{q} + \sum_{i=1}^{2k}\mathbf{a}^{(i)})} = \prod_{j=1}^m\left(\sum_{p\in\mathbb{F}_d}e^{i\frac{2\pi}{d}p(kq_j + \sum_{i=1}^{2k}a^{(i)}_j)}\right) - 1= d^{m}\delta_{k\mathbf{q} + \sum_{i=1}^{2k}\mathbf{a}^{(i)}\mod d,0} - 1.
    \end{equation}
    \textbf{Case 1:} If $k\neq 0\mod d$, we can have a unique $g \in \mathbb{F}_d^m$ such that $gk = 1\mod d$ giving a unique $q_j = -g\sum_{i=1}^{2k}a_j^{(i)}$ for all $j$. Note that in the case that $\sum_{i=1}^{2k}\mathbf{a}^{(i)}\equiv \mathbf{0}\mod d$, this value also equals zero. Using the permutation matrix $P$ defined in Lemma~\ref{lem:discr_disp_tensor}, we have
    \begin{equation}
        \sum_{\mathbf{q},\mathbf{p}\in\mathbb{F}^{m}_d\setminus\{\mathbf{0}\}}\hat{D}_{d,m}(\mathbf{q},\mathbf{p})^{\otimes 2k}\ket{\mathbf{a}} =\begin{cases}
            -\ket{\mathbf{a}}& \text{if }\sum_{i=1}^{2k}\mathbf{a}^{(i)} =\mathbf{0}\mod d\\d^m\ket{P\mathbf{a}} - \ket{\mathbf{a}} & \text{otherwise}            
        \end{cases},
    \end{equation}
which shows that the operator $\sum_{\mathbf{q},\mathbf{p}\in\mathbb{F}^{m}_d\setminus\{\mathbf{0}\}}\hat{D}_{d,m}(\mathbf{q},\mathbf{p})^{\otimes 2k}+\mathbb{I}$ equals $d^m$ times a permutation matrix over a subspace of $\mathcal{H}_{d,m}$, hence must have operator norm of $d^m$. By triangle inequality, this gives 
\begin{equation}
    \left\|\sum_{\mathbf{q},\mathbf{p}\in\mathbb{F}^{m}_d\setminus\{\mathbf{0}\}}\hat{D}_{d,m}(\mathbf{q},\mathbf{p})^{\otimes 2k}+\mathbb{I}\right\|_{\mathrm{op}} = d^m\implies \left\|\sum_{\mathbf{q},\mathbf{p}\in\mathbb{F}^{m}_d\setminus\{\mathbf{0}\}}\hat{D}_{d,m}(\mathbf{q},\mathbf{p})^{\otimes 2k}\right\|_{\mathrm{op}}\leq d^m + 1
\end{equation}
\textbf{Case 2:} if $k =0\mod d$, we make note of the fact that this is a sum of $(d^m -1)^2$ unitary operators. This gives the triangle inequality upper bound 
\begin{equation}
    \left\|\sum_{\mathbf{q},\mathbf{p}\in\mathbb{F}^{m}_d\setminus\{\mathbf{0}\}}\hat{D}_{d,m}(\mathbf{q},\mathbf{p})^{\otimes 2k}\right\|_{\mathrm{op}}\leq (d^m - 1)^2.
\end{equation}
Using a similar argument as the proof for Corollary~\ref{corr:quditEopnorm}, we can upper bound the operator norm for $\hat{O}$ as $2^k$ times the operator $\sum_{\mathbf{q},\mathbf{p}\in\mathbb{F}^{m}_d\setminus\{\mathbf{0}\}}\hat{D}_{d,m}(\mathbf{q},\mathbf{p})^{\otimes 2k}$. Plugging in the relations derived here, this gives us the claimed result for $\|\hat{O}\|_{\mathrm{op}}$.
\end{proof}

\begin{theorem}[No access to joint measurements: qudit $\to$ qudit] Let $\mathcal{E}\in \CPTP(\mathcal{H}_{d,m},\mathcal{H}_{d',m'})$ ($d'$ and $d$ being prime) which also satisfies the condition that $\mathcal{E}=\mathcal{E}^*$ (it equals its own complex-conjugate). Taking $0 < \epsilon\leq1/16$, assume an adaptive ancilla-assisted learning scheme which is allowed one use of the channel per measurement. If the learner can use this scheme to produce an estimate $\hat{C}_{\mathcal{E}}((\mathbf{q},\mathbf{p}),(\mathbf{q}',\mathbf{p}'))$ such that for a query of $\mathbf{q},\mathbf{p}\in \mathbb{F}^{m}_{d}$ and $\mathbf{q}',\mathbf{p}'\in\mathbb{F}^{m'}_{d'}$ it satisfies $\left||\hat{C}_{\mathcal{E}}((\mathbf{q},\mathbf{p}),(\mathbf{q}',\mathbf{p}'))| - |C_{\mathcal{E}}((\mathbf{q},\mathbf{p}),(\mathbf{q}',\mathbf{p}'))|\right|\leq \epsilon$ with a 2/3 success probability, then this learning scheme requires at least $N = \Omega(d^{m}d'^{m'}\epsilon^{-2})$ uses of the channel to complete the task.\label{thm:qudit2quditselfconj}
\end{theorem}
\begin{proof}
Consider the following states for $(\mathbf{q}_2,\mathbf{p}_2)\in\mathbb{F}_{d'}^{m'}\times\mathbb{F}_{d'}^{m'}$
\begin{equation}
    \hat{\rho}_{\pm,(\mathbf{q}_2,\mathbf{p}_2)} = \frac{\mathbb{I}\pm\epsilon_0\hat{E}_{d',m'}(\mathbf{q}_2,\mathbf{p}_2)}{d'^{m'}}.
\end{equation}
Now consider the following measurement operators for $(\mathbf{q}_1,\mathbf{p}_1)\in\mathbb{F}_{d}^{m}\times\mathbb{F}_{d}^{m}$
\begin{equation}
    \hat{\Pi}_{\pm,(\mathbf{q}_1,\mathbf{p}_1)} = \frac{\mathbb{I}\pm 2^{-1/2}\hat{E}_{d,m}(\mathbf{q}_1,\mathbf{p}_1)}{2}.
\end{equation}
Note that the set $\{(\hat{\Pi}_{+,(\mathbf{q}_1,\mathbf{p}_1)}+\hat{\Pi}_{+,(-\mathbf{q}_1,\mathbf{p}_1)})/2,(\hat{\Pi}_{-,(\mathbf{q}_1,\mathbf{p}_1)}+\hat{\Pi}_{-,(-\mathbf{q}_1,\mathbf{p}_1)})/2\}$ is a complete POVM set. Using this POVM and state output, we define the following CPTP maps
\begin{align}
    &\mathcal{E}_0(\cdot) = \text{Tr}(\cdot)\frac{\mathbb{I}}{d'^{m'}},\\
    &\mathcal{E}_{(\mathbf{q}_1,\mathbf{p}_1),(\mathbf{q}_2,\mathbf{p}_2)}(\cdot)\nonumber \\&=\text{Tr}\left[\frac{\hat{\Pi}_{+,(\mathbf{q}_1,\mathbf{p}_1)}}{2}(\cdot)\right]\hat{\rho}_{+,(\mathbf{q}_2,\mathbf{p}_2)}+\text{Tr}\left[\frac{\hat{\Pi}_{+,(-\mathbf{q}_1,\mathbf{p}_1)}}{2}(\cdot)\right]\hat{\rho}_{+,(-\mathbf{q}_2,\mathbf{p}_2)} \nonumber\\&\quad\quad+\text{Tr}\left[\frac{\hat{\Pi}_{-,(\mathbf{q}_1,\mathbf{p}_1)}}{2}(\cdot)\right]\hat{\rho}_{-,(\mathbf{q}_2,\mathbf{p}_2)}+\text{Tr}\left[\frac{\hat{\Pi}_{-,(-\mathbf{q}_1,\mathbf{p}_1)}}{2}(\cdot)\right]\hat{\rho}_{-,(-\mathbf{q}_2,\mathbf{p}_2)} \nonumber\\
    &= \text{Tr}(\cdot)\frac{\mathbb{I}}{d'^{m'}} + \frac{\epsilon_0}{2\sqrt{2}}\text{Tr}\left[(\hat{E}_{d,m}(\mathbf{q}_1,\mathbf{p}_1))(\cdot)\right]\left(\frac{\hat{E}_{d',m'}(\mathbf{q}_2,\mathbf{p}_2)}{d'^{m'}}\right)+ \frac{\epsilon_0}{2\sqrt{2}}\text{Tr}\left[(\hat{E}_{d,m}(-\mathbf{q}_1,\mathbf{p}_1))(\cdot)\right]\left(\frac{\hat{E}_{d',m'}(-\mathbf{q}_2,\mathbf{p}_2)}{d'^{m'}}\right).
\end{align}
Note that $\text{Tr}\left[(\hat{E}_{d,m}(-\mathbf{q}_1,\mathbf{p}_1))\hat{\sigma}\right]=\text{Tr}\left[(\hat{E}_{d,m}(\mathbf{q}_1,\mathbf{p}_1))\hat{\sigma}^T\right]$. Using this, we observe that for all input states $\hat{\sigma}$, we have $\mathcal{E}_{(\mathbf{q}_1,\mathbf{p}_1),(\mathbf{q}_2,\mathbf{p}_2)}(\hat{\sigma}^T)^T = \mathcal{E}_{(\mathbf{q}_1,\mathbf{p}_1),(\mathbf{q}_2,\mathbf{p}_2)}(\hat{\sigma})$ and hence the channel is self-conjugate.

\textbf{Success in Problem~\ref{prob:manyrevel} through learning $|C_{\mathcal{E}}((\mathbf{q}_1,\mathbf{p}_1),(\mathbf{q}_2,\mathbf{p}_2))|$}: Note that for $\mathbf{q}_1\in \mathbb{F}_d^m\setminus\{\mathbf{0}\}$, $\mathbf{p}_1\in \mathbb{F}_d^m\setminus\{\mathbf{0}\}$, $\mathbf{q}_2\in \mathbb{F}_{d'}^{m'}\setminus\{\mathbf{0}\}$ and $\mathbf{q}_1\in \mathbb{F}_{d'}^{m'}\setminus\{\mathbf{0}\}$, we have
\begin{equation}
    |C_{\mathcal{E}_{(\mathbf{q}_1,\mathbf{p}_1),(\mathbf{q}_2,\mathbf{p}_2)}}((\mathbf{q}_1,\mathbf{p}_1),(\mathbf{q}_2,\mathbf{p}_2))| = \frac{\epsilon_0}{4\sqrt{2}},\quad |C_{\mathcal{E}_{0}}((\mathbf{q}_1,\mathbf{p}_1),(\mathbf{q}_2,\mathbf{p}_2))| = 0.
\end{equation}
Hence this succeeds in the hypothesis testing task with $2/3$ probability as long as $\epsilon < \epsilon_0/(8\sqrt{2})$ and $\epsilon_0\leq 1/\sqrt{2}$ for a valid state. This gives $\epsilon\leq 1/16$. Notably, the probability distribution for the sampling of the parameters $(\mathbf{q}_1,\mathbf{p}_1)$, $(\mathbf{q}_2,\mathbf{p}_2)$ have to be slightly modified from their version in \ref{thm:qudit2qudit} to ensure distinguishability. We define their probability distributions as follows
\begin{equation}
    P((\mathbf{q}_1,\mathbf{p}_1)) = \begin{cases}
        0 & \text{if } \mathbf{q}_1 = \mathbf{0} \text{ or } \mathbf{p}_1 = \mathbf{0}\\
        \frac{1}{(d^m -1)^2} & \text{otherwise}
    \end{cases},\quad
    P((\mathbf{q}_2,\mathbf{p}_2)) = \begin{cases}
        0 & \text{if } \mathbf{q}_2 = \mathbf{0} \text{ or } \mathbf{p}_2 = \mathbf{0}\\
        \frac{1}{(d'^{m'} -1)^2} & \text{otherwise}
    \end{cases}
\end{equation}

\textbf{Application of Lemma~\ref{lem:master}}: 
Note that the structure of the channels is the same as in Eq.~\eqref{eq:Epar1par2} with $l = 2$. The setting we are examining is the 1-copy access case, which sets $c=1$. Since the expansion has three terms ($x = 0,1,2$), we can define
\begin{equation}
    \hat{K}_{x,(\mathbf{q}_1,\mathbf{p}_1)} = \begin{cases}
        \mathbb{I} & x=0\\
        \hat{E}_{d,m}(\mathbf{q}_1,\mathbf{p}_1) &  x= 1\\
        \hat{E}_{d,m}(-\mathbf{q}_1,\mathbf{p}_1) &  x= 2
    \end{cases},
\end{equation}
\begin{equation}
    \hat{\omega}_{x,(\mathbf{q}_2,\mathbf{p}_2)} = \hat{\rho}_0^{1/2}\hat{W}_{x,(\mathbf{q}_2,\mathbf{p}_2)}\hat{\rho}_0^{1/2},\quad \hat{W}_{x,(\mathbf{q}_2,\mathbf{p}_2)} = \begin{cases}
        \mathbb{I} & x=0\\
        \frac{\epsilon_0}{2\sqrt{2}}\hat{E}_{d',m'}(\mathbf{q}_2,\mathbf{p}_2) &  x= 1\\
        \frac{\epsilon_0}{2\sqrt{2}}\hat{E}_{d',m'}(-\mathbf{q}_2,\mathbf{p}_2) &  x= 2
    \end{cases},
\end{equation}

and now observe that we have the exact same setting as Corollary~\ref{corr:master}. Note that we use the result from Corollary~\ref{corr:quditEopnorm}
\begin{equation}\label{eq:quditEtranspnorm}
    \left\|\mathbb{E}_{(\mathbf{q},\mathbf{p})}\left[\hat{E}_{d,m}(\mathbf{q},\mathbf{p})^{\otimes 2}\right]\right\|_{\mathrm{op}}=\left\|\mathbb{E}_{(\mathbf{q},\mathbf{p})}\left[\hat{E}_{d,m}(-\mathbf{q},\mathbf{p})^{\otimes2}\right]\right\|_{\mathrm{op}} \leq  \frac{2}{d^m 
    -1}\frac{d^m+1}{d^m -1}\leq \frac{6}{d^{m}-1},
\end{equation}
since we assume that $d$ is prime. From this we substitute this in the weaker upper bound for $\Delta$ from Corollary~\ref{corr:master} to obtain
\begin{equation}
\begin{aligned}
    \Delta &\leq \sqrt{\left\|\mathbb{E}_{(\mathbf{q}_1,\mathbf{p}_1)}\left[\hat{K}_{1,(\mathbf{q}_1,\mathbf{p}_1)}^{\otimes 2}\right]\right\|_{\mathrm{op}}^2\left\|\mathbb{E}_{(\mathbf{q}_2,\mathbf{p}_2)}\left[\hat{W}_{1,(\mathbf{q}_2,\mathbf{p}_2)}^{\otimes 2}\right]\right\|_{\mathrm{op}}^2} + \sqrt{\left\|\mathbb{E}_{(\mathbf{q}_1,\mathbf{p}_1)}\left[\hat{K}_{2,(\mathbf{q}_1,\mathbf{p}_1)}^{\otimes 2}\right]\right\|_{\mathrm{op}}^2\left\|\mathbb{E}_{(\mathbf{q}_2,\mathbf{p}_2)}\left[\hat{W}_{2,(\mathbf{q}_2,\mathbf{p}_2)}^{\otimes 2}\right]\right\|_{\mathrm{op}}^2} \\
    &\leq 9\epsilon_0^2\left(\sqrt{((d^{m}-1)(d'^{m'}-1))^{-2}}+\sqrt{((d^{m}-1)(d'^{m'}-1))^{-2}}\right),
\end{aligned}
\end{equation}
\begin{equation}
     \Delta \leq \frac{9\epsilon_0^2}{(d^{m}-1)(d'^{m'}-1)}\implies N=\Omega(d^{m}d'^{m'}\epsilon^{-2}).
\end{equation}
\end{proof}
\begin{theorem}[No access to joint measurements: qudit $\to$ bosonic] Let $\mathcal{E}\in \CPTP(\mathcal{H}_{d,m},\mathcal{H}_{\infty,m'})$ (dimension $d$ being prime) and $m'\geq 8$ which also satisfies the condition that $\mathcal{E}=\mathcal{E}^*$ (it equals its own complex conjugate). Taking $0 < \epsilon\leq0.061$, assume an adaptive ancilla-assisted learning scheme which is allowed one use of the channel per measurement. If the learner can use this scheme to produce an estimate $\hat{C}_{\mathcal{E}}((\mathbf{q},\mathbf{p}),\beta)$ such that for a query of $\mathbf{q},\mathbf{p}\in\mathbb{F}^m_d$ and $|\beta|^2\leq \kappa m'$ ($\kappa>0$) it satisfies $\left||\hat{C}_{\mathcal{E}}((\mathbf{q},\mathbf{p}),\beta)| - |C_{\mathcal{E}}((\mathbf{q},\mathbf{p}),\beta)|\right|\leq \epsilon$ with a 2/3 success probability, then this learning scheme requires at least $N = \Omega(d^m\epsilon^{-2}(1+(0.99\kappa)^2)^{m'/2})$ samples to complete the task.\label{thm:qudit2bosonselfconj}
\end{theorem}
\begin{proof}Consider the two channels for $(\mathbf{q},\mathbf{p})\in\mathbb{F}^{m}_d\times\mathbb{F}^{m}_d$ and $\gamma\in\mathbb{C}^{m'}$ defined as
\begin{align}
    &\mathcal{E}_0(\cdot) = \text{Tr}(\cdot)\hat{\rho}_0,\\
    &\mathcal{E}_{(\mathbf{q},\mathbf{p}),\gamma}(\cdot)\nonumber \\&=\text{Tr}\left[\frac{\hat{\Pi}_{+,(\mathbf{q},\mathbf{p})}}{2}(\cdot)\right]\hat{\rho}_{\gamma}+\text{Tr}\left[\frac{\hat{\Pi}_{+,(-\mathbf{q},\mathbf{p})}}{2}(\cdot)\right]\hat{\rho}_{-\gamma^*} +\text{Tr}\left[\frac{\hat{\Pi}_{-,(\mathbf{q},\mathbf{p})}}{2}(\cdot)\right]\hat{\rho}_{-\gamma}+\text{Tr}\left[\frac{\hat{\Pi}_{-,(-\mathbf{q},\mathbf{p})}}{2}(\cdot)\right]\hat{\rho}_{\gamma^*} \nonumber\\
    &= \text{Tr}(\cdot)\hat{\rho}_0 + \frac{1}{2\sqrt{2}}\text{Tr}\left[\hat{E}_{d,m}(\mathbf{q},\mathbf{p})(\cdot)\right]\left(\hat{\rho}_{\gamma} - \hat{\rho}_0\right)+ \frac{1}{2\sqrt{2}}\text{Tr}\left[\hat{E}_{d,m}(-\mathbf{q},\mathbf{p})(\cdot)\right]\left(\hat{\rho}_{-\gamma^*} - \hat{\rho}_0\right),
\end{align}
where we define
\begin{equation}
    \hat{\Pi}_{\pm,(\mathbf{q},\mathbf{p})} = \frac{\mathbb{I} \pm 2^{-1/2}\hat{E}_{d,m}(\mathbf{q},\mathbf{p})}{2},\quad \hat{\rho}_\gamma = (1-\nu^2)^{m'}\nu^{\hat n}(\hat{D}(0) + 2i\epsilon_0(\hat{D}(\gamma) - \hat{D}^\dagger(\gamma)))\nu^{\hat n}.
\end{equation}
Note that the channel $\mathcal{E}_{(q,p),\gamma}(\hat{\rho}^T)^T = \mathcal{E}_{(q,p),\gamma}(\hat{\rho})$ where the input state transposition is in the standard basis for $\mathcal{H}_{d,m}$ and the transposition for the output is in the Fock basis.

We consider the setting of Problem~\ref{prob:manyrevel} where $A$ samples $(\mathbf{q},\mathbf{p})\in\mathbb{F}^{m}_d\times\mathbb{F}^{m}_d$ and $\gamma\in\mathbb{C}^{m'}$ according to 
\begin{equation}
    P((\mathbf{q},\mathbf{p})) = \begin{cases}
        0&\text{if }\mathbf{q} = \mathbf{0}\text{ or }\mathbf{p} = \mathbf{0}\\
        \frac{1}{d^{2m}-1}&\text{otherwise}
    \end{cases},\quad P(\gamma) = \frac{1}{(2\pi\sigma_{\gamma}^2)^{m'}}\exp(-\frac{|\gamma|^2}{2\sigma^2_\gamma}).
\end{equation}
 
\textbf{Success in Problem~\ref{prob:manyrevel} through channel learning}: Note that we have the following hold
\begin{equation}
    |C_{\mathcal{E}_{(\mathbf{q},\mathbf{p}),\gamma}}((\mathbf{q},\mathbf{p}),\beta)| = \frac{1}{4}|\chi_{\hat{\rho}_{\gamma} - \hat{\rho}_0}(\beta)|,\quad |C_{\mathcal{E}_0}((\mathbf{q},\mathbf{p}),\beta)| = 0.
\end{equation}
Using the discrimination procedure highlighted in the revealed hypothesis test to tell apart the characteristic function of three peak state $\chi_{\hat{\rho}_{\gamma}}$ from $\chi_{\hat{\rho}_{0}}$ in Eq. (S85) of \cite{coroi2025exponentialadvantagecontinuousvariablequantum}, the above discrimination also succeeds after having $\epsilon_0
      =4\epsilon e^{\kappa m'/\Sigma^2}/0.98$ and $\Sigma^2  = \frac{\kappa m'}{L} =\frac{1+\nu}{1-\nu}$. To succeed with at least probability of $1/2(1+1/6)$, we can set $2\sigma_{\gamma}^2 = 0.99\kappa$.

\textbf{Application of Lemma~\ref{lem:master}}: 
By defining the following operators
\begin{equation}
    \hat{K}_{x,(\mathbf{q},\mathbf{p})} = \begin{cases}
        \mathbb{I} & x = 0\\
        \hat{E}_{d,m}(\mathbf{q},\mathbf{p}) & x = 1\\
        \hat{E}_{d,m}(-\mathbf{q},\mathbf{p}) & x = 2
    \end{cases},
\end{equation}
\begin{equation}
    \hat{\omega}_{x,\gamma} = \hat{\rho}_0^{1/2}\hat{W}_{x,\gamma}\hat{\rho}_0^{1/2},\quad \hat{W}_{x,\gamma} = \begin{cases}
        \mathbb{I} & x = 0\\
        \frac{\epsilon_0}{\sqrt{2}}i\left(\hat{D}(\gamma) - \hat{D}^\dagger(\gamma)\right) & x= 1\\
        \frac{\epsilon_0}{\sqrt{2}}i\left(\hat{D}(-\gamma^*) - \hat{D}^\dagger(-\gamma^*)\right) & x =2
    \end{cases},
\end{equation}
we note that we have the exact setting for Lemma~\ref{lem:master} specifically with $c = 1$ and $l = 2$. Observe that using Corollary~\ref{corr:bosonicEnorm} we have
\begin{equation}
    \left\|\mathbb{E}_{\gamma}\left[\hat{W}_{1,\gamma}^{\otimes 2}\right]\right\|_{\mathrm{op}} =\frac{\epsilon_0^2}{2}\left\|\mathbb{E}_{\gamma}\left[(i\hat{D}(\gamma)-i\hat{D}^\dagger(\gamma))^{\otimes 2}\right]\right\|_{\mathrm{op}}\leq \frac{2\epsilon_0^2}{(1+4\sigma_{\gamma}^4)^{m'/2}},
\end{equation}
and since $\hat{W}_{2,\gamma}^T = \hat{W}_{1,\gamma}$ we have the same inequality hold for $\left\|\mathbb{E}_{\gamma}\left[\hat{W}_{2,\gamma}^{\otimes 2}\right]\right\|_{\mathrm{op}}$. We can then reuse the upper bounds in Eq.~\eqref{eq:quditEtranspnorm} and apply the case of Corollary~\ref{corr:master} to obtain the weaker upper bound on $\Delta$ as
\begin{equation}
\begin{aligned}
    \Delta &\leq \sqrt{\left\|\mathbb{E}_{(\mathbf{q},\mathbf{p})}\left[\hat{K}_{1,(\mathbf{q},\mathbf{p})}^{\otimes 2}\right]\right\|_{\mathrm{op}}^2\left\|\mathbb{E}_{\gamma}\left[\hat{W}_{1,\gamma}^{\otimes 2}\right]\right\|_{\mathrm{op}}^2} + \sqrt{\left\|\mathbb{E}_{(\mathbf{q},\mathbf{p})}\left[\hat{K}_{2,(\mathbf{q},\mathbf{p})}^{\otimes 2}\right]\right\|_{\mathrm{op}}^2\left\|\mathbb{E}_{\gamma}\left[\hat{W}_{2,\gamma}^{\otimes 2}\right]\right\|_{\mathrm{op}}^2} \\
    &\leq 12\epsilon_0^2\left(\sqrt{((d^{m}-1)(1+(0.99\kappa)^2)^{m'/2})^{-2}}+\sqrt{((d^{m}-1)(1+(0.99\kappa)^2)^{m'/2})^{-2}}\right),
\end{aligned}
\end{equation}
\begin{equation}
     \Delta \leq \frac{24\epsilon_0^2}{(d^{m}-1)(1+(0.99\kappa)^2)^{m'/2}}\implies N=\Omega(d^{m}(1+(0.99\kappa)^2)^{m'/2}\epsilon^{-2}).
\end{equation}
\end{proof}
\begin{theorem}[No access to joint measurements: bosonic $\to$ bosonic] Let $\mathcal{E}\in \CPTP(\mathcal{H}_{\infty,m},\mathcal{H}_{\infty,m'})$ and $m,m'\geq 8$ which also satisfies the condition that $\mathcal{E}=\mathcal{E}^*$ (it equals its own complex conjugate). Taking $0 < \epsilon\leq0.025$, and $\min(\kappa^2 m,\kappa'^2m') > 2/0.99$, assume an adaptive ancilla-assisted learning scheme which only uses $\mathcal{E}$ per measurement. If the learner can use this scheme to produce an estimate $\hat{C}^{\mathrm{TMSV},r}_{\mathcal{E}}(\alpha,\beta)$ (where  $\cosh(2r)>1.06\kappa m$) such that for a query of $\alpha\in\mathbb{C}^m,\beta\in\mathbb{C}^{m'}$ with $|\alpha|^2\leq \kappa m,|\beta|^2\leq \kappa' m'$ it satisfies $\left||\hat{C}^{\mathrm{TMSV},r}_{\mathcal{E}}(\alpha,\beta)| - |C^{\mathrm{TMSV},r}_{\mathcal{E}}(\alpha,\beta)|\right|\leq \epsilon$ with a 2/3 success probability, then this learning scheme requires at least $N = \Omega\left(\epsilon^{-2}(1+(0.99\kappa')^2)^{m'/2}(1+(0.99\tanh^2(2r)\kappa)^2)^{m/2}\right)$ samples to complete.\label{thm:boson2bosonselfconj}
\end{theorem}
\begin{proof}
We define the following channels where $\gamma_1\in \mathbb{C}^m,\gamma_2\in \mathbb{C}^{m'}$
\begin{align}
    &\mathcal{E}_0(\cdot) = \text{Tr}(\cdot)\hat{\rho}_0,\\
    &\mathcal{E}_{\gamma_1,\gamma_2}(\cdot)\nonumber \\&=\text{Tr}\left[\frac{\hat{\Pi}_{+,\gamma_1}}{2}(\cdot)\right]\hat{\rho}_{\gamma_2}+\text{Tr}\left[\frac{\hat{\Pi}_{+,-\gamma_1^*}}{2}(\cdot)\right]\hat{\rho}_{-\gamma_2^*} +\text{Tr}\left[\frac{\hat{\Pi}_{-,\gamma_1}}{2}(\cdot)\right]\hat{\rho}_{-\gamma_2}+\text{Tr}\left[\frac{\hat{\Pi}_{-,-\gamma_1^*}}{2}(\cdot)\right]\hat{\rho}_{\gamma_2^*} \nonumber\\
    &= \text{Tr}(\cdot)\hat{\rho}_0 + \frac{1}{2\sqrt{2}}\text{Tr}\left[\hat{E}(\gamma_1)(\cdot)\right]\left(\hat{\rho}_{\gamma_2} - \hat{\rho}_0\right)+ \frac{1}{2\sqrt{2}}\text{Tr}\left[\hat{E}(-\gamma_1^*)(\cdot)\right]\left(\hat{\rho}_{-\gamma_2^*} - \hat{\rho}_0\right),
\end{align}
where we define
\begin{equation}
    \hat{\Pi}_{\pm,\gamma_1} = \frac{\mathbb{I} \pm \hat{E}(\gamma_1)/\sqrt{2}}{2},\quad\hat{\rho}_{\gamma_2} = (1-\nu^2)^{m'}\nu^{\hat n}(\hat{D}(0) + i\epsilon_0(\hat{D}(\gamma_2) - \hat{D}^\dagger(\gamma_2))+i\epsilon_0(\hat{D}(-\gamma_2^*) - \hat{D}^\dagger(-\gamma_2^*)))\nu^{\hat n}.
\end{equation}
where $\hat{n}$ is the total photon number operator over $m'$ modes. Observe that the channel $\mathcal{E}_{\gamma_1,\gamma_2}(\hat{\rho}^T)^T = \mathcal{E}_{\gamma_1,\gamma_2}(\hat{\rho})$ and hence it is self-conjugate. Similarly we consider $A$ to randomly sample $\gamma_1$ and $\gamma_2$ according to the probability distributions
\begin{equation}
    P(\gamma_1) = \frac{1}{(2\pi\sigma_1^2)^m}e^{-\frac{|\gamma_1|^2}{2\sigma_1^2}},\quad P(\gamma_2) = \frac{1}{(2\pi\sigma_2^2)^{m'}}e^{-\frac{|\gamma_2|^2}{2\sigma_2^2}},
\end{equation}
according to the setting of Problem~\ref{prob:manyrevel} to send copies of one of the channels to $B$.

\textbf{Success in Problem~\ref{prob:manyrevel} through learning $C^{\mathrm{TMSV},r}_{\mathcal{E}}$}: Note that we have the following hold
\begin{align}
    &C^{\mathrm{TMSV},r}_{\mathcal{E}_{\gamma_1,\gamma_2}}(\gamma_1/\tanh(2r),\gamma_2)\nonumber\\
    &= e^{-\frac{|\gamma_1|^2\cosh(2r)}{2\sinh^2(2r)}}\left(\frac{1+i}{4\sqrt{2}} + \frac{1-i}{4\sqrt{2}}e^{-2|\gamma_1|^2\cosh(2r)}\right)(\chi_{\hat{\rho}_{\gamma_2}}(\gamma_2) - \chi_{\hat{\rho}_0}(\gamma_2))\nonumber\\
    &\quad\quad+e^{-\frac{|\gamma_1|^2\cosh(2r)}{2\sinh^2(2r)}}\left(\frac{1+i}{4\sqrt{2}}e^{-2\Re(\gamma_1)^2\cosh(2r)} + \frac{1-i}{4\sqrt{2}} e^{-2\Im(\gamma_1)^2\cosh(2r)}\right)(\chi_{\hat{\rho}_{-\gamma_2^*}}(\gamma_2) - \chi_{\hat{\rho}_0}(\gamma_2))\nonumber\\
    &\quad\quad+ e^{-\frac{|\gamma_1|^2\cosh(2r)}{2\tanh^2(2r)}}\chi_{\hat{\rho}_0}(\gamma_2)=-i\epsilon_0 t'_+ + t_0 + \epsilon_0t'_-,
    \\&C^{\mathrm{TMSV},r}_{\mathcal{E}_{0}}(\gamma_1/\tanh(2r),\gamma_2) =e^{-\frac{|\gamma_1|^2\cosh(2r)}{2\tanh^2(2r)}}\chi_{\hat{\rho}_0}(\gamma_2) = t_0,
    \end{align}
    where we define
    \begin{equation}
    \begin{aligned}
        t'_\pm &= \frac{e^{-\frac{|\gamma_1|^2\cosh(2r)}{2\sinh^2(2r)}}}{2\sqrt{2}}\Bigg[
        (1\pm e^{-2|\gamma_1|^2\cosh(2r)})e^{-\frac{|\gamma_2|^2}{\Sigma^2}}\left(1-e^{-2\frac{|\gamma_2|^2}{\sigma^2}}\right)\\&\hspace{80pt}+(e^{-2\Re(\gamma_1)^2\cosh(2r)}\pm e^{-2\Im(\gamma_1)^2\cosh(2r)})e^{-\frac{|\gamma_2|^2}{\Sigma^2}}\left(e^{-2\frac{\Re(\gamma_2)^2}{\sigma^2}}-e^{-2\frac{\Im(\gamma_2)^2}{\sigma^2}}\right)\Bigg]\\
        &\geq \frac{e^{-\frac{|\gamma_1|^2\cosh(2r)}{2\sinh^2(2r)}}e^{- \frac{|\gamma_2|^2}{\Sigma^2}}}{2\sqrt{2}}\Bigg[
        (1 - e^{-2|\gamma_1|^2\cosh(2r)})\left(1-e^{-2\frac{|\gamma_2|^2}{\sigma^2}}\right)\\&\hspace{110pt}-(e^{-2\Re(\gamma_1)^2\cosh(2r)}+ e^{-2\Im(\gamma_1)^2\cosh(2r)})\left(e^{-2\frac{\Re(\gamma_2)^2}{\sigma^2}}+e^{-2\frac{\Im(\gamma_2)^2}{\sigma^2}}\right)\Bigg].
    \end{aligned}
    \end{equation}
    We restrict the range of values to satisfy
    \begin{equation}\begin{aligned}
        \label{eq:conditionsthmb12}
        \frac{2}{\cosh(2r)}\leq\Re(\gamma_1)^2\leq \frac{\kappa m\tanh^2(2r)}{2}&,\quad \frac{2}{\cosh(2r)}\leq\Im(\gamma_1)^2\leq\frac{\kappa m\tanh^2(2r)}{2},\\2\sigma^2\leq\Re(\gamma_2)^2\leq\frac{\kappa' m'}{2}&,\quad 2\sigma^2\leq\Im(\gamma_2)^2\leq\frac{\kappa' m'}{2}.
    \end{aligned}
    \end{equation}
    Note that for the conditions above to be valid, we require that $\kappa m\tanh(2r)\sinh(2r) > 2$. Over this range we have the following hold
    \begin{equation}
        \begin{aligned}
            t_{\pm}' \geq \frac{1}{2\sqrt{2e}}e^{-L}\left[(1-e^{-4})^2  - 4e^{-8}\right].
        \end{aligned}
    \end{equation}
To make this discrimination possible we also choose $\sinh(2r)\tanh(2r)\geq \kappa m$. Note that we have
\begin{equation}
         \left||C^{\mathrm{TMSV},r}_{\mathcal{E}_{\gamma_1,\gamma_2}}(\gamma_1/\tanh(2r),\gamma_2)| - |C^{\mathrm{TMSV},r}_{\mathcal{E}_{0}}(\gamma_1/\tanh(2r),\gamma_2)|\right| \geq \epsilon_0t'_-.
     \end{equation}
     Hence for success in the learning scheme, we have to ensure that $2\epsilon < \epsilon_0t'_-$, which we can ensure by taking
     \begin{equation}
         \epsilon = \epsilon_0\frac{1}{4\sqrt{2e}}\left[(1-e^{-4})^2 - 4e^{-8}\right]e^{-L},
     \end{equation}
     and by taking the hardest learning case of $L\to 0$, we can formulate a version of this problem for any $\epsilon\leq 0.0257$. $B$ can perform the learning scheme to discriminate under the conditions in Eq.~\eqref{eq:conditionsthmb12} which would mean succeeding with probability
      \begin{equation}                   
      \begin{aligned}
          \Pr(\mathrm{success}) &= \frac{1}{2} + \frac{1}{6}\Big(\Pr_{\gamma_1}( 1< |\Re(\gamma_1)|^2 \leq \frac{1}{2}\kappa m\tanh^2(2r), 1< |\Im(\gamma_1)|^2 \leq \frac{1}{2}\kappa m\tanh^2(2r))\\
          &\quad\quad\quad\quad\quad\times\Pr_{\gamma_2}( 2\sigma^2< |\Re(\gamma_2)|^2 \leq \frac{1}{2}\kappa' m', 2\sigma^2< |\Im(\gamma_2)|^2 \leq \frac{1}{2}\kappa' m')\Big)\\
          &\geq \frac{1}{2} + \frac{0.298}{6},
      \end{aligned}
      \end{equation}
where we make the choice of $\sigma_1^2 = 0.99\kappa\tanh^2(2r)/2$ and $\sigma_2^2=0.99\kappa'/2$.In the range of $L<\kappa'm'<1.16$, we always have $2\sigma^2<\frac{2L}{\kappa'm'}$ which along with ensuring $\kappa'^2m'>2/0.99$ and $\kappa m>2/0.99$ will give $0.99\kappa'>2\sigma^2$ and $0.99\kappa\tanh^2(2r)\geq \frac{2}{\cosh(2r)}$ (since $\sinh(2r)\tanh(2r)\geq \kappa m$). We now can 
Lemma~\ref{lem:gaussprob} to obtain the lower bound on the success probability. Further in the relevant parameter range, $\cosh(2r)\geq 1.06\kappa m$ is sufficient to ensure $\sinh(2r)\tanh(2r)\geq \kappa m$.

\textbf{Application of Lemma~\ref{lem:master}}: 
By defining the following operators
\begin{equation}
    \hat{K}_{x,\gamma_1} = \begin{cases}
        \mathbb{I} & x = 0\\
        \hat{E}(\gamma_1) & x = 1\\
        \hat{E}(-\gamma_1^*) & x = 2
    \end{cases},
\end{equation}
\begin{equation}
    \hat{\omega}_{x,\gamma_2} = \hat{\rho}_0^{1/2}\hat{W}_{x,\gamma_2}\hat{\rho}_0^{1/2},\quad \hat{W}_{x,\gamma_2} = \begin{cases}
        \mathbb{I} & x = 0\\
        \frac{\epsilon_0}{\sqrt{2}}i\left(\hat{D}(\gamma_2) - \hat{D}^\dagger(\gamma_2)\right) & x= 1\\
        \frac{\epsilon_0}{\sqrt{2}}i\left(\hat{D}(-\gamma_2^*) - \hat{D}^\dagger(-\gamma_2^*)\right) & x =2
    \end{cases},
\end{equation}
we note that we have the exact setting for Lemma~\ref{lem:master} specifically with $c = 1$ and $l = 2$. Observe that using Corollary~\ref{corr:bosonicEnorm} we have
\begin{equation}
    \left\|\mathbb{E}_{\gamma_1}\left[\hat{K}_{1,\gamma_1}^{\otimes 2}\right]\right\|_{\mathrm{op}} =\left\|\mathbb{E}_{\gamma_1}\left[\hat{K}_{2,\gamma_1}^{\otimes 2}\right]\right\|_{\mathrm{op}} \leq\frac{2}{(1+4\sigma_1^4)^{m/2}},
\end{equation}
\begin{equation}
    \left\|\mathbb{E}_{\gamma_2}\left[\hat{W}_{1,\gamma_2}^{\otimes 2}\right]\right\|_{\mathrm{op}} =\frac{\epsilon_0^2}{2}\left\|\mathbb{E}_{\gamma_2}\left[(i\hat{D}(\gamma_2)-i\hat{D}^\dagger(\gamma_2))^{\otimes 2}\right]\right\|_{\mathrm{op}}\leq \frac{2\epsilon_0^2}{(1+4\sigma_{2}^4)^{m'/2}},
\end{equation}
and since $\hat{W}_{2,\gamma_2}^T = \hat{W}_{1,\gamma_2}$ the same inequality holds for $\left\|\mathbb{E}_{\gamma_2}\left[\hat{W}_{2,\gamma_2}^{\otimes 2}\right]\right\|_{\mathrm{op}}$. Combining these two bounds and applying the case of Corollary~\ref{corr:master}, we obtain the weaker upper bound on $\Delta$ as
\begin{equation}
\begin{aligned}
    \Delta &\leq \sqrt{\left\|\mathbb{E}_{\gamma_1}\left[\hat{K}_{1,\gamma_1}^{\otimes 2}\right]\right\|_{\mathrm{op}}^2\left\|\mathbb{E}_{\gamma_2}\left[\hat{W}_{1,\gamma_2}^{\otimes 2}\right]\right\|_{\mathrm{op}}^2} + \sqrt{\left\|\mathbb{E}_{\gamma_1}\left[\hat{K}_{2,\gamma_1}^{\otimes 2}\right]\right\|_{\mathrm{op}}^2\left\|\mathbb{E}_{\gamma_2}\left[\hat{W}_{2,\gamma_2}^{\otimes 2}\right]\right\|_{\mathrm{op}}^2} \\
    &\leq 2\epsilon_0^2\left(\sqrt{((1+(0.99\kappa\tanh^2(2r))^2)^{m/2}(1+(0.99\kappa')^2)^{m'/2})^{-2}}+\sqrt{((1+(0.99\kappa\tanh^2(2r))^2)^{m/2}(1+(0.99\kappa')^2)^{m'/2})^{-2}}\right),
\end{aligned}
\end{equation}
\begin{equation}
\begin{aligned}
    \Delta &\leq \frac{4\epsilon_0^2}{(1+(0.99\kappa')^2)^{m'/2}(1+(0.99\kappa\tanh^2(2r))^2)^{m/2}}\\&\implies N= \Omega(1/\Delta) = \Omega(\epsilon^{-2}(1+(0.99\kappa')^2)^{m'/2}(1+(0.99\tanh^2(2r)\kappa)^2)^{m/2}).
\end{aligned}
\end{equation}
\end{proof}

\subsection{Lower bounds with access to conjugate of the channel}\label{app:lowerbound_conjugate}
The channel learning problem for the qubit case has already been analyzed in \cite{10.1145/3670418}; however, we note that there are certain subtleties in relation to the conjugate channel as a resource for the case of qubits. First, we observe that the phase space displacement operators for $m$ qubits will have 
\begin{equation}
    \hat{D}_{2,m}(\mathbf{q},\mathbf{p}) = (-1)^{\mathbf{p}\cdot\mathbf{q}}\hat{D}_{2,m}(\mathbf{q},\mathbf{p})^T = (-1)^{\mathbf{p}\cdot\mathbf{q}}\hat{D}_{2,m}(-\mathbf{q},\mathbf{p}),
\end{equation}
for all $\mathbf{q},\mathbf{p}\in \mathbb{F}_2^{m}$. Hence any channel $\mathcal{E}\in \CPTP(\mathcal{H}_{2,m},\mathcal{H}_{2,m'})$ will always satisfy
\begin{equation}
    C_{\mathcal{E}}((-\mathbf{q},\mathbf{p}),(-\mathbf{q}',\mathbf{p}')) = (-1)^{\mathbf{p}\cdot\mathbf{q}+\mathbf{p}'\cdot\mathbf{q}'}C_{\mathcal{E}}((\mathbf{q},\mathbf{p}),(\mathbf{q}',\mathbf{p}')).
\end{equation}
for all $(\mathbf{q},\mathbf{p})\in \mathbb{F}_2^{m}\times\mathbb{F}_2^m$ and $(\mathbf{q}',\mathbf{p}')\in \mathbb{F}_2^{m'}\times \mathbb{F}_2^{m'}$. Note that the efficiency in Thm.~\ref{thm:efficientlearnconj} comes from the conjugate channel having $C_{\mathcal{E}^*}((-\mathbf{q},\mathbf{p}),(-\mathbf{q}',\mathbf{p}')) = C_{\mathcal{E}}((\mathbf{q},\mathbf{p}),(\mathbf{q}',\mathbf{p}'))$ since this makes the Bell measurement possible, allowing efficient estimation of $C_{\mathcal{E}}((\mathbf{q},\mathbf{p}),(\mathbf{q}',\mathbf{p}'))^2$ for all possible queries. Hence whether provided $2$ copies of the same channel $\mathcal{E}$ or $1$ copy of $\mathcal{E}\otimes\mathcal{E}^*$, these cases have effectively the same performance for the channel learning Problem~\ref{prob:channel_learn}. We now make these observations concrete with the analysis of the general multi-qudit/bosonic cases.

\begin{lemma}\label{lem:discr_disp_transtensor}
    Assuming prime $d$ and $k,k'$ being positive integers,
    \begin{equation}
        \left\|\sum_{\mathbf{q},\mathbf{p}\in\mathbb{F}^{m}_d}\hat{D}_{d,m}(\mathbf{q},\mathbf{p})^{\otimes 2k}\otimes \hat{D}_{d,m}(-\mathbf{q},\mathbf{p})^{\otimes 2k'}\right\|_\mathrm{op} = \begin{cases}
            d^m,& k\neq k' \mod d\\
            d^{2m}, & k =k'\mod d
        \end{cases}.
    \end{equation}
\end{lemma}
\begin{proof}
We first note that in the qubit case $\hat{D}_{2,m}(-\mathbf{q},\mathbf{p}) = (-1)^{\mathbf{q}\cdot\mathbf{p}}\hat{D}_{2,m}(\mathbf{q},\mathbf{p})$. Hence $\hat{D}_{2,m}(\mathbf{q},\mathbf{p})^{\otimes 2k}\otimes \hat{D}_{2,m}(-\mathbf{q},\mathbf{p})^{\otimes 2 k'} = \hat{D}_{2,m}(\mathbf{q},\mathbf{p})^{\otimes 2(k+k')}$ since the $(-1)^{\mathbf{q}\cdot\mathbf{p}}$ phase would be irrelevant due to being tensored $2k'$ times. Hence this is the exact same formulation as Lemma~\ref{lem:discr_disp_tensor} which has two cases of $k +k' = 0\mod 2$ and $k+k' = 1\mod 2$ which are equivalent to $k  =k'\mod 2$ and $k\neq k'\mod 2$ yielding operator norms of $2^{2m}$ and $2^m$ respectively.

For qutrits and beyond, we first examine how this operator acts on a basis state in $\mathcal{H}_{d,m}$ defined as follows where $\mathbf{a}^{(i)}\in\mathbb{F}^m_d$
    \begin{equation}
        \begin{aligned}
            &\sum_{\mathbf{q},\mathbf{p}\in\mathbb{F}^{m}_d}\hat{D}_{d,m}(\mathbf{q},\mathbf{p})^{\otimes 2k}\otimes \hat{D}_{d,m}(-\mathbf{q},\mathbf{p})^{\otimes 2k'}\bigotimes_{i=1}^{2k}\ket{\mathbf{a}^{(i)}}\bigotimes_{i'=1}^{2k'}\ket{\mathbf{a}^{(i')}} \\&=\sum_{\mathbf{q},\mathbf{p}\in\mathbb{F}^{m}_d}e^{i\frac{2\pi}{d}\mathbf{p}\cdot((k-k')\mathbf{q} + \sum_{i=1}^{2k}\mathbf{a}^{(i)}+\sum_{i'=1}^{2k'}\mathbf{a}^{(i')})}\bigotimes_{i=1}^{2k}\ket{\mathbf{a}^{(i)} + \mathbf{q}}\bigotimes_{i'=1}^{2k'}\ket{\mathbf{a}^{(i')} - \mathbf{q}}.
        \end{aligned}
    \end{equation}
    Observe that
    \begin{equation}
    \begin{aligned}
        \sum_{\mathbf{p}\in\mathbb{F}^{m}_d}e^{i\frac{2\pi}{d}\mathbf{p}\cdot((k-k')\mathbf{q} + \sum_{i=1}^{2k}\mathbf{a}^{(i)}+ \sum_{i'=1}^{2k'}\mathbf{a}^{(i')})} &= \prod_{j=1}^m\left(\sum_{p\in\mathbb{F}_d}e^{i\frac{2\pi}{d}p\left((k-k')q_j + \sum_{i=1}^{2k}a^{(i)}_j+\sum_{i'=1}^{2k'}a^{(i')}_j\right)}\right) \\&=d^{m}\delta_{(k-k')\mathbf{q} + \sum_{i=1}^{2k}\mathbf{a}^{(i)}+\sum_{i'=1}^{2k'}\mathbf{a}^{(i')}\mod d,0}.
    \end{aligned}
    \end{equation}
    \textbf{Case 1:} $k \neq k'\mod d$. We have a $g\in\mathbb{F}_{d}$ such that $g(k-k') = 1\mod d$ which then means that there is a unique $q_j = -g(\sum_{i=1}^{2k}a^{(i)}_j+\sum_{i'=1}^{2k'}a^{(i')}_j)$ for all $j$ hence giving
    \begin{equation}
        \sum_{\mathbf{q},\mathbf{p}\in\mathbb{F}^{m}_d}\hat{D}_{d,m}(\mathbf{q},\mathbf{p})^{\otimes 2k}\otimes \hat{D}_{d,m}(-\mathbf{q},\mathbf{p})^{\otimes 2k'}\bigotimes_{i=1}^{2k}\ket{\mathbf{a}^{(i)}}\bigotimes_{i'=1}^{2k'}\ket{\mathbf{a}^{(i')}} = d^m\bigotimes_{i=1}^{2k}\left|\mathbf{a}^{(i)} -gS\right\rangle\bigotimes_{i'=1}^{2k'}\left|\mathbf{a}^{(i')} 
        +gS\right\rangle,\quad S = \sum_{i=1}^{2k}\mathbf{a}^{(i)}+\sum_{i'=1}^{2k'}\mathbf{a}^{(i')}.
    \end{equation}
which means that this operation divided by $d^m$ maps basis states to other basis states. This mapping can be expressed as a matrix over $\mathbb{F}^{(2k+2k')m}_{d}$ given by $P = \mathbb{I}_{(2k+2k')m} - g(\mathbf{1}_{2k,2k'})\otimes\mathbb{I}_m$ where the matrix $\mathbf{1}_{2k,2k'}$ is a $2(k+k')\times 2(k+k')$ matrix given by
\begin{equation}
    \mathbf{1}_{2k,2k'} = \begin{pmatrix}
        1 & 1 & \cdots & 1\\
        \vdots & 2k-1\text{ more rows of }1\\
        -1& -1 &\cdots & -1\\
        \vdots & 2k'-1\text{ more rows of }-1\\
    \end{pmatrix}.
\end{equation}
The matrix $P$ is self-inverse over $\mathbb{F}^{(2k+2k')m}_d$ and so the operator
\begin{equation}
        \sum_{\mathbf{q},\mathbf{p}\in\mathbb{F}^{m}_d}\hat{D}_{d,m}(\mathbf{q},\mathbf{p})^{\otimes 2k}\otimes \hat{D}_{d,m}(-\mathbf{q},\mathbf{p})^{\otimes 2k'}\ket{\mathbf{a}} = d^m\ket{P\mathbf{a}}.
    \end{equation}
is $d^m$ times a permutation operation over the standard basis hence giving operator norm $d^m$.

\textbf{Case 2:} $k = k'\mod d$. In this case we have
\begin{equation}
    \begin{aligned}
            &\sum_{\mathbf{q},\mathbf{p}\in\mathbb{F}^{m}_d}\hat{D}_{d,m}(\mathbf{q},\mathbf{p})^{\otimes 2k}\otimes \hat{D}_{d,m}(-\mathbf{q},\mathbf{p})^{\otimes 2k'}\bigotimes_{i=1}^{2k}\ket{\mathbf{a}^{(i)}}\bigotimes_{i'=1}^{2k'}\ket{\mathbf{a}^{(i')}} \\&=\begin{cases}
        d^m\sum_{\mathbf{q}\in\mathbb{F}^m_d}\bigotimes_{i=1}^{2k}\ket{\mathbf{a}^{(i)} + \mathbf{q}}\bigotimes_{i'=1}^{2k'}\ket{\mathbf{a}^{(i')} - \mathbf{q}} & \sum_{i=1}^{2k}\mathbf{a}^{(i)}+\sum_{i'=1}^{2k'}\mathbf{a}^{(i')} = 0\mod d\\
        0 & \mathrm{otherwise}
    \end{cases}.
        \end{aligned} 
\end{equation}
We can check that a trivial eigenstate of this operator is the fully entangled input state  given by
\begin{equation}
    \ket{\Psi} =\frac{1}{
    d^{m/2}} \sum_{\mathbf{z}\in\mathbb{F}^m_d}\ket{\mathbf{z}}^{\otimes 2k}\ket{-\mathbf{z}}^{\otimes 2k'},\quad \sum_{\mathbf{q},\mathbf{p}\in\mathbb{F}^{m}_d}\hat{D}_{d,m}(\mathbf{q},\mathbf{p})^{\otimes 2k}\otimes \hat{D}_{d,m}(-\mathbf{q},\mathbf{p})^{\otimes 2k'}\ket{\Psi}  = d^{2m}\ket{\Psi}.
\end{equation}
which can be checked to yield an operator norm of $d^{2m}$. Note that since the operator in question is summing $d^{2m}$ unitary operators, the triangle inequality yields an upper bound of $d^{2m}$ which can be shown to be saturated by the choice of this input state.
\end{proof}

\begin{corollary}
     Assuming $d$ prime and $k,k',r,r'$ being positive integers with $1\leq r\leq 2k$ and $1\leq r'\leq 2k'$,
    \begin{equation}
        \begin{aligned}
            \left\|\sum_{\mathbf{q},\mathbf{p}\in\mathbb{F}^{m}_d}\hat{D}_{d,m}(\mathbf{q},\mathbf{p})^{\otimes r}\otimes \hat{D}_{d,m}(-\mathbf{q},-\mathbf{p})^{\otimes 2k-r}\otimes \hat{D}_{d,m}(-\mathbf{q},\mathbf{p})^{\otimes r'}\otimes \hat{D}_{d,m}(\mathbf{q},-\mathbf{p})^{\otimes 2k'-r'}\right\|_\mathrm{op} = \begin{cases}
            d^m,& k\neq k' \mod d\\
            d^{2m}, & k =k'\mod d
        \end{cases}.
        \end{aligned}
    \end{equation}
\end{corollary}
\begin{proof}
    This is a trivial extension using the parity operation defined in Corollary~\ref{corr:dispparityflip}.
\end{proof}

\begin{corollary}\label{corr:quditEtransopnorm}
    For $d$ prime, $k,k'$ positive integers and any $\theta\in\mathbb{R}$, the operators $\hat{E}_{d,m}$ of Eq.~\eqref{eq:quditHWobs} satisfy
    \begin{equation}
        \frac{1}{d^{2m}-1}\left\|\sum_{\mathbf{q},\mathbf{p}\in\mathbb{F}^m_d,(\mathbf{q},\mathbf{p})\neq (0,0)}\hat{E}_{d,m}(\mathbf{q},\mathbf{p})^{\otimes 2k}\otimes \hat{E}_{d,m}(-\mathbf{q},\mathbf{p})^{\otimes 2k'}\right\|_{\mathrm{op}} \leq \begin{cases}\frac{2^{k+k'}}{d^m-1}& k\neq k' \mod d\\
        2^{k+k'} & k = k'\mod d
        \end{cases}.
    \end{equation}
\end{corollary}
\begin{proof}
Note that
    \begin{equation}
    \begin{aligned}
        &\frac{1}{d^{2m} - 1}\sum_{\mathbf{q},\mathbf{p}\in\mathbb{F}^m_d,(\mathbf{q},\mathbf{p})\neq (0,0)}\hat{E}_{d,m}(\mathbf{q},\mathbf{p})^{\otimes 2k}\otimes \hat{E}_{d,m}(-\mathbf{q},\mathbf{p})^{\otimes 2k'}\\
        &= \frac{1}{d^{2m} - 1}\sum_{\mathbf{q},\mathbf{p} \in\mathbb{F}^m_{d}} \left(\frac{e^{i\theta}\hat{D}_{d,m}(\mathbf{q},\mathbf{p}) + e^{-i\theta}\hat{D}_{d,m}(-\mathbf{q},-\mathbf{p})}{\sqrt{2}}\right)^{\otimes 2k}\otimes\left(\frac{e^{i\theta}\hat{D}_{d,m}(-\mathbf{q},\mathbf{p}) + e^{-i\theta}\hat{D}_{d,m}(\mathbf{q},-\mathbf{p})}{\sqrt{2}}\right)^{\otimes 2k'} \\&\quad\quad- \frac{(\sqrt{2}\cos(\theta))^{2k+2k'}}{d^{2m}-1}\mathbb{I}  .
    \end{aligned} 
    \end{equation}
    By expanding the above sum in terms of displacements and then applying the triangle inequality we have the following for $k\neq k'\mod d$
    \begin{equation}
    \begin{aligned}
        &\frac{1}{d^{2m}-1}\left\|\sum_{\mathbf{q},\mathbf{p}\in\mathbb{F}^m_d,(\mathbf{q},\mathbf{p})\neq (0,0)}\hat{E}_{d,m}(\mathbf{q},\mathbf{p})^{\otimes 2k}\otimes \hat{E}_{d,m}(-\mathbf{q},\mathbf{p})^{\otimes 2k'}\right\|_{\mathrm{op}} \\&\leq \sum_{r=0}^{2k}\sum_{r'=0}^{2k'}\frac{\binom{2k}{r}\binom{2k'}{r'}}{2^{k+k'}(d^{2m}-1)} \left\|\sum_{\mathbf{q},\mathbf{p}\in\mathbb{F}^{m}_d}\hat{D}_{d,m}(\mathbf{q},\mathbf{p})^{\otimes r}\otimes \hat{D}_{d,m}(-\mathbf{q},-\mathbf{p})^{\otimes 2k-r}\otimes \hat{D}_{d,m}(-\mathbf{q},\mathbf{p})^{\otimes r'}\otimes \hat{D}_{d,m}(\mathbf{q},-\mathbf{p})^{\otimes 2k'-r'}\right\|_\mathrm{op}\\&\quad\quad + \frac{2^{k+k'}\cos^{2k+2k'}(\theta)}{d^{2m}-1} \\&\leq \frac{2^{k+k'}d^m}{d^{2m}-1} + \frac{2^{k+k'}}{d^{2m}-1} =\frac{2^{k+k'}}{d^{m}-1}.
    \end{aligned}
    \end{equation}
    For the case of $k = k'\mod d$, we can note that the operator can be expressed as the sum of $2^{2k+2k'}(d^{2m}-1)$ operators, each of the form $\frac{\hat{D}}{2^{k+k'}(d^{2m}-1)}$ where $\hat{D}$ is a displacement. Hence by the triangle inequality it trivially follows that $\|\hat{O}\|_{\mathrm{op}}\leq 2^{k+k'}$ for all choices of of $k$ and $k'$.\end{proof}
\begin{lemma}\label{lem:disptransopnorm}
Taking $P(\gamma)$ as a Gaussian distribution for $\gamma\in\mathbb{C}^m$ with spread $\sigma$, the following holds for $k,k'\geq 1$
\begin{equation}
    \left\|\int d^{2m}\gamma P(\gamma)\hat{D}^{\otimes k}(\gamma)\otimes\hat{D}^{\otimes k'}(-\gamma^*)\right\|_{\mathrm{op}} \leq \frac{1}{(1+(k-k')^2\sigma^4)^{m/2}}.
\end{equation}
\end{lemma}
\begin{proof}
    We define the operator $\hat{O}$ to be given by
    \begin{equation}
        \hat{O} = \int d^{2m}\gamma P(\gamma)\hat{D}^{\otimes k}(\gamma)\otimes \hat{D}^{\otimes k'}(-\gamma^*).
    \end{equation}
    By the definition of the operator norm, we have
    \begin{align}
\|\hat{O}\|^2_{\mathrm{op}} &= \sup_{\hat{\rho}\in D(\mathcal
        H^{\otimes (k+k')})}\text{Tr}\left[\hat{O}\hat{\rho}\hat{O}^\dagger\right]\\
        &= \sup_{\hat{\rho}\in D(\mathcal
        H^{\otimes (k+k')})}\int d^{2m}\gamma_1 d^{2m}\gamma_2 P(\gamma_1)P(\gamma_2)e^{\frac{k-k'}{2}i\Omega(\gamma_2,\gamma_1)}\text{Tr}\left[\hat{D}^{\otimes k}(\gamma_1-\gamma_2)\otimes \hat{D}^{\otimes k'}(-\gamma_1^*+\gamma_2^*)\hat{\rho}\right].
    \end{align}
    We define $\gamma_+ = \frac{\gamma_1+\gamma_2}{2}$ and $\gamma_- = \gamma_1 - \gamma_2$. Note that
    \begin{align}
        P(\gamma_1)P(\gamma_2) &= \frac{1}{(2\pi\sigma^2)^{2m}}\exp(-\frac{|\gamma_1|^2 + |\gamma_2|^2}{2\sigma^2})\\
        &= \left(\frac{1}{(\pi\sigma^2)^m}\exp(-\frac{|\gamma_+|^2}{\sigma^2})\right)\times\left(\frac{1}{(4\pi\sigma^2)^m}\exp(-\frac{|\gamma_-|^2}{4\sigma^2})\right),
    \end{align}
    and note that $i\Omega(\gamma_2,\gamma_1) = i\Omega(\gamma_+,\gamma_-)$, from which we use the Gaussian Fourier transform formula to note that
    \begin{align}
        \|\hat{O}\|_{\mathrm{op}}^2 &= \sup_{\hat{\rho}\in D(\mathcal
        H^{\otimes (k+k')})}\int d^{2m}\gamma_+ d^{2m}\gamma_- \left(\frac{1}{(\pi\sigma^2)^m}\exp(-\frac{|\gamma_+|^2}{\sigma^2})\right)\nonumber\\&\quad\quad\quad\times\left(\frac{1}{(4\pi\sigma^2)^m}\exp(-\frac{|\gamma_-|^2}{4\sigma^2})\right) e^{\frac{k-k'}{2}i\Omega(\gamma_+,\gamma_-)}\text{Tr}\left[\hat{D}^{\otimes k}(\gamma_-)\otimes \hat{D}^{\otimes k'}(\gamma_-^*)\hat{\rho}\right]\\
        &= \sup_{\hat{\rho}\in D(\mathcal
        H^{\otimes (k+k')})}\int d^{2m}\gamma_- \frac{1}{(4\pi\sigma^2)^m}\exp(-\frac{|\gamma_-|^2}{4}\left(\frac{1}{\sigma^2}+(k-k')^2\sigma^2\right)) \text{Tr}\left[\hat{D}^{\otimes k}(\gamma_-)\otimes \hat{D}^{\otimes k'}(-\gamma_-^*)\hat{\rho}\right]\\
        &\leq \int d^{2m}\gamma_- \frac{1}{(4\pi\sigma^2)^m}\exp(-\frac{|\gamma_-|^2}{4}\left(\frac{1}{\sigma^2}+(k-k')^2\sigma^2\right))
        \\&= \frac{1}{(1+(k-k')^2\sigma^4)^{m}}.
    \end{align}
    For the final inequality we use the fact that the characteristic function is a complex number with norm bounded by 1. Note that this holds true for arbitrary states $\hat{\rho}$ without any assumptions on the energy of the state.
\end{proof}
\begin{remark}
    Note that for Lemma~\ref{lem:disptransopnorm} the case of $k=k'$ makes the upper bound trivial which is important in showing the advantage of access to the complex-conjugate channel.
\end{remark}
\begin{corollary}\label{corr:bosonicEnormconj}
Taking $P(\gamma)$ as a Gaussian distribution for $\gamma\in\mathbb{C}^m$ with spread $\sigma$, the following holds for $k,k'\geq 1$ and any $\theta\in\mathbb{R}$

\begin{equation}
    \left\|\int d^{2m}\gamma P(\gamma)\left(\frac{e^{i\theta}\hat{D}(\gamma) + e^{-i\theta}\hat{D}^\dagger(\gamma)}{\sqrt{2}}\right)^{\otimes k}\otimes\left(\frac{e^{i\theta}\hat{D}(-\gamma^*) + e^{-i\theta}\hat{D}^\dagger(-\gamma^*)}{\sqrt{2}}\right)^{\otimes k'}\right\|_{\mathrm{op}} \leq \frac{2^{(k+k')/2}}{(1+(k-k')^2\sigma^4)^{m/2}}.
\end{equation}
\end{corollary}
\begin{proof}
    Note that 
    \begin{align}
        \left(\frac{e^{i\theta}\hat{D}(\gamma) + e^{-i\theta}\hat{D}^\dagger(\gamma)}{\sqrt{2}}\right)^{\otimes k}\otimes\left(\frac{e^{i\theta}\hat{D}(-\gamma^*) + e^{-i\theta}\hat{D}^\dagger(-\gamma^*)}{\sqrt{2}}\right)^{\otimes k'}= \frac{1}{2^{(k+k')/2}}\sum_{s\in\{-1,1\}^{k+k'}}e^{i\sum_{i=1}^{k+k'}s_i\theta}\bigotimes_{i = 1}^{k}\hat{D}(s_i\gamma)\bigotimes_{i = k+1}^{k+k'}\hat{D}(-s_i\gamma^*).
    \end{align}
    Note that $\|\hat{O}\|_{\mathrm{op}} = \|\hat{U}^\dagger\hat{O}\hat{U}\|_{\mathrm{op}}$ where $\hat{U}$ is a unitary and $\hat{D}(\gamma) = e^{i\pi\hat{n}}\hat{D}^\dagger(\gamma)e^{-i\pi\hat{n}}$. As a result for all $s\in\{-1,1\}^{k+k'}$, we have $\|\int d^{2m}\gamma P(\gamma)\hat{D}^{\otimes k}(\gamma)\otimes\hat{D}^{\otimes k'}(-\gamma^*)\|_{\mathrm{op}} = \|\int d^{2m}\gamma P(\gamma)\bigotimes_{i = 1}^{k}\hat{D}(s_i\gamma)\bigotimes_{i = k+1}^{k+k'}\hat{D}(-s_i\gamma^*)\|_{\mathrm{op}}$ since they are connected by unitary $\hat{U} = \bigotimes_{i=1}^{k+k'} e^{i\pi (1-s_i)\hat{n}/2}$. By triangle inequality, we have
    \begin{align}
       &\left\|\int d^{2m}\gamma P(\gamma)\left(\frac{e^{i\theta}\hat{D}(\gamma) + e^{-i\theta}\hat{D}^\dagger(\gamma)}{\sqrt{2}}\right)^{\otimes k}\otimes\left(\frac{e^{i\theta}\hat{D}(-\gamma^*) + e^{-i\theta}\hat{D}^\dagger(-\gamma^*)}{\sqrt{2}}\right)^{\otimes k'}\right\|_{\mathrm{op}} \nonumber\\
       &\leq \frac{1}{2^{(k+k')/2}}\sum_{s\in\{-1,1\}^{k+k'}}\left\|\int d^{2m}\gamma P(\gamma)e^{i\sum_{i=1}^{k+k'}s_i\theta}\bigotimes_{i = 1}^{k}\hat{D}(s_i\gamma)\bigotimes_{i = k+1}^{k+k'}\hat{D}(-s_i\gamma^*)\right\|_{\mathrm{op}}\\
       &= 2^{(k+k')/2}\left\|\int d^{2m}\gamma P(\gamma)\hat{D}^{\otimes k}(\gamma)\otimes\hat{D}^{\otimes k'}(-\gamma^*)\right\|_{\mathrm{op}} \leq \frac{2^{(k+k')/2}}{(1+(k-k')^2\sigma^4)^{m/2}}.
    \end{align}
\end{proof}

\begin{theorem}\label{thm:qudit2quditwithconj}
Consider a $\mathcal{E}\in \CPTP(\mathcal{H}_{d,m},\mathcal{H}_{d',m'})$ ($d,d'$ are prime) and its conjugate $\mathcal{E}^*$ with conjugation defined in the standard computational bases of both $\mathcal{H}_{d,m}$ and $\mathcal{H}_{d',m'}$. Taking $\epsilon \leq1/8$ and $c\leq \min(d',d) - 1$, consider any adaptive ancilla-assisted learning scheme that is allowed to perform arbitrary measurements along with arbitrary inputs to parallel uses of the channel $(\mathcal{E}\otimes \mathcal{E}^*)^{\otimes c}$. If the learner can use this scheme to produce an estimate of $|C_{\mathcal{E}}(\mathbf{q},\mathbf{p},\mathbf{q}',\mathbf{p}')|$ such that for a query of $\mathbf{q},\mathbf{p}\in\mathbb{F}_{d}^{m}$, $\mathbf{q}',\mathbf{p}'\in\mathbb{F}_{d'}^{m'}$ it satisfies $\left||\hat{C}_{\mathcal{E}}(\mathbf{q},\mathbf{p},\mathbf{q}',\mathbf{p}')| - |C_{\mathcal{E}}(\mathbf{q},\mathbf{p},\mathbf{q}',\mathbf{p}')|\right|\leq \epsilon$ with a 2/3 success probability, then this learning scheme requires at least $N$ uses of $\mathcal{E}\otimes\mathcal{E}^*$ where 
$$N = \Omega\left(\min\left(\frac{(d'^{m'}-1)(d^{m}-1)}{c\epsilon^2},\frac{64}{c^3\epsilon^4}\right)\right).$$
\end{theorem}
\begin{proof}
    Consider the following two channels
    \begin{gather}
        \mathcal{E}_0(\cdot) = \text{Tr}(\cdot)\frac{\mathbb{I}}{d'^{m'}},\quad\mathcal{E}_{(\mathbf{q}_1,\mathbf{p}_1),(\mathbf{q}_2,\mathbf{p}_2)}(\cdot) = \text{Tr}\left[\hat{\Pi}_{+,(\mathbf{q}_1,\mathbf{p}_1)}(\cdot)\right]\hat{\rho}_{+,(\mathbf{q}_2,\mathbf{p}_2)} + \text{Tr}\left[\hat{\Pi}_{-,(\mathbf{q}_1,\mathbf{p}_1)}(\cdot)\right]\hat{\rho}_{-,(\mathbf{q}_2,\mathbf{p}_2)},\\
        \hat{\Pi}_{\pm,(\mathbf{q}_1,\mathbf{p}_1)}  = \frac{\mathbb{I} \pm \hat{E}_{d,m}(\mathbf{q}_1,\mathbf{p}_1)/\sqrt{2}}{2},\\
        \hat{\rho}_{\pm,(\mathbf{q}_2,\mathbf{p}_2)} = \frac{\mathbb{I} \pm \epsilon_0\hat{E}_{d',m'}(\mathbf{q}_2,\mathbf{p}_2)}{d'^{m'}}.
    \end{gather}
    This simplifies the channel $\mathcal{E}_{(\mathbf{q}_1,\mathbf{p}_1),(\mathbf{q}_2,\mathbf{p}_2)}^*(\cdot)$ as 
    \begin{equation}
        \mathcal{E}^*_{(\mathbf{q}_1,\mathbf{p}_1),(\mathbf{q}_2,\mathbf{p}_2)}(\cdot)= \text{Tr}(\cdot)\frac{\mathbb{I}}{d'^{m'}} + \frac{\epsilon_0}{d'^{m'}\sqrt{2}}\text{Tr}(\hat{E}_{d,m}(-\mathbf{q}_1,\mathbf{p}_1)(\cdot))\hat{E}_{d',m'}(-\mathbf{q}_2,\mathbf{p}_2),
    \end{equation}
    which means that we can define the following superoperators
    \begin{equation}
    \begin{aligned}
        \Lambda_{(\mathbf{q}_1,\mathbf{p}_1),(\mathbf{q}_2,\mathbf{p}_2)}(\cdot)&=(\mathcal{E}_{(\mathbf{q}_1,\mathbf{p}_1),(\mathbf{q}_2,\mathbf{p}_2)}\otimes \mathcal{E}^*_{(\mathbf{q}_1,\mathbf{p}_1),(\mathbf{q}_2,\mathbf{p}_2)})(\cdot)\\
        &=\text{Tr}(\cdot)\frac{\mathbb{I}\otimes\mathbb{I}}{d'^{2m'}} + \frac{\epsilon_0}{d'^{2m'}\sqrt{2}}\text{Tr}\left[(\hat{E}_{d,m}(\mathbf{q}_1,\mathbf{p}_1)\otimes\mathbb{I})(\cdot)\right](\hat{E}_{d',m'}(\mathbf{q}_2,\mathbf{p}_2)\otimes\mathbb{I})\\
        &\quad\quad + \frac{\epsilon_0}{d'^{2m'}\sqrt{2}}\text{Tr}\left[(\mathbb{I}\otimes\hat{E}_{d,m}(-\mathbf{q}_1,\mathbf{p}_1))(\cdot)\right](\mathbb{I}\otimes\hat{E}_{d',m'}(-\mathbf{q}_2,\mathbf{p}_2))\\
        &\quad\quad+\frac{\epsilon_0^2}{2d'^{2m'}}\text{Tr}\left[(\hat{E}_{d,m}(\mathbf{q}_1,\mathbf{p}_1)\otimes\hat{E}_{d,m}(-\mathbf{q}_1,\mathbf{p}_1))(\cdot)\right](\hat{E}_{d',m'}(\mathbf{q}_2,\mathbf{p}_2)\otimes\hat{E}_{d',m'}(-\mathbf{q}_2,\mathbf{p}_2)),
    \end{aligned}
    \end{equation}
    \begin{equation}
        \Lambda_{0}(\cdot) = \text{Tr}\left(\cdot\right) \frac{\mathbb{I}\otimes\mathbb{I}}{d'^{2m'}}.
    \end{equation}
    \textbf{Succeeding in Problem~\ref{prob:manyrevel} by learning $|C_{\mathcal{E}}((\mathbf{q}_1,\mathbf{p}_1),(\mathbf{q}_2,\mathbf{p}_2))|$}: Note that we have the exact same setting as Problem~\ref{prob:manyrevel} for telling apart $\Lambda_0$ and $\Lambda_{(\mathbf{q}_1,\mathbf{p}_1),(\mathbf{q}_2,\mathbf{p}_2)}$. Clearly these can be told apart by learning $|C_{\mathcal{E}}((\mathbf{q}_1,\mathbf{p}_1),(\mathbf{q}_2,\mathbf{p}_2))|$ as shown in Thm.~\ref{thm:qudit2qudit} by picking in this case $\epsilon = \frac{\epsilon_0}{4\sqrt{2}} \leq \frac{1}{8}$ and so we now continue to find a lower bound through the use of Lemma~\ref{lem:master}.

    \textbf{Application of Lemma~\ref{lem:master}}: The channel $\Lambda_{(\mathbf{q}_1,\mathbf{p}_1),(\mathbf{q}_2,\mathbf{p}_2)}$ is a sum of 4 terms giving $l =3$ when written in the form of Eq.~\eqref{eq:Epar1par2}. We can also equivalently break it into two bits $x$ and $y$ in the summation giving 
    \begin{equation}
        \Lambda_{(\mathbf{q}_1,\mathbf{p}_1),(\mathbf{q}_2,\mathbf{p}_2)}(\cdot) = \sum_{x,y\in\{0,1\}}\text{Tr}\left[(\cdot)\hat{K}_{x,y,(\mathbf{q}_1,\mathbf{p}_1)}\right]\sqrt{\hat{\rho}_0}\hat{W}_{x,y,(\mathbf{q}_2,\mathbf{p}_2)}\sqrt{\hat{\rho}_0}
    \end{equation}
    where $\hat{\rho}_0 = \mathbb{I}_{d',2m'}/d'^{2m'}$ and we define
    \begin{align}
        2^{(x+y)/2}\hat{K}_{x,y,(\mathbf{q},\mathbf{p})} &= \begin{cases}
            \mathbb{I}\otimes \mathbb{I} & x,y = 0,0\\
            \mathbb{I}\otimes \hat{E}_{d,m}(-\mathbf{q},\mathbf{p})& x,y = 0,1\\
            \hat{E}_{d,m}(\mathbf{q},\mathbf{p})\otimes \mathbb{I} & x,y = 1,0\\
            \hat{E}_{d,m}(\mathbf{q},\mathbf{p})\otimes \hat{E}_{d,m}(-\mathbf{q},\mathbf{p}) & x,y=1,1
        \end{cases},\\
        \frac{1}{\epsilon_0^{x+y}}\hat{W}_{x,y,(\mathbf{q}_2,\mathbf{p}_2)} &= \begin{cases}
            \mathbb{I}\otimes \mathbb{I} & x,y = 0,0\\
            \mathbb{I}\otimes \hat{E}_{d',m'}(-\mathbf{q}_2,\mathbf{p}_2) & x,y = 0,1\\
            \hat{E}_{d',m'}(\mathbf{q}_2,\mathbf{p}_2)\otimes \mathbb{I} & x,y = 1,0\\
            \hat{E}_{d',m'}(\mathbf{q}_2,\mathbf{p}_2)\otimes \hat{E}_{d',m'}(-\mathbf{q}_2,\mathbf{p}_2) & x,y=1,1
        \end{cases}.
    \end{align}
    We can equivalently rewrite the $\Delta$ from Lemma~\ref{lem:master} as follows
    \begin{equation}
    \begin{aligned}
        \Delta^{1/2}\leq &\sum_{\pmb{x},\pmb{y}\in\{0,1\}^c,|\pmb{x}|+|\pmb{y}|>0}\Bigg\{\left\|\mathbb{E}_{(\mathbf{q}_1,\mathbf{p}_1)}[\hat{K}_{\pmb{x},\pmb{y},(\mathbf{q}_1,\mathbf{p}_1)}^{\otimes2}]\right\|_{\mathrm{op}}\times\left\|\mathbb{E}_{(\mathbf{q}_2,\mathbf{p}_2)}[\hat{W}_{\pmb{x},\pmb{y},(\mathbf{q}_2,\mathbf{p}_2)}^{\otimes 2}]\right\|_{\mathrm{op}}\Bigg\}^{1/2},
    \end{aligned}
    \end{equation}
    where we define
    \begin{equation}
        \hat{K}_{\pmb{x},\pmb{y},(\mathbf{q}_1,\mathbf{p}_1)} = \bigotimes_{i=1}^c \hat{K}_{x_i,y_i,(\mathbf{q}_1,\mathbf{p}_1)},\quad  \hat{W}_{\pmb{x},\pmb{y},(\mathbf{q}_2,\mathbf{p}_2)} = \bigotimes_{i=1}^c \hat{W}_{x_i,y_i,(\mathbf{q}_2,\mathbf{p}_2)}.
    \end{equation}
    From this we can observe that
    \begin{equation}
    \begin{aligned}
        \left\|\mathbb{E}_{(\mathbf{q}_1,\mathbf{p}_1)}[\hat{K}_{\pmb{x},\pmb{y},(\mathbf{q}_1,\mathbf{p}_1)}^{\otimes2}]\right\|_{\mathrm{op}}  &= \frac{1}{d^{2m}-1}\left\|\sum_{\mathbf{q}_1,\mathbf{p}_1\in\mathbb{F}_{d}^{m}}\hat{K}_{\pmb{x},\pmb{y},(\mathbf{q}_1,\mathbf{p}_1)}\otimes \hat{K}_{\pmb{x},\pmb{y},(\mathbf{q}_1,\mathbf{p}_1)}\right\|_{\mathrm{op}}\\
        &= \frac{1}{2^{|\pmb{x}|+|\pmb{y}|}(d^{2m} - 1)}\left\|\sum_{\mathbf{q}_1,\mathbf{p}_1\in\mathbb{F}_{d}^{m}}\hat{E}_{d,m}(\mathbf{q}_1,\mathbf{p}_1)^{\otimes2|\pmb{x}|}\otimes \hat{E}_{d,m}(-\mathbf{q}_1,\mathbf{p}_1)^{\otimes2|\pmb{y}|}\right\|_{\mathrm{op}}\\
        &\leq \begin{cases}
        \frac{1}{d^{m}-1} & |\pmb{x}| \neq |\pmb{y}| \mod d\\
        1 & |\pmb{x}| = |\pmb{y}| \mod d
                \end{cases}.
    \end{aligned}
    \end{equation}
    Similarly we also have
    \begin{equation}
    \begin{aligned}
        \left\|\mathbb{E}_{(\mathbf{q}_2,\mathbf{p}_2)}[\hat{W}_{\pmb{x},\pmb{y},(\mathbf{q}_2,\mathbf{p}_2)}^{\otimes 2}]\right\|_{\mathrm{op}} &\leq \frac{1}{d'^{2m'}-1}\left\|\sum_{\mathbf{q}_2,\mathbf{p}_2\in\mathbb{F}_{d'}^{m'}}\hat{W}_{\pmb{x},\pmb{y},(\mathbf{q}_2,\mathbf{p}_2)}\otimes\hat{W}_{\pmb{x},\pmb{y},(\mathbf{q}_2,\mathbf{p}_2)}\right\|_{\mathrm{op}}\\
        &= \frac{\epsilon_0^{2|\pmb{x}|+2|\pmb{y}|}}{(d'^{2m'} - 1)}\left\|\sum_{\mathbf{q}_2,\mathbf{p}_2\in\mathbb{F}_{d'}^{m'}}\hat{E}_{d',m'}(\mathbf{q}_2,\mathbf{p}_2)^{\otimes2|\pmb{x}|}\otimes \hat{E}_{d',m'}(-\mathbf{q}_2,\mathbf{p}_2)^{\otimes2|\pmb{y}|}\right\|_{\mathrm{op}}\\
        &\leq \begin{cases}
        \frac{(\epsilon_0\sqrt{2})^{2|\pmb{x}|+2|\pmb{y}|}}{d'^{m'}-1} & |\pmb{x}| \neq |\pmb{y}| \mod d'\\
       (\epsilon_0\sqrt{2})^{2|\pmb{x}|+2|\pmb{y}|} & |\pmb{x}| = |\pmb{y}| \mod d'
                \end{cases}.
    \end{aligned}
    \end{equation}
    Note that in the expression for $\Delta$, the only case where $|\pmb{x}|=|\pmb{y}|\mod d$ or $|\pmb{x}|=|\pmb{y}|\mod d'$ is if $|\pmb{x}|= |\pmb{y}|$ since $c < \min(d',d)$. 
    By substituting these bounds into the expression of $\Delta$ and adding extra terms for the case of $|\pmb{x}| = |\pmb{y}|$, we get
    \begin{equation}
    \begin{aligned}
        \Delta &\leq \left(\sum_{k,l=0,(k,l)\neq(0,0)}^{c}\binom{c}{k}\binom{c}{l}\frac{(\epsilon_0\sqrt{2})^{k+l}}{\sqrt{(d^{m}-1)(d'^{m'}-1)}} + \sum_{k=1}^c \binom{c}{k}^2(\epsilon_0\sqrt{2})^{2k}\right)^2\\
        &\leq \left(\frac{(1+\epsilon_0\sqrt{2})^c - 1}{(d'^{m'}-1)^{1/2}(d^{m}-1)^{1/2}} + (1-2\epsilon_0^2)^cL_c\left(\frac{1+2\epsilon_0^2}{1-2\epsilon_0^2}\right) -1\right)^2\\
        &\leq \left(\frac{c\epsilon_0\sqrt{2}e^{c\epsilon_0\sqrt{2}}}{(d'^{m'}-1)^{1/2}(d^{m}-1)^{1/2}} + 2c^2\epsilon_0^2e^{2c\epsilon_0\sqrt{2}}\right)^2,
    \end{aligned}
    \end{equation}
    where we make use of Lemma~\ref{lem:epssum}. This gives the lower bound on sample complexity to be
    \begin{equation}
        N = \Omega\left(\min\left(\frac{(d'^{m'}-1)(d^{m}-1)}{c\epsilon^2},\frac{64}{c^3\epsilon^4}\right)\right).
    \end{equation}
\end{proof}
Note that we do not examine the case of $c\geq\min(d,d')$ since we already examine these cases in the absence of the complex-conjugate channel, and do not expect the lower bounds in that case to offer any additional insights.
\begin{theorem}\label{thm:qudit2bosonwithconj}
Consider a $\mathcal{E}\in \CPTP(\mathcal{H}_{d,m},\mathcal{H}_{\infty,m'})$ with $d$ being prime and $m'\geq 8$ bosonic modes and its conjugate $\mathcal{E}^*$ with conjugation defined in the standard computational basis of both $\mathcal{H}_{d,m}$ and Fock basis of $\mathcal{H}_{\infty,m'}$. Taking $\epsilon < 0.1225$, $c\leq d-1$, consider any adaptive ancilla-assisted learning scheme that is allowed uses of the channel $(\mathcal{E}\otimes\mathcal{E}^*)^{\otimes c}$ for each measurement. If the learner can use this scheme to produce an estimate $|\hat{C}_{\mathcal{E}}((\mathbf{q},\mathbf{p}),\beta)|$ such that for a query of $\mathbf{q},\mathbf{p}\in\mathbb{F}_d^m$ and $|\beta|^2\leq \kappa m'$ it satisfies $\left||\hat{C}_{\mathcal{E}}((\mathbf{q},\mathbf{p}),\beta)| - |C_{\mathcal{E}}((\mathbf{q},\mathbf{p}),\beta)|\right|\leq \epsilon$ with a 2/3 success probability, then this learning scheme requires at least $N=\Omega\left(\min\left(\frac{(d^m-1)(1+(0.99\kappa)^2)^{m'/2}}{c(\epsilon/0.1225)^2},\frac{1}{c^3(\epsilon/0.1225)^4}\right)\right)$ uses of $\mathcal{E}\otimes\mathcal{E}^*$.
\end{theorem}
\begin{proof}
Consider the following two channels
    \begin{gather}
        \mathcal{E}_0(\cdot) = \text{Tr}(\cdot)\hat{\rho}_0,\quad\mathcal{E}_{(\mathbf{q},\mathbf{p}),\gamma}(\cdot) = \text{Tr}\left[\hat{\Pi}_{+,(\mathbf{q},\mathbf{p})}(\cdot)\right]\hat{\rho}_\gamma + \text{Tr}\left[\hat{\Pi}_{-,(\mathbf{q},\mathbf{p})}(\cdot)\right]\hat{\rho}_{-\gamma},\\
        \hat{\Pi}_{\pm,(\mathbf{q},\mathbf{p})}  = \frac{\mathbb{I} \pm \hat{E}_{d,m}(\mathbf{q},\mathbf{p})/\sqrt{2}}{2},\\
        \hat{\rho}_{\gamma} = (1-\nu^2)^{m'}\nu^{\hat{n}}(\hat{D}(0) + 2i\epsilon_0(\hat{D}(\gamma) - \hat{D}(-\gamma)))\nu^{\hat{n}},
    \end{gather}
where $\hat{n}$ is the total photon number operator in $m'$ bosonic modes. We now define the complex-conjugate channel
\begin{equation}
    \mathcal{E}^*_{(\mathbf{q},\mathbf{p}),\gamma}(\cdot) = \text{Tr}(\cdot)(1-\nu^2)^{m'}\nu^{2\hat n} + i\epsilon_0\sqrt{2}\text{Tr}\left[\hat{E}_{d,m}(-\mathbf{q},\mathbf{p})(\cdot)\right](1-\nu^2)^{m'}\nu^{\hat n}(\hat{D}(-\gamma^*)-\hat{D}(\gamma^*))\nu^{\hat n}.
\end{equation}
Through this we define the following superoperators
\begin{equation}
    \begin{aligned}
        &\Lambda_{(\mathbf{q},\mathbf{p}),\gamma}(\cdot)=(\mathcal{E}_{(\mathbf{q},\mathbf{p}),\gamma}\otimes \mathcal{E}^*_{(\mathbf{q},\mathbf{p}),\gamma})(\cdot)\\
        &=\text{Tr}(\cdot)(1-\nu^2)^{2m'}\nu^{2\hat n}\otimes\nu^{2\hat n} + i\epsilon_0\sqrt{2}(1-\nu^2)^{2m'}\text{Tr}\left[(\hat{E}_{d,m}(\mathbf{q},\mathbf{p})\otimes\mathbb{I})(\cdot)\right](\nu^{\hat n}\otimes\nu^{\hat n})((\hat{D}(\gamma) - \hat{D}(-\gamma))\otimes\mathbb{I})(\nu^{\hat n}\otimes\nu^{\hat n})\\
        &\quad\quad + i\epsilon_0\sqrt{2}(1-\nu^2)^{2m'}\text{Tr}\left[(\mathbb{I}\otimes\hat{E}_{d,m}(-\mathbf{q},\mathbf{p}))(\cdot)\right](\nu^{\hat n}\otimes\nu^{\hat n})(\mathbb{I}\otimes (\hat{D}(-\gamma^*) - \hat{D}(\gamma^*)))(\nu^{\hat n}\otimes\nu^{\hat n})\\
        &\quad\quad+i\epsilon_0\sqrt{2}(1-\nu^2)^{2m'}\text{Tr}\left[(\hat{E}_{d,m}(\mathbf{q},\mathbf{p})\otimes\hat{E}_{d,m}(-\mathbf{q},\mathbf{p}))(\cdot)\right](\nu^{\hat n}\otimes\nu^{\hat n})((\hat{D}(\gamma) - \hat{D}(-\gamma))\otimes (\hat{D}(-\gamma^*) - \hat{D}(\gamma^*)))(\nu^{\hat n}\otimes\nu^{\hat n}),
    \end{aligned}
\end{equation}
\begin{equation}
    \Lambda_{0}(\cdot) = \text{Tr}(\cdot)(1-\nu^2)^{2m'}\nu^{2\hat n}\otimes \nu^{2\hat n}.
\end{equation}

\textbf{Succeeding in the hypothesis task by learning $|C_{\mathcal{E}_{(\mathbf{q},\mathbf{p}),\gamma}}|$}: Suppose $\gamma$ is sampled from a Gaussian distribution of spread $\sigma_{\gamma}$ with $2\sigma_{\gamma}^2 = 0.99\kappa$ and $(\mathbf{p},\mathbf{q})$ is uniformly random over $\mathbb{F}^m_d\times\mathbb{F}^m_d\setminus\{(0,0)\}$. This gives the same setting as many-one channel discrimination with revelation (Problem~\ref{prob:manyrevel}) but for telling apart $\Lambda_{(\mathbf{q},\mathbf{p}),\gamma}$ and $\Lambda_0$ which can trivially be told apart by learning $|C_{\mathcal{E}_{(\mathbf{q},\mathbf{p}),\gamma}}|$ to accuracy $\epsilon<0.1225$ as shown in Thm.~\ref{thm:qudit2boson}.

\textbf{Application of Lemma~\ref{lem:master}}:
The channel $\Lambda_{(\mathbf{q},\mathbf{p}),\gamma}$ is a sum of $4$ terms when expressed in the form of Eq.~\eqref{eq:Epar1par2} giving $l=3$. This can alternatively be written as a sum over two bitwise variables $x,y$ as follows
\begin{equation}
        \Lambda_{(\mathbf{q},\mathbf{p}),\gamma}(\cdot) = \sum_{x,y\in\{0,1\}}\text{Tr}\left[(\cdot)\hat{K}_{x,y,(\mathbf{q},\mathbf{p})}\right]\sqrt{\hat{\rho}_0}\hat{W}_{x,y,\gamma}\sqrt{\hat{\rho}_0}
    \end{equation}
where $\hat{\rho}_0 = (1-\nu^2)^{2m'}\nu^{2\hat{n}}\otimes\nu^{2\hat{n}}$, and we have
\begin{align}
    2^{(x+y)/2}\hat{K}_{x,y,(\mathbf{q}_1,\mathbf{p}_1)} &= \begin{cases}
            \mathbb{I}\otimes \mathbb{I} & x,y = 0,0\\
            \mathbb{I}\otimes \hat{E}_{d,m}(-\mathbf{q}_1,\mathbf{p}_1)& x,y = 0,1\\
            \hat{E}_{d,m}(\mathbf{q}_1,\mathbf{p}_1)\otimes \mathbb{I} & x,y = 1,0\\
            \hat{E}_{d,m}(\mathbf{q}_1,\mathbf{p}_1)\otimes \hat{E}_{d,m}(-\mathbf{q}_1,\mathbf{p}_1) & x,y=1,1
        \end{cases},\\
        \frac{1}{(2\epsilon_0)^{x+y}}\hat{W}_{x,y,\gamma} &= \begin{cases}
            \mathbb{I}\otimes\mathbb{I} & x,y = 0,0\\
            \mathbb{I}\otimes (i\hat{D}(-\gamma^*) - i\hat{D}(\gamma^*)) & x,y = 0,1\\
            (i\hat{D}(\gamma) - i\hat{D}(-\gamma))\otimes \mathbb{I} & x,y = 1,0\\
            (i\hat{D}(\gamma) - i\hat{D}(-\gamma))\otimes (i\hat{D}(-\gamma^*) - i\hat{D}(\gamma^*)) & x,y=1,1
        \end{cases},
\end{align}
We equivalently rewrite $\Delta$ from Lemma~\ref{lem:master} as follows
\begin{equation}
    \begin{aligned}
     \Delta^{1/2}\leq &\sum_{\pmb{x},\pmb{y}\in\{0,1\}^c,|\pmb{x}|+|\pmb{y}|>0}\Bigg\{\left\|\mathbb{E}_{(\mathbf{q},\mathbf{p})}[\hat{K}_{\pmb{x},\pmb{y},(\mathbf{q},\mathbf{p})}^{\otimes2}]\right\|_{\mathrm{op}}\times\left\|\mathbb{E}_{\gamma}[\hat{W}_{\pmb{x},\pmb{y},\gamma}^{\otimes 2}]\right\|_{\mathrm{op}}\Bigg\}^{1/2},
    \end{aligned}
    \end{equation}
    where we define
    \begin{equation}
        \hat{K}_{\pmb{x},\pmb{y},(\mathbf{q},\mathbf{p})} = \bigotimes_{i=1}^c \hat{K}_{x_i,y_i,(\mathbf{q},\mathbf{p})},\quad  
        \hat{W}_{\pmb{x},\pmb{y},\gamma} = \bigotimes_{i=1}^c \hat{W}_{x_i,y_i,\gamma}.
    \end{equation}
    As observed in Theorem~\ref{thm:qudit2quditwithconj}
    \begin{equation}
    \begin{aligned}
        &\left\|\mathbb{E}_{(\mathbf{q},\mathbf{p})}[\hat{K}_{\pmb{x},\pmb{y},(\mathbf{q},\mathbf{p})}^{\otimes2}]\right\|_{\mathrm{op}} \leq \begin{cases}
        \frac{1}{d^{m}-1} & |\pmb{x}| \neq |\pmb{y}| \mod d\\
        1 & |\pmb{x}| = |\pmb{y}| \mod d
                \end{cases}.
    \end{aligned}
    \end{equation}
    We now expand the expression for
    \begin{equation}
    \begin{aligned}
        \left\|\mathbb{E}_{\gamma}[\hat{W}_{\pmb{x},\pmb{y},\gamma}^{\otimes 2}]\right\|_{\mathrm{op}}&= (2\sqrt{2}\epsilon_0)^{2|\pmb{x}|+2|\pmb{y}|}\left\|\int d^{2m'}\gamma P(\gamma)\left(\frac{e^{i\theta}\hat{D}(\gamma) + e^{-i\theta}\hat{D}^\dagger(\gamma)}{\sqrt{2}}\right)^{\otimes 2|\pmb{x}|}\otimes\left(\frac{e^{i\theta}\hat{D}(-\gamma^*) + e^{-i\theta}\hat{D}^\dagger(-\gamma^*)}{\sqrt{2}}\right)^{\otimes 2|\pmb{y}|}\right\|_{\mathrm{op}}\\
        &\leq \frac{(4\epsilon_0)^{2|\pmb{x}|+2|\pmb{y}|}}{(1+ 4\sigma_{\gamma}^4(|\pmb{x}|-|\pmb{y}|)^2)^{m'/2}} \leq\begin{cases}
            \frac{(4\epsilon_0)^{2|\pmb{x}|+2|\pmb{y}|}}{(1+4\sigma_{\gamma}^4)^{m'/2}} & |\pmb{x}|\neq |\pmb{y}|\\
            (4\epsilon_0)^{2|\pmb{x}|+2|\pmb{y}|}& |\pmb{x}| = |\pmb{y}|
        \end{cases}.
    \end{aligned}
    \end{equation}
    Plugging these into the expression of $\Delta$ and noting that $|\pmb{x}|=|\pmb{y}|\mod d$ can only occur for $|\pmb{x}|=|\pmb{y}|$ in this case since $c\leq d-1$ we get
    \begin{equation}
        \Delta\leq \left(\sum_{k,l =0,(k,l)\neq (0,0)}^c \binom{c}{k}\binom{c}{l}\frac{(4\epsilon_0)^{k+l}}{(1+4\sigma_{\gamma}^4)^{m'/4}(d^{m}-1)^{1/2}} + \sum_{k=1}^c\binom{c}{k}^2(4\epsilon_0)^{2k}\right)^2.
    \end{equation}
    This follows the same reduction as $\Delta$ in the previous theorem to give
    \begin{equation}
        N = \Omega\left(\min\left(\frac{(d^m-1)(1+(0.99\kappa)^2)^{m'/2}}{c(\epsilon/0.1225)^2},\frac{1}{c^3(\epsilon/0.1225)^4}\right)\right).
    \end{equation}
\end{proof}

\begin{theorem}[Access to the complex-conjugate channel: bosonic $\to$ bosonic]\label{thm:boson2bosonwithconj}
Let $\mathcal{E}\in \CPTP(\mathcal{H}_{\infty,m},\mathcal{H}_{\infty,m'})$ with $m,m'\geq 8$ and $\mathcal{E}^*$ to be its complex-conjugate channel. Taking $\epsilon \leq 0.05$, $c = O(\epsilon^{-1})$, $\min(\kappa^2 m,\kappa'^2m') > 2/0.99$, consider any adaptive ancilla-assisted learning scheme that is allowed uses of the channel $(\mathcal{E}\otimes\mathcal{E}^*)^{\otimes c}$. If the learner can use this scheme to produce an estimate $|\hat{C}^{\mathrm{TMSV},r}_{\mathcal{E}}(\alpha,\beta)|$ (where $\cosh(2r)\geq 1.06\kappa m$) such that for a query of $\alpha\in\mathbb{C}^m,\beta\in\mathbb{C}^{m'}$ with $|\alpha|^2\leq \kappa m,|\beta|^2\leq \kappa' m'$ it satisfies $$\left||\hat{C}^{\mathrm{TMSV},r}_{\mathcal{E}}(\alpha,\beta)| - |C^{\mathrm{TMSV},r}_{\mathcal{E}}(\alpha,\beta)|\right|\leq \epsilon$$ with a 2/3 success probability, then this learning scheme requires at least $N$ uses of $\mathcal{E}\otimes\mathcal{E}^*$ where $$N = \Omega\left(\min\left(\frac{(1+(0.99\kappa')^2)^{m'/2}(1+(0.99\kappa\tanh^2(2r))^2)^{m/2}}{c(\epsilon/0.05)^2},\frac{1}{c^3(\epsilon/0.05)^4}\right)\right).$$ 
\end{theorem}
\begin{proof}
    We consider two channels defined as follows
    \begin{gather}
        \mathcal{E}_0(\cdot) = \text{Tr}(\cdot)\hat{\rho}_0,\quad\mathcal{E}_{\gamma_1,\gamma_2}(\cdot) = \text{Tr}\left[\hat{\Pi}_{+,\gamma_1}(\cdot)\right]\hat{\rho}_{\gamma_2} + \text{Tr}\left[\hat{\Pi}_{-,\gamma_1}(\cdot)\right]\hat{\rho}_{-\gamma_2},\\
        \hat{\Pi}_{\pm,\gamma_1} = \frac{\mathbb{I} \pm \hat{E}(\gamma_1)/\sqrt{2}}{2},\quad\hat{\rho}_{\gamma_2} = (1-\nu^2)^{m'}\nu^{\hat n}(\hat{D}(0) + 2i\epsilon_0(\hat{D}(\gamma_2) - \hat{D}^\dagger(\gamma_2)))\nu^{\hat n},
    \end{gather}
    from which we obtain that
    \begin{equation}
        \mathcal{E}^{*}_{\gamma_1,\gamma_2}(\cdot) = \text{Tr}\left(\cdot\right)\hat{\rho}_0 + i\epsilon_0\sqrt{2}\text{Tr}\left[\hat{E}(-\gamma_1^*)(\cdot)\right](1-\nu^2)^{m'}\nu^{\hat n}(\hat
        D(-\gamma_2^*) - \hat{D}(\gamma_2^*))\nu^{\hat n}.
    \end{equation}
    Using this we define the following channels
    \begin{equation}
        \Lambda_{0}(\cdot) = (\mathcal{E}_{0}\otimes\mathcal{E}_0^*)(\cdot) = \text{Tr}(\cdot)(1-\nu^2)^{2m'}\nu^{2\hat n}\otimes \nu^{2\hat n},
    \end{equation}
    \begin{equation}
        \begin{aligned}
            &\Lambda_{\gamma_1,\gamma_2}(\cdot) = (\mathcal{E}_{\gamma_1,\gamma_2}\otimes\mathcal{E}_{\gamma_1,\gamma_2}^*)(\cdot)\\
            &=\text{Tr}(\cdot)(1-\nu^2)^{2m'}\nu^{2\hat n}\otimes\nu^{2\hat n} + i\epsilon_0\sqrt{2}(1-\nu^2)^{2m'}\text{Tr}\left[(\hat{E}(\gamma_1)\otimes\mathbb{I})(\cdot)\right](\nu^{\hat n}\otimes\nu^{\hat n})((\hat{D}(\gamma_2) - \hat{D}(-\gamma_2))\otimes\mathbb{I})(\nu^{\hat n}\otimes\nu^{\hat n})\\
        &\quad\quad + i\epsilon_0\sqrt{2}(1-\nu^2)^{2m'}\text{Tr}\left[(\mathbb{I}\otimes\hat{E}(-\gamma_1^*))(\cdot)\right](\nu^{\hat n}\otimes\nu^{\hat n})(\mathbb{I}\otimes (\hat{D}(-\gamma_2^*) - \hat{D}(\gamma_2^*)))(\nu^{\hat n}\otimes\nu^{\hat n})\\
        &\quad\quad+i\epsilon_0\sqrt{2}(1-\nu^2)^{2m'}\text{Tr}\left[(\hat{E}(\gamma_1)\otimes\hat{E}(-\gamma_1^*))(\cdot)\right](\nu^{\hat n}\otimes\nu^{\hat n})((\hat{D}(\gamma_2) - \hat{D}(-\gamma_2))\otimes (\hat{D}(-\gamma_2^*) - \hat{D}(\gamma_2^*)))(\nu^{\hat n}\otimes\nu^{\hat n}).
        \end{aligned}
    \end{equation}
\textbf{Success in Problem~\ref{prob:manyrevel} through learning $|C_{\mathcal{E}}^{\mathrm{TMSV},r}(\alpha,\beta)|$}: Taking $\gamma_1\in \mathbb{C}^m$ and $\gamma_2\in \mathbb{C}^{m'}$ to be sampled from Gaussian distributions of spread $\sigma_1$ and $\sigma_2$ respectively where $\sigma_1^2 = 0.99\kappa\tanh^2(2r)/2$ and $\sigma_2^2 = 0.99\kappa'/2$ we get a setting for many-one channel discrimination with revelation (Problem~\ref{prob:manyrevel}) for telling apart $\Lambda_0$ and $\Lambda_{\gamma_1,\gamma_2}$. Note that this task is solved by learning $|C_{\mathcal{E}}^{\mathrm{TMSV},r}(\alpha,\beta)|$ as described in Thm.~\ref{thm:boson2boson}. We similarly define the accuracy parameter $\epsilon = \epsilon_0\frac{1}{2\sqrt{2e}}(1-e^{-4})^2e^{-L}$ ensuring that $4\epsilon_0 < \epsilon/0.05$ by choosing an appropriate value of $L$.

\textbf{Application of Lemma~\ref{lem:master}}:
Within the framing of Problem~\ref{prob:manyrevel}, the channel $\Lambda_{\gamma_1,\gamma_2}$ is a sum of 4 terms giving $l =3$ when written in the form of Eq.~\eqref{eq:Epar1par2}. We can also equivalently break it into two bits $x$ and $y$ in the summation giving 
    \begin{equation}
        \Lambda_{\gamma_1,\gamma_2}(\cdot) = \sum_{x,y\in\{0,1\}}\text{Tr}\left[(\cdot)\hat{K}_{x,y,\gamma_1}\right]\sqrt{\hat{\rho}_0}\hat{W}_{x,y,\gamma_2}\sqrt{\hat{\rho}_0},
    \end{equation}
    where $\hat{\rho}_0 = (1-\nu^2)^{2m'}\nu^{2\hat{n}}\otimes\nu^{2\hat{n}}$ and we define

    \begin{align}
    2^{(x+y)/2}\hat{K}_{x,y,\gamma_1} &= \begin{cases}
            \mathbb{I}\otimes \mathbb{I} & x,y = 0,0\\
            \mathbb{I}\otimes \hat{E}(-\gamma_1^*)& x,y = 0,1\\
            \hat{E}(\gamma_1)\otimes \mathbb{I} & x,y = 1,0\\
            \hat{E}(\gamma_1)\otimes \hat{E}(-\gamma_1^*) & x,y=1,1
        \end{cases},\\
        \frac{1}{(2\epsilon_0)^{x+y}}\hat{W}_{x,y,\gamma_2} &= \begin{cases}
            \mathbb{I}\otimes\mathbb{I} & x,y = 0,0\\
            \mathbb{I}\otimes (i\hat{D}(-\gamma_2^*) - i\hat{D}(\gamma_2^*)) & x,y = 0,1\\
            (i\hat{D}(\gamma_2) - i\hat{D}(-\gamma_2))\otimes \mathbb{I} & x,y = 1,0\\
            (i\hat{D}(\gamma_2) - i\hat{D}(-\gamma_2))\otimes (i\hat{D}(-\gamma_2^*) - i\hat{D}(\gamma_2^*)) & x,y=1,1
        \end{cases}.
    \end{align}
    We equivalently rewrite $\Delta$ from Lemma~\ref{lem:master} as follows
\begin{equation}
    \begin{aligned}
         \Delta^{1/2}\leq &\sum_{\pmb{x},\pmb{y}\in\{0,1\}^c,|\pmb{x}|+|\pmb{y}|>0}\Bigg\{\left\|\mathbb{E}_{\gamma_1}[\hat{K}_{\pmb{x},\pmb{y},\gamma_1}^{\otimes2}]\right\|_{\mathrm{op}}\times\left\|\mathbb{E}_{\gamma_2}[\hat{W}_{\pmb{x},\pmb{y},\gamma_2}^{\otimes 2}]\right\|_{\mathrm{op}}\Bigg\}^{1/2},
    \end{aligned}
    \end{equation}
    where we define
    \begin{equation}
        \hat{K}_{\pmb{x},\pmb{y},\gamma_1} = \bigotimes_{i=1}^c \hat{K}_{x_i,y_i,\gamma_1},\quad \hat{W}_{\pmb{x},\pmb{y},\gamma_2} = \bigotimes_{i=1}^c \hat{W}_{x_i,y_i,\gamma_2}.
    \end{equation}
    Observe that by expanding the expression of $\hat{K}_{\pmb{x},\pmb{y},\gamma_1}$ we get
    \begin{equation}
    \begin{aligned}
        \left\|\mathbb{E}_{\gamma_1}[\hat{K}_{\pmb{x},\pmb{y},\gamma_1}^{\otimes2}]\right\|_{\mathrm{op}} &= \frac{1}{2^{|\pmb{x}|+|\pmb{y}|}}\left\|\mathbb{E}_{\gamma_1}\left[\hat{E}(\gamma_1)^{\otimes2|\pmb{x}|}\otimes\hat{E}(-\gamma_1^*)^{\otimes 2|\pmb{y}|}\right]\right\|_{\mathrm{op}}\\
        &\leq \frac{1}{(1+4\sigma_1^4)^{m/2}}.
    \end{aligned}
    \end{equation}
    As observed in Theorem~\ref{thm:qudit2bosonwithconj} 
    \begin{equation}
    \begin{aligned}
        &\left\|\mathbb{E}_{\gamma_2}[\hat{W}_{\pmb{x},\pmb{y},\gamma_2}^{\otimes 2}]\right\|_{\mathrm{op}}
        \leq \frac{(4\epsilon_0)^{2|\pmb{x}|+2|\pmb{y}|}}{(1+ 4\sigma_{2}^4(|\pmb{x}|-|\pmb{y}|)^2)^{m'/2}} \leq\begin{cases}
            \frac{(4\epsilon_0)^{2|\pmb{x}|+2|\pmb{y}|}}{(1+4\sigma_2^4)^{m'/2}} & |\pmb{x}|\neq |\pmb{y}|\\
            (4\epsilon_0)^{2|\pmb{x}|+2|\pmb{y}|}& |\pmb{x}| = |\pmb{y}|
        \end{cases}.
    \end{aligned}
    \end{equation}
    Plugging these into the expression of $\Delta$ we get
    \begin{equation}
        \Delta\leq \left(\sum_{k,l =0,(k,l)\neq (0,0)}^c \binom{c}{k}\binom{c}{l}\frac{(4\epsilon_0)^{k+l}}{(1+4\sigma_{1}^4)^{m/4}(1+4\sigma_{2}^4)^{m'/4}} + \sum_{k=1}^c\binom{c}{k}^2(4\epsilon_0)^{2k}\right)^2.
    \end{equation}
    This follows the same reduction as $\Delta$ in the previous theorem to give
    \begin{equation}
        N = \Omega\left(\min\left(\frac{(1+(0.99\kappa')^2)^{m'/2}(1+(0.99\kappa\tanh^2(2r))^2)^{m/2}}{c(\epsilon/0.05)^2},\frac{1}{c^3(\epsilon/0.05)^4}\right)\right).
    \end{equation}
\end{proof}

\subsection{Lower bounds for state learning}\label{app:statelearn}
We will make use of Corollary~\ref{corr:statelearn} to derive sharper bounds for the state learning problem as has been studied in \cite{coroi2025exponentialadvantagecontinuousvariablequantum,PRXQuantum.5.040301,ller2025infinitehierarchymulticopyquantum}.
\begin{theorem}\label{thm:qudit_statelearn}
    Consider a state $\hat{\rho}\in D(\mathcal{H}_{d,m})$ where $d$ is a prime number. Consider any adaptive learning scheme that is allowed to perform arbitrary measurements on the state $\hat{\rho}^{\otimes c}$ where $c\leq \frac{1}{4e\epsilon}$ (here $\epsilon\leq \frac{1}{4e}$). If the learner can use this scheme to produce an estimate of $|\tilde{\chi}_{\hat{\rho}}(\mathbf{q},\mathbf{p})|$ for a query of $\mathbf{q},\mathbf{p}\in \mathbb{F}^m_d$ such that $\left||\tilde{\chi}_{\hat{\rho}}(\mathbf{q},\mathbf{p})| - |\chi_{\hat{\rho}}(\mathbf{q},\mathbf{p})|\right| \leq \epsilon$ with a $2/3$ success probability, then this learning scheme requires at least $N$ uses of $\hat{\rho}$ where
    \begin{itemize}
        \item If $c\leq d-1$
        $$N = \Omega(d^m c^{-1}\epsilon^{-2}).$$
        \item If $c\geq d$
        $$N = \Omega\left(c\min\left[\frac{d^m-1}{16c^2\epsilon^2e^{2/e}},\left(\frac{d}{4ec\epsilon}\right)^{2d}\left(\frac{3}{4}\right)^2\right]\right).$$
    \end{itemize}
\end{theorem}
\begin{proof}
Consider the following two states
\begin{equation}
    \hat{\rho}_0 = \frac{\mathbb{I}}{d^m}, \quad\hat{\rho}_{(\mathbf{q},\mathbf{p})} = \frac{1}{d^m}\left(\mathbb{I}+ \epsilon_0\hat{E}_{d,m}(\mathbf{q},\mathbf{p})\right).
\end{equation}
Consider a state discrimination task where $N$ copies of either of the above states are prepared after party $A$ samples $(\mathbf{q},\mathbf{p})$ uniformly at random from $\mathbb{F}^{m}_d\times \mathbb{F}^m_d\setminus\{(0,0)\}$. Framing the above state discrimination as a reduction of Problem~\ref{prob:manyrevel} as per Corollary~\ref{corr:statelearn}, we will be able to lower bound the sample complexity by showing that learning $|\chi_{\hat{\rho}}(\mathbf{q},\mathbf{p})|$ suffices to succeed at this problem.

\textbf{Succeeding in revealed discrimination by learning $|\chi_{\hat{\rho}}((\mathbf{q},\mathbf{p}))|$}: note that for the two states we have
\begin{equation}
    |\chi_{\hat{\rho}_0}((\mathbf{q},\mathbf{p}))| = 0,\quad |\chi_{\hat{\rho}_{(\mathbf{q},\mathbf{p})}}((\mathbf{q},\mathbf{p}))|=\frac{\epsilon_0}{\sqrt{2}},
\end{equation}
and so if the learner can estimate $|\chi_{\hat{\rho}}((\mathbf{q},\mathbf{p}))|$ to $\epsilon$ accuracy where $\epsilon < \epsilon_0/\sqrt{8}$, they would succeed in the discrimination task. The largest value $\epsilon_0$ can take while still keeping the state physical is $1/\sqrt{2}$, hence we can always frame such a problem for any choice of $\epsilon \leq \frac{1}{4}$ (this holds true since we choose $\epsilon\leq \frac{1}{4e}$).

\textbf{Application of Corollary~\ref{corr:statelearn}}: Note that here the state can be expressed as $\hat{\rho}_{(\mathbf{q},\mathbf{p})} = \sum_{x\in\{0,1\}}\sqrt{\hat{\rho}_0}\hat{W}_{x,(\mathbf{q},\mathbf{p})}\sqrt{\hat{\rho}_0}$. We will make use of the quantity $\Delta$ defined in Corollary~\ref{corr:statelearn}
\begin{equation}
    \Delta \leq \left(\sum_{k=1}^c\binom{c}{k}\epsilon_0^{k}\sqrt{\left\|\mathbb{E}_{(\mathbf{q},\mathbf{p})}\left[\hat{E}_{d,m}(\mathbf{q},\mathbf{p})^{\otimes 2k}\right]\right\|_{\mathrm{op}}}\right)^2.
\end{equation}
\begin{itemize}
    \item Case 1: $c\leq d-1$, this means that $k$ is never $0\mod d$ and so we get
    \begin{equation}
        \Delta \leq \left(\frac{(1+\sqrt{2}\epsilon_0)^c -1}{(d^{m}-1)^{1/2}}\right)^2\leq \frac{2c^2\epsilon_0^2 e^{2\sqrt{2}c\epsilon_0}}{d^m-1} \implies N = \Omega(d^mc^{-1}\epsilon^{-2}).
    \end{equation}
    \item Case 2: $c\geq d$. In this case, we would have to include the case that $k = 0\mod d$ in the summation. This gives us
 \begin{equation}
    \begin{aligned}
        \Delta &\leq \left(\sum_{k=1,k \neq 0\mod d}^c\binom{c}{k}\frac{(\epsilon_0\sqrt{2})^k}{(d^{m}-1)^{1/2}} + \sum_{k=1}^{\lfloor c/d\rfloor}\binom{c}{kd}(\epsilon_0\sqrt{2})^{kd}\right)^2\\
        &\leq \left(\frac{4c\epsilon e^{1/e}}{(d^{m}-1)^{1/2}} + \frac{4}{3}\left(\frac{4ec\epsilon}{d}\right)^{d}\right)^2,
    \end{aligned}
    \end{equation}
    this gives the bound
    \begin{equation}
        N = \Omega\left(c\min\left[\frac{d^m-1}{16c^2\epsilon^2e^{2/e}},\left(\frac{d}{4ec\epsilon}\right)^{2d}\left(\frac{3}{4}\right)^2\right]\right).
    \end{equation}
\end{itemize}
\end{proof}

The above theorem shows the tightness for the $\epsilon^{-2d}$ scaling for the $d$-copy state learning scheme as found in \cite{ller2025infinitehierarchymulticopyquantum} while also offering slight improvements to the lower bound in the same work in the case of $c\leq d-1$. Since the above lower bound is derived for the absolute value estimation task, it naturally holds true for the estimation with phase information. The estimation with phase information as discussed in \cite{ller2025infinitehierarchymulticopyquantum} can be done by creating a hypothesis state $\hat{\sigma}\in D(\mathcal{H}_{d,m})$ using the absolute values of the displacement observables and then performing measurements over the state $\hat{\sigma}^{\otimes(d-1)}\otimes \hat{\rho}$.

\begin{theorem}\label{thm:boson_statelearn}
    Consider an $m$-mode bosonic state $\hat{\rho}\in D(\mathcal{H}_{\infty,m})$. Consider any adaptive learning scheme that is allowed to perform arbitrary measurements on the state $\hat{\rho}^{\otimes c}$ where $c = O(1/\epsilon)$ ($\epsilon<0.1725$). If the learner can use this scheme to produce an estimate of $|\tilde{\chi}_{\hat{\rho}}(\alpha)|$ for a query $\alpha\in\mathbb{C}^m$ that satisfies $|\alpha|^2\leq \kappa m$ ($\kappa>0$) such that $\left||\tilde{\chi}_{\hat{\rho}}(\alpha)| - |\chi_{\hat{\rho}}(\alpha)|\right| \leq \epsilon$ with a $2/3$ success probability, then this learning scheme requires at least $N = \Omega((1+0.99\kappa)^mc^{-1}\epsilon^{-2})$ uses of $\hat{\rho}$.
\end{theorem}
\begin{proof}
    Consider the following two states
    \begin{equation}
        \hat{\rho}_0 = (1-\nu^2)^m\nu^{2\hat{n}},\quad \hat{\rho}_\gamma= (1-\nu^2)^m\nu^{\hat{n}}(\hat{D}(0) + 2i\epsilon_0(\hat{D}(\gamma) - \hat{D}(-\gamma)))\nu^{\hat{n}}.
    \end{equation}
    Consider a state discrimination task where $N$ copies of either of the above states are prepared after party $A$ samples $\gamma\in \mathbb{C}^m$ as a Gaussian random variable with spread $\sigma_{\gamma}^2$. Framing the above state discrimination as a reduction of Problem~\ref{prob:manyrevel} as per Corollary~\ref{corr:statelearn}, we will be able to lower bound the sample complexity by showing that learning $|\chi_{\hat{\rho}}(\alpha)|$ for queries with $|\alpha|^2\leq \kappa m$ suffices to succeed at this problem.
    
    \textbf{Succeeding in revealed discrimination by learning $|\chi_{\hat{\rho}}(\alpha)|$}: 
    Defining the parameters
    \begin{equation}
        \Sigma^2 = \frac{1+\nu}{1-\nu} = \frac{\kappa m}{L},\quad 2\sigma^2 = \frac{1}{\nu}- \nu =\frac{4\kappa m L}{(\kappa m)^2 - L^2},
    \end{equation}
    where we choose constant $L>0$, we get
    \begin{equation}
        \chi_{\hat{\rho}_0}(\gamma) = e^{-\frac{|\gamma|^2}{2\Sigma^2}}e^{-\frac{|\gamma|^2}{2\sigma^2}}= t_0,\quad \chi_{\hat{\rho}_{\gamma}}(\gamma) = 2i\epsilon_0e^{-\frac{|\gamma|^2}{\Sigma^2}}(1 - e^{-2|\gamma|^2/\sigma^2})+e^{-\frac{|\gamma|^2}{2\Sigma^2}}e^{-\frac{|\gamma|^2}{2\sigma^2}} = t_0 + i\epsilon_0t_1.
    \end{equation}
    To enable discrimination through estimation of the absolute values to $\epsilon$ accuracy, we would like for the absolute values of $|\chi_{\hat{\rho}_0}(\gamma)|$ and $|\chi_{\hat{\rho}_\gamma}(\gamma)|$ to differ by $2\epsilon$ at least. For now we assume that the value of $|\gamma|^2$ satisfies
    \begin{equation}\label{eq:b20succcondition}
        1\leq |\gamma|^2\leq \kappa m.
    \end{equation}
    Observe that 
    \begin{equation}
        \frac{t_0}{\epsilon_0t_1} \leq \frac{e^{\frac{|\gamma|^2}{2\Sigma^2}}}{2\epsilon_0}\frac{e^{-\frac{|\gamma|^2}{2\sigma^2}}}{(1-e^{-2|\gamma|^2/\sigma^2})} \leq \frac{e^\frac{L}{2}}{2\epsilon_0}\frac{e^{-\frac{1}{2\sigma^2}}}{(1-e^{-2/\sigma^2})}.
    \end{equation}
    By assuming that $\kappa m\geq 2$ (and assuming $L<1$) which gives $2\sigma^2\geq \frac{8L}{3}$ and noting that the function $e^{-x}/(1-e^{-4x})$ is monotonically decreasing in $x$, we further get
    \begin{equation}
        \frac{t_0}{\epsilon_0t_1} \leq\frac{e^L}{2\epsilon_0}\frac{e^{-\frac{3}{8L}}}{(1-e^{-\frac{3}{2L}})} \leq \frac{e}{2(1-e^{-3/2})}\frac{e^{-\frac{3}{8L}}}{\epsilon_0},\quad t_1\geq 2e^{-L}(1 - e^{-\frac{3}{2L}}),
    \end{equation}
    where we have made use of $L<1$. We can now further choose a value of $L$ small enough that 
    \begin{equation}
        L<\frac{3}{8\log(\frac{2}{e\epsilon_0(1-e^{-3/2})})} \implies \frac{t_0}{\epsilon_0t_1}\leq 0.25,
    \end{equation}
    \begin{equation}
        \left||\chi_{\hat{\rho}_0}(\gamma)| - |\chi_{\hat{\rho}_\gamma}(\gamma)|\right| =\epsilon_0t_1\left(\sqrt{1+\left(\frac{t_0}{\epsilon_0 t_1}\right)^2} - \frac{t_0}{\epsilon_0t_1}\right)>1.39\epsilon_0.
    \end{equation}
    Hence if we can learn to accuracy of $\epsilon = 0.69\epsilon_0 \leq 0.1725$, then $B$ succeeds in the discrimination task. By choosing $2\sigma_\gamma^2 = 0.99\kappa $, the success probability can be guaranteed to be well above $1/2$ since the condition in Eq.~\eqref{eq:b20succcondition} is a subset of conditions from those in Eq.~\eqref{eq:b10succcondition} used in Thm.~\ref{thm:boson2boson}.
    
    \textbf{Application of Corollary~\ref{corr:statelearn}}: We define the following operators for $\pmb{x}\in\{0,1\}^c$
    \begin{equation}
        \hat{\omega}_{\pmb{x}} = \bigotimes_{i=1}^c \hat{\omega}_{x_i}, \quad \hat{\omega}_0 = \hat{\rho}_0,\quad \hat{\omega}_1 =\hat{\rho}_{\gamma} - \hat{\rho}_0.
    \end{equation}
    From \cite{coroi2025exponentialadvantagecontinuousvariablequantum}, following equations (S198) to (S216) yields the upper bound for all positive bounded operators $\hat{\Pi}$,
    \begin{equation}
        \frac{\mathbb{E}_{\gamma}\left[\text{Tr}\left[(\hat{\omega}_{\pmb x}\otimes\hat{\omega}_{\pmb x})(\hat{\Pi}\otimes\hat{\Pi})\right]\right]}{\text{Tr}\left[\hat{\rho}_0^{\otimes c}\hat{\Pi}\right]^2} \leq (4\epsilon_0)^{2|\pmb{x}|}\left(\frac{1+\nu}{1+\nu+2\nu\sigma^2_{\gamma,|\pmb{x}|}}\right)^m,\quad \sigma^2_{\gamma,|\pmb{x}|} =\left(\frac{1}{\Sigma^2} + \frac{1}{2|\pmb{x}|\sigma_{\gamma}^2}\right)^{-1}.
    \end{equation}
    Using the above bound in combination with Corollary~\ref{corr:statelearn} for the parameter $\Delta$, we get
    \begin{equation}
    \begin{aligned}
         \Delta &\leq \left(\sum_{k=1}^c\binom{c}{k}(4\epsilon_0)^{k}\left(\frac{1+\nu}{1+\nu+2\nu\sigma^2_{\gamma,k}}\right)^{m/2}\right)^2\\
        &\leq \left((1+4\epsilon_0)^c-1\right)^2\left(\frac{1+\nu}{1+\nu+2\nu\sigma^2_{\gamma,1}}\right)^{m}\\
        &\leq \frac{16c^2\epsilon_0^2e^{8c\epsilon_0}}{(1+0.99\kappa)^{m}}.
    \end{aligned}
    \end{equation}
    Since our substitutions for $\Sigma^2$ and $2\sigma^2_{\gamma}$ are the same as those in \cite{coroi2025exponentialadvantagecontinuousvariablequantum}, we directly use their simplification to obtain the final bound. Since we have assumed that $c = O(1/\epsilon)$ and $\epsilon = 0.69\epsilon_0$, we have
    \begin{equation}
        N = \Omega((1+0.99\kappa)^mc^{-1}\epsilon^{-2}).
    \end{equation}
    
\end{proof}
The above bound improves upon the bound obtained in Thm. 2 of \cite{coroi2025exponentialadvantagecontinuousvariablequantum} and additionally considers the simpler task of only estimating the absolute value. We can similarly extend the above methods to show the tightness of $\epsilon^{-4}$ scaling when using $\hat{\rho}\otimes \hat{\rho}^*$ which we exclude since we have already demonstrated this tightness in the channel learning scenario.

\subsection{Sufficient conditions for efficient learning protocols}\label{app:suff_cond_effic}
\begin{definition}[Complex-conjugate channel]
    For a given $n$-mode CPTP map $\mathcal{E}$, we define its complex-conjugate channel $\mathcal{E}^*$ to be the CPTP map that satisfies $(\mathcal{E}^*(\hat{\rho}^T))^T = \mathcal{E}(\hat{\rho})$ for all $\hat{\rho}\in L_1(\mathcal{H})$, for a fixed choice of basis for the transpose.
\end{definition}
It can be checked that if $\mathcal{E}$ is a well-defined CPTP channel, so is $\mathcal{E}^*$ since it satisfies the same trace-preserving constraint and has a phase-space reflection which does not affect positivity of the output states.

\begin{theorem}\label{thm:conjeff}
    There exists a non-adaptive ancilla-assisted learning scheme with 1-copy access to $\mathcal{E}\otimes\mathcal{E}^*$ that can estimate $M$ unique displacement observables on the state $(\mathcal{E}\otimes\mathbb{I})(\hat{\sigma})$ with accuracy $\epsilon$ (up to a minus sign) and success probability $2/3$, using $N = \mathcal{O}(\epsilon^{-4}\log(M))$ channel copies.
\end{theorem}
\begin{proof}
    Using $\mathcal{E}\otimes\mathcal{E}^*$, we can prepare the state $(\mathcal{E}\otimes\mathbb{I})(\hat{\sigma})\otimes(\mathcal{E}^*\otimes\mathbb{I})(\hat{\sigma}^T)$. Note that the identity channel is self-conjugate so this prepares the state $\hat{\sigma}_{\mathrm{choi}}\otimes\hat
    \sigma_{\mathrm{choi}}^T$ which can have the displacement observables estimated up to a sign for all possible dimensions as shown in \cite{PRXQuantum.5.040301,coroi2025exponentialadvantagecontinuousvariablequantum}.
\end{proof}
As can be seen from Thm.~\ref{thm:conjeff}, this implies that the $1/\epsilon^4$ scaling is tight for this particular task assuming access to $\mathcal{E}\otimes\mathcal{E}^*$.

\begin{theorem}\label{thm:quditeff}
For a channel $\mathcal{E}$ with input Hilbert space $\mathcal{H}_{d,m}$ and output Hilbert space $\mathcal{H}_{d,m'}$, there exists a learning scheme which successfully estimates all possible values $|C_{\mathcal{E}}((\mathbf{q},\mathbf{p}),(\mathbf{q}',\mathbf{p}'))|$ to accuracy $\epsilon$ with success probability $1-\delta$ for $(\mathbf{q},\mathbf{p})\in\mathbb{F}_{d}^{m}\times \mathbb{F}_{d}^{m}$ and $(\mathbf{q}',\mathbf{p}')\in\mathbb{F}_{d}^{m'}\times \mathbb{F}_{d}^{m'}$ using $O(d(m'+m)\log(d/\delta)\epsilon^{-2d})$ copies of $\mathcal{E}$.
\end{theorem}
\begin{proof}
    This can be done by learning the Choi-state of $\mathcal{E}$ as $\hat{\rho} = (\mathcal{E}\otimes\mathbb{I})(|\Phi_d\rangle\langle\Phi_d|^{\otimes m})$ and using Lemma~\ref{lem:dcopylearning} with accuracy of $\epsilon' = \epsilon^{d}$ since $|a-b|\leq |a^{d}-b^{d}|^{1/d}$ for $a,b>0$.
\end{proof}
Note that this means that for the setting of $c = d$ in Thm.~\ref{thm:qudit2qudit}, the lower bound is tight with respect to $\epsilon$ scaling.

\subsection{Extension to qudits of square-free dimension}\label{app:squarefree}

The infinite learning hierarchy of quantum states where there exists a learning task that is hard with $(d-1)$-copy access but becomes easy with $d$-copy access was shown for not just prime values of $d$ but also for any square-free $d$ in \cite{ller2025infinitehierarchymulticopyquantum}. While the lower bound remains exponentially large for $(d-1)$-copy learning tasks, it does not scale as $d^m$ any more but rather can only be guaranteed to scale as $d_{\mathrm{min}}^m$ where $d_{\mathrm{min}}$ is the smallest prime number that divides $d$. Consider that the number $d = \prod_{i=1}^{f} d_i$ where the $d_i$ are distinct prime numbers. Considering $m$ qudits of $d$ levels each, we can equivalently consider each $d$-level qudit to be composed of $f$ qudits with each qudit being of local dimension $d_i$. We can generalize this principle for an arbitrary $z$-level qudit where the number of copies required for efficient learning is now $d$ where $d$ is the product of all prime numbers dividing $z$. Hence assuming $z = \prod_{i=1}^fd_i^{r_i}$, we are concerned with a Hilbert space given by $\bigotimes_{i=1}^{f}\mathcal{H}_{d_i,mr_i}$. In this scenario, the learning of a state $\hat{\rho}\in D(\bigotimes_{i=1}^{f}\mathcal{H}_{d_i,mr_i})$ is defined as the estimation of the expectation values of operators $\bigotimes_{i=1}^{f}\hat{D}_{d_i,mr_i}(\mathbf{q}_i,\mathbf{p}_i)$. Observe that we have the following hold for these operators in terms of operator norm.

\begin{lemma}\label{lem:disp_sqfree}
    Consider a number with prime factorization $z =\prod_{i=1}^fd_i^{r_i}$ which can be used to define square-free $d = \prod_{i=1}^f d_i$ for unique prime numbers $d_i$ from $i=1$ to $f$. Let $d_{\mathrm{min}}$ be a prime number such that $d_{\min}^{r_{\mathrm{min}}}$is the smallest number that divides $z$ but is coprime to $z/d_{\min}^{r_{\mathrm{min}}}$. For any positive integer $k$,
    \begin{equation}
        \left\|\bigotimes_{i=1}^{f}\left(\sum_{\mathbf{q}_i,\mathbf{p}_i\in\mathbb{F}^{mr_i}_{d_i}}\hat{D}_{d_i,mr_i}(\mathbf{q}_i,\mathbf{p}_i)^{\otimes 2k}\right)\right\|_\mathrm{op} \leq \begin{cases}
            z^{2m}/d_{\mathrm{min}}^{r_{\mathrm{min}}m},& k\neq 0 \mod d\\
            z^{2m}, & k =0\mod d
        \end{cases}.
    \end{equation}
\end{lemma}
\begin{proof}
Due to the tensor product nature of the operator in question, we use the fact that $\|A\otimes B\|_{\mathrm{op}}\leq \|A\|_{\mathrm{op}}\|B\|_{\mathrm{op}}$. This gives us
\begin{equation}
    \left\|\bigotimes_{i=1}^{f}\left(\sum_{\mathbf{q}_i,\mathbf{p}_i\in\mathbb{F}^{mr_i}_{d_i}}\hat{D}_{d_i,mr_i}(\mathbf{q}_i,\mathbf{p}_i)^{\otimes 2k}\right)\right\|_\mathrm{op}\leq \prod_{i=1}^{f}\left\|\sum_{\mathbf{q}_i,\mathbf{p}_i\in\mathbb{F}^{mr_i}_{d_i}}\hat{D}_{d_i,mr_i}(\mathbf{q}_i,\mathbf{p}_i)^{\otimes 2k}\right\|_{\mathrm{op}}.
\end{equation}
Unless $k =0\mod d$, the exponent $k$ cannot satisfy $k = 0\mod d_i$ simultaneously for all $i = 1$ to $f$. Applying this with Lemma~\ref{lem:discr_disp_tensor} we get that for any $k\neq 0\mod d$ at least one of the factors $d_i$ has $k\neq 0\mod d_i$ and so the largest possible value of the RHS in the above inequality is $z^{2m}/d_{\mathrm{min}}^{r_{\mathrm{min}}m}$ and if $k = 0\mod d$ we obtain the RHS to equal $z^{2m}$.
\end{proof}

We can generalize the POVM given by $\hat{\Pi}_{\pm,\theta,(\mathbf{q},\mathbf{p})}^{(d,m)}$ in Eq.~\eqref{eq:qudit2outcome} for the Hilbert space of $\bigotimes_{i=1}^{f}\mathcal{H}_{d_i,m}$ to be given by
\begin{equation}
    \hat{\Pi}_{\pm,\theta}^{(d,m)} = \frac{\mathbb{I}\pm \frac{e^{-i\theta}\bigotimes_{i=1}^f\hat{D}_{d_i,m}(\mathbf{q}_i,\mathbf{p}_i) + e^{i\theta}\bigotimes_{i=1}^f\hat{D}_{d_i,m}(\mathbf{q}_i,\mathbf{p}_i)^\dagger}{2}}{2},
\end{equation}
which would have a very similar implementation as that depicted in Fig.~\ref{fig:hardtolearn}(a) with the difference being the operator $(\bigotimes_{i=1}^f\hat{D}_{d_i,m}(\mathbf{q}_i,\mathbf{p}_i))^{\pm 1/2}$ as the controlled displacement operation conditioned on whether the ancillary qubit is in $\ket{+}$ or $\ket{-}$. Using this, we now establish that for any square-free number $d$, there is a $(d-1)$-copy hard channel learning task.

\begin{theorem}
    Consider two composite numbers with prime factorizations $z = \prod_{i=1}^{f}d_i^{r_i}$ and $z' = \prod_{i=1}^{f}d_{i}'^{r'_i}$ which yield the square-free numbers with prime factorizations $d = \prod_{i=1}^f d_i$, $d' = \prod_{i=1}^{f'}d_i'$ with $d_i,d_i'$ each being prime. We define $d_{\min} = \min_{1\leq i\leq f}(d_i)$ and $d'_{\min} = \min_{1\leq i\leq f'}(d_i')$, and $r_{\min}$ as the largest number for which $d_{\min}^{r_{\min}}$ divides $z$ and $r'_{\min}$ as the largest number for which $d'^{r'_{\min}}_{\min}$ divides $z'$. Let $\mathcal{E}\in \CPTP(\bigotimes_{i=1}^f\mathcal{H}_{d_i,mr_i},\bigotimes_{i=1}^{f'}\mathcal{H}_{d_i',m'r'_i})$. Consider any $c$-copy ($c\leq \frac{1}{\epsilon}$, $\epsilon<1$) learning protocol capable of producing an estimate $\hat{C}$ to $\left|\text{Tr}\left[\bigotimes_{i=1}^{f'}\hat{D}_{d'_i,m'r'_i}(\mathbf{q}'_i,\mathbf{p}'_i)\mathcal{E}\left(\bigotimes_{i=1}^{f}\hat{D}_{d_i,mr_i}(\mathbf{q}_i,\mathbf{p}_i)\right)\right]\right|$ such that for a query of $(\mathbf{q}_i,\mathbf{p}_i)\in \mathbb{F}_{d_i}^{mr_i}\times\mathbb{F}_{d_i}^{mr_i}$ (for $i=1$ to $f$), $(\mathbf{q}_i',\mathbf{p}_i')\in \mathbb{F}_{d_i'}^{m'r'_i}\times\mathbb{F}_{d_i'}^{m'r'_i}$ (for $i = 1$ to $f'$) we have $\left|\hat{C} - \left|\text{Tr}\left[\bigotimes_{i=1}^{f'}\hat{D}_{d'_i,m'r'_i}(\mathbf{q}'_i,\mathbf{p}'_i)\mathcal{E}\left(\bigotimes_{i=1}^{f}\hat{D}_{d_i,mr_i}(\mathbf{q}_i,\mathbf{p}_i)\right)\right]\right|\right|\leq \frac{\epsilon}{8e}$ with a success probability of $2/3$. We have the following hold for the number of uses of the channel $N$ for any such learning protocol.
    \begin{itemize}
        \item For $c\leq \min(d,d')-1$ we have
        $$N = \Omega\left(\frac{d_{\min}^{mr_{\mathrm{min}}} d_{\min}'^{m'r'_{\mathrm{min}}}}{c\epsilon^2}\right)$$
        \item Assuming $d = d'$, for $c\geq d$ we have
        $$N = \Omega\left(c\min\left[\frac{(d_{\min}^{mr_{\min}}-1) (d_{\min}'^{m'r'_{\min}}-1)}{c^2\epsilon^2},\left(\frac{d}{c\epsilon}\right)^{2d}\left(\frac{3}{4}\right)^2\right]\right)$$
    \end{itemize}
\end{theorem}
\begin{proof}
    We define a shorthand of $\mathbf{u}\in \mathbb{F}_{d_1}^{2mr_1}\times \mathbb{F}_{d_2}^{2mr_2}\times\dots\times\mathbb{F}_{d_f}^{2mr_f}$ and $\mathbf{v}\in \mathbb{F}_{d_1'}^{2m'r'_1}\times \mathbb{F}_{d_2'}^{2m'r'_2}\times\dots\times\mathbb{F}_{d'_{f'}}^{2m'r'_{f'}}$. Using this shorthand we define displacement operators for $\mathbf{u} = ((\mathbf{q}_1,\mathbf{p}_1),(\mathbf{q}_2,\mathbf{p}_2),\cdots,(\mathbf{q}_f,\mathbf{p}_f))$ and $\mathbf{v} = ((\mathbf{q}'_1,\mathbf{p}'_1),(\mathbf{q}'_2,\mathbf{p}'_2),\cdots,(\mathbf{q}'_{f'},\mathbf{p}'_{f'}))$ as $\hat{D}_{z,m}(\mathbf{u}) = \bigotimes_{i=1}^{f}\hat{D}_{d_i,mr_i}(\mathbf{q}_i,\mathbf{p}_i)$ and $\hat{D}_{z',m'}(\mathbf{v}) = \bigotimes_{i=1}^{f'}\hat{D}_{d'_i,m'r'_i}(\mathbf{q}'_i,\mathbf{p}'_i)$ respectively. We define the following channels 
    \begin{equation}
        \mathcal{E}_{0}(\cdot) = \text{Tr}\left[\cdot\right]\frac{\bigotimes_{i=1}^{f'}\mathbb{I}_{d_i',m'r_i'}}{z'^{m'}},\quad \mathcal{E}_{\mathbf{u},\mathbf{v}}(\cdot) = \text{Tr}\left[\hat{\Pi}_{+,\mathbf{u}}(\cdot)\right]\hat{\rho}_{+,\mathbf{v}} + \text{Tr}\left[\hat{\Pi}_{-,\mathbf{u}}(\cdot)\right]\hat{\rho}_{-,\mathbf{v}}
    \end{equation}
    where we define the operators
    \begin{equation}
        \hat{\Pi}_{\pm,\mathbf{u}}  = \frac{\bigotimes_{i=1}^{f}\mathbb{I}_{d_i,mr_i} \pm \frac{e^{i\pi/4}\hat{D}_{z,m}(\mathbf{u})+ e^{-i\pi/4}\hat{D}_{z,m}(\mathbf{u})^\dagger}{2}}{2},\quad \hat{\rho}_{\pm,\mathbf{v}} = \frac{\bigotimes_{i=1}^{f'}\mathbb{I}_{d_i',m'r'_i} \pm \epsilon_0(e^{i\pi/4}\hat{D}_{z',m'}(\mathbf{v})+ e^{-i\pi/4}\hat{D}_{z',m'}(\mathbf{v})^\dagger)/\sqrt{2}}{z'^{m'}}.
    \end{equation}
    From this, we can see that the channels take on the same form as Eq.~\eqref{eq:Epar1par2} with $\hat{\rho}_0 = \frac{\bigotimes_{i=1}^{f'}\mathbb{I}_{d_i',m'r'_i}}{z'^{m'}}$ and
    \begin{equation}
        \hat{K}_{1,\mathbf{u}} = \frac{e^{i\pi/4}\hat{D}_{z,m}(\mathbf{u})+ e^{-i\pi/4}\hat{D}_{z,m}(\mathbf{u})^\dagger}{2},\quad \hat{W}_{1,\mathbf{v}} = \epsilon_0(e^{i\pi/4}\hat{D}_{z',m'}(\mathbf{v})+ e^{-i\pi/4}\hat{D}_{z',m'}(\mathbf{v})^\dagger)/\sqrt{2},
    \end{equation}
    using which we rewrite
    \begin{equation}
        \mathcal{E}_{\mathbf{u},\mathbf{v}}(\cdot) = \sum_{x\in\{0,1\}}\text{Tr}\left[\hat{K}_{x,\mathbf{u}}(\cdot)\right]\sqrt{\hat{\rho}_0}\hat{W}_{x,\mathbf{v}}\sqrt{\hat{\rho}_0}.
    \end{equation}
\textbf{Success in Problem~\ref{prob:manyrevel} through learning task:} Consider $A$ to sample $\mathbf{u}$ from $\mathbb{F}_{d_1}^{2mr_1}\times \mathbb{F}_{d_2}^{2mr_2}\times\dots\times\mathbb{F}_{d_f}^{2mr_f}\setminus\{\mathbf{0}\}$ uniformly randomly and $\mathbf{v}$ from $\mathbb{F}_{d_1'}^{2m'r'_1}\times \mathbb{F}_{d_2'}^{2m'r'_2}\times\dots\times\mathbb{F}_{d_f'}^{2m'r'_{f'}}\setminus\{\mathbf{0}\}$ uniformly randomly. $B$ performs the learning protocol following which $A$ reveals the values of $\mathbf{u}$ and $\mathbf{v}$. Similar to Thm.~\ref{thm:qudit2qudit} we note 
\begin{equation}
\left|\text{Tr}\left[\hat{D}_{z',m'}(\mathbf{v})\mathcal{E}_0\left(\hat{D}_{z,m}(\mathbf{u})\right)\right]\right| = 0,\quad \left|\text{Tr}\left[\hat{D}_{z',m'}(\mathbf{v})\mathcal{E}_{\mathbf{u},\mathbf{v}}\left(\hat{D}_{z,m}(\mathbf{u})\right)\right]\right| = \frac{\epsilon_0}{2\sqrt{2}},
\end{equation}  
through which we see that setting $\epsilon = e\epsilon_0\sqrt{2}$ makes an estimation accuracy of $\epsilon/8e$ sufficient to succeed in the hypothesis test.

\textbf{Lower bound using Lemma~\ref{lem:master}}:
Observe that the operator $\hat{K}_{1,\mathbf{u}}^{\otimes 2k}$ when averaged over all possible $\mathbf{u}$ values has a very similar form to that of the operator norm studied in Corollary~\ref{corr:quditEopnorm}, with the tensor-product structure controlled by Lemma~\ref{lem:disp_sqfree}. Using the fact that there is a parity operator that conjugates $\hat{D}_{z,m}(\mathbf{u})$ to $\hat{D}_{z,m}(\mathbf{u})^\dagger$, we can apply the triangle inequality to obtain (assuming $k\neq 0\mod d$)
\begin{equation}
\begin{aligned}
    \left\|\mathbb{E}_{\mathbf{u}}\hat{K}_{1,\mathbf{u}}^{\otimes 2k}\right\|_{\mathrm{op}} &\leq \frac{1}{z^{2m}-1}\left\|\sum_{\mathbf{u}\in \mathbb{F}_{d_1}^{2mr_1}\times \mathbb{F}_{d_2}^{2mr_2}\times\dots\times\mathbb{F}_{d_f}^{2mr_f}}\hat{D}_{z,m}(\mathbf{u})^{\otimes 2k}\right\|_{\mathrm{op}} + \frac{\left\|\hat{D}_{z,m}(\mathbf{0})\right\|_{\mathrm{op}}}{z^{2m}-1}\\
    &\leq \frac{{d_{\min}}^{mr_{\min}} +(z/d_{\mathrm{min}}^{r_{\min}})^{-2m}}{{d_{\min}}^{2mr_{\min}} -(z/d^{r_{\min}}_{\mathrm{min}})^{-2m}}\leq \frac{1}{d_{\mathrm{min}}^{mr_{\min}} - (z/d_{\mathrm{min}}^{r_{\min}})^{-m}} \leq \frac{1}{d_{\min}^{mr_{\min}} - 1}
\end{aligned}
\end{equation}
For $k = 0\mod d$, we note that the upper bound is always an operator norm of $1$ by the triangle inequality on the operator $\mathbb{E}_{\mathbf{u}}\hat{K}_{1,\mathbf{u}}^{\otimes 2k}$. Hence we have
\begin{equation}
    \left\|\mathbb{E}_{\mathbf{u}}\hat{K}_{1,\mathbf{u}}^{\otimes 2k}\right\|_{\mathrm{op}} \leq \begin{cases}
        \frac{1}{d_{\min}^{mr_{\min}} - 1}, & k\neq 0\mod d\\
        1, & k = 0\mod d
    \end{cases}
\end{equation}
which similarly extends to give
\begin{equation}
    \left\|\mathbb{E}_{\mathbf{v}}\hat{W}_{1,\mathbf{v}}^{\otimes 2k}\right\|_{\mathrm{op}} \leq \begin{cases}
        \frac{(\epsilon_0\sqrt{2})^{2k}}{d_{\min}'^{m'r'_{\min}}-1}, & k \neq 0\mod d'\\
    (\epsilon_0\sqrt{2})^{2k}, & k = 0\mod d'
    \end{cases}.
\end{equation}
Substituting this in the expression of the $\chi^2$-divergence we get
\begin{equation}
    \Delta \leq \left(\sum_{k = 1}^{c}\binom{c}{k}\sqrt{\left\|\mathbb{E}_{\mathbf{u}}\hat{K}_{1,\mathbf{u}}^{\otimes 2k}\right\|_{\mathrm{op}}\left\|\mathbb{E}_{\mathbf{v}}\hat{W}_{1,\mathbf{v}}^{\otimes 2k}\right\|_{\mathrm{op}}}\right)^{2}.
\end{equation}
If $c \leq \min(d,d') - 1$, then $k \neq 0\mod d$ and $k \neq 0 \mod d'$, and so we get
\begin{equation}
    \Delta \leq \left(\sum_{k = 1}^{c}\binom{c}{k}\frac{(\epsilon_0\sqrt{2})^k}{\sqrt{(d_{\min}^{mr_{\min}} -1)(d{'}^{m'r'_{\min}}_{\min}-1)}}\right)^2 = \frac{((1+\epsilon_0\sqrt{2})^c - 1)^2}{(d_{\min}^{mr_{\min}}-1)(d{'}^{m'r'_{\min}}_{\min}-1)} \implies N = \Omega\left(\frac{d_{\min}^{mr_{\min}} d{'}^{m'r'_{\min}}_{\min}}{c\epsilon^2}\right).
\end{equation}
If $c \geq d$ and $d = d'$, we separately deal with the $k = 0\mod d$ terms to obtain
\begin{equation}
\begin{aligned}
    \Delta &\leq \left(\sum_{k = 1}^{c}\binom{c}{k}\frac{(\epsilon_0\sqrt{2})^k}{\sqrt{(d_{\min}^{mr_{\min}} -1)(d{'}_{\min}^{m'r'_{\min}}-1)}} + \sum_{k=1}^{\lfloor c/d\rfloor}\binom{c}{kd}(\epsilon_0\sqrt{2})^{kd}\right)^2\\
    &\leq \left(\frac{c\epsilon}{\sqrt{(d_{\min}^{mr_{\min}} - 1)(d{'}_{\min}^{m'r'_{\min}} - 1)}} + \frac{4}{3}\left(\frac{c\epsilon}{d}\right)^d\right)^2,
\end{aligned}
\end{equation}
which yields the lower bound
\begin{equation}
    N = \Omega\left(c\min\left[\frac{(d_{\min}^{mr_{\min}} - 1)(d{'}_{\min}^{m'r'_{\min}} - 1)}{c^2\epsilon^2},\left(\frac{d}{c\epsilon}\right)^{2d}\left(\frac{3}{4}\right)^2\right]\right).
\end{equation}
\end{proof}

The above theorem shows that if we consider channels acting on $m$ $z$-level qudits, these will give a channel learning task that is $c$-copy hard for all $c\leq d$ where $d$ is the product of all the prime numbers that divide $z$. Note that if $d,d'$ are coprime, the cases of $d\neq d'$ and $c \geq \min(d,d')$ would follow exactly as those shown in Thm.~\ref{thm:qudit2qudit} where each $d$ is replaced by $d_{\min}$ and $d'$ by $d'_{\min}$. We expect that if $d,d'$ are not coprime, the lowest common multiple of them would be the number of copies required for efficient learning to be possible and leave this case to future exploration.

\section{Limitations of Choi-state learning}
\subsection{Inadequacy for describing arbitrary bosonic states}\label{app:inadequacy}
\subsubsection{Description of transfer functions for bosonic channels}
Consider channel $\mathcal{E}^{\mathrm{qudit}}$ with input Hilbert space $\mathcal{H}_{d,m}$ and output Hilbert space $\mathcal{H}_{d',m'}$. The transfer function given by $C_{\mathcal{E}^{\mathrm{qudit}}}((\mathbf{q},\mathbf{p}),(\mathbf{q}',\mathbf{p}')) = \text{Tr}\left[\hat{D}_{d',m'}(\mathbf{q}',\mathbf{p}')\mathcal{E}(\hat{D}_{d,m}(\mathbf{q},\mathbf{p}))\right]/{d^{m}}$ for $\mathbf{q},\mathbf{p}\in\mathbb{F}^{m}_{d}$ and $\mathbf{q}',\mathbf{p}'\in\mathbb{F}^{m'}_{d'}$ offers a complete description for the channel since for any arbitrary input state $\hat{\rho}\in D(\mathcal{H}_{d,m})$ we have
\begin{equation}
    \mathcal{E}^{\mathrm{qudit}}(\hat{\rho}) = \sum_{\mathbf{q},\mathbf{p}\in\mathbb{F}^{m}_{d},\mathbf{q}',\mathbf{p}'\in\mathbb{F}^{m'}_{d'}}C_{\mathcal{E}^{\mathrm{qudit}}}((\mathbf{q},\mathbf{p}),(\mathbf{q}',\mathbf{p}'))\frac{\hat{D}_{d',m'}^\dagger(\mathbf{q}',\mathbf{p}')}{d'^{m'}} \text{Tr}\left[\hat{D}_{d,m}(\mathbf{q},\mathbf{p})\hat{\rho}\right].
\end{equation}
Alternatively if we consider a channel $\mathcal{E}^{\mathrm{qudit-boson}}$ with input space $\mathcal{H}_{d,m}$ and output space of $m$ bosonic modes $\mathcal{H}_{\infty,m}$ and we similarly define $C_{\mathcal{E}^{\mathrm{qudit-boson}}}((\mathbf{q},\mathbf{p}),\beta) = \text{Tr}\left[\hat{D}(\beta)\mathcal{E}(\hat{D}_{d,m}(\mathbf{q},\mathbf{p}))\right]/d^m$ for $\mathbf{q},\mathbf{p}\in\mathbb{F}^{m}_{d}$ and $\beta\in\mathbb{C}^m$, for any arbitrary input state $\hat{\rho}\in D(\mathcal{H}_{d,m})$ we have
\begin{equation}
    \mathcal{E}(\hat{\rho}) = \sum_{\mathbf{q},\mathbf{p}\in\mathbb{F}^{m}_{d}}\int d^{2m}\beta \hat{D}^\dagger(\beta)C_{\mathcal{E}^{\mathrm{qudit-boson}}}((\mathbf{q},\mathbf{p}),\beta)\text{Tr}\left[\hat{D}_{d,m}(\mathbf{q},\mathbf{p})\hat{\rho}\right].
\end{equation}
This kind of expansion for arbitrary input states is made possible since the qudit phase space displacement operators offer an orthonormal operator basis to represent an arbitrary input state. While this still holds true using the bosonic characteristic function for the case of bosonic Hilbert spaces, we must note that the object we are capable of learning does not actually take form of $\text{Tr}\left[\hat{D}(\beta)\mathcal{E}(\hat{D}(\alpha))\right]$ on account of this function being ill defined.

We now consider a channel $\mathcal{E}^{\mathrm{boson}}$ with input and output space being $m$ bosonic modes $\mathcal{H}_{\infty,m}$. Recall that the channel learning defined over TMSV input states learns the function
\begin{equation}
    C^{\mathrm{TMSV},r}_{\mathcal{E}^{\mathrm{boson}}}(\alpha,\beta) = e^{-\frac{|\alpha|^2}{2\cosh(2r)}}\text{Tr}\left[\hat{D}(\beta)\mathcal{E}\left(\hat{D}^\dagger(\alpha\tanh(2r))\hat{\rho}_{\mathrm{th},\alpha}\right)\right], \quad \hat{\rho}_{\mathrm{th},\alpha} = \hat{D}(\alpha\tanh(2r)/2)\frac{(\tanh^2(r))^{\hat{n}}}{\cosh^{2m}(r)}\hat{D}^\dagger(\alpha\tanh(2r)/2).
\end{equation}
The set of operators $e^{-\frac{|\alpha|^2}{2\cosh(2r)}}\hat{D}^\dagger(\alpha\tanh(2r))\hat{\rho}_{\mathrm{th},\alpha}$ can be used to represent certain $m$-mode states by noting that
\begin{equation}
\begin{aligned}
     \hat{\rho} &= \int d^{2m}\gamma f_{\hat{\rho}}(\gamma)e^{-\frac{|\gamma|^2}{2\cosh(2r)}}\hat{D}^\dagger(\gamma\tanh(2r))\hat{\rho}_{\mathrm{th},\gamma},\\
     f_{\hat{\rho}}(\gamma) &= e^{\frac{|\gamma|^2}{2\cosh(2r)}}\left(\frac{\sinh^2(2r)}{2\pi\cosh(2r)}\right)^m\int d^{2m}\alpha W_{\hat{\rho}}(\alpha)e^{i\Omega({{\gamma}\tanh(2r),\alpha})}e^{\frac{2|{\alpha}|^2}{\cosh(2r)}}.
\end{aligned}
\end{equation}
where we define $W_{\hat{\rho}}(\alpha) = \text{Tr}\left[e^{i\pi\hat n}\hat{D}(\alpha)\hat{\rho}\hat{D}^\dagger(\alpha)\right]\frac{2^m}{\pi^m}$ which is the Wigner function of the state. For the expansion $f_{\hat{\rho}}(\gamma)$ to be valid, we require $W_{\hat{\rho}}(\alpha)$ to decay significantly faster than $e^{-\frac{2|{\alpha}|^2}{\cosh(2r)}}$ which naturally places a limitation on the states $\hat{\rho}$ that can be represented in this form. Assuming that the function $f_{\hat{\rho}}(\gamma)$ is well defined, we can write
\begin{equation}\label{eq:char_by_CTMSV}
    \mathcal{E}^{\mathrm{boson}}(\hat{\rho}) = \frac{1}{\pi^m}\int d^{2m}\beta d^{2m}\alpha \hat{D}^\dagger(\beta)C_{\mathcal{E}^{\mathrm{boson}}}^{\mathrm{TMSV},r}(\alpha,\beta)f_{\hat{\rho}}(\alpha),
\end{equation}
which notably can only be written down assuming the Wigner function of the state $\hat{\rho}$ decays fast enough.

While the transfer matrix description may be complete, our efficient protocols would still fail at obtaining a tomographic description of the channel in terms of diamond norm accuracy. The main reason for this is that we assume each learning step only to be allowed $O(1/\epsilon)$ copies of the channel where $\epsilon$ is the additive accuracy for the estimation of the transfer function $C_{\mathcal{E}}$. Considering the basic case of the channel $\mathcal{E}^{\mathrm{qudit}}$, if we wish to have a diamond norm accuracy of $\tilde{\epsilon}$ by only learning the function $C_{\mathcal{E}^{\mathrm{qudit}}}$, it would be sufficient to require our queries to be accurate to $\epsilon = \Theta(\tilde{\epsilon}d_{\mathrm{in}}^{-2}d_{\mathrm{out}}^{-1})$. The efficient schemes we have would clearly still have a much higher sample complexity than the optimal schemes highlighted in \cite{mele2026optimallearningquantumchannels}. We explain this by the fact that our efficient schemes rely on using as little memory as possible which is a hindrance in the regime of exponentially small accuracy requirements. Further, these efficient schemes make use of many more than $O(1/\epsilon)$ copies in their learning, since they involve parallel use of $O(d_{\mathrm{in}}d_{\mathrm{out}}k/\epsilon^2)$ copies ($k$ is the Kraus rank of the channel) that are input to the random purification channel. This basic calculation highlights that our channel learning goal will ultimately not be well suited to channel tomography but is suited for cases where learning these transfer functions offers useful information by itself. 

\subsubsection{A no-go result for learning a bosonic channel transfer function}\label{app:bosonicnogo}
Much like Pauli transfer matrices \cite{10.1145/3670418} that are used to characterize finite-dimensional CPTP maps, we can formulate a continuous-variable version of this by allowing the use of tempered distributions. We note that for any channel $\mathcal{E}\in \CPTP(\mathcal{H}_{\infty,m},\mathcal{H}_{\infty,m})$ we can write down
\begin{equation}
    \mathcal{E}(\hat{D}^\dagger({\alpha})) = \int d^{2m}{\beta} \Lambda_{\mathcal{E}}({\alpha},{\beta})\hat{D}^\dagger({\beta}),
\end{equation}
where effectively $\Lambda_{\mathcal{E}}({\alpha},{\beta})$ is the continuous version of the Pauli transfer matrix. This can be directly obtained from the channel as
\begin{equation}
    \Lambda_{\mathcal{E}}({\alpha},{\beta}) = \frac{1}{\pi^m} \text{Tr}(\hat{D}({\beta})\mathcal{E}(\hat{D}^\dagger({\alpha}))).
\end{equation}
Alternatively, bosonic channels can also be represented as follows \cite{PRXQuantum.5.010331}
\begin{equation}
    \mathcal{E}(\cdot) = \int d^{2m}\alpha_1 d^{2m}\alpha_2K_{\mathcal{E}}(\alpha_1,\alpha_2) \hat{D}(\alpha_1)(\cdot)\hat{D}^\dagger(\alpha_2).
\end{equation}
It can be checked that one can recover the transfer function from the above representation by the following relation
\begin{equation}
\begin{aligned}
    \Lambda_{\mathcal{E}}(\alpha,\beta) = \frac{1}{\pi^m}\int d^{2m}\alpha_1d^{2m}\alpha_2 K_{\mathcal{E}}(\alpha_1,\alpha_2)\text{Tr}\left[\hat{D}(\beta)\hat{D}(\alpha_1)\hat{D}^\dagger(\alpha)\hat{D}^\dagger(\alpha_2)\right].
\end{aligned}
\end{equation}
Note that
\begin{equation}\label{eq:4disptrace}
    \text{Tr}(\hat{D}({\beta})\hat{D}({\alpha}_1)\hat{D}^\dagger({\alpha})\hat{D}^\dagger({\alpha_2})) = \pi^m\delta^{(2m)}({\beta}+{\alpha}_1 -{\alpha} - {\alpha}_2)e^{\frac{1}{2}(i\Omega({\alpha}_1,{\beta}) + i\Omega({\alpha}_2,{\alpha}) + i\Omega({\beta}+{\alpha}_1,{\alpha}+{\alpha}_2))},
\end{equation}
following which we do a variable change to ${\xi} = \frac{{\alpha}_1+{\alpha}_2}{2}$ and ${\xi}' = {\alpha}_1-{\alpha}_2$. The Dirac delta chooses ${\xi}' = {\alpha}-{\beta}$ in the above integral giving
\begin{equation}
    \begin{aligned}
        \Lambda_{\mathcal{E}}(\alpha,\beta) &= \int d^{2m}{\xi} K_{\mathcal{E}}\left({\xi} + \frac{{\alpha}-{\beta}}{2},{\xi}+\frac{{\beta}-{\alpha}}{2}\right) e^{i\Omega({\xi},\frac{{\alpha}+{\beta}}{2})}.
    \end{aligned}
\end{equation}
Unfortunately, both $K_{\mathcal{E}}(\alpha,\alpha')$ and $\Lambda_{\mathcal{E}}(\alpha,\beta)$ are numerically ill defined since displacements are unbounded operators and are not trace class resulting in both functions having tempered distribution like behaviors. Note that the above definition is motivated by finding the coefficient proportional to identity in the expression $\hat{D}({\beta})\mathcal{E}(\hat{D}^\dagger({\alpha}))$ which is not trace class for infinite-dimensional systems. Taking $\hat{n}$ to be the total photon number operator over $m$ modes, we define the operator
\begin{align}
\hat B_{\lambda}(\alpha)=2^m\hat D^\dagger(\alpha)e^{(i\pi-\lambda)\hat n}\hat D(\alpha),
\end{align}
which admits a simple interpretation in the displaced frame centered at \(-\alpha\). Defining
\begin{align}
\hat n_{-\alpha}
:=\hat D^\dagger(\alpha)\hat n\hat D(\alpha)
=(\hat a^\dagger+\alpha^*)(\hat a+\alpha),
\end{align}
one obtains
\begin{align}
\hat B_{\lambda}(\alpha)
=2^me^{(i\pi-\lambda)\hat n_{-\alpha}}
=2^m(-1)^{\hat n_{\alpha} }e^{-\lambda \hat n_{-\alpha} }.
\end{align}
Hence, $\hat B_{\lambda}(\alpha)$ can be viewed as the parity operator about the phase-space point $-\alpha$, regularized by the damping factor $e^{-\lambda\hat n_{-\alpha} }$. In the limit $\lambda\to0^+$,
\begin{align}
\hat B(\alpha)
=\hat D^\dagger(\alpha)(-1)^{\hat n}\hat D(\alpha)
=(-1)^{\hat n_{-\alpha} },
\end{align}
which is precisely the displaced parity operator centered at $-\alpha$. Geometrically, it implements the point-reflection transformation $\gamma \mapsto -2\alpha-\gamma$ in phase space, or equivalently, in local coordinates $\gamma=-\alpha+\beta$,
$\beta\mapsto -\beta$. Note that this convergence does not hold true in the operator norm sense since we are taking $\lim_{\lambda\to0^+}e^{-\lambda\hat
n} = \mathbb{I}$. The Wigner function can be written as
\begin{align}
W_{\hat\rho}(\alpha)
=
\frac{1}{\pi^m}\mathrm{Tr}\!\left[\hat\rho\,\hat B(\alpha)\right]
=
\left(\frac{2}{\pi}\right)^m
\mathrm{Tr}\!\left[\hat\rho\,\hat\Pi_\alpha\right],\quad \hat\Pi_\alpha := 2^{-m}\hat{B}(\alpha),
\end{align}
showing that $W_{\hat\rho}(\alpha)$ is the expectation value of the local parity about $\alpha$. Equivalently, it measures the imbalance between the locally even and locally odd components of $\hat{\rho}$ under point reflection about $\alpha$.

We introduced a family of operators we refer to as $\hat{B}({\alpha})$ and their regularized version of $\hat{B}_{\lambda}({\alpha})$. Since $\hat{B}_{\lambda}({\alpha})$ is a trace-class operator, we can write down its characteristic function as
\begin{equation}
    \hat{B}_{\lambda}({\alpha}) = \frac{1}{\pi^m}\int d^{2m}{\beta}e^{i\Omega({\beta},{\alpha})}e^{-|{\beta}|^2\tanh(\lambda/2)/2}\hat{D}^\dagger({\beta}),
\end{equation}
which in the limit of $\lambda\to0^+$ can be thought of as a Fourier transform to displacement operators as
\begin{equation}
    \hat{D}({\beta}) = \frac{1}{\pi^m}\int d^{2m}{\alpha}\hat{B}({\alpha})e^{i\Omega({\beta},{\alpha})}.
\end{equation}
Under this construction, we obtain the following relations.
\begin{equation}
\begin{aligned}
    \hat{B}({\alpha}_1)\hat{B}({\alpha}_2) &= \frac{1}{\pi^{2m}}\int d^{2m}{\beta}_1 d^{2m}{\beta}_2 \hat{D}({\beta}_1)\hat{D}({\beta}_2)e^{i\Omega({\alpha}_1,{\beta}_1)+i\Omega({\alpha}_2,{\beta}_2)}\\
    &= \frac{1}{\pi^{2m}}\int d^{2m}{\beta}_+ d^{2m}{\beta}_- \hat{D}({\beta}_+) e^{i\Omega(\frac{{\beta}_+}{2} + {{\alpha}}_1-{{\alpha}}_2,{\beta}_-) + \frac{1}{2}i\Omega({\alpha}_1+{\alpha}_2,{\beta}_+)}\\
    &= 2^{2m}\hat{D}(2({\alpha}_2-{\alpha}_1))e^{2i\Omega({\alpha}_2,{\alpha}_1)},
\end{aligned}
\end{equation}
\begin{equation}
    \text{Tr}(\hat{B}({\alpha}_1)\hat{B}({\alpha}_2)) = \pi^m\delta^{(2m)}({\alpha}_2 -{\alpha}_1).
\end{equation}
Here we have used the substitution $\beta_+ = \frac{\beta_1+\beta_2}{2}$ and $\beta_- = \beta_1 - \beta_2$, and transformed the integral appropriately. Using the above equations, we note the following properties
\begin{equation}
    \hat{D}({\beta})\hat{B}({\alpha})\hat{D}^\dagger({\beta}) = \hat{B}({\alpha}-{\beta}),\quad\hat{B}({\alpha})\hat{D}({\beta}) = e^{i\Omega({\beta},{\alpha})}\hat{B}({\alpha}+\frac{{\beta}}{2}),\quad \frac{\hat{B}({\alpha}')}{2^m}\hat{B}({\alpha})\frac{\hat{B}({\alpha}')}{2^m} = \hat{B}(2{\alpha}' - {\alpha}).
\end{equation}
We observe that $\hat{B}({\alpha})^2 = 2^{2m}\mathbb{I}$ for all ${\alpha}$. Along with the fact that these are Hermitian operators, this means that the following operators are idempotent projectors belonging to a 2-element POVM
\begin{equation}
    \hat{P}_\pm =\frac{\mathbb{I}\pm 2^{-m}\hat{B}({\alpha})}{2},
\end{equation}
and the operator $\hat{B}({\alpha})/2^m$ is a unitary and Hermitian operator. The infinite-energy TMSV state can be expressed in this basis of operators as
\begin{equation}
\begin{aligned}
    |\Psi\rangle\langle\Psi| &= \frac{1}{\pi^m}\int d^{2m}{\beta}\hat{D}
    ({\beta})\otimes\hat{D}({\beta}^{*}) \\
    &= \frac{1}{\pi^{3m}}\int d^{2m}{\alpha}d^{2m}{\alpha}'d^{2m}{\beta}e^{i\Omega({\beta},{{\alpha}})}e^{i\Omega(({\alpha}')^{*},{\beta})}\hat{B}({\alpha})\otimes\hat{B}({\alpha}')\\
    &=\frac{1}{\pi^m}\int d^{2m}{\alpha} \hat{B}({\alpha})\otimes \hat{B}({\alpha}^{*}),
\end{aligned}
\end{equation}
and we can similarly define a POVM set using this state as
\begin{equation}
\begin{aligned}
    \hat{\Pi}({\zeta}) &= \frac{1}{\pi^m}(2^{-m}\hat{B}({\zeta}/2)\otimes \mathbb{I})|\Psi\rangle\langle\Psi|(2^{-m}\hat{B}^\dagger({\zeta}/2)\otimes\mathbb{I})\\
    &= \frac{1}{\pi^{2m}}\int d^{2m}{\alpha}\hat{B}({\zeta}-{\alpha})\otimes\hat{B}({\alpha}^{*}),
\end{aligned}
\end{equation}
which can be verified to be a valid POVM set since
\begin{equation}
    \int d^{2m}{\zeta}\hat{\Pi}({\zeta}) = \mathbb{I}\otimes\mathbb{I}.
\end{equation}
Also note that the Wigner function of a state represents it in the basis of operators $\hat{B}({\alpha})$ as can be noted by
\begin{equation}
    W_{\hat{\rho}}({\alpha}) = \frac{1}{\pi^{2m}}\int d^{2m}{\beta} e^{i\Omega({\alpha},{\beta})}\chi_{\hat{\rho}}({\beta}) = \frac{\text{Tr}(\hat{\rho}\hat{B}({\alpha}))}{\pi^{m}},
\end{equation}
which using the orthogonality of $\hat{B}$ operators lets us write
\begin{equation}
    \hat{\rho} = \int d^{2m}{\alpha} W_{\hat{\rho}}({\alpha})\hat{B}({\alpha}),
\end{equation}
and as such we can interpret $\hat{B}({\alpha})$ as an operator with a Dirac-delta like distribution in phase space which cannot be understood as a physical state. 

The Hilbert--Schmidt inner product of $\hat{B}_{\lambda}({\alpha})$ with displacements is a well-defined quantity taking the limit $\lambda\to0^+$ which motivates us to define the following quantities
\begin{equation}
\begin{aligned}
\tilde{C}_{\mathcal{E},\lambda}(\alpha,\beta)&= e^{i\Omega({\alpha},{\beta})}\text{Tr}(\hat{D}({\beta})\mathcal{E}(\hat{B}_{\lambda}({\alpha}))),\\
    \tilde{C}_{\mathcal{E}}({\alpha},{\beta}) &= \lim_{\lambda\to0^+} \tilde{C}_{\mathcal{E},\lambda}(\alpha,\beta)\\    &= \int d^{2m}{\alpha}'d^{2m}{\alpha}_1 d^{2m}{\alpha}_2\text{Tr}(\hat{D}({\beta})\hat{D}({\alpha}_1)\hat{D}^\dagger({\alpha}')\hat{D}^\dagger({\alpha}_2)) K_{\mathcal{E}}({\alpha}_1,{\alpha}_2)e^{i\Omega({\alpha}'-{\beta},{\alpha})}.
\end{aligned}
\end{equation}
We now use Eq.~\eqref{eq:4disptrace} following which we do a variable change to ${\xi} = \frac{{\alpha}_1+{\alpha}_2}{2}$ and ${\xi}' = {\alpha}_1-{\alpha}_2$. The Dirac delta chooses ${\xi}' = {\alpha}'-{\beta}$ in the above integral giving
\begin{equation}
    \begin{aligned}
        \tilde{C}_{\mathcal{E}}({\alpha},{\beta}) &= \int d^{2m}{\alpha}' d^{2m}{\xi} K_{\mathcal{E}}\left({\xi} + \frac{{\alpha}'-{\beta}}{2},{\xi}+\frac{{\beta}-{\alpha}'}{2}\right) e^{i\Omega({\xi},\frac{{\alpha}'+{\beta}}{2})}e^{i\Omega({\alpha}'-{\beta},{\alpha})}\\
        &= \int d^{2m}{\alpha}''d^{2m}{\xi} K_{\mathcal{E}}\left({\xi}+\frac{{\alpha}''}{2},{\xi}-\frac{{\alpha}''}{2}\right)e^{\frac{1}{2}i\Omega({\xi},{\alpha}'')} e^{i\Omega(\xi,{\beta})}e^{i\Omega({\alpha}'',{\alpha})},
    \end{aligned}
\end{equation}
where we do a change of variables in the second step ${\alpha}''={\alpha}'-{\beta}$ to obtain the final expression. Note that the above expression is fully invertible to obtain the function $K_{\mathcal{E}}$ by the transformation
\begin{equation}
    K_{\mathcal{E}}\left({\xi}+\frac{{\alpha}''}{2},{\xi}-\frac{{\alpha}''}{2}\right) = \frac{e^{\frac{1}{2}i\Omega({\alpha}'',{\xi})}}{\pi^{4m}}\int d^{2m}{\alpha}d^{2m}{\beta} \tilde{C}_{\mathcal{E}}({\alpha},{\beta})e^{i\Omega({\beta},{\xi})}e^{i\Omega({\alpha},{\alpha}'')}.
\end{equation}
We observe that the trace-preserving condition is equivalent to $\tilde{C}_{\mathcal{E}}({\alpha},0) = \text{Tr}(\hat{B}({\alpha})) = 1$ for all ${\alpha}$. The hermiticity preserving condition of $\mathcal{E}$ enforces that $(\tilde{C}_{\mathcal{E}}({\alpha},{\beta}))^{*} = \tilde{C}_{\mathcal{E}}({\alpha},-{\beta})$.

Note that we have a diagonal representation for $\hat{B}({\alpha})$ from the fact that it is a displaced parity operation using an orthonormal basis. We represent it as follows
\begin{equation}
    \hat{B}({\alpha})=2^m \sum_{\mathbf{l}\in(\mathbb{Z}^{\geq})^m}(-1)^{\sum_{i}l_i}|\mathbf{l},-{\alpha}\rangle\langle\mathbf{l},-{\alpha}|, \quad |\mathbf{l},{\alpha}\rangle = \hat{D}({\alpha})|\mathbf{l}\rangle,
\end{equation}
where the state $|\mathbf{l}\rangle$ is the $m$-mode bosonic state with excitations $l_i$ in the $i$th mode which is the $i$th entry in the vector $\mathbf{l}$. Note that
\begin{equation}
    \langle\mathbf{l},-{\alpha}|\hat{D}({\beta})|\mathbf{l},-{\alpha}\rangle = e^{i\Omega({\beta},{\alpha}) - |{\beta}|^2/2}\prod_{i=1}^{m}L^{\mathrm{Lag}}_{l_i}(|\beta_i|^2),
\end{equation}
where $L^{\mathrm{Lag}}_n$ is the $n$th Laguerre polynomial. We now introduce the definition of a well-behaved channel, which is essentially one for which the quantity $\tilde{C}_{\mathcal{E}}({\alpha},{\beta})$ is well defined.

\begin{definition}[Well-behaved channel]
    We define a well-behaved channel to be a CPTP map $\mathcal{E}$ acting on $m$ bosonic modes such that for all ${\alpha},{\beta} \in \mathbb{C}^m$, the function $\tilde{C}_{\mathcal{E}}({\alpha},{\beta}) =\lim_{\lambda\to0^+}e^{i\Omega({\alpha},{\beta})}\text{Tr}(\hat{D}({\beta})\mathcal{E}(\hat{B}_{\lambda}({\alpha})))$ satisfies $|\tilde{C}_{\mathcal{E}}({\alpha},{\beta})| < \infty$. We denote the set of all $m$-mode well-behaved channels by the set $\mathcal{W}_m$.
\end{definition}

We first begin by a counterexample to the above set of channels. Consider the following channel $\mathcal{E}_{\mathrm{PMD}}$ which measures the parity of the state and prepares a state conditional to this as follows.
\begin{equation}
    \mathcal{E}_{\mathrm{PMD}}(\cdot) = (|0\rangle\langle0|)^{\otimes m}\text{Tr}((\cdot)\hat{P}_+) + (|1\rangle\langle1|)^{\otimes m}\text{Tr}((\cdot)\hat{P}_-),\quad \hat{P}_{\pm} = \frac{\mathbb{I} \pm 2^{-m}\hat{B}(0)}{2}.
\end{equation}
The above channel is clearly CPTP since it is simply a measure-and-prepare channel. Defining ${\beta}$ as some $m$-mode displacement satisfying for at least one mode $i$, $|\beta_i|^2 = 1$ (note that $L_1(1)=0$) we then get
\begin{equation}
\begin{aligned}
    \tilde{C}_{\mathcal{E}_{\mathrm{PMD}}}({0},{\beta}) &= \text{Tr}(\hat{D}({\beta})\mathcal{E}(\hat{B}({0})))\\
    &= \text{Tr}(\hat{D}({\beta})|0\rangle\langle0|^{\otimes m})\text{Tr}\left(\frac{\hat{B}(0)+2^m\mathbb{I}}{2}\right) + \text{Tr}(\hat{D}({\beta})|1\rangle\langle1|^{\otimes m})\text{Tr}\left(\frac{\hat{B}(0)-2^m\mathbb{I}}{2}\right)\\
    &= \frac{1}{2}e^{-|{\beta}|^2/2}(1+2^{m}\delta^{(2m)}(0)),
\end{aligned}
\end{equation}
where we use the fact that $\text{Tr}(\hat{D}({\beta})|1\rangle\langle1|^{\otimes m}) = 0$.

\begin{lemma}[Sufficient conditions for well-behaved channel]
    A CPTP map $\mathcal{E}$ is guaranteed to be well-behaved if it satisfies any one of the following conditions
    \begin{itemize}
        \item It can be represented as $\mathcal{E}(\cdot) = \hat{U}(\cdot)\hat{U}^\dagger$ where $\hat{U}$ is a Gaussian unitary.
        \item It is an erasure channel represented by $\mathcal{E}(\cdot) = (1-p)(\cdot) + p\text{Tr}(\cdot)\hat{\rho}$ where $0\leq p\leq1$ and $\hat{\rho}$ is a valid density matrix.
        \item It is a convex mixture of channels contained in $\mathcal{W}_m$.
    \end{itemize}
\end{lemma}
\begin{proof}
    For proving the first condition, it is sufficient to show that Gaussian unitaries are well-behaved. To see this we note that for a Gaussian unitary $U_{\mathcal{G}}$, we can always find a ${\beta}'$ and $\phi$ for a given ${\beta}$ such that 
    \begin{equation}
        U_{\mathcal{G}}^\dagger\hat{D}({\beta})U_{\mathcal{G}} = e^{i\phi}\hat{D}({\beta}'),
    \end{equation}
    holds true. This is because Gaussian unitaries transform displacement operators to other displacement operators. As such the quantity
    \begin{equation}
        \lim_{\lambda\to 0^{+}}\text{Tr}(\hat{D}({\beta})(U_{\mathcal{G}}\hat{B}_{\lambda}({\alpha})U_{\mathcal{G}}^\dagger)) = e^{i\Omega({\beta}',{\alpha})+i\phi},
    \end{equation}
    has a well-defined limit and so this unitary is well behaved. It is trivial to see then that a convex mixture of these unitaries will also result in a well-behaved channel.

    For the second condition, we note that $\lim_{\lambda\to0^+}\text{Tr}(\hat{B}_{\lambda}({\alpha})) = 1$ and so the function $\tilde{C}_{\mathcal{E}}({\alpha},{\beta}) = 1-p + pe^{i\Omega({\alpha},{\beta})}\chi_{\hat{\rho}}({\beta})$, which is clearly a complex number contained in the unit disk.
\end{proof}

\begin{theorem}[Hardness even if allowed the complex-conjugate channel and various joint measurements]\label{thm:bosonnogo}
    Let $\mathcal{E}$ be a well-behaved channel over $m\geq 8$ modes which is guaranteed to output a mixture of two known Gaussian states conditioned on a Gaussian measurement, and let $\lambda>0$ be the regularisation parameter of $\hat{B}_\lambda$. If there exists a learning scheme that uses $N$ copies of the channel and produces an estimate $\tilde{\tilde{C}}_{\mathcal{E},\lambda}(0,\beta)$ for a fixed real-valued $\beta$ with $\beta\cdot\beta = 1$, with success probability $\geq 2/3$, such that $\left||\tilde{\tilde{C}}_{\mathcal{E},\lambda}(0,\beta)|-|\tilde{C}_{\mathcal{E},\lambda}(0,\beta)|\right|\leq \epsilon$ (with $0<\epsilon\leq0.1$), then this requires that
    $$N = \Omega\left(\max\left(\frac{1}{10\epsilon},\frac{\pi m}{e}\left(\frac{0.99}{e\tanh
    (\lambda/2)}\right)^{m/2}\right)\right).$$
    This lower bound holds even while allowing infinite-energy input states and arbitrary kinds of joint measurements.
\end{theorem}
\begin{proof}
We define the following two measurement operations, in which $\hat{\pmb{n}} = (\hat{n}_1,\dots,\hat{n}_m)$ is the vector of per-mode photon-number operators and $\delta>0$ is fixed in Eq.~\eqref{eq:nogodelta} below,
\begin{equation}
    \hat{\Pi}^{\phi}_1 = \int_{\|\mathbf{x}\|_{2}\leq \delta} d^m\mathbf{x} e^{-i\phi\cdot \hat{\pmb{n}}}|\mathbf{x}\rangle\langle \mathbf{x}|e^{i\phi\cdot\hat{\pmb{n}}}, \quad \hat{\Pi}^{\phi}_0 = \mathbb{I} - \hat{\Pi}^{\phi}_1.
\end{equation}
Using which we define the following two channels
\begin{equation}
    \mathcal{E}_0(\cdot) =\text{Tr}(\cdot)\hat{\sigma}_0,\quad\mathcal{E}_{\phi,r}(\cdot) = \text{Tr}(\hat{\Pi}^{\phi}_0 (\cdot))\hat{\sigma}_0 +\text{Tr}(\hat{\Pi}^{\phi}_1 (\cdot))\hat{\sigma}_{r},
\end{equation}
\begin{equation}
    \hat{\sigma}_0 = |0\rangle\langle0|^{\otimes m}, \quad \hat{\sigma}_{r} = \left(\hat{S}(r)|0\rangle\langle0|\hat{S}^\dagger(r)\right)^{\otimes m},
\end{equation}
It can be checked that the characteristic function
\begin{equation}
    \chi_{\hat{\sigma}_{r}}(\beta) = \text{Tr}\left[\hat{D}(\beta)\hat{\sigma}_r\right] = \text{Tr}\left[\hat{D}(\Re(\beta)e^{-r} + i\Im(\beta)e^r)\hat{\sigma}_0\right] = \exp(-\frac{1}{2}(\Re(\beta)^2e^{-2r}+\Im(\beta)^2e^{2r})).
\end{equation}
hence the output states of both these channels are guaranteed reflective symmetry. The variable $\phi$ is uniformly random over $[-\pi,\pi]^m$. $A$ has promised to satisfy one of the following hypotheses
\begin{itemize}
    \item $H_0$: $N$ copies of the channel $\mathcal{E}_0$ are prepared
    \item $H_1$: $N$ copies of the channel $\mathcal{E}_{\phi,r}$ are prepared
\end{itemize}
These $N$ copies of the unknown channel are then provided to $B$, who is allowed to use them along with arbitrary inputs and measurements. After $B$ has completed the measurements, $A$ reveals the values of $r$ and $\phi$, following which $B$ has to distinguish between the two hypotheses $H_0$ and $H_1$.

\noindent\textbf{Showing that learning $\tilde{C}_{\mathcal{E},\lambda}(0,\beta)$ succeeds at this hypothesis test}:

Now we show that with $\epsilon$-accurate knowledge of $\tilde{C}_{\mathcal{E}_{\phi,r},\lambda}$ this can be made to succeed with probability greater than $2/3$. We first find the function $\tilde{C}_{\mathcal{E}_{\phi,r},\lambda}$ by noting that the value
\begin{equation}
    \text{Tr}(\hat{\Pi}^{\phi}_1\hat{B}(\alpha)) = \begin{cases}
        1 & \text{if}\quad\|\Re(\alpha e^{i\phi})  \|_{2} \leq \delta \\
        0 & \text{otherwise}
    \end{cases}.
\end{equation}
Based on the fact that $\langle\mathbf{x}|\hat{B}_{\lambda}(0)|\mathbf{x}\rangle$ is a Gaussian distribution over $m$ dimensions with spread $\sigma^2 = \tanh(\lambda/2)/2$, to ensure that $\text{Tr}(\hat{\Pi}_1^\phi \hat{B}_{\lambda}(0))\approx 1$, we just require $\delta > \sqrt{m\tanh(\lambda/2)/2}$ for a large enough $m$. To exactly quantify this, we note that $\text{Tr}(\hat{\Pi}_{1}^{\phi}\hat{B}(0))$ is the cumulative distribution function $\mathrm{CDF}_{\chi^2_m}(\delta^2/\sigma^2) = \tilde{p}$ for the distribution $\chi^2_m$. This means $\tilde{p} = 1 - Q(m/2,\delta^2/2\sigma^2)$ where $Q(m,x) = \frac{\Gamma(m,x)}{\Gamma(m)}$ is the regularized gamma function. 

Hence we have $\tilde{C}_{\mathcal{E},\lambda}(0,\beta) = \text{Tr}(\hat{D}(\beta)(\tilde{p}\hat{\sigma}_r +(1-\tilde{p})\hat{\sigma}_0)) = \tilde{p}\chi_{\hat{\sigma}_r}(\beta) + (1-\tilde{p})\chi_{\hat{\sigma}_0}(\beta)$. Note that the task only requires being able to estimate the magnitude of $\tilde{C}_{\mathcal{E},\lambda}(0,\beta)$ and for the chosen states this function is positive. 
\begin{equation}
      \left||\tilde{C}_{\mathcal{E}_{\phi,r},\lambda}(0,\beta)| - |\tilde{C}_{\mathcal{E}_0,\lambda}(0,\beta)|\right|  = \tilde{p} e^{-\beta^2/2}\left(\exp(\frac{\beta^2}{2}(1-e^{-2r})) -1\right).
\end{equation}
For the discrimination to succeed we require that
\begin{equation}
    \tilde{p} e^{-1/2}\left(\exp(\frac{1}{2}(1-e^{-2r})) -1\right) > 2\epsilon\implies r> -\frac{1}{2}\log(1-2\log(1+\frac{2\epsilon e^{1/2}}{\tilde{p}})),
\end{equation}
and by choosing
\begin{equation}\label{eq:nogodelta}
    \delta = \sqrt{\frac{m}{2\times 0.99}\tanh(\lambda/2)},
\end{equation}
this guarantees that $\tilde{p}>0.546$ by applying Lemma~\ref{lem:gaussprob} since $m\geq8$. From this we can see that this discrimination will always be possible as long as $\epsilon < \frac{1-e^{-1/2}}{2}\tilde{p}$ and so we choose the maximum value of $\epsilon$ to be $0.107$ which ensures that $r$ remains bounded by choosing $r = -\frac{1}{2}\log(1-2\log(1+\frac{2\epsilon e^{1/2}}{0.546}))$. We can further ensure that if $\epsilon\leq 0.1$, the squeezing doesn't get unnecessarily large.

\noindent\textbf{Sample lower bound by diamond norm}:

The fundamental limit to distinguish two quantum channels in a single shot is given by the diamond distance between them. For shorthand, we will represent the $m$ mode bosonic Hilbert space to be $\mathcal{H}$ in this section. The diamond distance satisfies
\begin{equation}
\begin{aligned}
    \|\mathcal{E}_{\phi,r} - \mathcal{E}_{0}\|_\diamond &= \max_{\rho_{AB}\in D(\mathcal{H}^{\otimes 2})}\|((\mathcal{E}_{\phi,r})_A\otimes\mathbb{I}_B)(\rho_{AB}) - ((\mathcal{E}_0)_A\otimes\mathbb{I}_B)(\rho_{AB})\|_1\\
    &= \max_{\rho_{AB}\in D(\mathcal{H}^{\otimes 2})}\|(\hat{\sigma}_{r}-\hat{\sigma}_0)_A\otimes\text{Tr}_A[((\hat{\Pi}_1^{\phi})_A\otimes\mathbb{I}_B)\hat{\rho}_{AB}]\|_1\\
    &=\max_{\rho_{AB}\in D(\mathcal{H}^{\otimes 2})} \|\hat{\sigma}_{r}-\hat{\sigma}_0\|_1\text{Tr}(\hat{\Pi}_1^\phi\hat{\rho}_A).
\end{aligned}
\end{equation}
Here we have used the fact that the operator $\text{Tr}_A[((\Pi_1)_A\otimes\mathbb{I}_B)\rho_{AB}]$ is a positive operator since it is proportional to the post-measurement state for outcome $\Pi_1$ since this is an idempotent POVM. Notably even with access to an ancillary system, there is no advantage in the distinguishing task here due to the structure of the measure-and-prepare channel.

Note that all adaptive and feedforward strategies with parallel channel use with $N$ copies of the channels will still be upper bounded by the diamond distance between the channels $\mathcal{E}_0^{\otimes N}$ and $\mathcal{E}_{\phi,r}^{\otimes N}$ instead. Taking the subsystems $A$ and $B$ as Hilbert spaces of $\mathcal{H}^{\otimes N}$ and any arbitrary state $\rho_{AB}$ we have
\begin{equation}
    \begin{aligned}
      &\|((\mathcal{E}_{\phi,r}^{\otimes N})_A\otimes \mathbb
        I_B)(\rho) - ((\mathcal{E}_{0}^{\otimes N})_A\otimes \mathbb
        I_B)(\rho)\|_1\\
        &= \left\|\sum_{\mathbf{b}\in\{0,1\}^{N}}\left\{\left[\bigotimes_{i=1}^{N}\hat{\sigma}_{b_ir} - \hat{\sigma}_0^{\otimes N}\right]\otimes\text{Tr}_A\left[\left(\left(\bigotimes_{i=1}^N \hat{\Pi}_{b_i}^{\phi}\right)_A\otimes\mathbb{I}_B\right)\rho_{AB}\right]\right\}\right\|_1\\
        &\leq \sum_{\mathbf{b}\in\{0,1\}^{N}}\left\|\left\{\left[\bigotimes_{i=1}^{N}\hat{\sigma}_{b_ir} - \hat{\sigma}_0^{\otimes N}\right]\otimes\text{Tr}_A\left[\left(\left(\bigotimes_{i=1}^N \hat{\Pi}_{b_i}^{\phi}\right)_A\otimes\mathbb{I}_B\right)\rho_{AB}\right]\right\}\right\|_1\\
        &\leq \|\hat{\sigma}_{r}^{\otimes N} - \hat{\sigma}_0^{\otimes N}\|_1\left(1 - \text{Tr}[\hat{\Pi}_0^{\otimes N}\hat{\rho}_A]\right).
    \end{aligned}
\end{equation}
Here we first make use of the triangle inequality and then use the fact that $\|\bigotimes_i\hat{\sigma}_{b_ir} - \hat{\sigma}_0^{\otimes N}\|_1 \leq \|\hat{\sigma}_r^{\otimes N} - \hat{\sigma}_0^{\otimes N}\|_1$ for all bitstrings $\mathbf{b}\in\{0,1\}^N$. Since the hypothesis test is performed without the actual knowledge of $\phi,\gamma$, we need to maximize the following quantity to obtain the highest success probability 
\begin{equation}
\begin{aligned}
    p_{\mathrm{success}}&\leq \max_{\rho\in D(\mathcal{H}^{\otimes2 N})}\frac{1}{2}\left(1+\frac{1}{2}\mathbb{E}_{\phi,r}\Big[\|((\mathcal{E}_{\phi,r}^{\otimes N})_A\otimes \mathbb
        I_B)(\rho) - ((\mathcal{E}_{0}^{\otimes N})_A\otimes \mathbb
        I_B)(\rho)\|_1\Big]\right)\\&\leq \max_{\hat{\rho}\in D(\mathcal{H}^{\otimes 2N})}\frac{1}{2}\left(1+\frac{1}{2}\mathbb{E}_{\phi,r}\Big[\|\hat{\sigma}_{r}^{\otimes N} - \hat{\sigma}_0^{\otimes N}\|_1\left(1 - \text{Tr}[\hat{\Pi}_0^{\otimes N}\hat{\rho}_A]\right)\Big]\right).
\end{aligned}
\end{equation}

\noindent\textbf{Upper bound on the trace distance $\|\hat{\sigma}_r^{\otimes N} - \hat{\sigma}_{0}^{\otimes N}\|_1$}:
Note that here the states $\hat{\sigma}_r^{\otimes N}$ and $\hat{\sigma}_{0}^{\otimes N}$ are Gaussian states. Hence we make use of the tight trace distance bound from \cite{Bittel2025optimalestimatesof} to obtain
\begin{equation}
    \|\hat{\sigma}_{r}^{\otimes N} - \hat{\sigma}_0^{\otimes N}\|_1 \leq  N\frac{1+\sqrt{3}}{8}\max(\|V_r\|_{\infty},\|V_0\|_{\infty})\|V_{r} - V_0\|_1,
\end{equation}
where $V_r$ and $V_0$ are the covariance matrices of the states $\hat{\sigma}_r$ and $\hat{\sigma}_0$. Since the squeezed state only has squeezing on one mode, this effectively means that we are distinguishing $N$ copies of squeezed vacuum with squeezing $r$. Hence we have
\begin{equation}
\|V_r\|_{\infty} = e^{2r},\quad \|V_0\|_{\infty} = 1,\quad \|V_{r} - V_0\|_1 = 2\sinh(2r),
\end{equation}
which gives us
\begin{equation}
    \|\hat{\sigma}_{r}^{\otimes N} - \hat{\sigma}_0^{\otimes N}\|_1 \leq \frac{1+\sqrt{3}}{8}N\left(\frac{1}{(1 - 2\log(1 + 2\epsilon e^{1/2}/0.546))^2} - 1\right)< 10N\epsilon,
\end{equation}
where the last bound is only useful for $10N\epsilon\leq 2$ since we always have the bound $\|\cdot\|_1\leq 2$ hold. This implies that if we wish to have a success probability above $1/2+1/6$, we must have $N  \geq 6\times (1/10\epsilon)$.

\noindent\textbf{Upper bound on probability of input state having distinguishable output states}:

Now we note that the operator $\hat{\Pi}_0$ is always idempotent. For a state $\hat{\rho}_{A_1\dots A_N}$ defined over subsystems $A_1,\dots,A_N$, each with Hilbert space $\mathcal{H}$, and for any two-outcome POVM $\{\hat{\Pi}_0,\hat{\Pi}_1\}$, we have the upper bound
\begin{equation}
    \text{Tr}[(\mathbb{I} - (\hat{\Pi}_0)^{\otimes N})\hat{\rho}_{A_1\dots A_N}] \leq \min\left(1,\sum_{i=1}^N \text{Tr}(\hat{\Pi}_1\hat{\rho}_{A_i})\right),
\end{equation}
which follows from the fact that the operator is positive semidefinite
\begin{equation}
\sum_{i=1}^{N}\left[(\hat{\Pi}_1)_{A_i}\otimes\mathbb{I}_{A_{j\neq i}}\right] - (\mathbb{I} - \hat{\Pi}^{\otimes N}_{0}) \geq 0,
\end{equation}
for any valid two-outcome POVM $\{\hat{\Pi}_0,\hat{\Pi}_1\}$. To find the input state which maximizes $\sum_{i=1}^N \text{Tr}(\hat{\Pi}_1\hat{\rho}_{A_i})$, we consider the problem defined by
\begin{equation}
    \hat{\rho}_{\mathrm{opt}}= \argmax_{\hat{\rho}\in D(\mathcal{H})}\text{Tr}(\hat{\Pi}_{\mathrm{avg}}\hat{\rho}),
\end{equation}
\begin{equation}
    \hat{\Pi}_{\mathrm{avg}} = \int_{[-\pi,\pi
    )^m} \frac{d^m\phi}{(2\pi)^m}  \hat{\Pi}^{\phi}_{1}.
\end{equation}
Since $\hat{\Pi}_{\mathrm{avg}}$ is positive and the constraint $\text{Tr}(\hat{\rho})=1$ always holds, it is sufficient to assume that the optimal state is pure (rank one). Defining Fock states $\ket{\mathbf{l}} = \ket{l_1,\dots, l_m}$, we can obtain a matrix representation of $\hat{\Pi}_{\mathrm{avg}}$ as
\begin{equation}
\begin{aligned}
\langle\mathbf{l}'|\hat{\Pi}_{\mathrm{avg}}|\mathbf{l}\rangle &=  \int \frac{d^m\phi}{(2\pi)^m} e^{i\phi\cdot(\mathbf{l}-\mathbf{l}')} \int_{\|\mathbf{x}\|_{2}\leq \delta} d^m\mathbf{x} \langle\mathbf{l}'|\mathbf{x}\rangle\langle \mathbf{x}|\mathbf{l}\rangle \\
&= \delta_{\mathbf{l},\mathbf{l}'} \int_{\|\mathbf{x}\|_{2}\leq \delta} d^m\mathbf{x} \langle\mathbf{l}|\mathbf{x}\rangle\langle \mathbf{x}|\mathbf{l}\rangle.
\end{aligned}
\end{equation}
Note that by Cram\'er's inequality the wavefunction satisfies $|\langle l|x\rangle|\leq \pi^{-1/4}$, hence $|\langle l|x\rangle|^2\leq \pi^{-1/2}$, for all $l\in \mathbb{Z}_{\geq 0}$. By substituting this, we have 
\begin{equation}
\langle\mathbf{l}|\hat{\Pi}_{\mathrm{avg}}|\mathbf{l}\rangle \leq  \frac{1}{\pi^{m/2}}\int_{\|\mathbf{x}\|_{2}\leq \delta} d^m \mathbf{x} = \frac{\delta^m}{\Gamma(\frac{m}{2}+1)}.
\end{equation}

\noindent\textbf{Final lower bound on sample complexity}:

Putting all of this together, we have the success probability satisfy
\begin{equation}
    p_{\mathrm{success}}\leq \frac{1}{2} + \frac{1}{2}\min\left(2,10N\epsilon\right)\min\left(1,N\frac{\delta^m}{\Gamma(\frac{m}{2}+1)}\right),
\end{equation}
which means that for a success probability of $2/3$, we require $N = \Omega(\Gamma(\frac{m}{2}+1)\delta^{-m})$ which we can expand more carefully by noting that for $m>0$
\begin{equation}
    \Gamma\left(\frac{m}{2}+1\right) > \sqrt{\frac{\pi}{e}m\left(\frac{m}{2e}\right)^m},
\end{equation}
which then gives us the final lower bound of $N$ in terms of $\lambda$ as
\begin{equation}
    N = \Omega\left(\max\left(\frac{1}{10\epsilon},\sqrt\frac{\pi m}{e}\left(\frac{0.99}{e\tanh
    (\lambda/2)}\right)^{m/2}\right)\right),
\end{equation}
which for $\lambda<0.763$ grows exponentially in $m$.
\end{proof}

We note a few remarks on the above result. The first obvious implication is that this shows a clear no-go for learning the transfer function for $\lambda\to0^+$ as a consequence of unbounded sample complexity. Further, we note that this hardness is shown for a very restricted kind of channel which only outputs a mixture of vacuum and squeezed vacuum. Notably, the key aspect of hardness is in learning an unknown binary POVM. We anticipate this bound can be made much tighter with various extensions to larger families of channels and leave this task for future exploration.

\subsection{Channel tomography using Choi-state tomography}\label{app:bosonicACID}
\begin{definition}[Finite-energy Choi-state]
    For $\mathcal{E}\in \CPTP(\mathcal{H}_{\infty,m},\mathcal{H}_{\infty,m})$, we define the finite-energy Choi-state as
    \begin{equation}
        J_r(\mathcal{E}) = (\mathcal{E}\otimes \mathbb{I})(|\psi_r\rangle\langle\psi_r|),
    \end{equation}
    where the state $\ket{\psi_r} =\ket{\Phi_{r}^{\mathrm{TMSV}}}^{\otimes m}$ is a $2m$-mode state formed by $m$ two-mode squeezed vacuum states with squeezing parameter $r$.
\end{definition}
\begin{theorem}
    Consider the learning task for obtaining a classical description of an unknown channel $\mathcal{E}\in \CPTP(\mathcal{H}_{\infty,m},\mathcal{H}_{\infty,m})$ ($m\geq8$) using a tomographic description of $J_r(\mathcal{E})$ labeled as $\hat{\sigma}_{\mathrm{Choi}}$ with a trace distance accuracy promise. If the task wishes to succeed with probability $\geq 2/3$ in producing a classical description of the channel $\tilde{\mathcal{E}}$ with energy-constrained diamond norm accuracy of $\|\tilde{\mathcal{E}}-\mathcal{E}\|_{\diamond,E} \leq \epsilon$ where $E$ is the constraint placed on average total photon number, this would necessarily require the tomographic description of $J_r(\mathcal{E})$ to satisfy
    $$\|\hat{\sigma}_{\mathrm{Choi}} - J_r(\mathcal{E})\|_1   < \epsilon\sqrt\frac{1.85e}{\pi m}\left(\frac{em}{3.96 E}\right)^{m/2}.$$

\end{theorem}
\begin{proof}
We define $\hat{n}$ to be the total photon number over all $m$ modes. The operator $\pmb{\hat{n}}$ is the vector of operators with each entry being the photon number in the $i$th mode for $1\leq i\leq m$. We now define
\begin{equation}
\begin{aligned}
    \hat{\Pi}^{\phi}_1 = \int_{\|\mathbf{x}\|_{2}\leq \delta} d^{m}\mathbf{x} e^{-i\phi\cdot \hat{\pmb{n}}}|\mathbf{x}\rangle\langle \mathbf{x}|e^{i\phi\cdot\hat{\pmb{n}}}, \quad \hat{\Pi}^{\phi}_0 = \mathbb{I} - \hat{\Pi}^{\phi}_1, \quad\delta = \sqrt{\frac{me^{-2r}}{2\times0.99}}.
\end{aligned}
\end{equation}
Using which we define the following two channels
\begin{equation}
    \mathcal{E}_0(\cdot) =\text{Tr}(\cdot)\hat{\sigma}_0,\quad\mathcal{E}_{1,\phi}(\cdot) = \text{Tr}(\hat{\Pi}^{\phi}_0 (\cdot))\hat{\sigma}_0 +\text{Tr}(\hat{\Pi}^{\phi}_1 (\cdot))\hat{\sigma}_{1},
\end{equation}
while choosing the states $\hat{\sigma}_1$ and $\hat{\sigma}_0$ to satisfy $\|\hat{\sigma}_1 - \hat{\sigma}_0\|_1 = 3.7\epsilon$. For simplicity we take these states to be Gaussian. Suppose we prepare $m$ modes of the single-mode squeezed state with squeezing parameter $r$. We can then check that for this input state of rotated squeezed vacuum $e^{-i\phi\cdot\hat{\pmb{n}}}|r\rangle\langle r|^{\otimes m}e^{i\phi\cdot\hat{\pmb{n}}}$, we have
\begin{equation}
    \|\mathcal{E}_0(e^{-i\phi\cdot\hat{\pmb{n}}}|r\rangle\langle r|^{\otimes m}e^{i\phi\cdot\hat{\pmb{n}}}) - \mathcal{E}_1(e^{-i\phi\cdot\hat{\pmb{n}}}|r\rangle\langle r|^{\otimes m}e^{i\phi\cdot\hat{\pmb{n}}})\|_1  = 3.7\epsilon(1 - Q(m/2,\delta^2 e^{2r})) > 2\epsilon,
\end{equation}
where we have used the fact that $m\geq 8$ and applied Lemma~\ref{lem:gaussprob}. This means that we have the promise that $\|\mathcal{E}_0 - \mathcal{E}_{1}\|_{\diamond,E} > 2\epsilon$ assuming that $E\geq m \sinh^2(r)$. We will consider the hardest setting where $E = m\sinh^2(r)$ which will then give $\delta$ as a function of $E$. If the algorithm is able to successfully learn channels up to the energy-constrained diamond norm of $\epsilon$, this would mean we can learn a classical description of both of these channels such that 
    \begin{equation}
        \|\tilde{\mathcal{E}}_0 - \mathcal{E}_0\|_{\diamond,E}\leq\epsilon, \quad \|\tilde{\mathcal{E}}_1 -\mathcal{E}_1\|_{\diamond,E} \leq \epsilon.
    \end{equation}
    This means that there exists a state $\hat{\rho}\in D(\mathcal{H}^{\otimes 2})$ such that
    \begin{equation}
    \begin{aligned}
        \|(\tilde{\mathcal{E}}_0\otimes\mathbb{I})(\hat{\rho})-(\tilde{\mathcal{E}}_1\otimes\mathbb{I})(\hat{\rho})\|_1 &\geq \left|\|(\mathcal{E}_0\otimes\mathbb{I})(\hat{\rho})-(\mathcal{E}_{1,\phi}\otimes\mathbb{I})(\hat{\rho})\|_1 - \left(\|\tilde{\mathcal{E}}_1 -\mathcal{E}_{1,\phi}\|_{\diamond,E} +\|\tilde{\mathcal{E}}_0 -\mathcal{E}_0\|_{\diamond,E}\right)\right|  >0.
    \end{aligned}
    \end{equation}
    Hence the classical descriptions $\tilde{\mathcal{E}}_0$ and $\tilde{\mathcal{E}}_1$ will always be distinguishable under the given promises. 
    
    Consider the scenario of two parties $A$ and $B$. $B$ has access to an algorithm that produces a classical tomographic description of an unknown channel $\mathcal{E}$ using a classical description of $J_r(\mathcal{E})$. Assume that $B$ has complete knowledge of the states $\hat{\sigma}_0$ and $\hat{\sigma}_1$. $A$ sends a classical description of one of the two states $\hat{\rho}_0$ or $\hat{\rho}_1$ that both are promised to satisfy
    \begin{equation}
        \|J_r(\mathcal{E}_0) - \hat{\rho}_0\|_1\leq \epsilon',\quad \|J_r(\mathcal{E}_{1,\phi}) - \hat{\rho}_1\|_1\leq \epsilon',
    \end{equation}
    where $\epsilon'$ is the Choi-state accuracy required for the tomography algorithm to return a channel description accurate to $\epsilon$ in energy-constrained diamond norm.
    
    What this means is that if $B$ employs the described algorithm for channel learning, they would be able to successfully distinguish the case of whether they received an approximated description of $J_r(\mathcal{E}_0)$ or $J_r(\mathcal{E}_{1})$ with success probability of at least $2/3$. Due to the entanglement-breaking structure of the channel, we note that regardless of $\phi$ we have
    \begin{equation}
    \begin{aligned}
        \|J_{r}(\mathcal{E}_0) - J_{r}(\mathcal{E}_{1,\phi})\|_1 &=3.7\epsilon\text{Tr}(\hat{\Pi}_{1}^\phi(1-\nu^2)^m\nu^{2{\hat{n}}})\\
        &= 3.7\epsilon \left((1-\nu^2)^{m}\int_{\|\mathbf x\|_2\leq \delta}d^{m}\mathbf{x}\langle\mathbf{x}|\nu^{2\hat{n}}|\mathbf{x}\rangle\right)\\
        &= 3.7\epsilon\sum_{\mathbf{l}\in \mathbb{Z}_{\geq 0}^{m}} p_{\mathrm{thermal}}(\mathbf{l}) \int_{\|\mathbf x\|_2\leq \delta}d^{m}\mathbf{x}\langle\mathbf{l}|\mathbf{x}\rangle\langle\mathbf{x}|\mathbf{l}\rangle \leq 3.7\epsilon\frac{\delta^{m}}{\Gamma(\frac{m}{2}+1)},
    \end{aligned}
    \end{equation}
    where $p_{\mathrm{thermal}}(\mathbf{l})$ is the thermal distribution obtained by measuring photon number on the thermal state $(1-\nu^2)^{m}\nu^{2\hat n}$ following which we use the fact $|\langle l|x\rangle|^2\leq \frac{1}{\sqrt{\pi}}$ by Cram\'er's inequality.
    
    If $B$ is able to succeed in distinguishing the two channels, this implies that the classical descriptions that are possible to receive must also be distinguishable. Consider $A$ to be an adversary who chooses $\hat{\rho}_0$ and $\hat{\rho}_1$ to be as close together as the promise allows. Any $\epsilon'$ for which the scheme succeeds must also be able to distinguish the two descriptions given by this adversarial choice. By triangle inequality, for every $\phi$ there is a choice of $\hat{\rho}_0$ and $\hat{\rho}_1$ with 
    \begin{equation}
        \|\hat{\rho}_0 - \hat{\rho}_1\|_1 = 
        \max\left(0,\ \|J_r(\mathcal{E}_{1,\phi})-J_r(\mathcal{E}_0)\|_1 - \|\hat{\rho}_1 - J_r(\mathcal{E}_{1,\phi})\|_1 - \|\hat{\rho}_0 - J_r(\mathcal{E}_0)\|_1\right). 
    \end{equation}
    The choice of $\hat{\rho}_0$ and $\hat{\rho}_1$ which satisfies the above equality is
    \begin{equation}
        \hat{\rho}_0 = pJ_r(\mathcal{E}_0) +(1-p)J_r(\mathcal{E}_{1,\phi}),\quad\hat{\rho}_1 = pJ_r(\mathcal{E}_{1,\phi}) + (1-p)J_r(\mathcal{E}_0),
    \end{equation}
    \begin{equation}
        p = \min\left(\frac{1}{2},1 - \frac{\epsilon'}{\|J_r(\mathcal{E}_{1,\phi}) - J_r(\mathcal{E}_0)\|_1}\right).
    \end{equation}
    Hence $\|\hat{\rho}_0 - \hat{\rho}_1\|_1$ would be zero unless 
    \begin{equation}
        2\epsilon' \leq 3.7\epsilon\frac{\delta^{m}}{\Gamma(\frac{m}{2}+1)}\leq \epsilon\sqrt{\frac{3.7e}{\pi m}\left(\frac{e}{0.99 e^{2r}}\right)^m} < \epsilon\sqrt{\frac{3.7e}{\pi m}}\left(\frac{em}{3.96 E}\right)^{m/2}.
    \end{equation}
    Hence unless $\epsilon'$ satisfies the above condition, the classical descriptions $\hat{\rho}_0$ and $\hat{\rho}_1$ are not guaranteed to be distinguishable. Hence, the channels are no longer distinguishable implying that the learning task fails.
\end{proof}
We note that the above bound gives us a non-trivial relation between the accuracy required for the Choi-state and the target channel tomography accuracy (in energy-constrained diamond distance) only when we have the constraint $E > \frac{e}{3.96}m$ which is roughly $0.69$ photons per mode. We believe that this can be made tighter, and leave further extensions to future work on the general bosonic channel tomography problem, where the natural measure would have to be the energy-constrained diamond norm accuracy.
\end{document}